\documentclass[runningheads]{llncs} % to get page numbers

\usepackage{amsmath,amssymb,amsfonts}
\usepackage{graphicx}
\usepackage{textcomp}
\usepackage{xcolor}
\def\BibTeX{{\rm B\kern-.05em{\sc i\kern-.025em b}\kern-.08em
    T\kern-.1667em\lower.7ex\hbox{E}\kern-.125emX}}

\usepackage{hyperref}

\usepackage{mathpartir}
\usepackage{color}
\usepackage{url}
\usepackage{todonotes} 
\usepackage{enumerate}
\usepackage{algpseudocode}
\usepackage{stmaryrd} %\rightarrowtriangle

\usepackage{listings}
\newtheorem{assumption}[theorem]{Assumption}

\def\EmptyRule#1#2{\ensuremath{\inferrule*[Left={}]{#1}{#2}}}
\def\LabelRule#1#2#3{\ensuremath{{\inferrule*[left={#1}]{#2}{#3}}}}
\def\LabelRuleProof#1#2#3{\ensuremath{{\inferrule*[Left={#1}]{#2}{#3}}}}

\def\redtext#1{{\color{black}{#1}}}        %texts are important and need to be revised.
\def\bluetext#1{{\color{black}{#1}}}       
\def\purpletext#1{{\color{black}{#1}}}     
\def\airforcetext#1{{\color{black}{#1}}}   %for tracking new notation
\def\orangetext#1{{\color{black}{#1}}}
\def\tealtext#1{{\color{black}{#1}}}       %for lemma, theorem
\def\newtext#1{{\color{black}{#1}}}
\def\trimming#1{{\color{black}{#1}}}
\def\emcase#1{\bluetext{{\textbf{#1}}}}

\def\noteinline#1{\todo[inline,caption={},backgroundcolor = white]{#1}}

\def\tbupdated{{\color{red}{TBU}}}

\newif\ifshow
\showfalse

\def\sep{\ensuremath{\ |\  }}

\def\appapp{\mathbin{\cdot}} % application of applicative functor
\def\intType{\ensuremath{\textbf{int}}}
\def\prdVal#1{\ensuremath{\langle #1 \rangle}}

\def\bind#1#2#3{\ensuremath{\text{bind}\ #1 = #2\ \text{in}\ e_2}}
\def\prj#1#2{\ensuremath{\pi_{#1}#2}}
\def\prjtwo#1#2{\ensuremath{\pi_{#1}(#2)}}
\newcommand{\Abstract}[2]{\ensuremath{\lambda #1.#2}}

\def\transit{\ensuremath{\rightarrowtriangle}}
\def\transitml{\ensuremath{\rightarrowtriangle}}

\def\termEnv{\ensuremath{\Gamma}}
\def\typeEnv{\ensuremath{\Delta}}
\def\typeEval#1#2{\ensuremath{#1 \vdash #2}}
\def\typeEvalI#1#2{\ensuremath{#1 \vdash_{\textsf{I}} #2}}
\def\typeEvalP#1#2{\ensuremath{#1 \vdash_{\textsf{P}} #2}}
\def\typeEvalK#1#2{\ensuremath{#1 \vdash_{\mathsf{\kappa}} #2}}
\def\wfmod#1#2#3{\ensuremath{#1 \vdash_{\mathsf{#2}} #3}}

\def\smalleq{\ensuremath{\sqsubseteq}}

\def\valrelation#1#2{\ensuremath{[\![#1]\!]_{#2}}}

\def\termrelation#1#2{\ensuremath{[\![#1]\!]_{#2}^{\mathsf{ev}}}}

\def\trg#1{{\color{black}{\ensuremath{#1}}}}
\def\typeVar{\ensuremath{\alpha}}

\def\reduce{\ensuremath{\rightarrowtriangle^*}}

\def\dom#1{\ensuremath{\textit{dom}(#1)}}
\def\range#1{\ensuremath{\textit{range}(#1)}}

\def\tuple#1{\ensuremath{\langle #1 \rangle}}
\def\tupletwo#1{\ensuremath{( #1 )}}
\def\inttype{\ensuremath{\textbf{int}}}

\def\respect{\ensuremath{\models}}
\def\respectfull{\ensuremath{\models}^{\mathsf{full}}}

\def\indval#1{\ensuremath{\mathcal{I}_V[\![#1]\!]}}
\def\indterm#1{\ensuremath{\mathcal{I}_E[\![#1]\!]}}

\def\policy{\ensuremath{\mathcal{P}}}

\def\encpub#1{\ensuremath{\langle\!\langle #1 \rangle \!\rangle_P}}
\def\enccon#1{\ensuremath{\langle\!\langle #1 \rangle \!\rangle_C}}
\def\conView#1{\ensuremath{\Gamma^{#1}_C}}
\def\pubViewType#1{\ensuremath{\Delta^{#1}_P}}
\def\pubViewTerm#1{\ensuremath{\Gamma^{#1}_P}}

\def\confInp{\ensuremath{\mathbf{V}}}
\def\decFunc{\ensuremath{{\mathbf{F}}}}
\def\TRNI#1{\ensuremath{\text{TRNI}(#1)}}

\def\set#1{\ensuremath{\{#1\}}}

\def\varPolicy#1{\ensuremath{\confInp_{#1}}}
\def\decPolicy#1{\ensuremath{\decFunc_{#1}}}
\def\deltaPol#1{\ensuremath{\delta_{#1}}}

\def\rhoPol#1{\ensuremath{\rho_{#1}}}

\def\logeq{\ensuremath{\sim}}

\def\relDel#1{\ensuremath{\textit{Rel}(#1)}}

\def\unitkind{\ensuremath{\mathsf{1}}}
\def\unitcon{\ensuremath{\star}}
\def\unittype{\ensuremath{\textbf{unit}}}
\def\unitval{\ensuremath{\star}}
\def\unitmod{\ensuremath{\star}}
\def\unitsign{\ensuremath{\textsf{1}}}
\def\basekind{\ensuremath{\mathsf{T}}}
\def\singleton#1{\ensuremath{{\mathsf{S}}(#1)}}
\def\pack#1#2{\ensuremath{\textsf{pack}[#1]\ \textsf{as}\ #2 }}
\def\unpack#1#2#3{\ensuremath{\textsf{unpack}[#1] = #2\ \textsf{in}\ #3 }}
\def\fix#1#2{\ensuremath{\textsf{fix}_{#1}#2}}
\def\letexp#1#2{\ensuremath{\textsf{let}\ #1\ \textsf{in}\ #2}}
\def\extract#1{\ensuremath{\textsf{Ext}\ #1}}
\def\extracttwo#1{\ensuremath{\textsf{Ext(#1)}}}
\def\seal{\ensuremath{:>}}
\def\atcmod#1{\ensuremath{{(\!|#1|\!)}}}
\def\attmod#1{\ensuremath{{\langle\!|#1|\!\rangle}}}

\def\atksign#1{\ensuremath{{(\!|#1|\!)}}}
\def\atcsign#1{\ensuremath{{\langle\!|#1|\!\rangle}}}

\def\lambdagn{\ensuremath{\lambda^{\text{gn}}}}
\def\lambdaap{\ensuremath{\lambda^{\text{ap}}}}
\def\Pign{\ensuremath{\Pi^{\text{gn}}}}
\def\Piap{\ensuremath{\Pi^{\text{ap}}}}
\def\kind{\ensuremath{\textsf{kind}}}
\def\ok{\ensuremath{\textsf{ok}}}
\def\sign{\ensuremath{\textsf{sig}}}

\def\fstsign#1{\ensuremath{{\textsf{Fst}}(#1)}}
\def\fstop#1#2{\ensuremath{{\textsf{Fst}}(#1) \gg  #2}}
\def\fstoptwo#1{\ensuremath{{\textsf{Fst}}(#1)}}

\def\inter#1#2{\ensuremath{[\![#1]\!]_{#2}}}
\def\interset#1#2{\ensuremath{[\![#1]\!]_{#2}^{\mathsf{set}}}}

\def\rev#1{\ensuremath{#1^{\mathsf{ev}}}}
\def\rst#1{\ensuremath{#1^{\mathsf{st}}}}
\def\rpev#1{\ensuremath{#1^{\mathsf{pev}}}}
\def\riev#1{\ensuremath{#1^{\mathsf{iev}}}}
\def\rarb#1#2{\ensuremath{#2^{\mathsf{#1}}}}

\def\precan#1{\ensuremath{\mathsf{PreCand}_{{#1}}}}
\def\simple#1{\ensuremath{\textit{simp}(#1)}}

\def\terminating#1{\ensuremath{#1\downarrow}}

\def\valSet{\ensuremath{\mathsf{Val}}}
\def\conSet{\ensuremath{\mathsf{Con}}}
\def\power#1{\ensuremath{\mathcal{P}(#1)}}

\def\interenv#1{\ensuremath{[\![#1]\!]}}
\def\interenvfull#1{\ensuremath{[\![#1]\!]^{\mathsf{full}}}}
\def\envfull#1{\ensuremath{[\![#1]\!]^{\mathsf{full}}}}

\def\code#1{\text{{\ttfamily{#1}}}}
\def\wrappub#1{\ensuremath{\textit{wrap}_{\policy}(#1)}}

\def\sigmaPol#1{\ensuremath{\sigma_{#1}}}
\def\sigmaPolCon#1{\ensuremath{\sigma_{#1}^{C}}}

\def\concat{\ensuremath{::}}

\def\emptylist{\ensuremath{[]}}

\definecolor{burntumber}{rgb}{0.54, 0.2, 0.14}
\definecolor{airforceblue}{rgb}{0.36, 0.54, 0.66}
\def\predenv#1#2{\ensuremath{{|#1|_{#2}}}}
\def\condenv#1#2#3{\ensuremath{{#1 \in \predenv{#2}{#3}}}}
\def\condenvalt#1#2{\ensuremath{{#1 \respectfull #2}}}

\def\WfeE{{{wf\_nil}}}
\def\WfeC{{{wf\_cn}}}
\def\WfeT{{{wf\_tm}}}
\def\WfeM{{{wf\_md}}}

\def\WfkB{{{ofk\_type}}}
\def\WfkS{{{ofk\_sing}}}
\def\WfkU{{{ofk\_one}}}
\def\WfkF{{{ofk\_pi}}}
\def\WfkP{{{ofk\_sigma}}}

\def\EqkR{{{eqk\_refl}}}
\def\EqkS{{{eqk\_symm}}}
\def\EqkT{{{eqk\_trans}}}
\def\EqkSg{{{eqk\_sing}}}
\def\EqkF{{{eqk\_pi}}}
\def\EqkP{{{eqk\_sigma}}}

\def\SubkE{{{subk\_refl}}}
\def\SubkT{{{subk\_trans}}}
\def\SubkSg{{{subk\_sing\_t}}}
\def\SubkDFn{{{subk\_pi}}}
\def\SubkDPr{{{subk\_sigma}}}

\def\WfcV{{{ofc\_var}}}
\def\WfcDFun{{{ofc\_lam}}}
\def\WfcDApp{{{ofc\_app}}}
\def\WfcPr{{{ofc\_pair}}}
\def\WfcPrj{{{ofc\_pi}}}
\def\WfcUc{{{ofc\_star}}}
\def\WfcUt{{{ofc\_unit}}}
\def\WfcInt{{{ofc\_int}}}
\def\WfcFn{{{ofc\_arrow}}}
\def\WfcPrd{{{ofc\_prod}}}
\def\WfcAbs{{{ofc\_all}}}
\def\WfcEx{{{ofc\_exists}}}
\def\WfcSg{{{ofc\_sing}}}
\def\WfcROne{{{ofc\_extpi}}}
\def\WfcRTwo{{{ofc\_extsigma}}}
\def\WfcSub{{{ofc\_subsume}}}

\def\EqcR{{{eqc\_refl}}}
\def\EqcS{{{eqc\_symm}}}
\def\EqcT{{{eqc\_trans}}}
\def\EqcDFn{{{eqc\_lam}}}
\def\EqcDApp{{{eqc\_app}}}
\def\EqcDPr{{{eqc\_pair}}}
\def\EqcPrj{{{eqc\_pi}}}
\def\EqcFn{{{eqc\_arrow}}}
\def\EqcPrd{{{eqc\_prod}}}
\def\EqcAbs{{{eqc\_all}}}
\def\EqcEx{{{eqc\_exists}}}
\def\EqcSg{{{eqc\_sing}}}
\def\EqcSgTwo{{{eqc\_singelim}}}
\def\EqcROne{{{eqc\_extpi}}}
\def\EqcRTwo{{{eqc\_extpiw}}}
\def\EqcRThree{{{eqc\_extsigma}}}
\def\EqcRFour{{{eqc\_extone}}}
\def\EqcSub{{{eqc\_subsume}}}
\def\EqcBeta{{{eqc\_beta}}}
\def\EqcRFive{{{eqc\_beta1}}}
\def\EqcRSix{{{eqc\_beta2}}}

\def\WfsOne{{{ofs\_one}}}
\def\WfsTwo{{{ofs\_stat}}}
\def\WfsThree{{{ofs\_dyn}}}
\def\WfsFour{{{ofs\_piapp}}}
\def\WfsFive{{{ofs\_pigen}}}
\def\WfsSix{{{ofs\_sigma}}}

\def\SeqR{{{eqs\_refl}}}
\def\SeqS{{{eqs\_symm}}}
\def\SeqT{{{eqs\_trans}}}
\def\SeqAtk{{{eqs\_stat}}}
\def\SeqAtc{{{eqs\_dyn}}}
\def\SeqDFnGn{{{eqs\_pigen}}}
\def\SeqDFnAp{{{eqs\_piapp}}}
\def\SeqDPr{{{eqs\_sigma}}}

\def\SubsEq{{{subs\_refl}}}
\def\SubsT{{{subs\_trans}}}
\def\SubsK{{{subs\_stat}}}
\def\SubsDFnGn{{{subs\_pigen}}}
\def\SubsDFnAp{{{subs\_piapp}}}
\def\SubsDPr{{{subs\_sigma}}}

\def\WttV{{{oft\_var}}}
\def\WttUnit{{{oft\_star}}}
\def\WttInt{{{oft\_int}}}
\def\WttAbs{{{oft\_lam}}}
\def\WttApp{{{oft\_app}}}
\def\WttPr{{{oft\_pair}}}
\def\WttPrj{{{oft\_pi}}}
\def\WttUnv{{{oft\_plam}}}
\def\WttInst{{{oft\_papp}}}
\def\WttPk{{{oft\_pack}}}
\def\WttUp{{{oft\_unpack}}}
\def\WttRc{{{oft\_fix}}}
\def\WttLtOne{{{oft\_lett}}}
\def\WttLtTwo{{{oft\_letm}}}
\def\WttExt{{{oft\_ext}}}
\def\WttEq{{{oft\_equiv}}}

\def\WfmV{{{ofm\_var}}}
\def\WfmU{{{ofm\_star}}}
\def\WfmAtc{{{ofm\_stat}}}
\def\WfmAtt{{{ofm\_dyn}}}
\def\WfmAbsI{{{ofm\_lamgn}}}
\def\WfmAppI{{{ofm\_appgn}}}
\def\WfmAbsP{{{ofm\_lamap}}}
\def\WfmAppK{{{ofm\_appap}}}
\def\WfmPr{{{ofm\_pair}}}
\def\WfmPrj{{{ofm\_pi}}}
\def\WfmUnp{{{ofm\_unpack}}}
\def\WfmLetOne{{{ofm\_lett}}}
\def\WfmLetTwo{{{ofm\_letm}}}
\def\WfmSeal{{{ofm\_seal}}}
\def\WfmROne{{{ofm\_extstat}}}
\def\WfmRTwo{{{ofm\_extpi}}}
\def\WfmRThree{{{ofm\_extsigma}}}
\def\WfmRFour{{{ofm\_forget}}}
\def\WfmSub{{{ofm\_subsume}}}

\def\modop{\ensuremath{\mathbin{\mathit{mod}}}}
\def\oprt{\ensuremath{\oplus}}

\def\policyoe{\ensuremath{{\policy_{\textit{OE}}}}}

\def\policyAut{\ensuremath{{\policy_{\text{Aut}}}}}

\def\comma{\ensuremath{\!:\!}}

\def\dummy{\ensuremath{V_{\policy}}}

\newcommand{\alphaPol}[1][\policy]{\ensuremath{\alpha_{#1}}}
\newcommand{\mPol}[1][\policy]{\ensuremath{m_{#1}}}

\newcommand{\tuplemt}[1]{{\ensuremath{\langle #1 \rangle}}}
\newcommand{\lattice}{\ensuremath{\mathcal{L}}}
\newcommand{\latticethree}{\ensuremath{\mathcal{L}_{3}}}
\newcommand{\latticedm}{\ensuremath{\mathcal{L}_{\diamond}}}

\newcommand{\inpSet}{\ensuremath{\mathbf{I}}}
\newcommand{\lvln}{\ensuremath{\textit{lvl}}}
\newcommand{\lvl}[1]{\ensuremath{\textit{lvl}(#1)}}

\newcommand{\plusOne}{\ensuremath{\textit{plus\_one}}}
\newcommand{\wrapPlusOne}{\ensuremath{\textit{wrap\_plus\_one}}}
\newcommand{\add}{\ensuremath{\textit{add}}}
\newcommand{\wrapAddC}{\ensuremath{\textit{wrap\_addc}}}
\newcommand{\expl}[1]{\quad\quad \text{(#1)}}
\newcommand{\lambdaeq}{\ensuremath{\cong}}
\newcommand{\wraptype}[2][l]{\ensuremath{\textsf{W}_{#1}\ #2}}

\newcommand{\unwrap}{\textit{unwrap}}
\newcommand{\prjcode}[2]{\ensuremath{\texttt{#1}_{#2}}}

\newcommand{\comp}[1]{{\ensuremath{{\mathsf{cp}}_{#1}}}}
\newcommand{\convup}[2]{{\ensuremath{{\mathsf{cvu}}_{#1}^{#2}}}}
\newcommand{\conv}[1]{{\ensuremath{{\mathsf{cv}}_{#1}}}}
\newcommand{\wrap}[1]{{\ensuremath{{\mathsf{wr}}_{#1}}}}
\newcommand{\smallst}{\ensuremath{\sqsubset}}

\newcommand{\contextni}{\ensuremath{\Gamma_{\textit{NI}}}}
\newcommand{\tcontextni}{\ensuremath{\Delta_{\textit{NI}}}}
\newcommand{\hinp}{\textit{hi}}
\newcommand{\minp}{\textit{mi}}
\newcommand{\minpl}{\textit{mi1}}
\newcommand{\minpr}{\textit{mi2}}
\newcommand{\linp}{\textit{li}}
\newcommand{\contextp}[1][\policy]{\ensuremath{\Gamma_{#1}}}
\newcommand{\tcontextp}[1][\policy]{\ensuremath{\Delta_{#1}}}

\newcommand{\unitType}{\ensuremath{\textbf{unit}}}
\newcommand{\unitVal}{\ensuremath{\tuple{}}}
\newcommand{\InferRule}[3][]{\ensuremath{{\inferrule*[left={#1}]{#2}{#3}}}}
\newcommand{\Rel}[1]{\ensuremath{\textit{Rel}(#1)}}
\newcommand{\restrict}[1]{\ensuremath{#1|_{\textit{tv}}}}
\newcommand{\logrel}{\ensuremath{\sim}}

\newcommand{\obs}{\ensuremath{\zeta}}

\newcommand{\full}[1]{\ensuremath{\mathsf{Full}_{#1}}}
\newcommand{\respectni}[1][\obs]{\ensuremath{\models_{#1} \text{NI}}}
\newcommand{\respectfullni}[1][\obs]{\ensuremath{\models_{#1}^{\text{full}} \text{NI}}}

\newcommand{\compcr}{{\ensuremath{\mathsf{comp}}}}
\newcommand{\convcr}{{\ensuremath{\mathsf{conv}}}}
\newcommand{\convupcr}{{\ensuremath{\mathsf{convup}}}}
\newcommand{\wrapcr}{{\ensuremath{{\mathsf{wrap}}}}}

\newcommand{\indvalp}[3][\obs]{\ensuremath{\textit{I}_{#3}^{#1}[\![#2]\!]}}
\newcommand{\indtermp}[3][\obs]{\ensuremath{\textit{I}_{#3}^{#1}[\![#2]\!]^{\textit{ev}}}}
\newcommand{\respectp}[1][\obs]{\ensuremath{\models_{#1}}}
\newcommand{\respectfullp}[1][\obs]{\ensuremath{\models_{#1}^{\text{full}}}}

\newcommand{\dec}[1]{\ensuremath{\textsf{dec}_{#1}}}

\newcommand{\compdec}[1]{\ensuremath{\textsf{cpdec}_{#1}}}

\newcommand{\decSet}{\ensuremath{\mathbf{Dec}}}

\newcommand{\decMap}{\ensuremath{\mathbf{F}}}

\newcommand{\policymul}{{\ensuremath{\mathcal{P}_{\textit{M1}}}}}
\newcommand{\policyglb}{{\ensuremath{\mathcal{P}_{\textit{Av}}}}}

\newcommand{\relDec}[1][f]{\ensuremath{R_{#1}^{\bullet}}}

\newcommand{\btsApply}{\textit{cpNewSt}}
\newcommand{\btsApplyWrap}{\textit{wrapCpNewSt}}
\newcommand{\btsGetMove}{\textit{getMv}}
\newcommand{\btsGetMoveWrap}{\textit{wrapGetMv}}

\newcommand{\btsEndorse}{\textit{endorse}}
\newcommand{\btsDeclassify}{\textit{declassify}}
\newcommand{\btsNotDone}{\textit{notDone}}
\newcommand{\btsNewState}{\textit{newSt}}
\newcommand{\btsNewMove}{\textit{newMv}}

\newcommand{\indvalpnew}[3][\obs]{\ensuremath{\textit{I}_{#3}^{#1}[\![#2]\!]}}
\newcommand{\indtermpnew}[3][\obs]{\ensuremath{\textit{I}_{#3}^{#1}[\![#2]\!]^{\textit{ev}}}}

\title{Type-based Declassification for Free {\normalsize (with appendix)}}

\author{Minh Ngo\inst{1,2}\and
David A. Naumann\inst{1} \and
Tamara Rezk\inst{2}}

\institute{Stevens Institute of Technology \and Inria}

\begin{document}

\maketitle

\begin{abstract}
This work provides a study to demonstrate the potential of
using off-the-shelf programming languages 
and their theories 
to build sound language-based-security tools.
Our study focuses on information flow security  
encompassing declassification policies that allow us to express 
flexible security policies needed for practical requirements. 
We translate security policies, with declassification, 
into an interface for which an unmodified standard typechecker 
can be applied to a source program---if the program typechecks, it provably satisfies the policy.
Our proof reduces security soundness---with declassification---to the mathematical foundation of data abstraction, Reynolds' abstraction theorem.
\end{abstract}

\begin{center}
\color{blue} To appear in ICFEM 2020
\end{center}

\section{Introduction}
\label{sec:intro}
A longstanding challenge for software systems  is the enforcement of security in applications implemented in conventional general-purpose programming languages.
For high assurance, precise mathematical definitions are needed for policies, enforcement mechanism, and program semantics.  The latter, {in particular}, is a major challenge for languages in practical use. {In order to} minimize the cost of assurance, especially over time as systems evolve, it is desirable to leverage work on formal modeling with other goals such as functional verification, equivalence checking, and compilation.  
%This paper reports on progress and challenges in leveraging type abstraction as basis for IF policy, %enforcement, and assurance.

To be auditable by stakeholders, policy should be expressed in an accessible way.
This is one of several reasons why types play an important role in many works on information flow (IF) security.
For example, Flowcaml~\cite{PottierSimonet:flowcaml} and Jif~\cite{Myers:jif} express policy using types that include IF labels.  
They statically enforce policy using dedicated IF type checking and inference.
%---mature techniques widely available in some form in many practical languages.  
%Techniques from type theory are used in the security proofs for Flowcaml and for the influential calculus %DCC.  
Techniques  from type theory are also used in security proofs such as those for Flowcaml and the calculus DCC~\cite{Abadi-etal-99-POPL}.

IF is typically formalized as the preservation of indistinguishability relations between executions.
Researchers have hypothesized that this should be 
an instance of a celebrated semantics basis in type theory: relational parametricity~\cite{Reynolds-83}.
Relational parametricity 
provides an effective basis for formal reasoning about program transformations 
(``theorems for free''~\cite{Wadler-89-FPCA}),
representation independence and information hiding for program verification~\cite{Mitchell96,BanerjeeNaumann02c}.
The connection between IF and relational parametricity has been made precise in 2015, for DCC, by translation to the calculus $F_\omega$ and use of the existing parametricity 
theorem for $F_\omega$~\cite{Bowman-Ahmed-15-ICFP}.
The connection is also made, perhaps more transparently, in a translation of DCC to dependent type theory, 
specifically the calculus of constructions and its parametricity theorem~\cite{AlgehedBernardy19}. 

In this work, we advance the state of the art in the connection between IF and relational parametricity,
guided by three main goals.
One of the goals motivating our work is to \emph{reduce the burden of defining dedicated type checking, inference, and security proofs} for high assurance in programming languages. 
A promising approach towards this goal is the idea of leveraging type abstraction to enforce policy, and in particular, \emph{leveraging the parametricity theorem 
 to obtain security guarantees}.
% and prove security of enforcement, for industrially relevant languages.}
A concomitant goal is \emph{to do so for practical IF policies}
that encompass selective declassification, which is needed for most policies in practice. 
For example,  a password checker program or a program that  calculates aggregate or statistical information  must be considered  insecure
 without declassification. 

%One of the goals motivating our work is to \emph{leverage type abstraction to enforce policy and prove %correctness of enforcement, for industrially relevant languages.}
%A second goal is \emph{to do so for practical policies}
%that encompass selective downgrading, which is needed for the vast majority of policies of practical %interest.
%A related practical need is for output channels at multiple levels of security.  
%Our third goal needs explanation.

To build on the type system and theory of a language without \emph{a priori}\, IF features, policy needs to be encoded somehow, and the program may need to be transformed.
For example, to prove that a typechecked DCC term is secure with respect to the policy expressed by its type, 
Bowman and Ahmed~\cite{Bowman-Ahmed-15-ICFP} encode the typechecking judgment by nontrivial 
translation of both types and terms into $F_\omega$.
%both types and terms into types and terms of $F_\omega$.
Any translation becomes part of the assurance argument.
Most likely, complicated translation will also make it more difficult to use extant 
type checking/inference (and other development tools) in diagnosing security errors and developing secure code.
This leads us to highlight a third goal, needed to achieve the first goal, namely to
\emph{minimize the complexity of translation.}

There is a major impediment to leveraging type abstraction: few languages are relationally parametric or have parametricity theorems.   
The lack of parametricity can be addressed by focusing on well behaved subsets and leveraging additional features like ownership types that may be available for other purposes 
(e.g., in the Rust language).  
As for the paucity of parametricity theorems, we take hope in the recent advances in machine-checked metatheory, such as correctness of the CakeML and CompCert compilers,
the VST logic for C, the relational logic of Iris.
For parametricity specifically, the most relevant work is Crary's formal proof of parametricity for the ML module calculus~\cite{Crary-POPL-17}.
%One of our results is a reduction of IF to parametricity in that calculus.  

% DN moved some of the following to rw  
%% Before elaborating on our contributions let us review some prior work.
%% The calculus DCC expresses policy using monad types indexed on levels in a lattice of security levels with the usual interpretation that flows are only allowed between levels in accord with the ordering.  
%% While DCC is a theoretical calculus, its monadic types fit nicely with the monads and monad transformers used by the Haskell language for computational effects like state and I/O.  
%% Algehed and Russo encode the typing judgment of DCC in Haskell 
%% using closed type families, one of the type system extensions supported by GHC~\cite{AlgehedR17}.
%% However, they do not prove security;
%% and DCC expresses strict noninterference, with no form of declassification.

\paragraph*{Contributions.}  
Our \emph{first contribution} is to translate policies with declassification---in the style of relaxed noninterference~\cite{Li-Zdancewic-05-POPL}---into abstract types in a functional language, in such a way that
typechecking the original program implies its security.  For doing so, we neither rely on  a specialized security type system~\cite{Bowman-Ahmed-15-ICFP} nor on modifications of existing type systems~\cite{Cruz-etal-17-ECOOP}.
% DN removed:  of languages with a parametricity result.
A program that typechecks may use the secret inputs parametrically, e.g., storing in data structures, but cannot look at the data until declassification has been applied.
Our \emph{second contribution} is to prove security by direct application of a parametricity theorem.  
We carry out this development for the polymorphic lambda calculus, using the original theorem of Reynolds. 
We also provide an appendix that shows this development  for the ML module calculus using Crary's theorem~\cite{Crary-POPL-17}, enabling the use of ML to check security.

\section{Background: Language and Abstraction Theorem}
\label{sec:abstraction}
%Following Li and Zdancewic~\cite{Li-Zdancewic-05-POPL}, 
%we formulate RNI for a typed, pure functional language.
%We consider the simply typed and call-by-value lambda calculus with integers and type variables.%,
To present our results we choose the simply typed and call-by-value lambda calculus, with integers and type variables, for two reasons: (1) the chosen language is similar to the language used in the paper of Reynolds \cite{Reynolds-83} where the abstraction theorem was first proven, and (2) we want to illustrate our encoding approach (\S\ref{sec:trni}) in a minimal calculus. % that supports our encoding approach.
This section defines the language we use and recalls the abstraction theorem, a.k.a. parametricity.  
%These results are basically standard. 
Our language is very close to the one in Reynolds~\cite[\S~2]{Reynolds-83}; we prove the abstraction theorem using contemporary notation.\footnote{
Some readers may find it helpful to consult the following references  for background on logical relations and parametricity:~\cite[Chapt.~49]{harper2016practical},
\cite[Chapt.~8]{Mitchell96}, 
\cite{Crary-05-chapter}, 
\cite{Pitts-05-chapter}.
}  

%\subsubsection
\paragraph{Language.}

The syntax of the language is as below, where $\typeVar$ denotes a type variable, $x$ a term variable, and $n$ an integer value.
A value is {\em closed} when there is no free term variable in it.
A type is {\em closed} when there is no type variable in it.
\begin{align*}
{\tau}     ::= &\ {\intType}\ |\ {\typeVar}\ |\ {\tau_1 \times \tau_2}\ |\ {\tau_1 \rightarrow \tau_2}\ & &\text{Types} \\
{v}     ::= &\ {n}\ |\  {\prdVal{v,v}}\ |\ {\Abstract{x:\tau}{e}} & & \text{Values}\\
{e}     ::= &\ x\ |\ {v}\ |\ \tuple{e,e}\ |\ {\prj{i}{e}}\ |\ {e_1e_2} & & \text{Terms}\\
{E}           ::= &\ {[.]}\ |\ \tuple{E,e}\ |\ \tuple{v,E}\ |\ \trg{\prj{i}{E}}\ |\ \trg{E\ e}\ |\ v\ E\ & &\text{Eval.\ Contexts} 
\end{align*}
%Our main interest is terms without type variables, which represent source programs.
% DN moved following to next section where it's relevant 
%We consider terms without type variables as source programs (the role of type variables is to encode policies, as explained in due course).
%%  \footnote{In the definition of substitution, we have $($ and $)$ which are not in the syntax of the language since we just want to make the definition clear.}.
%% \begin{align*}
%% n [x\mapsto e] & = n\\
%% x [x\mapsto e] & = e\\
%% y [x\mapsto e] & = y\ \text{if $y \neq x$}\\
%% (\lambda y:\tau.e)[x\mapsto e] & = \lambda y:\tau.(e[x\mapsto e]) \ \text{if $x \neq y$ and $y \not\in FV(e)$}\\
%% (e_1\ e_2)[x\mapsto e] &= e_1[x\mapsto e]\ e_2[x\mapsto e]\\
%% \tuple{e_1,e_2}[x\mapsto e] & = \tuple{e_1[x\mapsto e],e_2[x\mapsto e]}\\
%% \prj{i}{e'}[x\mapsto e] &= \prj{i}{e'[x\mapsto e]}
%% \end{align*}
We use small-step semantics, with the reduction relation $\rightarrowtriangle$ defined inductively by these rules.
\begin{mathpar}
\inferrule{}
{\trg{\prj{i}{\tuple{v_1,v_2}} \rightarrowtriangle v_i }}\and
%%%%%%%%%
\inferrule{}
{\trg{(\Abstract{x:\tau}{e})v \rightarrowtriangle e[x \mapsto v]} } \and
%%%%%%%%%%%
\inferrule{\trg{e \rightarrowtriangle e'}}
{\trg{E[e] \rightarrowtriangle E[e']}}
%%%%%%%%%%%
\end{mathpar}
% DN save space
%% \begin{mathpar}
%% \EmptyRule{~}
%% {\trg{\prj{i}{\tuple{v_1,v_2}} \rightarrowtriangle v_i }}\and
%% %%%%%%%%%
%% \EmptyRule{~}
%% {\trg{(\Abstract{x:\tau}{e})v \rightarrowtriangle e[x \mapsto v]} } \and
%% %%%%%%%%%%%
%% \EmptyRule{\trg{e \rightarrowtriangle e'}}
%% {\trg{E[e] \rightarrowtriangle E[e']}}
%% %%%%%%%%%%%
%% \end{mathpar}

We write $e [x\mapsto e']$ for capture-avoiding substitution of $e'$ for free occurrences of $x$ in $e$.
We use parentheses to disambiguate term structure and write
$\rightarrowtriangle^*$ for the reflexive, transitive closure of 
$\rightarrowtriangle$.

%\paragraph{Typing rules} 

A {\em typing context} \trg{\typeEnv} is a set of type variables.
A {\em \purpletext{term context}} \trg{\termEnv} is a mapping from term variables to types,
written like $x:\intType,y:\intType\rightarrow\intType$.
%\begin{align*}
%\trg{\typeEnv} ::= & \trg{.\sep \typeEnv, \typeVar} & \text{Typing Contexts}\\
%\trg{\termEnv} ::= & \trg{.\sep \termEnv, x:\tau} & \text{Term Contexts}
%\end{align*}
We write \typeEval{\Delta}{\tau} to mean that $\tau$ is {\em well-formed w.r.t. $\Delta$},
that is, all type variables in $\tau$ are in $\Delta$.
% DN removed since it's standard and obvious
%The definition of \typeEval{\Delta}{\tau} is described below. The definition is standard and it amounts to the requirement that type variables in $\tau$ are in $\Delta$.
We  say that $e$ is {\em typable w.r.t. $\Delta$ and $\Gamma$} (denoted by \typeEval{\Delta,\Gamma}{e}) when there exists a well-formed type $\tau$ such that \typeEval{\Delta,\Gamma}{e:\tau}.
%
%% %\begin{figure}
%% %\vspace{-10pt}
%% \begin{mathpar}
%% \small
%% \EmptyRule{~}{\typeEval{\Delta}{\intType}}\and
%% \EmptyRule{~}{\typeEval{\Delta,\alpha}{\alpha}}\and
%% \EmptyRule{\typeEval{\Delta}{\tau_1}\\\typeEval{\Delta}{\tau_2}}
%% {\typeEval{\Delta}{\tau_1\times \tau_2}}\and
%% \EmptyRule{\typeEval{\Delta}{\tau_1}\\\typeEval{\Delta}{\tau_2}}
%% {\typeEval{\Delta}{\tau_1\rightarrow \tau_2}}
%% \end{mathpar}
%% %\vspace{-4ex}
%% %\caption{Well-formed type}
%% %\label{fig:wellformedType}
%% %\end{figure}
%
The derivable typing judgments are defined inductively in Fig.~\ref{fig:typing-rule}.
The rules are to be instantiated only with $\Gamma$ that is well-formed under $\Delta$,
in the sense that \typeEval{\Delta}{\Gamma(x)} for all $x\in\dom{\Gamma}$.
When the \purpletext{term context} and the \purpletext{type context} are empty, we write \typeEval{}{e:\tau}. % instead of \typeEval{.,.}{e:\tau}.

\begin{figure}[!t]
%\vspace{-20pt}
\begin{mathpar}
%\small
%\boxed{\trg{\typeEval{\typeEnv,\termEnv}{e:t} }}\\
\LabelRule{FT-Int}
{~}
{\trg{\typeEval{\typeEnv,\termEnv}{n:\intType}}} \and
%%%%%%%%%%%%%%%%%%
\LabelRule{FT-Var}
{\trg{x:\tau \in \termEnv}}
{\trg{\typeEval{\typeEnv,\termEnv}{x:\tau}}}\and
%%%%%%%%%%%%%%%%%
\LabelRule{FT-Pair}
{\trg{\typeEval{\typeEnv,\termEnv}{e_1:\tau_1}}\\
\trg{\typeEval{\typeEnv,\termEnv}{e_2:\tau_2}}
}
{\trg{\typeEval{\typeEnv,\termEnv}{\prdVal{e_1,e_2}:\tau_1\times \tau_2}}}\and
%%%%%%%%%%%%%%%%%
\LabelRule{FT-Prj}
{\trg{\typeEval{\typeEnv,\termEnv}{e:\tau_1 \times \tau_2}}
}
{\trg{\typeEval{\typeEnv,\termEnv}{\prj{i}{e}:\tau_i}}}\and
%%%%%%%%%%%%%%%%%
\LabelRule{{FT-Fun}}
{\trg{\typeEval{\typeEnv,\termEnv,x:\tau_1}{e:\tau_2}} \and
%\trg{\typeEval{\typeEnv}{\tau_1::\prpkind}}
}
{\trg{\typeEval{\typeEnv,\termEnv}{\Abstract{x:\tau_1}{e}:\tau_1 \rightarrow \tau_2}}}\\
%%%%%%%%%%%%%%%%%
\LabelRule{FT-App}
{\trg{\typeEval{\typeEnv,\termEnv}{e_1:\tau_1 \rightarrow \tau_2}} \and
\trg{\typeEval{\typeEnv,\termEnv}{e_2:\tau_1}} 
}
{\trg{\typeEval{\typeEnv,\termEnv}{e_1\ e_2: \tau_2}}}
\end{mathpar}
%\vspace{-4ex}
\caption{Typing rules}
\label{fig:typing-rule}
%\vspace{-10pt}
\end{figure}

%\begin{restatable}[Canonical forms]{lem}{lemCanonicalForm}
%\label{lem:canonical-form}
%\tealtext{Suppose that \typeEval{}{v:\tau}
%It follows that:}
%\begin{itemize}
%\item \tealtext{if $\tau$ is $\intType$, $v$ is $n$ for some $n$},
%\item \tealtext{if $\tau$ is $\tau_1 \rightarrow \tau_2$, $v$ is $\lambda x:\tau_1.e$ for some $e$},
%\item \tealtext{if $\tau$  is $\tau_1 \times \tau_2$, it is \tuple{v_1,v_2} for some $v_1$ and $v_2$}.
%\end{itemize}
%\end{restatable}

%The static and dynamic semantics enjoy the following standard properties.
%\todoinline{DN}{The outline of the paper should end with the comment ``All results not proved in the paper are proved
%in the Appendix provided as supplementary material.
%I'm commenting-out the specific cross-references to the appendix, since they would be good in a long version without page limit.
%\newline
%[MN] Yes.
%}
% The proofs of the first four lemmas can be found in Appendix, Section~\ref{apd:proof:abstraction:basic-properties}.
% The proof of Lemma~\ref{lem:normalization:simple} can be found in Appendix, Section~\ref{apd:proof:abstraction:normalization}.
%\begin{restatable}[Normalization]{lem}{lemNormalizationSimple}
%\label{lem:normalization:simple}
%\tealtext{If \typeEval{}{e:\tau}, then there exists $v$ s.t. $e \reduce v$}.
%\end{restatable}

%\begin{lemma}[Normalization]
%\label{lem:normalization:simple}
%\tealtext{If \typeEval{}{e:\tau}, then there exists $v$ s.t. $e \reduce v$}.
%\end{lemma}

%\subsection
\paragraph{Logical relation.}

The logical relation is a type-indexed family of relations on values, 
parameterized by given relations for type variables.  
From it, we derive a relation on terms.
The abstraction theorem says the latter is reflexive.

Let $\gamma$ be a {\em term substitution}, i.e., a finite map from term variables to closed values, and $\delta$ be a {\em type substitution}, i.e., a finite map from type variables to closed types.  In symbols:
\begin{align*}
\gamma  ::= & ~.\sep \gamma, x \mapsto v & \text{Term Substitutions}\\
\delta  ::= &~ .\sep \delta, \alpha \mapsto \tau,\ \text{where \typeEval{}{\tau}} & \text{Type Substitutions}
\end{align*}

We say $\gamma$ {\em respects} $\Gamma$ (denoted by $\gamma\respect\Gamma$) 
when $\dom{\gamma} = \dom{\Gamma}$ and \typeEval{}{\gamma(x):\Gamma(x)} for any $x$.
We say $\delta$ {\em respects} $\Delta$ (denoted by $\delta \respect\Delta$) 
when $\dom{\delta} = \Delta$.
Let $\relDel{\tau_1,\tau_2}$ be the set of all binary relations over \bluetext{closed} values of closed types $\tau_1$ and $\tau_2$.
Let $\rho$ be an \airforcetext{\em environment}, a mapping from type variables to relations
$R \in \relDel{\tau_1,\tau_2}$.
%We say $\rho$ respects $\Delta$ (denoted by $\rho \respect \Delta$ when $\dom{\rho} = \Delta$.
We write $\rho \in \relDel{\delta_1,\delta_2}$ 
%as \bluetext{the point-wise extension of $R \in \relDel{\tau_1,\tau_2}$ for $\rho$}.
to say that $\rho$ is compatible with $\delta_1,\delta_2$ as follows: 
%\[ 
$
\rho \in \relDel{\delta_1,\delta_2} \triangleq  \dom{\rho} = \dom{\delta_1} = \dom{\delta_2}  \land  \forall \alpha \in \dom{\rho}.\, \rho(\alpha) \in \relDel{\delta_1(\alpha) ,\delta_2(\alpha)}
$.
%\]
The logical relation is inductively defined in Fig.~\ref{fig:logical-relation}, where $\rho\in \relDel{\delta_1,\delta_2}$ for some $\delta_1$ and $\delta_2$.
For any $\tau$, $\valrelation{\tau}{\rho}$ is a relation on closed values.
In addition, $\termrelation{\tau}{\rho}$ is a relation on terms.

% FUTURE NOTE: it would be more clear to use angle brackets in the calculus
% but parentheses for tuples in the metatheory, like Crary does.

%We next state a basic property of terms and values related by the logical relation.

%In the rules, \bluetext{we require that when $\tuple{v_1,v_2} \in \valrelation{\tau}{\rho}$, then \typeEval{}{v_i: \delta_i(\tau)}}; and when $\tuple{e_1,e_2} \in \termrelation{\tau}{\rho}$, then \typeEval{}{e_i: \delta_i(\tau)}.

%\begin{restatable}{lem}{lemLogRelatedTerm}
%\label{lem:logeq:related-term}
%\tealtext{Suppose that $\rho \in \relDel{\delta_1,\delta_2}$ for some $\delta_1$ and $\delta_2$. 
%For $i \in \set{1,2}$, it follows that:}
%\begin{itemize}
%\item \tealtext{if $\tuple{v_1,v_2} \in \valrelation{\tau}{\rho}$, then \typeEval{}{v_i: \delta_i(\tau)}, and}
%\item \tealtext{if $\tuple{e_1,e_2} \in \termrelation{\tau}{\rho}$, then \typeEval{}{e_i: \delta_i(\tau)}}.
%\end{itemize}
%\end{restatable}

\begin{lemma}
%{lemLogRelatedTerm}
\label{lem:logeq:related-term}
\tealtext{Suppose that $\rho \in \relDel{\delta_1,\delta_2}$ for some $\delta_1$ and $\delta_2$. 
For $i \in \set{1,2}$, it follows that:}
\begin{itemize}
\item \tealtext{if $\tuple{v_1,v_2} \in \valrelation{\tau}{\rho}$, then \typeEval{}{v_i: \delta_i(\tau)}, and}
\item \tealtext{if $\tuple{e_1,e_2} \in \termrelation{\tau}{\rho}$, then \typeEval{}{e_i: \delta_i(\tau)}}.
\end{itemize}
\end{lemma}

\begin{figure}[!t]
%\vspace{-20pt}
\begin{mathpar}
%\small
\LabelRule{FR-Int}
{~}
{\tuple{n,n} \in \valrelation{\intType}{\rho}}\and
%%%%%%%%%%%%%%%%%%%
\LabelRule{FR-Pair}{
\tuple{v_1,v_1'} \in \valrelation{\tau_1}{\rho} \\
\tuple{v_2,v_2'} \in \valrelation{\tau_2}{\rho}}
{\tuple{\tuple{v_1,v_2},\tuple{v_1',v_2'}}\in \valrelation{\tau_1 \times \tau_2}{\rho}}\and
%%%%%%%%%%%%%%%%%%%
\LabelRule{FR-Fun}
{\forall \tuple{v_1',v_2'} \in \valrelation{\tau_1}{\rho}.  \tuple{v_1\ v_1', v_2\ v_2'} \in \termrelation{\tau_2}{\rho} }
{\tuple{v_1, v_2}\in \valrelation{\tau_1 \rightarrow \tau_2}{\rho}}\and
%%%%%%%%%%%%%%%%%%%
\LabelRule{FR-Var}
{\tuple{v_1,v_2} \in R {\in Rel(\tau_1,\tau_2)}}
{\tuple{v_1,v_2} \in \valrelation{\alpha}{\rho[\alpha \mapsto R]}}
\and
%%%%%%%%%%%%%%%%%%%
\LabelRule{FR-Term}
{{\typeEval{}{e_1:\delta_1(\tau)}} \\
{\typeEval{}{e_2:\delta_2(\tau)}} \\
e_1 \reduce v_1\\ e_2 \reduce v_2\\ \tuple{v_1,v_2} \in \valrelation{\tau}{\rho}}
{\tuple{e_1,e_2} \in \termrelation{\tau}{\rho}}
\end{mathpar}
%\vspace{-4ex}
\caption{{The logical relation}}
\label{fig:logical-relation}
%\vspace{-10pt}
\end{figure}

We write $\delta(\Gamma)$ to mean a term substitution obtained from $\Gamma$ by applying $\delta$ on the range of $\Gamma$, i.e.:
\[ \dom{\delta(\Gamma)} = \dom{\Gamma} \mbox{ and } \forall x \in \dom{\Gamma}.\, \delta(\Gamma)(x) = \delta(\Gamma(x)).\] 

Suppose that \typeEval{\Delta,\Gamma}{e:\tau}, $\delta\respect \Delta$, and $\gamma \respect \delta(\Gamma)$.
Then we write $\delta\gamma(e)$ to mean the application of $\gamma$ and then $\delta$ to $e$.
{For example, suppose that $\delta(\alpha) = \intType$, $\gamma(x) = n$ for some $n$, and \typeEval{\alpha,x:\alpha}{\lambda y:\alpha.x:\alpha \rightarrow \alpha}, then $\delta\gamma(\lambda y:\alpha.x) = \lambda y:\intType.n$.
%(\typeEval{}{\delta\gamma(x):\intType}, and hence  \typeEval{}{n:\intType}).
}
We write $\tuple{\gamma_1, \gamma_2} \in \valrelation{\Gamma}{\rho}$ for some $\rho \in Rel(\delta_1,\delta_2)$ when $\gamma_1\respect \delta_1(\Gamma)$, $\gamma_2\respect \delta_2(\Gamma)$, and $\tuple{\gamma_1(x),\gamma_2(x)} \in \valrelation{\Gamma(x)}{\rho}$
for all $x \in \dom{\Gamma}$.

\begin{definition}[Logical equivalence]
Terms $e$ and $e'$ are {\em logically equivalent} at $\tau$ in $\Delta$ and $\Gamma$ (written \typeEval{\Delta,\Gamma}{e \logeq e':\tau}) if \typeEval{\Delta,\Gamma}{e:\tau},  \typeEval{\Delta,\Gamma}{e':\tau},  and for all $\delta_1,\delta_2 \respect \Delta$, all \bluetext{$\rho \in \relDel{\delta_1,\delta_2}$}, and all $\tuple{\gamma_1,\gamma_2} \in \valrelation{\Gamma}{\rho}$,
we have 
\( \tuple{\delta_1\gamma_1(e), \delta_2\gamma_2(e')} \in \termrelation{\tau}{\rho}. \)
\end{definition}

%\bluetext{Next, we state the abstraction theorem as in \cite{Reynolds-83}.}

%\begin{restatable}[Abstraction \cite{Reynolds-83}]{thm}{thmTypeAbstraction}
%\label{thm:type-abstraction}
%%If \typeEval{\Delta,\Gamma}{e:\tau}, $\delta_1,\delta_2 \respect \Delta$, \bluetext{$\rho \in \relDel{\delta_1,\delta_2}$}, and \redtext{$\tuple{\gamma_1,\gamma_2} \in \valrelation{\Gamma}{\rho}$},  then 
%%$$\tuple{\delta_1\gamma_1(e), \delta_2\gamma_2(e)} \in \termrelation{\tau}{\rho}.$$
%If \typeEval{\Delta,\Gamma}{e:\tau}, then \typeEval{\Delta,\Gamma}{e\logeq e:\tau}.
%\end{restatable}

\begin{theorem}[Abstraction \cite{Reynolds-83}]
%{thmTypeAbstraction}
\label{thm:type-abstraction}
%If \typeEval{\Delta,\Gamma}{e:\tau}, $\delta_1,\delta_2 \respect \Delta$, \bluetext{$\rho \in \relDel{\delta_1,\delta_2}$}, and \redtext{$\tuple{\gamma_1,\gamma_2} \in \valrelation{\Gamma}{\rho}$},  then 
%$$\tuple{\delta_1\gamma_1(e), \delta_2\gamma_2(e)} \in \termrelation{\tau}{\rho}.$$
If \typeEval{\Delta,\Gamma}{e:\tau}, then \typeEval{\Delta,\Gamma}{e\logeq e:\tau}.
\end{theorem}

\section{Declassification Policies}
\label{sec:local-policies}
Confidentiality policies can be expressed by 
information flows  of confidential sources to public sinks in programs. 
Confidential sources correspond to the secrets that the program receives and public sinks correspond
 to any results  given to a public observer, a.k.a. the attacker. 
These flows can  either  be direct ---e.g. if 
a function, whose result is  public, receives a confidential value as input and directly returns the secret---
 or indirect ---e.g. if a function, whose result is public, receives a confidential boolean value and returns 
 0 if the confidential value is false and 1 otherwise. 
Classification of program sources as confidential or public, a.k.a. {\it security policy},  must be  given by the programmer or security engineer: for a given security policy the program is said to be secure for {\it noninterference} if public resources do not depend on  confidential ones. Thus, noninterference for a program means total independence between
 public and confidential information. 
As simple and elegant as this information flow policy is, noninterference does not permit to consider 
as secure  programs that purposely need to release information in a controlled  way: for example a password-checker function that receives as confidential input a boolean value representing if the system password is equal to the user's input and returns 0  or 1 accordingly. 
In order to consider such intended dependences of public sinks from confidential sources, we need to  consider 
more relaxed security policies than noninterference, a.k.a. {\it declassification policies}.  
{Declassification security policies} allow us to specify controlled ways to release confidential inputs~\cite{sabelfeldSands:jcs2009}. 

Declassification policies that we consider in this work map  confidential inputs to functions, namely  {\em declassification functions}. % $f$.
These functions allow the programmer to specify what and how information can be released. 
The formal syntax for  declassification functions in this work is given
below,\footnote{In this paper, the type of confidential inputs is \intType. 
%We choose this since we are sticking with \cite{Li-Zdancewic-05-POPL}. The result presented in this paper can be generalized to accept confidential inputs of arbitrary types.}
} 
where $n$ is an integer value, and  $\oprt$ represents primitive arithmetic operators. 
\begin{align*}
\tau   & ::=  \intType \sep \tau \rightarrow \tau   & \text{Types}  \\
%\oprt  & :: = + \sep - \sep \dots                   &\text{Primitive Operators} \\
e      & ::= \lambda x:\tau.e \sep e\ e \sep x \sep n \sep e \oprt e  & \text{Terms}  \\
%a,b    & ::= \lambda x:\intType.e                   & \text{Actions}\\
f    & ::= \lambda x:\intType.e    & \text{Declass.\ Functions}
%\comb  & ::= \tuple{a,f}             & \text{Combinations}
\end{align*}
The static and dynamic semantics are standard.
% DN saving space
% and similar to the ones of the simply typed and call-by-value lambda calculus with type variables (see \S~\ref{sec:abstraction}).
To simplify the presentation we suppose that the applications of primitive operators on well-typed arguments terminates.
Therefore, the evaluations of declassification functions on values terminate. A policy simply defines which are the confidential variables and their authorized declassifications.
 For policies we refrain from using concrete syntax and instead give a simple
formalization that facilitates later definitions.
%We next define local policies. 

\begin{definition}[Policy]
\label{def:policy}
A policy \policy\ is a tuple  $\tuple{\varPolicy{\policy},\decPolicy{\policy}}$, where $\varPolicy{\policy}$ is a finite set of variables for confidential inputs, and \decPolicy{\policy} is a partial mapping from variables in \varPolicy{\policy} to declassification functions. 
\end{definition}

For simplicity we require that if $f$ appears in the policy then
it is a closed term of type $\intType \rightarrow \tau_f$ for some $\tau_f$. 
In the definition of policies, if a confidential input is not associated with a declassification function, then it cannot be declassified.
%Formally, a combination is mathematically defined as a pair $\tuple{f,a}$ but we write it as $f\circ a$ for clarity.

\begin{example}[Policy \policyoe\ using $f$]  %[Declassification via declassification function]
\label{ex:policy:odd-even}
Consider  policy \policyoe\ given by $\tuple{\varPolicy{\policyoe}, \decPolicy{\policyoe}}$ where $\varPolicy{\policyoe} = \set{x}$
 and $\decPolicy{\policyoe}(x) = f = \lambda x:\intType.\, x \modop 2$.
 Policy \policyoe\ states that only the parity of the confidential input $x$ can be released to a public observer. % or an odd number.
%W.r.t. this requirement, we have that $\varPolicy{\policyoe} = \set{x}$, $\dom{\decPolicy{\policyoe}} = \set{x}$, and $\decPolicy{\policyoe}(x) = %f=\lambda x:\intType. x \modop 2$.
\end{example}

%\begin{example}[Policy \policyhash\ using $f\circ a$] (inspired by \cite[Example 3.2.1]{Li-Zdancewic-05-POPL}) %[Declassification via action and declassification function]
%\label{ex:policy:hash}
%Assume that $\hash$ is a primitive operator.
%Consider  policy \policyhash\   given by $\tuple{\varPolicy{\policyhash}, \decPolicy{\policyhash}}$  where $\varPolicy{\policyhash} = \set{x}$,   $\decPolicy{\policyhash}(x) = f\circ a$, 
% $a$ is $\lambda x:\intType. \hash\ x$, and $f$ is $\lambda x: \intType.x\modop 2^{64}$.
%Policy \policyhash\ states that the hashed value of the confidential input $x$  cannot be released, but the lowest 64 bits of its hashed value can.
%\end{example}

%The notion of action can be generalized to multiple steps of declassification,
%for example to specify the correct order of application of sanitizers~\cite{HooimeijerLMSV11}.
%Our encoding can be extended straightforwardly to multiple steps, 
%at the cost of notational clutter we prefer to avoid in this presentation.

\ifshow
\begin{assumption}
\label{assumption:alpha-equivalence}
\redtext{We say $f$ {\em appears} in \decPolicy{\policy} when $f \in \range{\decPolicy{\policy}}$ or $f \circ b \in \range{\decPolicy{\policy}}$ (for some $b$)}.
We say that $a$ {\em appears} in \decPolicy{\policy} when $f\circ a \in \range{\decPolicy{\policy}}$ (for some $f$).

\redtext{Suppose that $f$ and $g$ appear in \decPolicy{\policy} and they are alpha-equivalent. By changing bound variables in $f$ and $g$, we can have functions which are syntactically equivalent.
To simplify the presentation, we suppose that for any $f$ and $g$ appearing in  \decPolicy{\policy}, $f$ and $g$ are either syntactically equal or not alpha-equivalent}.

Similarly, we require that if $a$ and $b$ appearing in \decPolicy{\policy}, $a$ and $b$ are either syntactically equal or not alpha-equivalent.
\end{assumption}
\fi

%\begin{remark}[\redtext{Connection to \cite{Li-Zdancewic-05-POPL}}].
%\redtext{The local policies presented here are different from the ones in \cite[Definition 5.1.1]{Li-Zdancewic-05-POPL}. In \cite{Li-Zdancewic-05-POPL}, there is no action in policy specifications.
%However, an action $a$ may appear in good programs (programs where confidential inputs are released correctly) when we have ``reclassification'' (or downgrading) as explained next}.
%
%\redtext{Suppose that we have $l_1 = \set{f}$, $l_2 = \set{g}$ where $f$ and $g$ are different and $g \circ a = f$, $x$ is at $l_1$. 
%Then $a\ x$ is considered/reclassified as a confidential input at $l_2$ and can be declassified via $g$ (notice that $g(a\ x) = f(x)$)}.
%%\redtext{In our current solution, we need to specify $a$ explicitly in the policy.
%%If $x$ is a confidential input at $l_1$, then our approach accepts $f\ x$ but not $g(a\ x)$.
%%An extension to address this problem is presented in Section~\ref{sec:extension}}.
%\end{remark}

\section{Type-based Declassification}
\label{sec:trni}

In this section, we show how to encode declassification policies as standard types in the language of \S~\ref{sec:abstraction}, we define and we prove our free theorem. 
We consider a termination-sensitive~\cite{imposs} information flow security property,\footnote{Our security property is termination sensitive but programs in the language always terminate.
In the development for ML (in an appendix), programs may not terminate and the security property is also termination sensitive.
%Notice that even if the security property is termination sensitive, expressions in our language always terminate.
} with declassification,  called type-based relaxed noninterference (TRNI) and  taken from Cruz et al~\cite{Cruz-etal-17-ECOOP}. 
 It is important to notice that our developement, in this section,
studies the reuse for security of standard programming languages type systems
together with  soundness proofs for security for free by using the abstraction theorem. 
In contrast, Cruz et al~\cite{Cruz-etal-17-ECOOP}
 use a modified type system for security and prove soundness from scratch, without apealing to parametricity. 

 Through this section, we consider a fixed policy \policy\ (see Def.~\ref{def:policy}) given by $\tuple{\varPolicy{\policy},\decPolicy{\policy}}$.
%The language for writing programs is presented in \S\ref{sec:abstraction}.
We treat free variables in a program as inputs and, without loss of generality, we assume that there are two kinds of inputs: integer values, which are considered as confidential,  and declassification functions, which are fixed according to policy. % \newtext{(and named, following~\cite{Cruz-etal-17-ECOOP}).}
A public input can be encoded as a confidential input that can be declassified via the identity function.
We consider terms without type variables as source programs.
That is we consider terms $e$ s.t. for all type substitutions $\delta$,  $\delta(e)$ is syntactically the same as $e$.\footnote{%
An example of a term with type variables is $\lambda x:\alpha.x$. We can easily check that there exists a type substitutions $\delta$ s.t. $\delta(e)$ is syntactically different from $e$ (e.g. for $\delta$ s.t. $\delta(\alpha) = \intType$, $\delta(e) = \lambda x:\intType.x$).}

%Without loss of generality, we assume that there are two types of inputs for a program: integer values, which are considered as confidential,  and declassification functions and actions. 

\subsection{Views and indistinguishability}

%\todoinline{DN}{``private'' versus ``confidential''; probably choose one and stick with it.
%\newline
%[MN] Yes. I chose ``confidential''.
%}

We provide two term contexts to define TRNI, called the confidential view and public view. The first view represents an observer that can access confidential inputs, while the second one represents an observer that can only observe declassified inputs.
% DN: we no longer have public inputs
The views are defined using fresh term and type variables.

\paragraph{Confidential view.}
Let $\confInp_\top = \set{x\sep x \in \varPolicy{\policy} \setminus \dom{\decPolicy{\policy}}}$ be the set of inputs that cannot be declassified.
%We have that for all $\dom{\conView{\policy}} = \policy.\confInp$ and for all $x \in \dom{\policy}$, $x:\intType \in \conView{\policy}$.
First we define the encoding for these inputs as a term context:
$$
\Gamma^{\policy}_{C,\top}  \triangleq \set{x:\intType \sep x \in \confInp_\top}.
$$
Next, we specify the encoding of confidential inputs that can be declassified.
To this end, define \enccon{\_,\_} as follows,
where $f:\inttype\to\tau_f$ is in \policy.
\begin{align*}
\enccon{x,f}        & \triangleq \set{x:\intType, x_f:\intType \rightarrow {\tau_f}}
%\enccon{x,f\circ a} & \triangleq \set{x: \intType, x_a: \intType \rightarrow \intType, x_f: \intType \rightarrow {\tau_f}}
\end{align*}
Finally, we write \conView{\policy} for the term context encoding the confidential view for \policy.
$$\conView{\policy} \triangleq 
\Gamma^{\policy}_{C,\top} \cup \bigcup_{x \in \dom{\decPolicy{\policy}}} \enccon{x,\decPolicy{\policy}(x)}.$$
We assume that, for any $x$, the variable $x_f$ in the result of \enccon{x,\decPolicy{\policy}(x)} is 
distinct from the variables in \varPolicy{\policy}, distinct from each other, 
and distinct from $x_{f'}$ for distinct $f'$.
We make analogous assumptions in later definitions.

%\todoinline{DN}{example for Example~\ref{ex:policy:hash}.
%\newline
%[MN] I added Example~\ref{ex:trni:con-view}.}

From the construction, \conView{\policy} is a mapping, and for any $x \in \dom{\conView{\policy}}$, it follows that $\conView{\policy}(x)$ is a closed type.
Therefore, \conView{\policy} is well-formed for the empty set of type variables, so it can be used in typing judgments of the form  \typeEval{\conView{\policy}}{e:\tau}.

%\begin{lemma}
%\label{lem:conview:correctness}
%If $x:\tau_1 \in \conView{\policy}$ and $x:\tau_2 \in \conView{\policy}$ then $\tau_1 = \tau_2$.
%\end{lemma}
%\begin{proof}
%\bluetext{The proof follows from the definition of \conView{\policy}}.
%\end{proof}

\begin{example}[Confidential view]
\label{ex:trni:con-view}
For \policyoe\ in Example~\ref{ex:policy:odd-even}, the confidential view is: $\conView{\policyoe} = x:\intType, x_f:\intType \rightarrow \intType$.
%For \policyhash\ described in Example~\ref{ex:policy:hash}, the confidential view is $\conView{\policyhash} = x:\intType, x_a:\intType \rightarrow \intType, x_f:\intType \rightarrow \intType$.
\end{example}

\paragraph{Public view.}
The basic idea is to encode policies by using type variables.
First we define the encoding for confidential inputs that cannot be declassified.
We define a set of type variables, $\Delta^{\policy}_{P,\top}$ 
and a mapping $\Gamma^{\policy}_{P,\top}$ 
for confidential inputs that cannot be declassified.
\[
\Delta^{\policy}_{P,\top}  \triangleq \set{\alpha_x \sep x\in \confInp_\top}
\qquad
\Gamma^{\policy}_{P,\top}  \triangleq \set{x:\alpha_x \sep x \in \confInp_\top}
\]
This gives the program access to $x$ at an opaque type.

In order to define the encoding for confidential inputs that can be declassified, 
we define \encpub{\_,\_}:
\begin{align*}
\encpub{x,f}        & \triangleq \tuple{\set{\alpha_f}, \set{x:\alpha_f, x_f:\alpha_f \rightarrow  {\tau_f}}}
%\encpub{x,f\circ a} & \triangleq \tuple{\set{\alpha_{f \circ a}, \alpha_f},\ \set{x: \alpha_{f\circ a},\ x_a: \alpha_{f \circ a} \rightarrow \alpha_f, \\
%	& \hspace{130pt} x_f: \alpha_{f} \rightarrow {\tau_f}}}
\end{align*}
The first form will serve to give the program access to $x$ only via function variable $x_f$ that we will 
ensure is interpreted as the policy function $f$.
%similarly for the second form.
We define a type context \pubViewType{\policy} and term context \pubViewTerm{\policy} 
that comprise the public view, as follows.
$$\tuple{\pubViewType{\policy}, \pubViewTerm{\policy}} \triangleq 
\tuple{\Delta^{\policy}_{P,\top},\Gamma^{\policy}_{P,\top}} \cup \bigcup_{x \in \dom{\decPolicy{\policy}}} \encpub{x,\decPolicy{\policy}(x)}, $$
where $\tuple{S_1,S_1'} \cup \tuple{S_2,S_2'} = \tuple{S_1\cup S_2, S_1'\cup S_2'}$.

%\noteinline{Notice that $f$ is a function. Thus, we use $x_f$ to denote the variable for $f$. To simplify the presentation, we require that $x_f$ is not in \varPolicy{\policy}.
%Similarly, $x_a$ is not in \varPolicy{\policy}.}

\ifshow
\noteinline{Should we keep the following lemma?}
\begin{lemma}[Correctness of the construction - 1]
\begin{enumerate}
\item $x:\alpha_x \in \pubViewTerm{\policy}$ and $\alpha_x \in \pubViewType{\policy}$ iff $x \in \varPolicy{\policy} \setminus \dom{\decPolicy{\policy}}$
\item $x:\alpha_f, x_f:\alpha_f \rightarrow \intType \in \pubViewTerm{\policy}$, and $\alpha_f \in \pubViewType{\policy}$ iff $\decPolicy{\policy}(x) = f$
\item $x:\alpha_{f\circ a}, x_a:\alpha_{f\circ a} \rightarrow \alpha_f, x_f:\alpha_f\rightarrow \intType \in \pubViewTerm{\policy}$, and $\alpha_{f\circ a}, \alpha_f  \in \pubViewType{\policy}$ iff $\decPolicy{\policy}(x) = f \circ a$.
\end{enumerate}
\end{lemma}
\begin{proof}
Directly from the definition of \pubViewTerm{\policy} and \pubViewType{\policy}.
\end{proof}
\fi

\begin{example}[Public view]
\label{ex:trni:pub-view}
For \policyoe, the typing context in the public view has one type variable: $\pubViewType{\policyoe}  = \alpha_f$.
The term context in the public view is $\pubViewTerm{\policyoe} = x:\alpha_{f},\  x_f: \alpha_f \rightarrow \intType$.

%For \policyhash, the typing context in the public view has two type variables: $\pubViewType{\policyhash}  = \alpha_{f\circ a}, \alpha_f$.
%The term context in the public view is $\pubViewTerm{\policyhash} = x:\alpha_{f\circ a},\ x_a:\alpha_{f\circ a} \rightarrow \alpha_a,\ x_f: \alpha_f \rightarrow \intType$.
%%Notice that for \policyoe\ and \policyhash, as described in public views, we can have well-typed programs like $x_f\ x$, 
\end{example}

From the construction, \pubViewTerm{\policy} is a mapping, and for any $x \in \dom{\pubViewTerm{\policy}}$, it follows that $\pubViewTerm{\policy}(x)$ is well-formed in \pubViewType{\policy} (i.e. $\typeEval{\pubViewType{\policy}}{\pubViewTerm{\policy}(x)}$).
Thus, \pubViewTerm{\policy} is well-formed in the typing context \pubViewType{\policy}.
Therefore, \pubViewType{\policy} and \pubViewTerm{\policy} can be used in typing judgments of the form  \typeEval{\pubViewType{\policy},\pubViewTerm{\policy}}{e:\tau}.

Notice that in the public view of a policy, types of variables for confidential inputs are not \intType. Thus, the public view does not allow programs where concrete declassifiers are applied to confidential input variables even when the applications are semantically correct according to the policy (e.g. for \policyoe, the program $f\ x$ does not typecheck in the public view). 
Instead, programs should apply named declassifers (e.g.  for \policyoe, the program $x_f\ x$ is well-typed in the public view).
%However, the public view does allow programs where confidential input variables are used in applications of declassifier variables associated with them 
%Variables for declassifiers and confidential inputs will be replaced with respectively concrete functions specified by the policy and values.
%This is consistent with the requirement we presented above: programs received integer values, declassification functions and actions as inputs.

%We prove that \pubViewType{\policy} and \pubViewTerm{\policy} can be used in typing judgments \typeEval{\pubViewType{\policy},\pubViewTerm{\policy}}{e:\tau}.
%That is \pubViewTerm{\policy} is a mapping and  for any $x$, $\pubViewTerm{\policy}(x)$ is well-formed under \pubViewType{\policy}.

%\todoinline{DN}{The following proofs are vague (and the results depend not only on the explicit definitions but also the assumption of uniqueness of the fresh variables.  Since these two lemmas don't seem to be referenced, I suggest we just say the properties in words, without calling them lemmas.
%\newline
%[MN] Yes. 
%}

%\begin{lemma}
%If $x:\tau_1 \in \pubViewTerm{\policy}$ and $x:\tau_2 \in \pubViewTerm{\policy}$ then $\tau_1 = \tau_2$.
%\end{lemma}
%\begin{proof}
%\bluetext{The proof follows from the definitions of \pubViewType{\policy} and  \pubViewTerm{\policy}}.
%\end{proof}
%
%\begin{lemma}
%If $x:\tau \in \pubViewTerm{\policy}$, then \typeEval{\pubViewType{\policy}}{\tau}.
%\end{lemma}
%\begin{proof}
%\bluetext{The proof follows from the definitions of \pubViewType{\policy} and  \pubViewTerm{\policy}}.
%\end{proof}

\paragraph{Indistinguishability.} 
The security property TRNI is defined in a usual way,
using partial equivalence relations called indistinguishability.
To define indistinguishability, we define a type substitution \deltaPol{\policy} such that
$\deltaPol{\policy} \respect \pubViewType{\policy}$, as follows:
%\begin{equation}\label{def:deltaPol}
%\mbox{for all $\alpha_x,\alpha_f,\alpha_{f\circ a}$ in $\pubViewType{\policy}$, 
%let $\deltaPol{\policy}(\alpha_x) = \deltaPol{\policy}(\alpha_{x_f}) = \deltaPol{\policy}(\alpha_{x_{f\circ a}}) = \intType$.}
%\end{equation}
\begin{equation}\label{def:deltaPol}
\text{for all}\ \alpha_x,\alpha_f\ \text{in}\ \pubViewType{\policy},\  \text{let}\ \deltaPol{\policy}(\alpha_x) = \deltaPol{\policy}(\alpha_{f})  = \intType.
\end{equation}

%\redtext{We write $\action{\alpha_f} = \set{a \sep \exists x_a: \pubViewTerm{\policy}(x_a) = \alpha_{f\circ a} \rightarrow \alpha_f}$ for the set of all actions whose results can be declassified via $f$, 
%and \valaction{\alpha_f}  = \set{n \sep \exists a \in \action{\alpha_f} \wedge n\ \text{in the range of $a$}} 
%as the set of all values in the ranges of actions in $\action{\alpha_f}$}.
%\redtext{We write \condtwo{\alpha_f,\policy} when there is no $x$ s.t. $\pubViewTerm{\policy}(x) = \alpha_f$ (that is there is no direct input at the type $\alpha_f$)}.
%\redtext{Notice that $f \in \range{\decPolicy{\policy}}$ is defined based on alpha-equivalent}.
%Notice that given a policy \policy, we can change bound variables of elements in the range of \decPolicy{\policy}.
%Therefore, we require that \redtext{if $f$ and $g$ are alpha-equivalent, then $f \in \range{\decPolicy{\policy}}$ iff $g \in \range{\decPolicy{\policy}}$, and if $f\circ a$ and $g\circ b$ are alpha-equivalent, $f \circ a \in \range{\decPolicy{\policy}}$ iff $g\circ b \in \range{\decPolicy{\policy}}$}.

The inductive definition of indistinguishability for a policy \policy\ is presented in Figure~\ref{fig:indistinguishability}, where $\alpha_x$ and $\alpha_f$ are from \pubViewType{\policy}.
Indistinguishability is defined for $\tau$ s.t. \typeEval{\pubViewType{\policy},\pubViewTerm{\policy}}{\tau}.
The definitions of indistinguishability for \intType\ and $\tau_1 \times \tau_2$ are straightforward.
We say that two functions are indistinguishable at $\tau_1 \rightarrow \tau_2$ if on any indistinguishable inputs they generate indistinguishable outputs.
Since we use $\alpha_x$  to encode confidential integer values that cannot be declassified, any integer values $v_1$ and $v_2$ are indistinguishable, according to rule Eq-Var1.
Notice that $\deltaPol{\policy}(\alpha_x) = \intType$.
Since we use $\alpha_f$ to encode confidential integer values that can be declassified via $f$ where \typeEval{}{f:\intType \rightarrow\tau_f}, we say that $\tuple{v_1,v_2} \in \indval{\alpha_f}$ when $\tuple{f\ v_1,f\ v_2} \in \indterm{\tau_f}$.
%The idea behind $\indval{\alpha_{f\circ a}}$ is similar.

\begin{example}[Indistinguishability]
\label{ex:trni:ind}
For \policyoe\  (of Example~\ref{ex:policy:odd-even}),
two values $v_1$ and $v_2$ are indistinguishable at $\alpha_f$ when both of them are even numbers or odd numbers.
\begin{equation*}
\indval{\alpha_f} = \{\tuple{v_1,v_2} \sep \typeEval{}{v_1:\intType},\ \typeEval{}{v_2:\intType},\ (v_1 \modop 2) =_\intType (v_2 \modop 2)\}.
\end{equation*}
We write $e_1 =_\intType e_2$ to mean that {$e_1 \reduce v$ and $e_2 \reduce v$ for some integer value $v$}.

%$$\indval{\alpha_f} = \{\tuple{v_1,v_2} \sep \typeEval{}{v_1:\intType}, \typeEval{}{v_2:\intType}, (v_1 \modop 2) =_\intType (v_2 \modop 2)\}.$$

%For \policyhash\  (of Example~\ref{ex:policy:hash}),
%two values $v_1$ and $v_2$ are indistinguishable at $\alpha_{f\circ a}$ when they are integer values and the lowest 64 bits of their hashed values are the same.
%\begin{multline*}
%\indval{\alpha_{f\circ a}} = \{\tuple{v_1,v_2} \sep \typeEval{}{v_1:\intType}, \typeEval{}{v_2:\intType}, \\
%((\hash\ v_1) \modop 2^{64}) =_\intType ((\hash\ v_2) \modop 2^{64})\}.
%\end{multline*}

%$$
%\indval{\alpha_{f\circ a}} = \{\tuple{v_1,v_2} \sep \typeEval{}{v_1:\intType}, \typeEval{}{v_2:\intType}, ((\hash\ v_1) \modop 2^{64}) =_\intType ((\hash\ v_2) \modop 2^{64})\}.
%$$
\end{example}

\begin{figure}[!t]
%\vspace{-20pt}
\begin{mathpar}
%\small
\LabelRule{Eq-Int}
{~}
{\tuple{n,n} \in \indval{\intType}}\and
%%%%%%%%%%%%%%%%%%%
\LabelRule{Eq-Pair}{
\tuple{v_1, v_1'} \in \indval{\tau_1} \\
\tuple{v_2, v_2'} \in \indval{\tau_2}}
{\tuple{\tuple{v_1,v_2},\tuple{v_1',v_2'}} \in \indval{\tau_1\times\tau_2}}\and
%%%%%%%%%%%%%%%%%%%
\LabelRule{Eq-Fun}
{\forall \tuple{v_1',v_2'}: \tuple{v_1',v_2'} \in \indval{\tau_1}.  \tuple{v_1\ v_1',  v_2\ v_2'} \in \indterm{\tau_2}}
{\tuple{v_1,v_2} \in \indval{\tau_1 \rightarrow \tau_2}}\and
%%%%%%%%%%%%%%%%%%
\LabelRule{Eq-Var1}
{\typeEval{}{v_1,v_2:\deltaPol{\policy}(\alpha_x)}
}
{\tuple{v_1,v_2} \in \indval{\alpha_x}} \and
%%%%%%%%%%%%%%%%%%
\LabelRule{Eq-Var2}
{\typeEval{}{v_1,v_2:\deltaPol{\policy}(\alpha_f)} \\ 
%\typeEval{}{f:\intType \rightarrow\tau_f}\\
\purpletext{\tuple{f\ v_1, f\ v_2} \in \indterm{\tau_f}}
}
{\tuple{v_1,v_2} \in \indval{\alpha_f}}\and
%%%%%%%%%%%%%%%%%%
%\LabelRule{Eq-Var3}
%{\typeEval{}{v_1,v_2:\deltaPol{\policy}(\alpha_{f\circ a})} \\ 
%\purpletext{\tuple{a\ v_1, a\ v_2} \in \indterm{\alpha_f}}
%}
%{\tuple{v_1,v_2} \in \indval{\alpha_{f\circ a}} } \and
%%%%%%%%%%%%%%%%%%
\LabelRule{Eq-Term}
{\typeEval{}{e_1,e_2:\deltaPol{\policy}(\tau)} \\ 
 e_1 \reduce v_1 \\ 
 e_2 \reduce v_2 \\ 
 \tuple{v_1,v_2} \in \indval{\tau}
}
{\tuple{e_1,e_2} \in \indterm{\tau}}
%%%%%%%%%%%%%%%%%%
\end{mathpar}
%\vspace*{-4ex}
\caption{Indistinguishability}
\label{fig:indistinguishability}
%\vspace{-10pt}
\end{figure}

%Let $\tuple{\Delta_P,\Gamma_P} = \bigcup_{f \in \dom{\policy}} \encpub{\policy(f)}$ and .

Term substitutions $\gamma_1$ and $\gamma_2$ are called {\em indistinguishable w.r.t. \policy} (denoted by $\tuple{\gamma_1,\gamma_2} \in \indval{\policy}$) if the following hold.
\begin{itemize}
\item \purpletext{$\gamma_1 \respect \deltaPol{\policy}(\pubViewTerm{\policy})$ and $\gamma_2 \respect \deltaPol{\policy}(\pubViewTerm{\policy})$},
\item for all $x_f \in \dom{\pubViewTerm{\policy}}$, $\gamma_1(x_f) = \gamma_2(x_f) = f$, 
%\item for all $x_a \in \dom{\pubViewTerm{\policy}}$ , $\gamma_1(x_a) = \gamma_2(x_a) = a$, 
\item for all other $x \in \dom{\pubViewTerm{\policy}}$, $\tuple{\gamma_1(x), \gamma_2(x)} \in \indval{\pubViewTerm{\policy}(x)}$.
\end{itemize}
Note that each $\gamma_i$ maps $x_f$ to the specific function $f$ in the policy.
Input variables are mapped to indistinguishable values.

We now define type-based relaxed noninterference w.r.t. \policy\ {for a type $\tau$ well-formed in \pubViewType{\policy}}.
It says that indistinguishable inputs lead to indistinguishable results.
%\todoinline{DN}{I think we should say $\TRNI{\policy,\tau}$ is defined on $\tau$ that is well-formed in \pubViewType{\policy}.
%\newline
%[MN] Yes. This is already in Definition~\ref{def:trni}. \redtext{I added more information to the paragraph above this definition}.
%}
\begin{definition} %[TRNI \cite{Cruz-etal-17-ECOOP}]
\label{def:trni}
A term $e$ is \TRNI{\policy,\tau} provided that \typeEval{\conView{\policy}}{e}, and \typeEval{\pubViewType{\policy}}{\tau}, and for all 
%$\gamma_1$ and $\gamma_2$ s.t. 
$\tuple{\gamma_1,\gamma_2} \in \indval{\policy}$
we have $\tuple{\gamma_1(e), \gamma_2(e)} \in \indterm{\tau}$.
\end{definition}

Notice that if a term is well-typed in the public view then by replacing all type variables in it with \intType, we get a term which is also well-typed in the confidential view (that is, if \typeEval{\pubViewType{\policy},\pubViewTerm{\policy}}{e:\tau}, then \typeEval{\conView{\policy}}{\delta(e):\delta(\tau)} where $\delta$ maps all type variables in \pubViewType{\policy} to \intType).
However, 
Definition~\ref{def:trni} also requires that the term $e$ is itself well-typed in the confidential view.
This merely ensures that the definition is applied, as intended, to programs that do not contain type variables.
% since we want our programs to be independent from our approach, i.e. there should be no type variable in programs.

%In our approach, the public view (i.e. \pubViewType{\policy} and \pubViewTerm{\policy}) is responsible for specifying local policy \policy.
%Thus, by having \typeEval{\pubViewType{\policy}}{\tau} in the definition of \TRNI{\policy,\tau}, we want to specify what can be done by a public observer.

{The definition of TRNI is indexed by a type for the result of the term.
The type can be interpreted as constraining the observations to be made by the public observer.
We are mainly interested in concrete output types, which express that the observer can do whatever they like and has full knowledge of the result.
Put differently, TRNI for an abstract type expresses security under the assumption that the observer is somehow forced to respect the abstraction.
Consider the policy \policyoe\ (of Example~\ref{ex:policy:odd-even}) where $x$ can be declassified via $f = \lambda x:\intType.x\modop 2$.
As described in Example~\ref{ex:trni:pub-view}, $\pubViewType{\policyoe}  = \alpha_f$ and $\pubViewTerm{\policyoe} = x:\alpha_{f},\  x_f: \alpha_f \rightarrow \intType$. 
We have that the program $x$ is \TRNI{\policyoe,\alpha_f} since the observer cannot do anything to $x$ except for applying $f$ to $x$ which is allowed by the policy.
This program, however, is not \TRNI{\policyoe,\intType} since the observer can apply any function of the type $\intType \rightarrow \tau'$ (for some closed $\tau'$), including the identity function, to $x$ and hence can get the value of $x$.}

\begin{example}
\label{ex:ml:trni:def}
The program $x_f\ x$ is \TRNI{\policyoe, \intType}.
Indeed, for any arbitrary $\tuple{\gamma_1,\gamma_2} \in \indval{\policy}$, we have that $\gamma_1(x_f) = \gamma_2(x_f) = f = \lambda x:\intType.x\modop 2$, and $\tuple{v_1,v_2} \in \indval{\alpha_f}$, where $\gamma_1(x) = v_1$ and $\gamma_2(x) = v_2$ for some $v_1$ and $v_2$.
When we apply $\gamma_1$ and $\gamma_2$ to the program, we get respectively $v_1 \modop 2$ and $v_2 \modop 2$.
Since $\tuple{v_1,v_2} \in \indval{\alpha_f}$, as described in Example~\ref{ex:trni:ind}, $(v_1 \modop 2) =_\intType (v_2 \modop 2)$.
Thus, $\tuple{\gamma_1(x_f\ x),\gamma_2(x_f\ x)} \in \indterm{\intType}$.
Therefore, the program $x_f\ x$ satisfies the definition of TRNI.
\end{example}

\ifshow
\begin{remark}[\orangetext{On encoding}]
\tealtext{We consider a policy where $x$ can be declassified via $f\circ a$ and $y$ can be declassified via $f \circ b$.
W.r.t. the encoding presented above, we have two term variables for $f$ and we can accept programs $f(a\ x)$ and $f(b\ y)$}. 

We consider a policy where $x$ can be declassified via $f \circ a$ and $y$ can be declassified via $g\circ a$ where $f\neq g$, where \typeEval{}{f,g:\intType \rightarrow \tau} for some $\tau$.
W.r.t. the encoding, we have:
\begin{itemize}
\item $x: \alpha_{f\circ a}$, $x_a: \alpha_{f\circ a} \rightarrow \alpha_f$, $x_f: \alpha_f \rightarrow\tau$,
\item $y: \alpha_{f\circ b}$, $y_a: \alpha_{f\circ a} \rightarrow \alpha_g$, $y_g: \alpha_g \rightarrow\tau$.
\end{itemize}

One may think that we should have only one term variable for $a$ (instead of having two) (e.g. $z_a$).
However, by doing this, we may accept programs $f(a\ y)$ and $g(a\ x)$ which are actually violate the policy.
\end{remark}
\fi

%If in fact the program $e$ is typable at $\tau$ in the public view, then it satisfies $\TRNI{\policy,\tau}$.

\subsection{Free theorem: typing in the public view implies security}
In order to prove security ``for free'', i.e., as consequence of Theorem~\ref{thm:type-abstraction},
we define $\rhoPol{\policy}$ as follows:
\begin{itemize}
\item for all $\alpha_x \in \pubViewType{\policy}$, $\rhoPol{\policy}(\alpha_x) = \indval{\alpha_x}$,
\item for all $\alpha_f \in \pubViewType{\policy}$, $\rhoPol{\policy}(\alpha_f) = \indval{\alpha_{f}}$.
%\item for all $\alpha_{f\circ a} \in \pubViewType{\policy}$, $\rhoPol{\policy}(\alpha_{f\circ a}) = \indval{\alpha_{f\circ a}}$.
\end{itemize} 
It is a relation on the type substitution $\deltaPol{\policy}$ defined in Eqn.~(\ref{def:deltaPol}).
\begin{lemma}
\label{lem:relation-assignment}
\tealtext{$\rhoPol{\policy}\in \relDel{\deltaPol{\policy},\deltaPol{\policy}}$}.
\end{lemma}

\purpletext{From Lemma~\ref{lem:relation-assignment}, we can write \valrelation{\tau}{\rhoPol{\policy}} or \termrelation{\tau}{\rhoPol{\policy}} for any $\tau$ such that \typeEval{\pubViewType{\policy}}{\tau}}.
We next establish the relation between \termrelation{\tau}{\rho} and \indterm{\tau}:
under the interpretation corresponding to the desired policy \policy, they are equivalent.
In other words, indistinguishability is an instantiation of the logical relation.

\begin{lemma}
\label{lem:logeqpol-indis:eq}
For any $\tau$ such that  $\typeEval{\pubViewType{\policy}}{\tau}$, we have
\( \tuple{v_1,v_2} \in \valrelation{\tau}{\rhoPol{\policy}} \mbox{ iff } \tuple{v_1,v_2} \in \indval{\tau} \),
and also 
\( \tuple{e_1,e_2} \in \termrelation{\tau}{\rhoPol{\policy}} \mbox{ iff } \tuple{e_1,e_2} \in \indterm{\tau} \).
\end{lemma}

{By analyzing the type of $\pubViewTerm{\policy}(x)$, we can establish the relation of $\gamma_1$ and $\gamma_2$ when $\tuple{\gamma_1,\gamma_2} \in \indval{\policy}$}.
\begin{lemma}
\label{lem:term-sub:indis:logeq}
\tealtext{If $\tuple{\gamma_1,\gamma_2} \in \indval{\policy}$, then $\tuple{\gamma_1,\gamma_2} \in \valrelation{\pubViewTerm{\policy}}{\rhoPol{\policy}}$}.
\end{lemma}
%The lemma is proven by case analysis on the type of .
%\begin{proof}
%By case analysis.
%\end{proof}

The main result of this section is that a term is TRNI at $\tau$ if it has type $\tau$ in the public view that encodes the policy.
%\footnote{\newtext{Notice that if $e$ has no type variable and $e$ is well-typed in the public view, then $e$ is well-typed in the confidential view; and if $e$ is well-typed in the confidential view, then $e$ has no type variable. 
%Thus, we can state Theorem~\ref{thm:trni} as: if $e$ is well-typed in the confidential view and its type in the public view is $\tau$, then $e$ is \TRNI{\policy,\tau}}.
%}.

\begin{theorem}
\label{thm:trni}
If $e$ has no type variables and $\typeEval{\pubViewType{\policy},\pubViewTerm{\policy}}{e:\tau}$, then $e$ is \TRNI{\policy,\tau}.
\end{theorem}
\begin{proof}
From the abstraction theorem (Theorem~\ref{thm:type-abstraction}), for all $\delta_1,\delta_2 \respect \pubViewType{\policy}$, for all {$\tuple{\gamma_1,\gamma_2} \in \valrelation{\pubViewTerm{\policy}}{\rho}$}, and for all \bluetext{$\rho \in Rel(\delta_1,\delta_2)$}, it follows that
$$\tuple{\delta_1\gamma_1(e), \delta_2\gamma_2(e)} \in \termrelation{\tau}{\rho}.$$

Consider $\tuple{\gamma_1,\gamma_2} \in \indval{\policy}$.
\purpletext{Since $\tuple{\gamma_1,\gamma_2} \in \indval{\policy}$, from Lemma~\ref{lem:term-sub:indis:logeq}, we have that $\tuple{\gamma_1,\gamma_2} \in \valrelation{\pubViewTerm{\policy}}{\rhoPol{\policy}}$}.
Thus, we have that  $\tuple{\deltaPol{\policy}\gamma_1(e), \deltaPol{\policy}\gamma_2(e)} \in \termrelation{\tau}{\rhoPol{\policy}}$.
%Since \typeEval{\conView{\policy}}{e}, we have that $\deltaPol{\policy}\gamma_i(e) = \gamma_i(e)$.
%\newtext{Since $e$ has no type variable}, we have that $\deltaPol{\policy}\gamma_i(e) = \gamma_i(e)$.
{Since $e$ has no type variable}, we have that $\deltaPol{\policy}\gamma_i(e) = \gamma_i(e)$.
Therefore, $\tuple{\gamma_1(e), \gamma_2(e)} \in \termrelation{\tau}{\rhoPol{\policy}}$.
Since $\tuple{\gamma_1(e), \gamma_2(e)} \in \termrelation{\tau}{\rhoPol{\policy}}$, from Lemma~\ref{lem:logeqpol-indis:eq}, it follows that $\tuple{\gamma_1(e),\gamma_2(e)} \in \indterm{\tau}$.
In addition, since $e$ has no type variable and \typeEval{\pubViewType{\policy},\pubViewTerm{\policy}}{e:\tau}, we have that \typeEval{\deltaPol{\policy}(\pubViewTerm{\policy})}{e:\deltaPol{\policy}(\tau)} and hence, \typeEval{\conView{\policy}}{e}.
Therefore, $e$ is \TRNI{\policy,\tau}.
\end{proof}

\begin{example}[Typing implies TRNI]
\label{ex:trni:program:trni}
Consider the policy \policyoe. As described in Examples~\ref{ex:trni:con-view} and~\ref{ex:trni:pub-view}, the confidential view
$\conView{\policyoe}$ is $x:\intType, x_f:\intType \rightarrow \intType$ 
and the public view $\pubViewType{\policyoe},\pubViewTerm{\policyoe}$ is 
$\alpha_f, x:\alpha_f, x_f: \alpha_f \rightarrow \intType$.
We look at the program $x_f\ x$. 
We can easily verify that \typeEval{\conView{\policyoe}}{x_f\ x: \intType} and \typeEval{\pubViewType{\policyoe},\pubViewTerm{\policyoe}}{x_f\ x: \intType}.
Therefore, by Theorem~\ref{thm:trni}, the program is \TRNI{\policyoe,\intType}.

%(see Example~\ref{ex:ml:trni:def})

%Similarly, we have that the program $x_f(x_a\ x)$ is well-typed in both views of \policyhash, and in the public view, its types is \intType.
%Thus, the program is \TRNI{\policyhash,\intType}.

%We consider the policy \policyhash\ and the program: $x_f(x_a\ x)$.
%In this program, the confidential input is declassified correctly and hence, it is TRNI.
%
%Indeed, this program is well-typed in both views and in the public view, its type is \intType.
%From Theorem~\ref{thm:trni}, the program is \TRNI{\policyhash, \intType}.
%
%We now check the definition of TRNI with the program $x_f(x_a\ x)$.
%For any $\tuple{\gamma_1,\gamma_2} \in \indval{\policy}$, we have that $\gamma_i(x_f(x_a\ x) = \hash(v_i) \modop 2^{64}$, where $\gamma_i(x) = v_i$ for some $v_i$ ($i \in \set{1,2}$) s.t. $v_1$ and $v_2$ are indistinguishable at $\alpha_{f\circ a}$ (i.e. $\tuple{v_1,v_2} \in \indval{\alpha_{f\circ a}}$).
%From the definition of \indval{\alpha_{f\circ a}} (\redtext{described in the previous example}), $\hash(v_i) \modop 2^{64} \reduce v_i'$ for some $v_i'$ and $v_1' = v_2'$.
%In other words, $\tuple{\hash(v_1) \modop 2^{64}, \hash(v_2) \modop 2^{64}} \in \indval{\intType}$.
\end{example}

\begin{example}
\label{ex:trni:program:non-trni}
If a program is well-typed in the confidential view but not % DN: the grammar is ok, and it fixes spacing
\TRNI{\policy,\tau} for some $\tau$ well-formed in the public view of \policy, then the type of the program in the public view is not $\tau$ or the program is not well-typed in the public view. 
In policy $\policyoe$,  from Example~\ref{ex:trni:program:trni}, the public view is $\alpha_f, x:\alpha_f, x_f: \alpha_f \rightarrow \intType$.
We first look at the program $x$ that is not \TRNI{\policyoe,\intType} since $x$ itself is confidential and cannot be directly declassified.
%%%
In the public view of the policy, 
the type of this program is $\alpha_{f}$ which is not $\intType$.
We now look at the program $x \modop 3$ that
 is not \TRNI{\policyoe,\alpha_f} since it takes indistinguishable inputs at $\alpha_f$ (e.g. $2$ and $4$) and produces results that are not indistinguishable at $\alpha_f$ (e.g. $2 = 2 \modop 3$, $1=4 \modop 3$, and $\tuple{2,1} \not\in \indval{\alpha_f}$). 
We can easily verify that this program is not well-typed in the public view since the type of $x$ in the public view is $\alpha_f$, while \modop\ expects arguments of the \intType\ type.
\end{example}

\begin{remark}[Extension]
\label{rem:extension}
Our encoding can be extended to support richer policies (details in appendix).
To support policies where an input $x$ can be declassified via two declassifiers $f:\intType\rightarrow \tau_f$ and $g:\intType \rightarrow \tau_g$ for some $\tau_f$ and $\tau_g$, we use type variable $\alpha_{f,g}$ as the type for $x$ and use $\alpha_{f,g} \rightarrow \tau_f$ and $\alpha_{f,g} \rightarrow \tau_g$ as types for $x_f$ and $x_g$.
%For indistinguihability, two values are indistinguishable at $\alpha_{f,g}$ if they cannot be distinguished by $f$ and $g$.
% DN hiding the following which seems redundant and was more relevant to actions
%% To support policies where input $x$ can be declassified via equivalent declassifiers, e.g. via $g$ and $f\circ a$ where $g$ and $f\circ a$ are equivalent, $g:\intType \rightarrow \tau$, $a: \intType \rightarrow \intType$, we use $\alpha_{g,f\circ a}$ as the type for $x$, $\alpha_{g,f\circ a} \rightarrow \alpha_f$ as the type for $x_a$ which is corresponding to $a$, and $\alpha_f \rightarrow \tau$ as the type for $x_f$.
To support policies where multiple inputs can be declassified via a declassifier, e.g. inputs $x$ and $y$ can be declassified via $f=\lambda z:\intType \times \intType.(\prj{1}{z}+\prj{2}{z})/2$, we introduce a new term variable $z$ which is corresponding to a tuple of two inputs $x$ and $y$ and we require that only $z$ can be declassified. The type of $z$ is $\alpha_f$ and two tuples \tuple{v_1,v_2} and \tuple{v_1',v_2'} are indistinguishable at $\alpha_f$ when $f\ \tuple{v_1,v_2} = f\ \tuple{v_1',v_2'}$.
%(1) an input can be declassified via multiple declassifiers (e.g. input $x$ can be declassified via $f$ or $g$); (2) an input can be declassified via equivalent declassifiers (e.g. input $x$ can be declassified via $g$ or $f \circ a$ where $g$ and $f\circ a$ are equivalent; (3) multiple inputs can be declassified via a declassifier (e.g. the average of inputs $x$ and $y$ can be declassified).
\end{remark}

\section{Related Work}
\label{sec:related-work}
%% \todo[inline]{Just keeping this for further use in related work BUT you can erase it if you don't need it: 
%% Bowman and Ahmed translate DCC to $F_\omega$ and prove the security theorem of DCC
%% as a consequence of parametricity of $F_\omega$~\cite{Bowman-Ahmed-15-ICFP}.  
%% (The original security proof for DCC does not leverage parametricity~\cite{Abadi-etal-99-POPL}.)  
%% DCC relies on a subsidiary judgment about types, called ``protected-at'', 
%% and the cited works rely on the power of a highly expressive target calculus to encode this judgment.
%% As we discuss in the related work section~\ref{sec:related-work}, prior attempts to formalize security of DCC using parametricity in less powerful target calculi encountered difficulties in connection with 
%% the ``protected-at'' judgment.  
%% Most information flow type systems address practical policies in which the sensitive data is first order; they express and check security more simply than DCC.
%% Our goals do not at all necessitate a system like DCC for policy.
%% }

%In this section, 
\iffalse
We focus on the closest related work regarding 
noninterference, declassification, and connections to the abstraction theorem. 
We refer the interested reader to ~\cite{Sabelfeld-Myers-03-JSAC} for the early history of language-based information flow security, and to ~\cite{sabelfeldSands:jcs2009} for a survey on declassification up to 2009.
\fi
\vspace{-1cm}. 

\paragraph{Typing secure information flow.} 

Pottier and Simonet~\cite{Pottier-Simonet-02-POPL} implement FlowCaml~\cite{PottierSimonet:flowcaml}, 
the first type system for information flow analysis dealing with a real-sized programming language (a large fragment of OCaml), and they prove soundness.  In comparison with our results, we do not consider any imperative features;
they do not consider any form of declassification, 
their type system significantly departs from ML typing, % rules,
and their security proof is not based on an abstraction theorem.
An interesting question is whether their type system can be 
translated to system F or some other calculus with an abstraction theorem.
%~\footnote{When FlowCaml was developed, the abstraction theorem for the module calculus had not yet been proven~\cite{Crary-POPL-17}.}.   
FlowCaml provides type inference for security types.
Our work relies on the Standard ML type system to enforce security. 
Standard ML provides type inference, which endows our approach with an inference mechanism.
%\omittable{
Barthe et al.~\cite{Barthe-etal-08-CSF} propose a modular method to reuse type systems and proofs for noninterference~\cite{htyping} for declassification. They also provide a method to conclude declassification soundness by using an existing noninterference theorem~\cite{thesis}. 
In contrast to our work, their type system significantly departs from standard typing rules, and does not make use of parametricity. 
%\omittable{
Tse and Zdancewic~\cite{Tse-Zdancewic-05-ESOP} %A Design for a Security-typed Language with Certificate-based Declassification
% "concisely expresses the decentralized label model in the polymorphic lambda cal- culus with subtyping"
 propose a security-typed language for robust declassification: declassification cannot be triggered unless there is a digital certificate to assert the proper authority. 
  Their language inherits many features from System F$_{<:}$ and  uses monadic labels as in DCC~\cite{Abadi-etal-99-POPL}.
In contrast to our work, security labels are based on the Decentralized Label Model (DLM)~\cite{Myers-Liskov-00-tosem}, and are not semantically unified with the standard safety types of the language.  
The Dependency Core Calculus (DCC)~\cite{Abadi-etal-99-POPL} expresses security policies using monadic types
indexed on levels in a security lattice with the usual interpretation that flows are only allowed between levels in accordance with the ordering. 
DCC does not include declassification and the noninterference theorem of~\cite{Abadi-etal-99-POPL} is proved from scratch (not leveraging parametricity).
% DN moved from intro:
While DCC is a theoretical calculus, its monadic types fit nicely with the monads and monad transformers used by the Haskell language for computational effects like state and I/O.  
Algehed and Russo~\cite{AlgehedR17} encode the typing judgment of DCC in Haskell 
using closed type families, one of the type system extensions supported by GHC that brings it close to dependent types.
However, they do not prove security.
Compared with type systems, relational logics can specify IF policy and prove more programs secure through semantic reasoning~\cite{NanevskiBG13,BanerjeeNN16,GrimmMFHMPRRSB18,BeckertU18},
but at the cost of more user guidance and less familiar notations.
Aguirre et al~\cite{AguirreBG0S17} use relational higher order logic to prove soundness of DCC essentially by formalizing the semantics of DCC~\cite{Abadi-etal-99-POPL}.

\paragraph{Connections between secure IF and type abstraction.}

Tse and Zdancewic~\cite{Tse-Zdancewic-04-ICFP} translate the recursion-free fragment of DCC to System F.
The main theorem for this translation aims to show that  parametricity of System F implies  noninterference.
Shikuma and Igarashi identify a mistake in the proof~\cite{Shikuma-Igarashi-08-LMCS}; they 
also give a noninterference-preserving translation for a version of DCC to the simply-typed lambda calculus. 
Although they make direct use of a specific logical relation,
their results are not obtained by instantiating a parametricity theorem. 
Bowman and Ahmed~\cite{Bowman-Ahmed-15-ICFP} finally provide a translation from 
the recursion-free fragment of DCC to System F$_{\omega}$,  proving that parametricity implies noninterference,
via a correctness theorem for the translation (which is akin to a full abstraction property).
Bowman and Ahmed's translation makes essential use of the power of System F$_{\omega}$ to encode judgments of DCC.
Algehed and Bernardy~\cite{AlgehedBernardy19} translate a label-polymorphic variant DCC (without recursion) into the calculus of constructions (CC) and prove noninterference directly from a parametricity result for CC~\cite{BernardyanssonPaterson2012}.
The authors note that it is not obvious this can be extended to languages with nontermination or other effects.
Their results have been checked in Agda and the presentation achieves elegance owing to the fact that
parametricity and noninterference can be explicitly defined in dependent type theory; indeed, CC terms can represent proof of parametricity~\cite{BernardyanssonPaterson2012}.
 Our goals do not necessitate a system like DCC for policy,
raising the question of whether a simpler target type system can suffice for security policies expressed differently from DCC.
We answer the question in the affirmative, and believe our results for polymorphic lambda (and for ML) provide transparent explication of noninterference by reduction to parametricity.
%No space: \todo[inline]{? Say the reason is that sensitive input data is first order, even if internally it is manipulated at higher types?}
 The preceding works on DCC are ``translating noninterference to parametricity'' in the sense of translating both programs and types.  
The implication is that one might leverage an existing type checker by translating both a program and its security policy into another program such that it's typability implies the original conforms to policy.
Our work aims to cater more directly for practical application, by minimizing the need to translate the program and hence avoiding the need to prove the correctness of a translation.
 Cruz et al.~\cite{Cruz-etal-17-ECOOP}  show that type abstraction implies relaxed noninterference.
Similar to ours, their definition of relaxed noninterference is a standard extensional semantics, using partial equivalence relations.
This is in contrast with Li and Zdancewic~\cite{Li-Zdancewic-05-POPL} where the semantics is entangled with typability. 
%Their definition of relaxed noninterference is extensional, like ours.
%Moreover, we give relaxed noninterference a standard extensional semantics,
%using partial equivalence relations---this is not a contribution in itself, but it is in contrast with~\cite{Li-Zdancewic-05-POPL} where the semantics is entangled with typability. 
%[\tbupdated: two channels].
\iffalse
Cruz et al.~\cite{Cruz-etal-17-ECOOP} formulation of policy in terms of allowed operations is attractive and seems adaptable to practical languages. 
The idea is close to the use of an explicit ``declass'' operation as in Jif~\cite{Myers:jif} and other works~\cite{sabelfeldSands:jcs2009},
while keeping policy distinct from program rather than embedded in it.
Although the object calculus enjoys a parametricity theorem \cite{Abadietal-96-book}, 
the security proof of Cruz et al is done from scratch.
Moreover they make a significant modification to the type system, introducing faceted types in order to express sensitivity from the perspective of observers at different levels.
This makes good use of subtyping, already present in the object calculus,
but is a concern with respect to our goals of leveraging existing tools and theorems.
We conjecture that our approach can also be applied in the context of the object calculus for relaxed noninterference as defined in~\cite{Cruz-etal-17-ECOOP}.
\fi 
 %Parametricity does not hold for languages with run-time type analysis.

Protzenko et al.~\cite{Protzenkoetal-17-ICFP} propose to use abstract types as the types for secrets and use standard type systems for security. 
This is very close in spirit to our work. 
Their soundness theorem is about  a property called ``secret independence'', very close to noninterference. 
In contrast to our work, there is no declassification and no use of the abstraction theorem.
%\omittable{
 Rajani and Garg~\cite{RajaniGargCSF18} connect fine- and coarse-grained type systems for information flow in a lambda calculus with general references, defining noninterference (without declassification) as a step-indexed Kripke logical relation that expresses indistinguishability.
Further afield, a connection between security and parametricity is made 
by Devriese et al~\cite{DevriesePP18}, featuring a negative result: 
System F cannot be compiled to the the Sumii-Pierce calculus of dynamic sealing~\cite{SumiiP04} 
(an idealized model of a cryptographic mechanism).
Finally, information flow analyses have also been put at the service of parametricity~\cite{Washburn-Weirich-05-LICS}. 
%Washburn and Weirich~\cite{Washburn-Weirich-05-LICS} generalize parametricity in the presence of runtime type analysis using security labels for data structures that should remain confidential in order to hide implementation details . 

%  \vspace{-1cm}. 

\paragraph{Abstraction theorems for other languages.}

%\omittable{
Parametricity remains an active area of study~\cite{SojakovaJ18}.
Vytiniotis and Weirich~\cite{Vytiniotis-Weirich-10-JFP} prove the abstraction theorem for R$_{\omega}$, which extends  F$_{\omega}$ with constructs that are useful for programming with type equivalence propositions. 
Rossberg et al~\cite{RossbergRD14} show another path to parametricity for ML modules, by translating them to $F_\omega$.
Crary's result~\cite{Crary-POPL-17} covers a large fragment of ML but without references and mutable state.
Abstraction theorems have been given for mutable state,
based on ownership types~\cite{BanerjeeNaumann02c} 
and on more semantically based reasoning~\cite{AhmedDR09,DreyerNRB10,banerjeeN13,Timany2017}.

%% Banerjee and Naumann~\cite{BanerjeeNaumann02c} prove an abstraction theorem for a sequential Java-like language, using a form of ownership types to enforce abstraction for dynamically allocated mutable objects, and in later work they prove similar results using program annotations to enforce abstraction~\cite{BanerjeeNaumann04,banerjeeN13}.
%% %(Around the same time, they proved noninterference for a security type system for a similar language, but from scratch rather than via an abstraction theorem~\cite{Banerjee-Naumann-05-JFP,accinfPVStphols}.)
%% Ahmed et al.~\cite{AhmedDR09} develop a step-indexed logical relation for a language with references. Based on that work, Dreyer et al.~\cite{DreyerNRB10} formulate a relational modal logic for proving contextual equivalence for the LADR language that has general recursive types and general ML-style references atop System F.
%% Timany et al~\cite{Timany2017} give a logical relation for a state monad and use it to prove contextual equivalences.

%These works are important steps towards the development of abstraction theorems
%for rich fragments of practical languages.

%Our results give a strong motivation for development of abstraction theorems for richer languages: the immediate benefit of strong security assurance from off the shelf type checkers.

\section{Discussion and Conclusion}
\label{sec:conclusion}
%% \todo[inline]{Keeping this if it some text is needed for conclusion: 
%% The ML result makes a strong connection with a large fragment of a ``real'' language, however, we fall short of our practical goals because our development does not account for programs with high (or multiple level) computation and output.
%% Although this is needed in general, there are many important programs where this does not matter  
%% such as data mining computations using sensitive inputs to calculate aggregate or statistical information,
%% and many mobile apps.
%% To solve this problem we could follow Cruz et al and introduce a notion of faceted types for ML, but this would undercut the goal of leveraging existing tools.
%% Instead we offer our \emph{third contribution}, which is simply to pose this open problem:
%% encode relaxed policies using type abstraction, encompassing multiple level computation and outputs while leveraging an existing parametricity theorem---or demonstrate that it cannot be done.  For practical relevance, the encoding should target a language like ML with efficient type checking.
%% }

% goals from intro
%goal -  \emph{reduce the burden of defining dedicated type checking, inference, and security proofs} for high assurance in programming languages. 
%goal - \emph{to do so for practical policies}
%goal - \emph{minimize the complexity of translation.}

In this work, we show how to express declassification policies by using standard types 
of the simply typed lambda calculus. 
By means of parametricity, we prove that type checking implies relaxed noninterference,
showing a direct connection between declassification  and parametricity. 
 Our approach should be applicable to other languages that have an abstraction theorem (e.g~\cite{banerjeeN13,AhmedDR09,DreyerNRB10,Timany2017}) with the potential benefit of strong security assurance from off-the-shelf type checkers. 
 In particular, we demonstrate (in an appendix) that the results can be extended to a large fragment of ML including general recursion. 
 Although in this paper we demonstrate our results using  confidentiality and declassification, our approach applies as well to integrity
and endorsement, as they have been shown to be information flow properties analog to confidentiality~\cite{Li03informationintegrity,DBLP:conf/popl/FournetR08,crypto2,hcrypto}.  
%Many endorsement policies involve sanitization functions that are explicitly applied in the code, 
%which fits well with the wrapper style of interface.

The simple encodings in the preceding sections do not support computation and output at multiple levels.
For example, consider a policy where $x$ is a confidential input that can be declassified via $f$ and we also want to do the computation $x + 1$ of which the result is at confidential level.
Clearly, $x+1$ is ill-typed in the public interface.
We provide (in an appendix) more involved encodings supporting computation at multiple levels.
%To support computation at multiple levels we develop a monadic encoding inspired by DCC,
%and a public interface that represents policy for multiple levels.
To have an encoding that support multiple levels, we add universally quantified types $\forall \alpha.\tau$ to the language presented in \S\ref{sec:abstraction}.
However, this goes against our goal of minimizing complexity of translation.
Observe that many applications are composed of programs which, individually, do not output at multiple levels; for example, the password checker, and data mining computations using sensitive inputs to calculate aggregate or statistical information.  For these the simpler encoding suffices.

Vanhoef et al.~\cite{Vanhoef-14-CSF} and others have proposed more expressive declassification policies than the ones in Li and Zdancewic~\cite{Li-Zdancewic-05-POPL}: policies that keep state and  can  be written as programs. We speculate that TRNI for stateful declassification policies 
can be obtained for free in a language with state---indeed, our work provides motivation for
development of abstraction theorems for such languages.

\iffalse
Another direction of future work is to formalize our results in Coq,
building on Crary's formalization~\cite{Crary-POPL-17}.
We do not see impediments to doing so, but it is a substantial development and 
in fact has not yet been ported to the current version of Coq~\cite{Crary18}.
\fi

\paragraph{Acknowledgements.}  We thank anonymous reviewers for their suggestions. This work was partially supported by CISC ANR-17-CE25-0014-01, IPL SPAI, the European Union's Horizon 2020 research and innovation programme under grant agreement No 830892,
and US NSF award CNS 1718713.

%Such a formalization would be especially valuable once we have 
%overcome some limitations of our current results, so that the approach 
%can be applied to real programs.

\bibliographystyle{splncs04}

\bibliography{icfemArxiv}

\newpage
\appendix

\section*{Contents of appendix}

\S\ref{sec:extension} describes extensions for more expressive policies, 
in terms of the encoding in \S\ref{sec:trni}.
\S\ref{sec:ml:trni:summary} is an overview of our results for the ML module calculus
and \S\ref{sec:monad} is an overview of our encoding for computation at multiple security levels.
\S\ref{sec:abstraction:proof} provides proofs for \S\ref{sec:abstraction}
and \S\ref{sec:trni:proof} provides proofs for \S\ref{sec:trni}.
The following sections present the ML encoding (\S\ref{sec:ml:language}--\S\ref{app:ml:trni}) 
and the multi-level encoding (\S\ref{sec:multi_level:encoding}--\S\ref{sec:multi_level:proof})
in detail.

\section{Extensions}
\label{sec:extension}
The extensions in this section are corresponding to the encoding in \S\ref{sec:trni}.
%\subsubsection{Declassification policies}
\subsection{Declassification policies}
\label{sec:extension:local-policy}
Variations of our encoding can support richer declassification policies and accept more secure programs.
We consider two ways to extend our encoding.

%The first extension allows policies where a confidential input may be declassified by more than one declassification or combination.
%In the second extension, relations between declassification functions are considered.
%These two extensions are presented with simple policies and they can be generalized.

\paragraph{More declassification functions.} 
The notation in~\cite{Li-Zdancewic-05-POPL} labels an input with a set of declassification functions,
so in general an input can be declassified in more than one way.  
To show how this can be accomodated,
we present an extension for a policy \policy\ where $\varPolicy{\policy} = \set{x}$, and $x$ can be declassified via $f$ or $g$ for some $f$ and $g$, where \typeEval{}{f: \intType \rightarrow \tau_f} and \typeEval{}{g:\intType \rightarrow \tau_g}.
The confidential view and the public view for this policy are as below:
\begin{align*}
\conView{\policy} & = x\comma\intType, x_f\comma \intType \rightarrow \tau_f, x_g\comma \intType \rightarrow \tau_g\\
%%%%%%%%%%%%%%%%%%
\pubViewType{\policy} &= \alpha_{f, g}\\
\pubViewTerm{\policy} &= x\comma \alpha_{f,g}, x_f\comma \alpha_{f,g} \rightarrow \tau_f,\ x_g\comma \alpha_{f,g} \rightarrow \tau_g
\end{align*}

%The \purpletext{term context} \pubViewTerm{\policytwo} in the public view is as below. 
%\begin{itemize}
%\item $x: \alpha_{f \circ a,g \circ b}$,
%\item $x_a: \alpha_{f \circ a,g\circ b} \rightarrow \alpha_f$,
%\item $x_b: \alpha_{f \circ a,g\circ b} \rightarrow \alpha_g$,
%\item $x_f: \alpha_f \rightarrow \tau_1$,
%\item $x_g: \alpha_g \rightarrow \tau_2$,
%\end{itemize}
%
%From \pubViewTerm{\policytwo}, the \purpletext{type context} in the public view is:
%$$\pubViewType{\policytwo} = \{\alpha_{f \circ a, g \circ b}, \alpha_f, \alpha_g\}.$$

We now have a new definition of indistinguishability.
The definition is similar to the one presented in \S\ref{sec:trni}, except that we add a new rule for $\alpha_{f,g}$.
\begin{mathpar}
%%%%%%%%%%%%%%%%%%%
\LabelRule{Eq-Var4}
{\typeEval{}{v_1,v_2:\intType} \\
\tuple{f\ v_1, f\ v_2} \in \indterm{\tau_f} \\ \tuple{g\ v_1, g\ v_2} \in \indterm{\tau_g}
}
{\tuple{v_1,v_2} \in \indval{\alpha_{f,g}} }
\end{mathpar}

With the new encoding and the new definition of indistinguishability, we can define \TRNI{\policy,\tau} as in Definition~\ref{def:trni}.
From the abstraction theorem, we again obtain that for any program $e$, if \typeEval{\conView{\policy}}{e}, and \typeEval{\pubViewType{\policy},\pubViewTerm{\policy}}{e:\tau}, then $e$ is \TRNI{\policy,\tau}.

For example, we consider programs $e_1 = x_f\ x$ and $e_2 = x_g\ x$.
These two programs are well-typed in both views of \policy, and in the public view, their types are respectively $\tau_f$ and $\tau_g$.
Thus, $e_1$ is \TRNI{\policy,\tau_f}, and $e_2$ is \TRNI{\policy,\tau_g}. 

%\typeEval{\conView{\policy}}{e_1:\tau_f} and \typeEval{\conView{\policy}}{e_2:\tau_g}, and  \typeEval{\pubViewType{\policy},\pubViewTerm{\policy}}{e_1:\tau_f} and  \typeEval{\conView{\policy}}{e_2:\tau_g},

%\redtext{In other words, these two programs are well-typed in the confidential view and their types in the public views are respectively $\tau_f$ and $\tau_g$.
%\tealtext{Then it follows that $e_1$ is \TRNI{\policy,\tau_f} and $e_2$ is \TRNI{\policy,\tau_g}}.}

%We also have that if a program is not \TRNI{\policy,\tau} for some $\tau$ well-formed in \pubViewType{\policy}, then the type of $e$ in the public view is not $\tau$ or $e$ is not well-typed in the public view.
%For example, the program $e_3 = x$ is not \TRNI{\policy,\intType} (since $x$ is a confidential input) and its type in the public view is $\alpha_{f\circ a,g\circ b}$ which is not $\intType$.
%The program $e_4 = x_f(x_b\ x)$ is not \TRNI{\policy,\tau_f}  (since $x$ is declassified incorrectly) and this program is not well-typed in the public view.

%\subsection{Extension 2: relations between declassification functions}

\paragraph{Using an equivalent function to declassify.}
In most type systems for declassification, the declassifier function or expression must be identical to the one in the policy.
Indeed, policy is typically expressed by writing a ``declassify'' annotation on the expression~\cite{sabelfeldSands:jcs2009}.
However, the type system presented in \cite[\S~5]{Li-Zdancewic-05-POPL} is more permissive: 
it accepts a declassification if it is semantically equivalent to the policy function, according to 
a given syntactically defined approximation of equivalence.  
Verification tools can go even further in reasoning with semantic equivalence~\cite{NanevskiBG13,GrimmMFHMPRRSB18},
but any automated checker is limited due to undecidability of semantic equivalence.

We consider a policy \policy\ where there are two confidential inputs $x$ and $y$, $x$ can be declassified via $f$, and $y$ can be declassified via $g$, ${f:\intType \rightarrow \tau}$, and ${g:\intType \rightarrow \tau}$ for some $\tau$. 
Suppose that there exists a function $a$ s.t.\ $f \circ a = g$ semantically.
With the encoding in \S\ref{sec:trni}, we accept $g\ y$, or rather $x_g\ y$, but we cannot accept $f(a\ y)$ 
even though it is semantically the same.

%To allow both, we have devised an encoding similar to the extension described above for policies with multiple declassifiers, where $y$ is viewed as a confidential input that can be declassified via $g$ or $f\circ a$.

%\paragraph{Extension.}
%We consider a policy \policy\ where there are two confidential inputs $x$ and $y$, $x$ can be declassified via $f$, and $y$ can be declassified via $g$, ${f:\intType \rightarrow \tau}$, and ${g:\intType \rightarrow \tau}$ for some $\tau$. 
%Suppose that there exists an action $a$ s.t.\ $f \circ a = g$ semantically.
%With the encoding in Section~\ref{sec:trni}, we accept $g\ y$, or rather $x_g\ y$, but we cannot accept $x_f(x_a\ y)$ 
%even though it is semantically the same.
%To allow both, we have devised an encoding similar to the extension described above for policies with multiple declassifiers, where $y$ is viewed as a confidential input that can be declassified via $g$ or $f\circ a$.

To accept programs like $f(a\ y)$, based on the idea of the first extension, we encode the policy as below, where $y$ is viewed as a confidential input that can be declassified via $g$ or $f \circ a$.
Note that in the following encoding, we have two type variables: $\alpha_{g,f\circ a}$ for the confidential input, and $\alpha_f$ for the result of $x_a\ x$ which can be declassified via $f$.

\begin{align*}
\conView{\policy} &= x\comma\intType,\ x_f\comma \intType \rightarrow \tau, y\comma \intType, y_g\comma \intType \rightarrow\tau, y_a\comma \intType \rightarrow \intType\\
\pubViewType{\policy} &= \alpha_f,\ \alpha_{g, f\circ a}\\
\pubViewTerm{\policy} &= x\comma \alpha_{f},\ x_f\comma \alpha_{f} \rightarrow \tau,\ y\comma \alpha_{g, f\circ a},\ y_g\comma \alpha_{g, f\circ a} \rightarrow \tau,\ x_a\comma \alpha_{g, f\circ a} \rightarrow \alpha_{f}
\end{align*}

Indistinguishability for this policy is defined similarly to the one in Section~\ref{sec:trni}, except that we have the following rule for $\alpha_{g,f\circ a}$.
%\newtext{Since $y$ can be declassified via $g$ or $f\circ a$ where $g = f\circ a$, two input values for $y$ are indistinguishable if by applying $g$ on them, the results are indistinguishable at $\tau$}. This requirement is captured by 
\begin{mathpar}
%%%%%%%%%%%%%%%%%%%
\LabelRule{Eq-Var6}
{\typeEval{}{v_1:\intType} \\ \typeEval{}{v_2:\intType} \\
\tuple{g\ v_1, g\ v_2} \in \indterm{\tau}
}
{\tuple{v_1,v_2} \in \indval{\alpha_{g,f\circ a}} }
\end{mathpar}

As in the first extension, we can define TRNI for a type $\tau$ well-formed in \pubViewType{\policy} and we have the free theorem stating that if \typeEval{\conView{\policy}}{e}, and \typeEval{\pubViewType{\policy},\pubViewTerm{\policy}}{e:\tau}, then $e$ is \TRNI{\policy,\tau}.

W.r.t. the new encoding, both $y_g\ y$ and $x_f(x_a\ y)$ are well-typed in the public view.
In other words, we accepts both $y_g\ y$ and $x_f(x_a\ y)$.

%\begin{itemize}
%\item \typeEval{\pubViewType{\policy}, \pubViewTerm{\policy}}{y_g\comma \alpha_{g, f\circ a} \rightarrow \tau} and \typeEval{\pubViewType{\policy}, \pubViewTerm{\policy}}{y\comma \alpha_{g, f\circ a}},
%
%\item \typeEval{\pubViewType{\policy}, \pubViewTerm{\policy}}{x_a:\alpha_{g, f\circ a} \rightarrow \alpha_{f}} and \typeEval{\pubViewType{\policy}, \pubViewTerm{\policy}}{x_f\comma \alpha_{f} \rightarrow \tau}. 
%\end{itemize}

%\newtext{We consider two programs $e_1' = y_g\ y$ and $e_2' = x_f(y_a\ y)$.
%From the encoding, we have that \typeEval{\pubViewType{\policy}, \pubViewTerm{\policy}}{e_1': \tau} and  \typeEval{\pubViewType{\policy}, \pubViewTerm{\policy}}{e_2': \tau}.
%We also have that \typeEval{\conView{\policy}}{e_1': \tau} and  \typeEval{\conView{\policy}}{e_2': \tau}.
%Therefore, $e_1'$ and $e_2'$ are well-typed in both views.}
%
%\newtext{We now consider an arbitrary $\gamma \respect \deltaPol{\policy}(\pubViewTerm{\policy})$ that maps declassifier variables to the desired ones (i.e. $\gamma(x_f) = f$, $\gamma(y_g) = g$, and $\gamma(y_a) = a$).
%Since $g$, $f$, and $a$ are closed functions, we have that $\gamma(e_1') = g\ v = \gamma(g\ y)$ and $\gamma(e_2') = f(a\ v) = \gamma(f(a\ y))$ where $\gamma(y) = v$ for some integer value $v$.
%Therefore, with the new encoding, we accept both $g\ y$ and $f(a\ y)$.}
%
%

Notice that as discussed in \cite{Li-Zdancewic-05-POPL}, the problem of establishing relations between declassification functions in general is undecidable. Thus, the relations should be provided or can be found in a predefined amount of time. Otherwise, the relations are not used in the encoding and programs like $x_f(y_a\ y)$ will not typecheck.

%\subsubsection{Global policies}
\subsection{Global policies}
\label{sec:extension:global}
{The policies considered in \S\ref{sec:trni} and \S\ref{sec:extension:local-policy}  are corresponding to local policies in \cite{Li-Zdancewic-05-POPL}}.
We now consider policies where a declassifier can involve more than one confidential input. 
To be consistent with \cite{Li-Zdancewic-05-POPL}, we call such policies global policies.
For simplicity, in this subsection, we consider a policy \policy\ where there are two confidential inputs, $x_1$ and $x_2$, which can be declassified via $f$ of the type $\intType_1 \times \intType_2 \rightarrow \tau_f$.\footnote{%
We can extend the encoding presented in this section to have policies where different subsets of \varPolicy{\policy} can be declassified and to have more than one declassifier associated with a set of confidential inputs.}
Notice that here we use subscripts for the input type of $f$ to mean that the confidential input $x_i$ is corresponding to $i$-th element of an input of $f$.

%To maintain this corresponednce 

%\redtext{For simplicity, we consider only declassification functions of the type $\intType \times \dots \times \intType \rightarrow \tau_f$ and we do not consider actions}.
%\redtext{To specify that a confidential input $x$ corresponds to the $i$-th parameter of a declassification function, we use a tuple \tuple{l,f} where $l$ is a list of confidential input of length $n$ and $f$ is a declassification function that take a tuple of $n$ elements as its parameter}.
%The basic idea here is that the $i$-th element in the parameter of $f$ is for the $i$-th confidential input in $l$.
%\tbupdated.
%For example, we consider $f = \lambda x: \intType \times \intType.\prj{1}{x} + (2*\prj{2}{x})$.
%In order to specify that the first element in the parameter of $f$ corresponds to $x_1$, and the second element corresponds to $x_2$, we have a tuple \tuple{l,f} where $l = x_1:x_2$.
%
%
%
%\begin{definition}
%\label{def:global-policy}
%A global policy is a mapping from sets of confidential inputs to sets of tuples \tuple{l,f}.
%\end{definition}

%For simplicity, we suppose that a set of confidential inputs can only be declassified via a declassification \footnote{We can use the idea in \S\ref{sec:extension} to have more than one declassifier associated with a set of confidential inputs.}.

\def\policyAve{\ensuremath{\policy_{\text{Ave}}}}

\begin{example}[Average can be declassified]
\label{ex:average:description}
We consider the policy \policyAve\ where there are two confidential inputs $x_1$ and $x_2$ and their average can be declassified.
That is $x_1$ and $x_2$ can be declassified via $f = \lambda x:\intType \times \intType.(\prj{1}{x}+\prj{2}{x})/2$.

In our encoding, we need to maintain the correspondence between inputs and arguments of the declassifier since we want to  prevent laundering attacks \cite{sabelfeldSands:jcs2009}.
A laundering attack occurs, for example, when the declassifier $f$ is applied to $\tuple{x_1,x_1}$,
since then the value of $x_1$ is leaked.
\end{example}

In the general case, to encode the requirement that a specific $n$-tuple of confidential inputs can be declassified via $f$, we introduce a new variable $y$.
The basic idea is that $y$ is corresponding to that $n$-tuple of confidential inputs, $x_i$ cannot be declassified, and only $y$ can be declassified via $f$.
Therefore, the confidential and public views are as below,
where for readability we show the case $n=2$.
%\footnote{
%If we remove all $x_i$ from the public and confidential views, the type system will not accept the program $x_i$ (for any $i \in \set{1,\dots,n}$).
%}.
\begin{align*}
\conView{\policy} & \triangleq 
 \{x_1:\intType, x_2:\intType, y:\intType \times \intType, y_f: \intType \times  \intType \rightarrow \tau_f \}\\
\pubViewType{\policy} & \triangleq 
 \{\alpha_{x_1}, \alpha_{x_2}, \alpha_f \}\\ 
\pubViewTerm{\policy} & \triangleq 
 \{x_1:\alpha_{x_1}, x_2:\alpha_{x_2}, y:\alpha_f, y_f: \alpha_f \rightarrow \tau_f \}
\end{align*}

%\begin{remark}
%\end{remark}

For each $i \in \set{1,\dots, n}$, since $x_i$ cannot be declassified, the indisinguishability for $\alpha_{x_i}$  is the same as the one for $\alpha_x$ described in Fig.~\ref{fig:indistinguishability}.
Since $y$ corresponds to the tuple of confidential inputs and only it can be declassified via $f$, indistinguishability for the type of $y$ in the public view $\alpha_f$ is as below (again, case $n=2$).
%$$
%\LabelRule{Eq-Var5}{\typeEval{}{v_i:\intType} \\ \typeEval{}{v_i';\intType}\\
%\tuple{f\ \tuple{v_1, \tuple{\dots, \tuple{v_{n-1}, v_n}\dots}}, f\ \tuple{v_1', \tuple{\dots, \tuple{v_{n-1}', v_n'}\dots }}} \in \indterm{\tau_f}
%}
%{\tuple{\tuple{v_1, \tuple{\dots, \tuple{v_{n-1}, v_n} \dots }}, \tuple{v_1', \tuple{\dots, \tuple{v_{n-1}', v_n'} \dots }}} \in \indval{\alpha_f}}
%$$

$$
\LabelRule{Eq-Var5}{\typeEval{}{v, v':\intType \times \intType} \\
\tuple{f\ v, f\ v'} \in \indterm{\tau_f}
}
{\tuple{v,v'} \in \indval{\alpha_f}}
$$

%We maintain the correspondence between inputs and argument of declassifier via term substitutions.
We next encode the correspondence between inputs and argument of the declassifier.
We say that a term substitution $\gamma$ is {\em consistent} w.r.t. $\pubViewTerm{\policy}$ if $\gamma \respect \delta_\policy(\pubViewTerm{\policy})$  and in addition, for all $i \in \set{1,2}$, $\prj{i}{(\gamma(y))} = \gamma(x_i)$.
As we can see, the additional condition takes care of the correspondence of inputs and the arguments of the intended declassifier.

We next define the type substitution and indistinguishable term substitutions for \policy.
We say that $\delta_\policy \respect \pubViewType{\policy}$ when $\delta_\policy(\alpha_f) = \intType\times \intType$ and  for all $\alpha_{x_i}$, $\delta_\policy(\alpha_{x_i}) = \intType$. 
%In order to define indistinguishable term substitution, we define an auxiliary notation which is about the relation between $y$ and all $x_i$.
We say that two term substitutions $\gamma_1$ and $\gamma_2$ are indistinguishable w.r.t. \policy\ (denoted by $\tuple{\gamma_1,\gamma_2} \in \indval{\policy}$) if  $\gamma_1$ and $\gamma_2$ are consistent w.r.t. \pubViewTerm{\policy}, 
$\gamma_1(y_f) = \gamma_2(y_f) = f$,
for all other $x \in \dom{\pubViewTerm{\policy}}$, $\tuple{\gamma_1(x),\gamma_2(x)} \in \indval{\pubViewTerm{\policy}(x)}$.

Then we can define \TRNI{\policy,\tau} as in Def.~\ref{def:trni} (except that we use the new definition of indistinguishable term substitutions). We also have the free theorem stating that if $e$ has no type variable and $\typeEval{\pubViewType{\policy},\pubViewTerm{\policy}}{e:\tau}$, then $e$ is \TRNI{\policy,\tau}. The proof goes through without changes.

\begin{example}[Average can be declassified - cont.]
Here we present the encoding for the policy \policyAve\ described in Example~\ref{ex:average:description}.
The confidential and public views for this policy is as below:
\begin{align*}
\conView{\policyAve} & \triangleq 
 \{x_1:\intType, x_2:\intType, y:\intType \times \intType, y_f: \intType \times \intType \rightarrow \intType \}\\
\pubViewType{\policyAve} & \triangleq 
 \{\alpha_{x_1}, \alpha_{x_2}, \alpha_f \}\\ 
\pubViewTerm{\policyAve} & \triangleq 
 \{x_1:\alpha_{x_1}, x_2:\alpha_{x_2}, y:\alpha_f, y_f: \alpha_f \rightarrow \intType \}
\end{align*}

We can easily check that the program $y_f\ y$ is \TRNI{\policyAve,\intType}; it is well-typed in both views, and in the public view its type is $\intType$.
%
%\redtext{Notice that in our formalization, it is not possible to apply any function other than the intended declassifiers on secrets.
%Thus, we discard laundering attacks by types \cite{sabelfeldSands:jcs2009}.
%A laundering attack in this program for example occurs when $x_1$ is made equal to $x_2$ and then the average leaks the secret value $x_2$.}
\end{example}

\section{TRNI for Module Calculus: summary}
\label{sec:ml:trni:summary}
%\noteinline{This section is for ML. Its presentation will be changed (ideas of encoding and wrappers) .}
This section recapitulates the development of \S\ref{sec:trni} but using an encoding suited to the module calculus
of Crary and Dreyer \cite{Dryer-phd,Crary-POPL-17}.\footnote{Our only change is to add \intType\ 
and arithmetic primitives,  for examples.}
It is a core calculus that models Standard ML including higher order generative and applicative functors, sharing constraints (via singleton kinds), and sealing. Sealing ascribes a signature to a module expression and thereby enforces data abstraction.

%/  give more expressive support for abstraction.
% 
%Crary's calculus  is adapted from Dreyer's thesis \cite{Dryer-phd}, and the reader should consult these %references for explanations and motivation.
%{The calculus is adapted from Dreyer's thesis~\cite{Dryer-phd}, and the reader should consult these references for explanations and motivation.

The syntax is in Fig.~\ref{fig:ml:core-calculus:main}.
The calculus has static expressions: kinds ($k$), constructors ($c$) and signatures ($\sigma$), and dynamic expressions: terms ($e$) and modules ($M$).  
The full formal system is given in \S\ref{sec:ml:language}; here we sketch highlights.

\begin{figure}[!t]
\small
%\vspace{-20pt}
\begin{equation*}
\hspace{-30pt}
\begin{aligned}
k ::= &           & \text{kind}\\
      & \unitkind & \text{unit kind}\\
      & \sep \basekind & \text{\purpletext{base kind}}\\
      & \sep \singleton{c} & \text{singleton kind} \\
      & \sep \Pi \alpha:k.k     & \text{dependent functions}\\
      & \sep \Sigma \alpha:k.k  & \text{dependent pairs}\\
c, \tau ::= & & \text{type constructor}\\
            & \alpha & \text{constructor variable} \\
            & \sep \unitcon & \text{unit constructor} \\
            & \sep \lambda \alpha:k.c \sep c\ c & \text{lambda, application}\\
            & \sep \tuple{c,c} & \text{pair}\\
            & \sep \prj{1}{c}  \sep \prj{2}{c} &\text{projection}\\
            & \sep \unittype & \text{unit type}\\
            & \sep \purpletext{\intType} & {\text{int type}}\\
            & \sep \tau_1 \rightarrow\tau_2 & \text{functions} \\
            & \sep \tau_1 \times \tau_2 & \text{products} \\
            & \sep \forall \alpha:k.\tau & \text{universal} \\
            & \sep \exists \alpha:k.\tau & \text{existential}\\
\sigma ::= & & \text{signature}\\
           & \unitsign & \text{unit signature} \\
           & \sep \atksign{k} & \text{atomic signature}\\
           & \sep \atcsign{\tau} & \text{atomic signature}\\
           & \sep \Pign\alpha:\sigma.\sigma & \text{generative functors}\\
           & \sep \Piap\alpha:\sigma.\sigma & \text{applicative functors}\\
           & \sep \Sigma\alpha:\sigma.\sigma & \text{pairs}\\
\Gamma ::= & & \text{context} \\
           & . & \text{empty context} \\
           & \sep \Gamma,\alpha:k & \text{constructor hypothesis}\\
           & \sep \Gamma, x:\tau & \text{term hypothesis} \\
           & \sep \Gamma, \alpha/m:\sigma & \text{module hypothesis}
\end{aligned}
\qquad
\begin{aligned}
e ::= &  & \text{term}\\
     & x & \text{term variable} \\
     & \sep \unitval & \text{unit term}\\
     & \sep {n}        & {\text{integer literal}}\\
     & \sep \lambda x:\tau.e \sep e\ e & \text{lambda, application} \\
     & \sep \tuple{e,e} & \text{pair}\\
     & \sep \prj{1}{e} \sep \prj{2}{e} & \text{projection} \\
     & \sep \Lambda \alpha:k.e & \text{polymorphic fun.} \\
     & \sep e[c] & \text{polymorphic app.}\\
     & \sep \pack{c,e}{\exists \alpha:k.\tau} & \text{existential package} \\
     & \sep \unpack{\alpha,x}{e}{e} & \text{unpack}\\
     & \sep \fix{\tau}{e} & \text{recursion} \\
     & \sep \letexp{x=e}{e} & \text{term binding} \\
     & \sep \letexp{\alpha/m = M}{e} & \text{module binding}\\
     & \sep \extract{M} & \text{extraction}\\
M ::= & & \text{module}\\
      & m & \text{module variable} \\
      & \sep \unitmod & \text{unit module} \\
      & \sep \atcmod{c} & \text{{atomic module}}\\
      & \sep \attmod{e} & \text{{atomic module}}\\
      & \sep \lambdagn \alpha/m:\sigma.M & \text{generative functor}\\
      & \sep M\ M & \text{generative app.} \\
      & \sep \lambdaap \alpha/m:\sigma.M & \text{applicative functor} \\
      & \sep M\appapp M & \text{applicative app.}\\
      & \sep \tuple{M,M} & \text{pair} \\
      & \sep \prj{1}{M} \sep \prj{2}{M} & \text{projection} \\
      & \sep \unpack{\alpha,x}{e}{(M:\sigma)} & \text{unpack}\\
      & \sep \letexp{x=e}{M} & \text{term binding} \\
      & \sep \letexp{\alpha/m=M}{(M:\sigma)} & \text{module binding} \\
      & \sep M\seal\sigma & \text{sealing}
%\Delta ::= & & \text{type environment} \\
%           & . & \text{empty environment} \\                      
\end{aligned}
\end{equation*}
%\vspace{-3ex}
\caption{Module calculus}
\label{fig:ml:core-calculus:main}
%\vspace{-10pt}
\end{figure}

%\paragraph{Kinds and constructors}
The unit kind \unitkind\ has only the unit constructor \unitcon.
The base kind, \basekind, is for types that can be used to classify terms. 
% {\cite{Pierce-02-book}}.  DN: please omit, it's standard 
%We use meta-variables $c$
By convention, we use the metavariable $\tau$ for constructors that are types (i.e. of the kind \basekind).
The singleton kind \singleton{c} %({where $c$ is of the base kind}) 
classifies constructors that are definitionally equivalent to $c$.
In addition, we have higher kinds: dependent functions $\Pi \alpha:k_1.k_2$ and dependent pairs $\Sigma \alpha:k_1.k_2$.
%A constructor $c$ of the kind $\Sigma \alpha:k_1.k_2$ has pairs of constructors where the first component \prj{1}{c} is of the kind $k_1$ and the second component \prj{2}{c} is of the kind $k_2[\alpha\mapsto\prj{1}{c}]$.
%A constructor $c$ of the kind $\Pi \alpha:k_1.k_2$ takes a constructor $c'$ of kind $k_1$ as a parameter and returns a constructor of the kind $k_2[\alpha\mapsto c']$.
%When $\alpha$ does not appear free in $k_2$, we write $k_1 \rightarrow k_2$ instead of $\Pi \alpha:k_1.k_2$ and $k_1 \times k_2$ instead of $\Sigma \alpha:k_1.k_2$.

%\paragraph{Terms and modules} 
The syntax for terms is standard and includes general recursion ($\fix{\tau}{e}$).
Module expressions include unit module ($\unitmod$), 
pairing/projection, 
atomic modules with a single static or dynamic component ($\atcmod{c}$, $\attmod{e}$),
generative and applicative functors ($\lambdagn \alpha/m:\sigma.M$, $\lambdaap \alpha/m:\sigma.M$,
the applications of which are written resp.\ $M_1\ M_2$, $M_1\appapp M_2$),
and unpacking ($\unpack{\alpha,x}{e}{(M:\sigma)}$). 
%pair ($\tuple{M,M}$), projection ($\prj{i}{M}$), 
While term binding is as usual ($\letexp{x=e}{M}$),
the module binding construct is unusual: $\letexp{\alpha/m=M_1}{(M_2:\sigma)}$ binds a pair of names,
where constructor variable $\alpha$ is used to refer to the static part of $M_1$ (and $m$ to the full module).
This is used to handle the phase distinction between compile-time and run-time expressions.  

A signature describes an interface for a module.
Signatures include unit signature, atomic kind and atomic type signature, 
generative and applicative functors, and dependent pairs ($\Sigma\alpha:\sigma_1.\sigma_2$).
%As described in \cite{Dryer-phd}, 
A signature $\sigma$ is {\em transparent} when it exposes the implementation of the static part of modules of $\sigma$.
A signature $\sigma$ is {\em opaque} when it hides some information about the static part of modules of $\sigma$.
The sealing construct, $M\seal\sigma$, ascribes a signature to the module in the sense of enforcing $\sigma$ as an abstraction boundary.  

%Since a module does not appear in static part of a signature, we have only $\alpha$ in dependent signatures (instead of twinned variables, e.g. $\alpha/m$, as in the case of modules).
%In the binding $\alpha:\sigma$ within a dependent signature, $\alpha$ corresponds to the static part of some module of the signature $\sigma$.
%Thus, $\alpha$ has the kind \fstsign{\sigma}.
% which is described in Fig.~\ref{fig:ml:static-portion:kind}.
%\redtext{Whenever $m:\sigma$ and \fstop{m}{c}, it follows that $c$ is a constructor of the kind extracted from the signature $\sigma$ (i.e. c:\fstsign{\sigma})}.

%Applications of generative functors and applicative functors are syntactically distinguished.
%Impure viewed as effect
%Pure 
%Sealing a module 

%M ::= & \ m  
%       \sep   
%       \sep  
%       \sep  
%       \sep  & \text{module}\\
%       &\ \sep   
%       \sep   \\
%       &\ \sep   
%        \sep   
%       \sep \prj{1}{M} \sep \prj{2}{M}  \\
%       &\ \sep   \\
%       &\ \sep   \\
%       &\ \sep   
%       \sep  \\

%\paragraph{Signature}

%\paragraph{Semantics}

\paragraph{Abstraction theorem.} 
The static semantics includes judgments
\typeEval{}{\Gamma\ \ok},
\typeEval{\Gamma}{e:\tau}, \typeEvalP{\Gamma}{M:\sigma}, and \typeEvalI{\Gamma}{M:\sigma} for resp.\ well-formed context, well-typed term, pure well-formed module, and impure well-formed module. 
The pure and impure judgment forms roughly correspond to unsealed and sealed modules; the formal system treats sealing as an effect, introduced by application of a generative functor as well as by the sealing construct.

% ?? Reference for Theorem 6.1. Some intro about Abstraction theorem.

The dynamic semantics is call-by value, with these values:
\begin{align*}
v := & \;  x \sep \unitval \sep n \sep \lambda x:\tau.e \sep \tuple{v,v} \sep \Lambda \alpha:k.e & \text{Term values} \\
&\ \sep \pack{c,v}{\exists \alpha:k.\tau} 
\\
V := & \; m \sep \unitmod \sep \atcmod{c} \sep \attmod{v} \sep \tuple{V,V} & \text{Module values}\\
&\ \sep \lambdagn \alpha/m:\sigma.M \sep \lambdaap \alpha/m:\sigma.M
\end{align*}

%\paragraph{Abstraction theorem} 

The logical relation for the calculus is more complicated than the one in \S\ref{sec:abstraction}.
Even so, the statement of the abstraction theorem for terms is similar to the one in \S\ref{sec:abstraction}.

\begin{theorem}[Abstraction theorem~\cite{Crary-POPL-17}]
\label{thm:ml:abstraction:main}
Suppose that \typeEval{}{\Gamma \ok}.
If \typeEval{\Gamma}{e:\tau}, then \typeEval{\Gamma}{e \logeq e:\tau}.
\end{theorem}

Modules and terms are interdependent, and Crary's theorem
includes corresponding results for pure and for impure modules.  
We express security in terms of sealed modules, but our security proof only relies on the abstraction theorem for expressions.

%\noteinline{}

%Our security proof relies on Crary's  abstraction theorem for terms.\footnote{In~ \cite{Crary-POPL-17}, the abstraction theorem is proven not only for terms but also for module expressions. 
%Terms and modules are interdependent.}  The statement of Crary's abstraction theorem is simple and is similar to the statement of Theorem~\ref{thm:type-abstraction}, i.e. a well-typed term $e$ is related to itself by the defined logical relation.
%(The logical relation and abstraction theorem for terms of the module calculus can be found in appendix.) 

\paragraph{Free theorem: TRNI for the module calculus.}
We present the idea of the encoding for the module calculus.
(Formalization of the encoding  can be found in \S\ref{sec:ml:trni}.)
%To facilitate the presentation of the idea of the encoding
To make the presentation easier to follow, in this section, we write examples in Standard ML (SML).
These examples are checked with SML of New Jersey, version 110.96~\cite{sml}. 

%Before presenting the idea of the encoding for the module calculus, we recall here the abstraction theorem for the module calculus in \cite{Crary-POPL-17}.

For a policy \policy, we construct the public view and the confidential view by using signatures containing type information of confidential inputs and their associated declassifiers.
%Different from \S\ref{sec:trni}, we group type information of confidential inputs and declassifiers in signatures.
In particular, the signature for the confidential view is a transparent signature which exposes the concrete type of confidential input, while the signature for the public view is an opaque one which hides the type information of confidential inputs.
For example, for the policy \policyoe\ (see Example~\ref{ex:policy:odd-even}), we have the following signatures, where \code{transOE} and \code{opaqOE} are respectively the transparent signature for the confidential view and the opaque signature for the public view.

\begin{tabular}{lll}
\begin{lstlisting}
signature transOE = 
 sig
  type t = int
  val x:t 
  val f:t->int
 end
\end{lstlisting} & &
\begin{lstlisting}
signature opaqOE =
 sig
  type t
  val x:t 
  val f:t->int
 end
\end{lstlisting}
\\ 
\end{tabular}

Different from \S\ref{sec:trni}, a program has only a module input which is of the transparent signature and contains all confidential inputs and their declassifiers. A program can use the input via the module variable \code{m}.
For example, for \policyoe, we have the program $\code{m.f\ m.x}$, which is corresponding to the program $x_f\ x$ in Example~\ref{ex:ml:trni:def}.

Using the result in \S\ref{sec:trni}, we define indistinguishability 
as an instantation of the logical relation,
and we say that a term $e$ is \TRNI{\policy,\tau} if on indistinguishable substitutions w.r.t. \policy, it generates indistinguishable outputs at $\tau$.
By using the abstraction theorem \ref{thm:ml:abstraction:main} for terms, we obtain our main result.
% where \pubViewTerm{\policy} is the public view of \policy\
%\footnote{Different from \S\ref{sec:trni}, there is no free type variable in the public view and hence, the public view is only \pubViewTerm{\policy}.  
%}.

\begin{theorem}
\label{thm:ml:free-thm:opaque}
If the type of $e$ in the public view is $\tau$, then $e$ is \TRNI{\policy,\tau}.
\end{theorem}

For the module calculus, when $e$ is well-typed in the public view, $e$ is also well-typed in the confidential view.
Therefore, different from Theorem~\ref{thm:trni} which requires that $e$ has no type variable, Theorem~\ref{thm:ml:free-thm:opaque} simply requires that $e$ is well-typed in the public view.
Our example program \code{m.f\ m.x} typechecks at \intType,
so by Theorem~\ref{thm:ml:free-thm:opaque} it is \TRNI{\policyoe,\intType}.

%\subsection{Usage of our approach}
\paragraph{Usage of our approach.}
We can use our approach with ordinary ML implementations.
In the case that the source programs are already parameterized by one module for their confidential inputs and their declassifiers, then there is no need to modify source programs at all.

For example, we consider \code{program} described below.
Here \code{M} is a module of the transparent signature \code{transOE}. By sealing this module with the opaque signature \code{opaqOE}, we get the module \code{opaqM}.
% DN now said at start of section
% \footnote{The module \code{opaqOE} can only be accessed via the interface defined by \code{opaqOE}. Details about sealing can be found in \cite{Dryer-phd}.}. 
Intuitively, \code{program} is \TRNI{\policyoe,\intType} since the declassifier \code{f} is applied to the confidential input \code{x}.
We also come to the same conclusion from the fact that the type of this program is \intType.

\begin{lstlisting}
structure M = struct
   type t = int
   val x : t = 1
   val f : t -> int = fn x => x mod 2
end
structure opaqM :> opaqOE = M
val program : int = opaqM.f opaqM.x   
\end{lstlisting}

So far our discussion is about open terms but the ML type checker only applies to closed terms. 
In the case that the client program is open 
(i.e. that it can receive any module of the transparent signature as an input, as in the program \code{m.f\ m.x} presented above), in order to be able to type check it for a policy, we need to close it
by putting in a closing context, which we call wrapper.  
For any program $e$ and policy \policy,
the wrapper is written using a functor  as shown below, where 
   \code{opaqP} is the opaque signature for the public view of \policy.
Type $\tau$ is the type at which we want to check security of $e$.
(The identifiers \code{program} and \code{wrapper} are arbitrary.)
\begin{lstlisting}
functor wrapper (structure m: opaqP) =
  struct 
      val program : (*$\tau$*) = (*$e$*)
  end
\end{lstlisting}
Note that $e$ is unchanged.

We have proved that if  the wrapper \wrappub{e} 
is of the signature from \code{opaqP} to $\tau$, then the type of $e$ in the public view is $\tau$.
Therefore, from Theorem~\ref{thm:ml:free-thm:opaque},  $e$ is TRNI at $\tau$.
For instance, for the policy \policyoe, we have that \wrappub{\code{m.f\ m.x}} is of the signature from \code{transOE} to $\intType$ and hence, we infer that the type of \code{m.f\ m.x} in the public view is \intType\ and hence, \code{m.f\ m.x} is \TRNI{\policyoe,\intType}.

%To construct the wrapper for a policy \policy, we first construct \dummy, a dummy module value of the transparent signature of \policy\ where \intType\ is used for the type of confidential inputs, $0$ (arbitrary choice) is used as values for confidential inputs and declassification functions and actions are as described in \policy. 
%For example, for \policyoe, the dummy value $V_\policyoe$ in SML looks like as below, where we omit the implementation of the declassifier \code{f}.
%\begin{lstlisting}
%struct
%  type t = int
%  val x:t = 0
%  val f:t->int = ...
%end
%\end{lstlisting}

%$\Pign \alpha_\policy/m_\policy:\sigmaPol{\policy}.\atcsign{\tau}$

%\subsection{Extension}
\paragraph{Extension.}
As in the case of the simple calculus, our encoding for ML can also be extended for policies where multiple inputs are declassified via a declassifier.
Here, for illustration purpose, we present the encoding for a policy which is inspired by two-factor authentication.

\begin{example}
\label{ex:ml_trni:two_factor}
The policy \policyAut\ involves two confidential passwords and two declassifiers \code{checking1} and \code{checking2} as below, where \code{input1} and \code{input2} are respectively the first input and the second input from a user.
Notice that \code{checking2} takes a tuple of two passwords as its input.
%The policy requires that the result of the comparison between the first password and the first input can be released.
%If the first password is equal to the first input, then the result of the comparison between the second password and the second input can be released.
\begin{lstlisting}
fun checking1(password1:int) = 
  if (password1 = input1) then 1 else 0
fun checking2(passwords:int*int) = 
  if ((#1 passwords) = input1) then 
    if ((#2 passwords) = input2) then 1 else 0
  else 2
\end{lstlisting}

We next construct the confidential view and the public view for the policy. 
{To encode the requirement that two passwords can be declassified via \code{checking2}, we introduce a new variable \code{passwords} which is corresponding to the tuple of the two passwords, and only \code{passwords} can be declassified via \code{checking2}.}
The transparent signature for the confidential view of \policyAut\ is below.
%in Fig.~\ref{fig:ml:policyaut}.

%\pagebreak

%\begin{figure}
%\centering
\begin{minipage}{1\columnwidth}
\begin{lstlisting}
signature transAut = sig
  type t1 = int
  val password1:t1 
  val checking1:t1->int
  type t2 = int
  val password2:t2
  type t3 = int * int 
  val passwords:t3
  val checking2:t3 ->int			
end
\end{lstlisting}
%% \begin{lstlisting}
%% signature opaqAut = sig
%%   type t1
%%   val password1:t1 
%%   val checking1:t1->int
%%   type t2
%%   val password2:t2
%%   type t3
%%   val passwords:t3
%%   val checking2:t3->int			
%% end
%% \end{lstlisting}
\end{minipage}
The signature \lstinline{opaqAut} for the public view is the same except
the types \lstinline{t1}, \lstinline{t2}, and \lstinline{t3} are opaque.

%\caption{Transparent and opaque signatures for \policyAut}
%\label{fig:ml:policyaut}
%\end{figure}

%By using the wrapper as presented in the previous sub-section, we can easily check that programs $\code{m.checking2}\ \code{m.passwords}$ and $\code{m.checking1}\ \code{m.password1}$, where \code{m} is a module variable of the transparent signature \code{transAut}, is \TRNI{\policyAut,\intType}.
We have that the programs $\code{m.checking2}\ \code{m.passwords}$ and $\code{m.checking1}\ \code{m.password1}$, where \code{m} is a module variable of the transparent signature \code{transAut}, have the type \intType\ in the public view. Hence both programs are \TRNI{\policyAut,\intType}.
\end{example}

\section{Computation at multiple security levels: summary}
\label{sec:monad}
%The current encoding only supports two security levels.
The encodings in the preceding sections do not support computation and output at multiple levels.
For example, consider a policy where $x$ is a confidential input that can be declassified via $f$ and we also want to do the computation $x + 1$ of which the result is at confidential level.
Clearly, $x+1$ is ill-typed in the public interface.
To support computation at multiple levels we develop a monadic encoding inspired by DCC,
and a public interface that represents policy for multiple levels.

To have an encoding that support multiple levels, we add universally quantified types $\forall \alpha.\tau$ to the language presented in \S\ref{sec:abstraction} (already present in ML).
In addition, to simplify the encoding, we add the unit type \unitType.
W.r.t. these new types, we have new values: the unit value \unitVal\ of \unitType, and values $\Lambda \alpha.e$ of  $\forall \alpha.\tau$.

To facilitate the presentation of the idea of the encoding, we consider a lattice \lattice\ with three different levels $L$, $M$, $H$ such that  $L \smallst M \smallst H$. 
We also use a simple policy \policy\ with three inputs \hinp, \minp\ and \linp\ at resp. $H$, $M$ and $L$, and \hinp\ can be declassified via $f: \intType \rightarrow \intType$ to $M$.
(The encoding with an arbitrary finite lattice and policy is in \S\ref{sec:multi_level:encoding}.)
For simplicity, we suppose that values on input and output channels are of \intType\ type.

To model multiple outputs we consider programs that return a tuple of values, one component for each output channel.  
To model channel access being associated with different security levels, the output values 
are wrapped, in the form $\lambda x:\unitType . n$. 
To read such a value, an observer needs to provide an appropriate key.  
By giving $x$ an abstract type corresponding to a security level, we can control access.

Similar to the previous sections, we assume that free variables in programs are their inputs,
but now the values will be wrapped integers.  
Intuitively, a wrapped value $v$ can be unwrapped by $\unwrap\ k\ v$, where $k$ is an appropriate key, and $\unwrap\ k\ v$ can be implemented as the application of $v$ on $k$ (i.e. $v\ k$).
Concretely, $k$ will be the unit value $\unitVal$.

We further assume that programs are executed in a context where there are several output channels, each corresponding to a security level. 
A program will compute a tuple of wrapped values, where each element of the output tuple can be unwrapped by using an appropriate key and the unwrapped value is sent to the channel.
In short, we assume the program of interest is executed in a context that wraps its inputs, and also
unwraps each components of the output tuple and sends the value on the corresponding channel.
This assumption is illustrated in the following pseudo program, where $e$ is the program of interest, $o$ is the computed tuple, $\code{Output.Channel}_{l}$ is an output channel at $l$, $k_{l}$ is a key to unwrap value at $l$, %\prjcode{pr}{l} 
and $\prj{l}$ projects the output value for the output channel $l$.

\begin{lstlisting}[escapechar=\%]
let o = e in 
   %$\code{Output.Channel}_{L}$% := %\unwrap% %$k_{L}$% (%\prj{L}{\ o}%)
   %$\code{Output.Channel}_{M}$% := %\unwrap% %$k_{M}$% (%\prj{M}{\ o}%)
   %$\code{Output.Channel}_{H}$% := %\unwrap% %$k_{H}$% (%\prj{H}{\ o}%)
\end{lstlisting}
The keys are not made directly available to $e$, which must manipulate its inputs via an interface described below.

%Therefore, we assume that an output of a program is of the type $\intType_{L} \times \intType_{M} \times \intType_{H}$: 
%{an output of a program contains three values: the first, second and the third value are computed values at $L$, $M$ and $H$ respectively}.

%From these assumptions, \redtext{our results are for programs of the type \redtext{$\wraptype[l_1]{\intType} \times  \dots\times \wraptype[l_n]{\intType}$} in the context with term variables as inputs}.
\paragraph{Encoding.}
Different from the previous sections, we use type variables $\alpha_H, \alpha_M$ and $\alpha_L$ as the types of keys for unwrapping wrapped values at $H$, $M$ and $L$. This idea is similar to the idea in \cite{Li-Zdancewic-05-POPL}. Different from \cite{Li-Zdancewic-05-POPL}, we do not translate DCC and we support declassification.

For an input at $l$ that cannot be declassified (\minp\ or \linp), its type in the public view is $\alpha_l \rightarrow \intType$.
For the input \hinp\ which can be declassified via $f$, we use another type variable (i.e. $\alpha_H^f$) as the type of key to unwrapped values.\footnote{If we use $\alpha_H$ instead, since this input can be declassified to $M$, the indistinguishability will be incorrect: all wrapped values at $H$ are wrongly indistinguishable to observer $M$.}
Similar to the previous sections, we use $\alpha^f$ to encode number values at $H$.
Therefore, the type of \hinp\ in the public view is $\alpha_H^f \rightarrow \alpha^f$.
As we use \unitType\ as the type for key,
in the confidential view the type of \hinp, \minp, and \linp\ is $\unitType \rightarrow \intType$.

As assumed above, a program computes an output which is a tuple of three wrapped values.
Since we use type variables as keys to unwrap wrapped values, the type of outputs of programs we consider is $(\alpha_H \rightarrow \intType) \times (\alpha_M \rightarrow \intType) \times (\alpha_L \rightarrow \intType)$.

%For an input at $l$ that cannot be declassified, its type is $\alpha_l \rightarrow \intType$.
%For an input at $l$ that can be declassified via $f$, its type is $\alpha_H \rightarrow \alpha_f$.
%Therefore, the types of \hinp, \minp, and \linp\ are $\alpha_H \rightarrow \alpha_f$, $\alpha_M \rightarrow \intType$, $\alpha_L \rightarrow \intType$.

To support computing outputs at a level $l$, by using the idea of monad, we have interfaces \comp{l} and \wrap{l} which are the bind and unit expressions for a monad.
In addition, to support converting a wrapped value at $l$ to $l'$ (where $l \smallst l'$), we have interfaces \convup{l}{l'}.
To use $\hinp$ in a computation at $H$, we have $\conv{f}$.
Similar to the previous sections, we have $\hinp_f$ for the declassfier $f$.
The types of there interfaces are described in Fig.~\ref{fig:encoding:multiple}.

\begin{figure}
\begin{align*}
{\tcontextp} =&\ \{\alpha_L, \alpha_M, \alpha_H, {\alpha_H^f},\alpha^f\}\\
{\contextp} =&\ \{\hinp: {\alpha_H^f} \rightarrow \alpha^f, \minp: \alpha_M \rightarrow \intType, \linp: \alpha_L \rightarrow \intType \}\ \cup \\
   &\ \{\comp{l}: \forall \beta_1,\beta_2. (\alpha_l \rightarrow \beta_1) \rightarrow \big( \beta_1 \rightarrow (\alpha_l \rightarrow \beta_2)  \big) \\
   & \hspace{135pt} \rightarrow \alpha_l \rightarrow\beta_2 \sep l \in \lattice \}\ \cup \\
   &\ \{\convup{l}{l'}: \forall \beta. (\alpha_l \rightarrow \beta) \rightarrow (\alpha_{l'} \rightarrow \beta) \sep l,l' \in \lattice \wedge l \smallst l'\}\ \cup\\
   &\ \{\wrap{l}: \forall \beta.\beta \rightarrow \alpha_l \rightarrow \beta \sep l \in \lattice\}\ \cup\\
   &\ \{\hinp_f: (\alpha_H^f \rightarrow \alpha^f) \rightarrow (\alpha_M \rightarrow \intType)\}\ \cup\\
   &\ \{\conv{f}: (\alpha_H^f \rightarrow \alpha^f) \rightarrow (\alpha_H \rightarrow \intType)\}
\end{align*}
\caption{Contexts for \policy}
\label{fig:encoding:multiple}
\end{figure}

%\redtext{We next illustrate this idea by using the defined interface to write a program to compute $h+(l+1)$, $f\ h$ and $l+1$ at resp. $H$, $M$ and $L$}.
%Example~\ref{ex:ni:running_example:comp} shows how to do a computation at a level.
%\redtext{Example~\ref{ex:ni:running_example:conv} shows the complete program}.

\begin{example}
\label{ex:ni:running_example:comp}
We illustrate the idea of the encoding by writing a program that computes the triple $\hinp+\linp+1$, $f\ \hinp$ and $\linp+1$ at resp. $H$, $M$, and $L$.

First, we will have $e_1$ that does the computation at $L$: $\linp + 1$.
{Let $\plusOne = \lambda x:\intType.x+1$}.
In order to use \comp{L}, we first wrap \plusOne\ by using \wrap{L}. 
$$\wrapPlusOne = \lambda x:\intType.\wrap{L}(\plusOne\ x)$$ 

Then $e_1$ is as below: 
%{of the type $\alpha_M \rightarrow \intType$} encoding the requirement that the function \plusOne\ is applied on $\linp$ and the result is at $L$.
$$e_1 = \comp{L}[\intType][\intType]\ {\linp}\ {\wrapPlusOne}$$

Next, we have $e_2$ that does the computation $\hinp + \linp$ at $H$. 
Let \add\ be a function of the type $\intType \rightarrow\intType\rightarrow \intType$.
From \add, we construct $\wrapAddC$ of the type $\intType \rightarrow \alpha_H \rightarrow \intType$, 
where $c$ is a variable of the type \intType.
$${\wrapAddC} = \lambda y:\intType.\wrap{H}(\add\ \ c\ \ y)$$ 

Then we have $e_2$ as below. 
Note that in order to use $\linp$ in \comp{H}, we  need to convert $\linp$ from level $L$ to level $H$ by using \convup{L}{H}.
\[ \begin{array}{l}
    e_2 = \\ 
    \comp{H}[\intType][\intType]\ {\hinp} \ 
      (\lambda c:\intType.(\comp{H}[\intType][\intType]\ (\convup{L}{H}\ {\linp})\ \wrapAddC))
      \end{array}
\]
%$$\hspace{-10pt} e_2 = \comp{H}[\intType][\intType]\ {\hinp} \ (\lambda c:\intType.(\comp{H}[\intType][\intType]\ (\convup{L}{H}\ {\linp})\ \wrapAddC))$$

At this point, we can write the program $e$ that computes $\hinp + \linp + 1$, $f\ \hinp$, and $\linp+1$ at resp. $H$, $M$, and $L$.
$$e = \big(\lambda \linp: \alpha_L \rightarrow \intType.\tuple{e_2 ,\, \tuple{\hinp_f\ \hinp,\, \linp}}\big )\ e_1$$
\end{example}

The implementations of the defined interfaces are straightforward.
For example, on a protected input of type $\unitType \rightarrow \beta_1$ and a continuation of type $\beta_1 \rightarrow (\unitType \rightarrow \beta_2)$, the implementation \compcr\ of \comp{l}\  first unfolds the protected input by applying it to the key \unitVal\ and then applies the continuation on the result.
$$\compcr = \Lambda \beta_1, \beta_2. \lambda x:\unitType \rightarrow \beta_1. \lambda f: \beta_1 \rightarrow (\unitType\rightarrow \beta_2).f(x\ \unitVal)
$$
The implementation of $\wrap{l}$ is $\Lambda \beta . \lambda x:\beta . \lambda \_ : \unitType . \, x$.
The conversions $\convup{l}{l'}$ are implemented by the identity function.

\paragraph{Indistinguishability.}
Different from \S\ref{sec:trni}, indistinguishability is defined for observer $\obs$ ($\obs \in \lattice$). \footnote{Following \cite{Bowman-Ahmed-15-ICFP}, we use \obs\ for observers.}
The indistinguishability relations for \obs\ at type $\tau$ on values  (denoted as \indvalpnew{\tau}{\policy}) is defined as an instance of the logical relation with a careful choice of interpretations for $\alpha_l$ and $\alpha_H^f$.
The idea is that if the observer \obs\ cannot observe data at $l$ (i.e. $l \not\smalleq \obs$), \obs\ does not have any key to unwrap these values and hence, all wrapped values at $l$ are indistinguishable to \obs. Thus, $\alpha_l$ is interpreted as the empty relation for \obs.
Otherwise, since \obs\ has key and can unwrapped values, values wrapped at $l$ are indistinguishable to \obs\ if they are equal and hence, $\alpha_l$ is interpreted as $\{\tuplemt{\unitVal,\unitVal}\}$ (note that the concrete type for key is \unitType).
Since wrapped numbers from \hinp\ are indistinguishable to observer \obs\ when they cannot be distinguish by the declassifier $f$, the interpretation of $\alpha_H^f$ for the observer $M$ is
\{\tuplemt{\unitVal,\unitVal}\}.\footnote{Therefore, if we used $\alpha_H \rightarrow \alpha_f$ for \hinp, all wrapped values from \hinp\ would be indistinguishable to the observer $M$ since this observer do not have any key to open  data at $H$ that cannot be declassified (to observer $M$, the interpretation of $\alpha_H$ is empty.}
By using the idea in the previous sections, based on \indvalpnew{\tau}{\policy}, we define indistinguishability relations for \obs\ at type $\tau$ on terms (denoted as \indtermpnew{\tau}{\policy})).

\paragraph{Free theorem.}
We write $\rho$\ as an environment that maps type variables to its interpretations of form $\tuple{\tau_1,\tau_2,R}$ and maps term variables to tuples of values. 
(This is similar to the formalization for \S\ref{sec:ml:trni}.) 
We write $\rho_L$ and $\rho_R$ for the mappings that map every variable in the domain of $\rho$ to respectively the first element and the second element of the tuple that $\rho$ maps that variable to.
The application of $\rho_L$ (resp. $\rho_R$) to $e$ is denoted by $\rho_L(e)$ (resp. $\rho_R(e)$) (this notations is similar to $\delta\gamma(e)$ in \S\ref{sec:abstraction}).
We write $\rho \respectfullp \policy$ to mean that $\rho$ maps inputs to tuples of indistinguishable values.

By leveraging the abstraction theorem, we get the free theorem saying that a well-typed program $e$ maps indistinguishable inputs to indistinguishable outputs.

\begin{theorem}
\label{thm:trni_free:maintext}
If \typeEval{\tcontextp,\contextp}{e:\tau}, then for any {$\obs \in \lattice$} and $\rho \respectfullp \policy$, 
$$\tuplemt{\rho_L(e),\rho_R(e)} \in \indtermpnew{\tau}{\policy}.$$
\end{theorem}
We state it this way to avoid spelling out the definition of TRNI for this encoding.
The encoding presented here can also be extended to support richer policies described in Remark~\ref{rem:extension}.

\begin{remark}
Our encoding supports declassification while DCC does not.
However, if we consider programs without declassification then DCC is more expressive since in our encoding, to use a wrapped value in a computation at $l$, this value must be wrapped at $l'$ such that $l' \smalleq l$.
However, in DCC, this is not the case due to the definition of the ``protected at'' judgment in DCC: if type $\tau$ is already protected at $l$ then so is $T_{l'}\tau$ for any $l'$.
Therefore, data protected at $l$ can be used in a computation protected at $l'$ even when $l' \not\smalleq l$.
% even when $l \not\smalleq l'$. 
For example, we consider the encoding for a policy defined in a lattice with four levels $\top$, $M_1$, $M_2$, $\bot$ where $\top \smallst M_i \smallst \top$ but $M_1$ and $M_2$ are incomparable.
We can have the following well-typed program in DCC (the program is written in the notations in \cite{Bowman-Ahmed-15-ICFP}).
$$bind\  y = (\eta_{M_1}1)\ in\ \eta_{M_2}(\eta_{M_1} (y + 1)) $$

In our encoding, this program can be rewritten as below, where $f:\intType \rightarrow \alpha_{M_2} \rightarrow \alpha_{M_1} \rightarrow \intType$ is from function $\lambda y:\intType. y+1$ (see a similar function in Example~\ref{ex:ni:running_example:comp}).
$$\comp{M_2}[\intType][\intType]\ (\wrap{M_1}{1}) \ f$$

This program is not well-typed in our encoding.
% DN added without a chance for Minh or Tamara to discuss; hope it's ok:
This feature of DCC, allowing multiple layers of wrapping, is needed to encode state-passing programs (in particular, to encode the Volpano-Smith system for while programs) where low data is maintained unchanged through high computations.  The feature seems unnecessary for functional programs.
%In other word
% We also cannot convert the type of $(\wrap{M_1}{1})$ to $\alpha_{M_2} \rightarrow \intType$ since $M_1$ and $M_2$ are incomparable.
\end{remark}

\section{Proofs of Section~\ref{sec:abstraction}}
\label{sec:abstraction:proof}

\textbf{Lemma~\ref{lem:logeq:related-term}.}
\tealtext{Suppose that $\rho \in \relDel{\delta_1,\delta_2}$ for some $\delta_1$ and $\delta_2$. 
For $i \in \set{1,2}$, it follows that:}
\begin{itemize}
\item \tealtext{if $\tuple{v_1,v_2} \in \valrelation{\tau}{\rho}$, then \typeEval{}{v_i: \delta_i(\tau)}, and}
\item \tealtext{if $\tuple{e_1,e_2} \in \termrelation{\tau}{\rho}$, then \typeEval{}{e_i: \delta_i(\tau)}}.
\end{itemize}
\begin{proof}
The second part of the lemma follows directly from rule FR-Term.
We prove the first part of the lemma by induction on structure of $\tau$.

\emcase{Case 1:} $\intType$.
We consider $\tuple{v_1,v_2} \in \valrelation{\intType}{\rho}$.
From FR-Int, we have that \typeEval{}{v_i:\intType}.
Since $\delta_i(\intType) = \intType$, we have that \typeEval{}{v_i:\delta(\intType)}.

\emcase{Case 2:} $\alpha$.
We consider $\tuple{v_1,v_2} \in \valrelation{\alpha}{\rho}$.
From the FR-Var rule, $\tuple{v_1,v_2} \in \rho(\alpha) \in \relDel{\delta_1(\alpha),\delta_2(\alpha)}$.
From the definition of $\relDel{\delta_1(\alpha),\delta_2(\alpha)}$, we have that \typeEval{}{v_i:\delta_i(\alpha)}.

\emcase{Case 3:} $\tau_1 \times \tau_2$.
We consider $\tuple{v_1,v_2} \in \valrelation{\tau_1 \times \tau_2}{\rho}$.
We then have that $v_1 = \tuple{v_{11},v_{12}}$ for some $v_{11}$ and $v_{12}$ and $v_2 = \tuple{v_{21}, v_{22}}$ for some $v_{21}$ and $v_{22}$.
From FR-Pair, it follows that $\tuple{v_{11},v_{21}} \in \valrelation{\tau_1}{\rho}$ and $\tuple{v_{12},v_{22}} \in \valrelation{\tau_2}{\rho}$.
From IH (on $\tau_1$ and $\tau_2$), we have that \typeEval{}{v_{11}:\delta_1(\tau_1)}, \typeEval{}{v_{21}:\delta_2(\tau_1)}, \typeEval{}{v_{12}:\delta_1(\tau_2)}, and \typeEval{}{v_{22}:\delta_2(\tau_2)}.
From FT-Pair, \typeEval{}{\tuple{v_{11},v_{12}}:\delta_1(\tau_1)\times \delta_1(\tau_2)} and \typeEval{}{\tuple{v_{21},v_{22}}:\delta_2(\tau_1)\times \delta_2(\tau_2)}.
Thus, \typeEval{}{v_1:\delta_1(\tau_1)\times \delta_1(\tau_2)} and \typeEval{}{v_2:\delta_2(\tau_1)\times \delta_2(\tau_2)}.
In other words, \typeEval{}{v_1:\delta_1(\tau_1 \times \tau_2)} and \typeEval{}{v_2:\delta_2(\tau_1 \times \tau_2)}.

\emcase{Case 4:} $\tau_1 \rightarrow \tau_2$.
We consider $\tuple{v_1,v_2} \in \valrelation{\tau_1 \rightarrow \tau_2}{\rho}$.
We now look at arbitrary $\tuple{v_1',v_2'} \in \valrelation{\tau_1}{\rho}$.
From FR-Fun, it follows that $\tuple{v_1v_1',v_2v_2'} \in \termrelation{\tau_2}{\rho}$.
From IH on $\tau_1$ and the second part of the lemma on $\tau_2$, we have that \typeEval{}{v_i':\delta_i(\tau_1)} and \typeEval{}{v_iv_i':\delta_i(\tau_2)} (for $i \in \set{1,2}$).
From FT-App, we have that \typeEval{}{v_i:\delta_i(\tau_1 \rightarrow\tau_2)}.
\end{proof}

%\thmTypeAbstraction*
%{thmTypeAbstraction}
\textbf{Theorem~\ref{thm:type-abstraction}.}
%If \typeEval{\Delta,\Gamma}{e:\tau}, $\delta_1,\delta_2 \respect \Delta$, \bluetext{$\rho \in \relDel{\delta_1,\delta_2}$}, and \redtext{$\tuple{\gamma_1,\gamma_2} \in \valrelation{\Gamma}{\rho}$},  then 
%$$\tuple{\delta_1\gamma_1(e), \delta_2\gamma_2(e)} \in \termrelation{\tau}{\rho}.$$
If \typeEval{\Delta,\Gamma}{e:\tau}, then \typeEval{\Delta,\Gamma}{e\logeq e:\tau}.
\begin{proof}
We prove the theorem by induction on typing derivation.

\emcase{Case 1:} Rule FT-Int.
$$\EmptyRule{~}{\typeEval{}{n:\intType}}$$

From FR-Int, we have that $\tuple{n,n} \in \valrelation{\intType}{\rho}$.
Therefore, from FR-Term, it follows that $\tuple{n,n} \in \termrelation{\intType}{\rho}$.
\purpletext{Hence, $\tuple{\delta_1\gamma_1(n),\delta_2\gamma_2(n)} \in \termrelation{\intType}{\rho}$}.
The proof is closed for this case.

\emcase{Case 2:} Rule FT-Var.
$$\EmptyRule{x:\tau \in \Gamma}{\typeEval{\Delta,\Gamma}{x:\tau}}$$

Since $\tuple{\gamma_1,\gamma_2} \in \valrelation{\Gamma}{\rho}$ and $x \in \dom{\gamma_1}= \dom{\gamma_2}$, {we have that $\tuple{\gamma_1(x), \gamma_2(x)}\in \valrelation{\Gamma(x)}{\rho}$}, and hence $\tuple{\gamma_1(x), \gamma_2(x)}\in \termrelation{\Gamma(x)}{\rho}$.

Since there is no type variable in $x$, we have that $\delta_1\gamma_1(x) = \gamma_1(x)$ and $\delta_2\gamma_2(x) = \gamma_2(x)$.
As proven above, $\tuple{\gamma_1(x), \gamma_2(x)}\in \termrelation{\tau}{\rho}$.
Thus, we have that $\tuple{\delta_1\gamma_1(x), \delta_2\gamma_2(x)}\in \termrelation{\tau}{\rho}$.
The proof is closed for this case.

\emcase{Case 3:} Rule FT-Pair.
$$\EmptyRule{\typeEval{\Delta,\Gamma}{e_1:\tau_1}\\\typeEval{\Delta,\Gamma}{e_2:\tau_2}}
{\typeEval{\Delta,\Gamma}{\tuple{e_1,e_2}:\tau_1 \times \tau_2}}$$

From IH, it follows that 
\begin{itemize}
\item $\tuple{\delta_1\gamma_1(e_1), \delta_2\gamma_2(e_1)} \in \termrelation{\tau_1}{\rho}$, and 
\item $\tuple{\delta_1\gamma_1(e_2), \delta_2\gamma_2(e_2)} \in \termrelation{\tau_2}{\rho}$.
\end{itemize}

%From the definition of \termrelation{\_}{\rho}, we have that 
From the FR-Term, we have that
\begin{itemize}
\item $\tuple{v_{11}, v_{12}} \in \valrelation{\tau_1}{\rho}$, {where $\tuple{\delta_1\gamma_1(e_1),\delta_2\gamma_2(e_1)} \reduce \tuple{v_{11},v_{12}}$}, and 
\item $\tuple{v_{21}, v_{22}} \in \valrelation{\tau_2}{\rho}$, {where $\tuple{\delta_1\gamma_1(e_2),\delta_2\gamma_2(e_2)} \reduce \tuple{v_{21},v_{22}}$}.
\end{itemize}

From FR-Pair, it follows that $\tuple{\tuple{v_{11},v_{21}},\tuple{v_{12},v_{22}}} \in \valrelation{\tau_1\times \tau_2}{\rho}$.
From FR-Term, we have that
$${\tuple{\tuple{\delta_1\gamma_1(e_1),\delta_1\gamma_1(e_2)},\tuple{\delta_2\gamma_2(e_1),\delta_2\gamma_2(e_2)}}  \in \termrelation{\tau_1\times \tau_2}{\rho}}.$$

Thus, $\tuple{\delta_1\gamma_1\tuple{e_1,e_2},\delta_2\gamma_2\tuple{e_1,e_2}}  \in \termrelation{\tau_1\times \tau_2}{\rho}$.
This case is closed.

\emcase{Case 4:} rule FT-Prj.
$$\EmptyRule
{\typeEval{\Delta,\Gamma}{e:\tau_1\times\tau_2}}
{\typeEval{\Delta,\Gamma}{\prj{i}{e}:\tau_i}}
$$

From IH, we have that $\tuple{\delta_1\gamma_1(e),\delta_2\gamma_2(e)} \in \termrelation{\tau_1\times\tau_2}{\rho}$.
Thus, $\tuple{\tuple{v_{11},v_{21}},\tuple{v_{12},v_{22}}} \in \valrelation{\tau_1\times\tau_2}{\rho}$, {where $\delta_1\gamma_1(e) \reduce \tuple{v_{11},v_{21}}$ and $\delta_2\gamma_2(e) \reduce \tuple{v_{12},v_{22}}$}.
From FR-Pair, we have that {$\tuple{v_{11},v_{12}} \in \valrelation{\tau_1}{\rho}$ and $\tuple{v_{21},v_{22}} \in \valrelation{\tau_2}{\rho}$.}

{Since $\delta_1\gamma_1(e) \reduce \tuple{v_{11},v_{21}}$ and $\delta_2\gamma_2(e) \reduce \tuple{v_{12},v_{22}}$}
\begin{itemize}
\item $\prj{1}{\delta_1\gamma_1(e)} \reduce v_{11}$, 
\item $\prj{2}{\delta_1\gamma_1(e)} \reduce v_{12}$,
\item $\prj{1}{\delta_2\gamma_2(e)} \reduce v_{12}$,
\item $\prj{2}{\delta_2\gamma_2(e)} \reduce v_{22}$.
\end{itemize}

Thus, \begin{itemize}
\item $\delta_1\gamma_1(\prj{1}{e}) \reduce v_{11}$, 
\item $\delta_1\gamma_1(\prj{2}{e}) \reduce v_{12}$,
\item $\delta_2\gamma_2(\prj{1}{e}) \reduce v_{12}$,
\item $\delta_2\gamma_2(\prj{2}{e}) \reduce v_{22}$.
\end{itemize}
 
Thus, 
\begin{itemize}
\item $\tuple{\delta_1\gamma_1(\prj{1}{e}),\delta_2\gamma_2(\prj{1}{e})} \reduce \tuple{v_{11},v_{12}}$ and 
\item $\tuple{\delta_1\gamma_1(\prj{2}{e}),\delta_2\gamma_2(\prj{2}{e})} \reduce \tuple{v_{12,v_{22}}}$.
\end{itemize}

As proven above {$\tuple{v_{11},v_{12}} \in \valrelation{\tau_1}{\rho}$ and $\tuple{v_{21},v_{22}} \in \valrelation{\tau_2}{\rho}$.}
From FR-Term, we have that $\tuple{\delta_1\gamma_1(\prj{1}{e}),\delta_2\gamma_2(\prj{1}{e})} \in \termrelation{\tau_1}{\rho}$ and $\tuple{\delta_1\gamma_1(\prj{2}{e}),\delta_2\gamma_2(\prj{2}{e})}  \in \termrelation{\tau_2}{\rho}$.
This case is closed.

\emcase{Case 5:} rule FT-Fun.
$$\EmptyRule
{\typeEval{\Delta,\Gamma,x:\tau_1}{e:\tau_2}}
{\typeEval{\Delta,\Gamma}{\lambda x:\tau_1.e: \tau_1 \rightarrow \tau_2}}
$$

{From IH, we have that $\tuple{\delta_1\gamma_1[x\mapsto v_1](e), \delta_2\gamma_2[x \mapsto v_2](e)} \in \termrelation{\tau_2}{\rho}$, where $\tuple{v_1,v_2} \in \termrelation{\tau_1}{\rho}$}.
Hence, we have that $\tuple{\delta_1\gamma_1(e[x\mapsto v_1]), \delta_2\gamma_2(e[x \mapsto v_2])} \in \termrelation{\tau_2}{\rho}$.
{From the semantics of the language and the FR-Term rule}, it follows that $$\tuple{\delta_1\gamma_1((\lambda x:\tau_1.e)v_1]), \delta_2\gamma_2((\lambda x:\tau_1.e) v_2)} \in \termrelation{\tau_2}{\rho}.$$

Since \tuple{v_1,v_2} are arbitrary, from FR-Fun, we have that $\tuple{\delta_1\gamma_1(\lambda x:\tau_1.e), \delta_2\gamma_2(\lambda x:\tau_1.e)} \in \valrelation{\tau_1\rightarrow\tau_2}{\rho}$, and hence $\tuple{\delta_1\gamma_1(\lambda x:\tau_1.e), \delta_2\gamma_2(\lambda x:\tau_1.e)} \in \termrelation{\tau_1\rightarrow\tau_2}{\rho}$.
The proof is closed for this case.

\emcase{Case 6:} rule FT-App.
$$
\EmptyRule
{\typeEval{\Delta,\Gamma}{e_1:\tau_1 \rightarrow \tau_2} \\
\typeEval{\Delta,\Gamma}{e_2: \tau_1}}
{\typeEval{\Delta,\Gamma}{e_1\ e_2: \tau_2}}
$$

The proof follows from IH and rule FR-Fun.
\end{proof}

\section{Proofs of Section~\ref{sec:trni}}
\label{sec:trni:proof}
\textbf{Lemma~\ref{lem:relation-assignment}}.
\tealtext{$\rhoPol{\policy}\in \relDel{\deltaPol{\policy},\deltaPol{\policy}}$}.
\begin{proof}
We need to prove that for any type variable $\alpha$, if $\tuple{v_1,v_2} \in \rhoPol{\policy}(\alpha)$, then \typeEval{}{v_i:\deltaPol{\policy}(\alpha)} and hence \typeEval{}{v_i:\intType} 
according to the definition (\ref{def:deltaPol}).
This follows directly from the definition of \rhoPol{\policy} and rules Eq-Var1, Eq-Var2, and Eq-Var3.
\end{proof}

\begin{lemma}
\label{lem:canonical-form}
\tealtext{Suppose that \typeEval{}{v:\tau}
It follows that:}
\begin{itemize}
\item \tealtext{if $\tau$ is $\intType$, $v$ is $n$ for some $n$},
\item \tealtext{if $\tau$ is $\tau_1 \rightarrow \tau_2$, $v$ is $\lambda x:\tau_1.e$ for some $e$},
\item \tealtext{if $\tau$  is $\tau_1 \times \tau_2$, it is \tuple{v_1,v_2} for some $v_1$ and $v_2$}.
\end{itemize}
\end{lemma}
\begin{proof}
We prove the case of $\intType$ first by case analysis on typing rules. 
The FT-Int rule gives us the desired result.
The other rules cannot be instantiated with an expression that is a value and of the \intType\ type.

%From the definition of the language, there are only three forms of values $n$, \tuple{v_1,v_2} and $\lambda x:\tau'.e$. 
%We have that \typeEval{}{n:\intType} and the other forms of values cannot be of the type $\intType$.

%\todoinline{DN}{Seems vague.  By cases on typing rules, we observe no other typing rules can be
%instantiated with an expression that is a value and gets type int.
%\newline
%[MN] I updated the proof.
%}

We prove the case of $\tau_1 \rightarrow \tau_2$ by case analysis on typing rules.
The FT-Fun rule gives us the desired result.
The other rules cannot be instantiated with an expression that is a value and of the $\tau_1 \rightarrow \tau_2$ type.

%From the definition of the language, there are only three forms of values $n$, \tuple{v_1,v_2} and $\lambda x:\tau_1.e$.
%We have that \typeEval{}{\lambda x:\tau_1.e} for some $e$ is of the type $\tau_1 \rightarrow \tau_2$ and the other forms of values cannot be of the type $\tau_1 \rightarrow \tau_2$.

The proof for the case $\tau_1 \times \tau_2$ is similar.
\end{proof}

\textbf{Lemma~\ref{lem:logeqpol-indis:eq}}.
%\begin{lemma}
%\label{lem:logeqpol-indis:eq}
\tealtext{It follows that}
\begin{itemize}
\item \tealtext{$\tuple{v_1,v_2} \in \valrelation{\tau}{\rhoPol{\policy}}$ {iff} $\tuple{v_1,v_2} \in \indval{\tau}$, and}
\item \tealtext{$\tuple{e_1,e_2} \in \termrelation{\tau}{\rhoPol{\policy}}$ {iff} $\tuple{e_1,e_2} \in \indterm{\tau}$}.
\end{itemize}
%\end{lemma}
\begin{proof}
We prove the lemma by induction on the structure of $\tau$.

\emcase{Case 1:} \intType.
We consider $\indval{\intType}$ and \valrelation{\intType}{\rhoPol{\policy}}
We have that $\tuple{n,n} \in \valrelation{\intType}{\rhoPol{\policy}}$ iff $\tuple{n,n} \in \indval{\intType}$.

We consider $\indterm{\intType}$ and \termrelation{\intType}{\rhoPol{\policy}}.
We consider \tuple{e_1,e_2}. 
\begin{itemize}
\item If $\tuple{e_1,e_2} \in \indval{\intType}$, from Eq-Term and Eq-Int, there exists $n$ s.t. $e_i \transit^* n$. 
From FR-Int and FR-Term, $\tuple{e_1,e_2} \in \termrelation{\intType}{\rhoPol{\policy}}$.

\item If $\tuple{e_1,e_2} \in \termrelation{\intType}{\rhoPol{\policy}}$, from FR-Term and FR-Int, there exists $n$ s.t. $e_i \transit^* n$. 
From Eq-Int and Eq-Term, $\tuple{e_1,e_2} \in \indterm{\intType}$.
\end{itemize}

\emcase{Case 2:} $\alpha_x$.
First we consider $\indval{\alpha_x}$ and \valrelation{\alpha_x}{\rhoPol{\policy}}.
From the definition of \rhoPol{\policy} and the FR-Var rule, $\tuple{v_1,v_2} \in \valrelation{\alpha_x}{\rhoPol{\policy}}$ iff $\tuple{v_1,v_2} \in \indval{\alpha_x}$.

\emcase{Case 3:} $\alpha_f$.
The proof is similar to the one of Case 2.

%\emcase{Case 4:} $\alpha_{f \circ a}$.
%The proof is similar to the one of Case 2.

\emcase{Case 4:} $\tau_1 \times \tau_2$.
We first consider $\indval{\tau_1 \times \tau_2}$ and \valrelation{\tau_1 \times \tau_2}{\rhoPol{\policy}}.
\begin{itemize}
\item Suppose that $\tuple{v_1,v_2} \in \indval{\tau_1 \times \tau_2}$. 
{From the definition of indistinguishability, we have that \typeEval{}{v_i:\deltaPol{\policy}(\tau_1 \times \tau_2)}}.
From Lemma~\ref{lem:canonical-form}, $v_1 = \tuple{v_{11},v_{12}}$ and $v_2 = \tuple{v_{21},v_{22}}$.
Thus, we have that 
$\tuple{\tuple{v_{11},v_{12}},\tuple{v_{21},v_{22}}} \in \indval{\tau_1 \times \tau_2}$.
From the Eq-Pair rule, we have that $\tuple{v_{11},v_{21}} \in \indval{\tau_1}$ and $\tuple{v_{12},v_{22}} \in \indval{\tau_2}$.
From IH on $\tau_1$ and $\tau_2$, we have that $\tuple{v_{11},v_{21}} \in \valrelation{\tau_1}{\rhoPol{\policy}}$ and $\tuple{v_{12},v_{22}} \in \valrelation{\tau_2}{\rhoPol{\policy}}$.
From the FR-Pair, it follows that $\tuple{\tuple{v_{11},v_{12}},\tuple{v_{21},v_{22}}} \in \valrelation{\tau_1\times\tau_2}{\rhoPol{\policy}}$.

\item Suppose that $\tuple{v_1,v_2} \in \valrelation{\tau_1 \times \tau_2}{\rhoPol{\policy}}$.
From Lemma~\ref{lem:logeq:related-term}, \typeEval{}{v_i:\deltaPol{\policy}(\tau_1\times \tau_2)}.
Thus, we have that $\tuple{\tuple{v_{11},v_{12}},\tuple{v_{21},v_{22}}} \in \valrelation{\tau_1 \times \tau_2}{\rhoPol{\policy}}$.
From the FR-Pair rule, we have that $\tuple{v_{11},v_{21}} \in \valrelation{\tau_1}{\rhoPol{\policy}}$ and $\tuple{v_{12},v_{22}} \in \valrelation{\tau_2}{\rhoPol{\policy}}$.
From IH on $\tau_1$ and $\tau_2$, we have that $\tuple{v_{11},v_{21}} \in \indval{\tau_1}$ and $\tuple{v_{12},v_{22}} \in \indval{\tau_2}$.
From the Eq-Pair, it follows that $\tuple{\tuple{v_{11},v_{12}},\tuple{v_{21},v_{22}}} \in \indval{\tau_1\times\tau_2}$.
\end{itemize}

We now consider $\indterm{\tau_1 \times \tau_2}$ and \termrelation{\tau_1 \times \tau_2}{\rhoPol{\policy}}.
\begin{itemize}
\item Suppose that $\tuple{e_1,e_2} \in \indterm{\tau_1 \times \tau_2}$.
From Eq-Term, we have that $e_i \reduce v_i$ for some $v_i$ and $\tuple{v_1,v_2} \in \indval{\tau_1 \times \tau_2}$.
As proven above, we have that $\tuple{v_1,v_2} \in \valrelation{\tau_1 \times \tau_2}{\rhoPol{\policy}}$.
From FR-Term, $\tuple{e_1,e_2} \in \termrelation{\tau_1 \times \tau_2}{\rhoPol{\policy}}$.

\item Suppose that $\tuple{e_1,e_2} \in \termrelation{\tau_1 \times \tau_2}{\rhoPol{\policy}}$.
From FR-Term, we have that $e_i \reduce v_i$ for some $v_i$ and $\tuple{v_1,v_2} \in \valrelation{\tau_1 \times \tau_2}{\rhoPol{\policy}}$.
As proven above, we have that $\tuple{v_1,v_2} \in \indval{\tau_1 \times \tau_2}$.
From Eq-Term, $\tuple{e_1,e_2} \in \indterm{\tau_1 \times \tau_2}$.
\end{itemize}

\emcase{Case 5:} $\tau_1 \rightarrow \tau_2$.
We first consider $\indval{\tau_1 \rightarrow \tau_2}$ and \valrelation{\tau_1 \rightarrow\tau_2}{\rhoPol{\policy}}.
\begin{itemize}
\item Suppose $\tuple{v_1,v_2} \in \indval{\tau_1 \rightarrow \tau_2}$.
We need to prove that for any $\tuple{v_1',v_2'} \in \valrelation{\tau_1}{\rhoPol{\policy}}$, $\tuple{v_1 v_1', v_2 v_2'} \in \termrelation{\tau_2}{\rhoPol{\policy}}$.

Since $\tuple{v_1',v_2'} \in \valrelation{\tau_1}{\rhoPol{\policy}}$, from IH on $\tau_1$, we have that $\tuple{v_1',v_2'} \in \indval{\tau_1}$.
Since $\tuple{v_1,v_2} \in \indval{\tau_1 \rightarrow \tau_2}$, from Eq-Fun, we have that $\tuple{v_1\ v_1', v_2\ v_2'} \in \indterm{\tau_2}$.
From IH on $\tau_2$, $\tuple{v_1\ v_1', v_2\ v_2'} \in \termrelation{\tau_2}{\rhoPol{\policy}}$.

\item Suppose that $\tuple{v_1,v_2} \in \valrelation{\tau_1 \rightarrow \tau_2}{\rhoPol{\policy}}$.
We need to prove that for any $\tuple{v_1',v_2'} \in \indval{\tau_1}$, $\tuple{v_1 v_1', v_2 v_2'} \in \indterm{\tau_2}$.

Since $\tuple{v_1',v_2'} \in \indval{\tau_1}$, from IH on $\tau_1$, we have that $\tuple{v_1',v_2'} \in \valrelation{\tau_1}{\rhoPol{\policy}}$.
Since $\tuple{v_1,v_2} \in \valrelation{\tau_1 \rightarrow \tau_2}{\rhoPol{\policy}}$, from Eq-Fun, we have that $\tuple{v_1\ v_1', v_2\ v_2'} \in \termrelation{\tau_2}{\rhoPol{\policy}}$.
From IH on $\tau_2$, $\tuple{v_1\ v_1', v_2\ v_2'} \in \indterm{\tau_2}$.

\end{itemize}

We now consider $\indterm{\tau_1 \rightarrow \tau_2}$ and \termrelation{\tau_1 \rightarrow\tau_2}{\rhoPol{\policy}}.
\begin{itemize}
\item Suppose that $\tuple{e_1,e_2} \in \indterm{\tau_1 \rightarrow\tau_2}$.
From Eq-Term, $e_i \reduce v_i$ for some $v_i$ and $\tuple{v_1,v_2} \in \indval{\tau_1 \rightarrow\tau_2}$.
As proven above, we have that $\tuple{v_1,v_2} \in \valrelation{\tau_1\rightarrow \tau_2}{\rhoPol{\policy}}$.
Thus, $\tuple{e_1,e_2} \in \termrelation{\tau_1 \rightarrow\tau_2}{\rhoPol{\policy}}$.

\item Suppose that $\tuple{e_1,e_2} \in \termrelation{\tau_1 \rightarrow\tau_2}{\rhoPol{\policy}}$.
From FR-Term, $e_i \reduce v_i$ for some $v_i$ and $\tuple{v_1,v_2} \in \valrelation{\tau_1 \rightarrow\tau_2}{\rhoPol{\policy}}$.
As proven above, we have that $\tuple{v_1,v_2} \in \indval{\tau_1\rightarrow \tau_2}$.
Thus, $\tuple{e_1,e_2} \in \indterm{\tau_1 \rightarrow\tau_2}$.
\end{itemize}
\end{proof}

\textbf{Lemma~\ref{lem:term-sub:indis:logeq}.}
%\begin{lemma}
%\label{lem:term-sub:indis:logeq}
\tealtext{If $\tuple{\gamma_1,\gamma_2} \in \indval{\policy}$, then $\tuple{\gamma_1,\gamma_2} \in \valrelation{\pubViewTerm{\policy}}{\rhoPol{\policy}}$}.
%\end{lemma}
\begin{proof}
We first prove that $\dom{\gamma_1} =\dom{\gamma_2} = \dom{\pubViewTerm{\policy}}$.
This is directly from the definition of $\tuple{\gamma_1,\gamma_2} \in \indval{\policy}$.

We now need to prove that $\tuple{\gamma_1(x),\gamma_2(x)} \in \valrelation{\pubViewTerm{\policy}}{\rhoPol{\policy}}$ for all $x$.
From the construction of \pubViewTerm{\policy}, we have the following cases.

\emcase{Case 1:} $\pubViewTerm{\policy}(x) = \alpha_x$.
From the assumption, we have that 
$$\tuple{\gamma_1(x),\gamma_2(x)} \in \indval{\alpha_x}.$$

From Lemma~\ref{lem:logeqpol-indis:eq}, it follows that $\tuple{\gamma_1(x),\gamma_2(x)} \in \valrelation{\alpha_x}{\rhoPol{\policy}}$.

\emcase{Case 2:} $\pubViewTerm{\policy}(x) = \alpha_f$
From the assumption, we have that 
$$\tuple{\gamma_1(x),\gamma_2(x)} \in \indval{\alpha_f}.$$

From Lemma~\ref{lem:logeqpol-indis:eq}, it follows that $\tuple{\gamma_1(x),\gamma_2(x)} \in \valrelation{\alpha_f}{\rhoPol{\policy}}$.

\emcase{Case 3:} $\pubViewTerm{\policy}(x_f) = \alpha_f \rightarrow \tau$.
{We need to prove that $\tuple{f,f} \in \valrelation{\alpha_f \rightarrow\tau}{\rhoPol{\policy}}$, where \typeEval{}{f:\intType \rightarrow\tau} for some $\tau$ s.t. \typeEval{}{\tau}}.
From FR-Fun, we need to prove that for any $\tuple{v_1,v_2} \in \valrelation{\alpha_f}{\rhoPol{\policy}}$, $\tuple{f\ v_1, f\ v_2}\in \termrelation{\tau}{\rhoPol{\policy}}$.

We consider an arbitrary $\tuple{v_1,v_2} \in \valrelation{\alpha_f}{\rhoPol{\policy}}$.
From the Eq-Var2 rule, we have that $\tuple{f\ v_1,f\ v_2} \in \indterm{\tau}$.
{From Lemma~\ref{lem:logeqpol-indis:eq}, $\tuple{f\ v_1,f\ v_2} \in \termrelation{\tau}{\rhoPol{\policy}}$}.

%\emcase{Case 4:} $\pubViewTerm{\policy}(x) = \alpha_{f \circ a}$
%From the assumption, we have that 
%$$\tuple{\gamma_1(x),\gamma_2(x)} \in \indval{\alpha_{f\circ a}}.$$
%
%From Lemma~\ref{lem:logeqpol-indis:eq}, it follows that $\tuple{\gamma_1(x),\gamma_2(x)} \in \valrelation{\alpha_{f\circ a}}{\rhoPol{\policy}}$.

\emcase{Case 4:} $\pubViewTerm{\policy}(x_a) = \alpha_{f\circ a} \rightarrow \alpha_{f}$. 
We need to prove that $\tuple{a,a} \in \valrelation{\alpha_{f\circ a} \rightarrow\alpha_f}{\rhoPol{\policy}}$, where \typeEval{}{a:\intType \rightarrow\intType}.
From FR-Fun, we need to prove that for any $\tuple{v_1,v_2} \in \valrelation{\alpha_{f\circ a}}{\rhoPol{\policy}}$, $\tuple{a\ v_1, a\ v_2}\in \termrelation{\alpha_f}{\rhoPol{\policy}}$.

We consider an arbitrary $\tuple{v_1,v_2} \in \valrelation{\alpha_{f\circ a}}{\rhoPol{\policy}}$.
From the Eq-Var3 rule, we have that $\tuple{a\ v_1,a\ v_2} \in \indterm{\alpha_f}$.
{From Lemma~\ref{lem:logeqpol-indis:eq}, $\tuple{a\ v_1,a\ v_2} \in \termrelation{\alpha_f}{\rhoPol{\policy}}$}.
\end{proof}

\iffalse
\section{Remarks and Extensions}
\label{sec:extension}
\input{extension}
\fi

\section{Module Calculus}
\label{sec:ml:language}
%\vspace{-5pt}
\subsection{Syntax and semantics}
This section presents a module calculus, essentially the same as that of Crary~\cite{Crary-POPL-17} except that we add \intType\ for integers.  
Crary's calculus is adapted from Dreyer's thesis~\cite{Dryer-phd}, and  
the reader should consult these references for explanations and motivation.
%\toapp{We do aim to include here all the technical details, and therefore draw heavily on \cite{Crary-POPL-17}.}
%Missing formal definitions in this section can be found in the appendix.
As described in \S\ref{sec:ml:trni:summary}, the calculus has static expressions: kinds ($k$), constructors ($c$) and signatures ($\sigma$), and dynamic expressions: terms ($e$) and modules ($M$).
The syntax is in Fig.~\ref{fig:ml:core-calculus:main}.

\paragraph{Kinds and constructors.}
The unit kind \unitkind\ has only the unit constructor \unitcon.
The kind \basekind\ have base types that can be used to classify terms. 
% {\cite{Pierce-02-book}}.  DN: please omit, it's standard 
The singleton kind \singleton{c} ({where $c$ is of the base kind}) has constructors that are definitionally equivalent to $c$.
In addition, we have higher kinds: dependent functions $\Pi \alpha:k_1.k_2$ and dependent pairs $\Sigma \alpha:k_1.k_2$.
A constructor $c$ of the kind $\Sigma \alpha:k_1.k_2$ has pairs of constructors where the first component \prj{1}{c} is of the kind $k_1$ and the second component \prj{2}{c} is of the kind $k_2[\alpha\mapsto\prj{1}{c}]$.
A constructor $c$ of the kind $\Pi \alpha:k_1.k_2$ takes a constructor $c'$ of kind $k_1$ as a parameter and returns a constructor of the kind $k_2[\alpha\mapsto c']$.
When $\alpha$ does not appear free in $k_2$, we write $k_1 \rightarrow k_2$ instead of $\Pi \alpha:k_1.k_2$ and $k_1 \times k_2$ instead of $\Sigma \alpha:k_1.k_2$.
We use the metavariable $\tau$ for constructors that are types (i.e. of the kind \basekind).

\paragraph{Terms and modules.} 
The syntax for terms is standard.
Modules can be unit module, static atomic module, dynamic atomic module, {generative functor}, {applicative functor}, application, pair, projection, unpack, term binding, module binding, and sealing.
Applications of generative functors and applicative functors are syntactically distinguished.

%\paragraph{Terms and modules.} 
Abstraction is introduced by {\em sealing}: in the module $M \seal \sigma$, access to the component $M$ is limited to the interface $\sigma$.
A term is extracted from a module $\attmod{e}$ by \extract{\attmod{e}}.
To extract the static part of a module, we have the operation $\typeEval{\Gamma}{\fstop{M}{c}}$ meaning that the static part of $M$ is $c$ in $\Gamma$. 
%\newtext{The definition of the operation is described in Fig.~\ref{fig:ml:static-portion:constructor}}.
When the context is empty, we write \fstoptwo{M} for $c$ where \typeEval{}{\fstop{M}{c}}.

Notice that every module variable is associated with a constructor variable that represents its static part \cite{Crary-POPL-17}.
The relation is maintained by twinned variables: $\alpha/m: \sigma$ meaning that $m$  has signature $\sigma$ and its static part is $\alpha$ which is of the kind \fstsign{\sigma}, where for any signature $\sigma$, \fstsign{\sigma} extracts the information about kind from $\sigma$.
%defined in Fig.~\ref{fig:ml:static-portion:kind}.
Whenever $m:\sigma$ and \fstop{m}{c}, it follows that $c$ is of the kind \fstsign{\sigma}.

\begin{figure}
\vspace{-10pt}
\small
\begin{mathpar}
\EmptyRule{\alpha/m \in \dom{\Gamma}}
{\typeEval{\Gamma}{\fstop{m}{\alpha}}}\and
%%%%%%%%%%%%%%
\EmptyRule{~}{
    \typeEval{\Gamma}{\fstop{\unitmod}{\unitcon}}
}\and
%%%%%%
\EmptyRule{~}{
    \typeEval{\Gamma}{\fstop{\atcmod{c}}{c}}
}\and
%%%%%%
\EmptyRule{~}{
    \typeEval{\Gamma}{\fstop{\attmod{e}}{\unitcon}}
}\and
%%%%%%
\EmptyRule{~}{
    \typeEval{\Gamma}{\fstop{\lambdagn \alpha/m:\sigma.M}{\unitcon}}
}\and
%%%%%%%%
\EmptyRule{
    \typeEval{\Gamma,\alpha/m:\sigma}{\fstop{M}{c}}
}
{\typeEval{\Gamma}{\fstop{\lambdaap \alpha/m:\sigma.M}}{\lambda \alpha:\fstsign{\sigma}.c}}\and
%%%%%%%%
\EmptyRule{
    \typeEval{\Gamma}{\fstop{M_1}{c_1}}\\
    \typeEval{\Gamma}{\fstop{M_2}{c_2}}
}
{\typeEval{\Gamma}{\fstop{M_1\appapp M_2}{c_1c_2}}}\and
%%%%%%%%
\EmptyRule{
    \typeEval{\Gamma}{\fstop{M_1}{c_1}}\\
    \typeEval{\Gamma}{\fstop{M_2}{c_2}}}
{\typeEval{\Gamma}{\fstop{\tuple{M_1,M_2}}{\tuple{c_1,c_2}}} }\and
%%%%%%%%
\EmptyRule{
    \typeEval{\Gamma}{\fstop{M}{c}}
}
{\typeEval{\Gamma}{\fstop{\prj{i}{M}}{\prj{i}{c}}}}\and
%%%%%%%%
%\EmptyRule{
%    \typeEval{\Gamma}{\fstop{M}{c}}
%}
%{\typeEval{\Gamma}{\fstop{\prj{2}{M}}{c}}}\and
%%%%%%%%
\EmptyRule{
    \typeEval{\Gamma}{\fstop{M}{c}}
}
{\typeEval{\Gamma}{\fstop{\letexp{x=e}{M}}{c}}}
\end{mathpar}
\vspace*{-4ex}
\caption{Extracting constructor information from modules}
\label{fig:ml:static-portion:constructor}
\vspace{-10pt}
\end{figure}

\begin{figure}
\small
\begin{align*}
%\fstsign{1} & \triangleq \unitkind \\
%\fstsign{\atksign{k}} & \triangleq k\\
%\fstsign{\atcsign{\tau}} & \triangleq \unitkind\\
%\fstsign{\Pign \alpha:\sigma_1.\sigma_2} & \triangleq \unitkind\\
%\fstsign{\Piap \alpha:\sigma_1.\sigma_2} & \triangleq \Pi \alpha:\fstsign{\sigma_1}.\fstsign{\sigma_2}\\
%\fstsign{\Sigma \alpha:\sigma_1.\sigma_2} & \triangleq \Sigma \alpha:
%\fstsign{\sigma_1}.\fstsign{\sigma_2}
\fstsign{1} & \triangleq \unitkind 
    & \fstsign{\Pign \alpha:\sigma_1.\sigma_2} & \triangleq \unitkind\\
\fstsign{\atksign{k}} & \triangleq k
    & \fstsign{\Piap \alpha:\sigma_1.\sigma_2} & \triangleq \Pi \alpha:\fstsign{\sigma_1}.\fstsign{\sigma_2}\\
\fstsign{\atcsign{\tau}} & \triangleq \unitkind & 
    \fstsign{\Sigma \alpha:\sigma_1.\sigma_2} & \triangleq \Sigma \alpha:
\fstsign{\sigma_1}.\fstsign{\sigma_2}
\end{align*}
%\vspace*{-4ex}
\caption{Extracting kind information from signatures}
\label{fig:ml:static-portion:kind}
%\vspace{-10pt}
\end{figure}

\paragraph{Signature.}
Signatures include unit signature, atomic kind signature, atomic type signature, signatures for generative functors, applicative functors and pairs.
Since a module does not appear in static part of a signature, we have only $\alpha$ in dependent signatures (instead of twinned variables, e.g. $\alpha/m$, as in the case of modules).
In the binding $\alpha:\sigma$ within a dependent signature, $\alpha$ corresponds to the static part of some module of the signature $\sigma$.
Thus, $\alpha$ has the kind \fstsign{\sigma}.
% which is described in Fig.~\ref{fig:ml:static-portion:kind}.
%\redtext{Whenever $m:\sigma$ and \fstop{m}{c}, it follows that $c$ is a constructor of the kind extracted from the signature $\sigma$ (i.e. c:\fstsign{\sigma})}.

As described in \cite{Dryer-phd}, a signature $\sigma$ is {\em transparent} when it exposes the implementation of the static part of modules of $\sigma$.
A signature $\sigma$ is {\em opaque} when it hides some information about the static part of modules of $\sigma$.

\vspace{-5pt}
\begin{example}
\label{ex:ml:module-signature}
We consider the following module and signatures. 
Suppose that $f = \lambda x:\intType.e$ for some $e$ which is a closed function of the type $\intType \rightarrow \tau_f$ for some closed $\tau_f$.

%\begin{figure}
%\vspace{-10pt}
\begin{tabular}{@{\!\!}lllll}
\begin{minipage}{0.35\columnwidth} % was .3
\small
\begin{lstlisting}
structure M = 
  struct
    type t = int
    val x:t = 0
    val f:t->int = ...
  end
\end{lstlisting} 
% DN changed to be consistent with Crary and Dreyer
%        type (*$\tau$*) = int
%        val x:(*$\tau$*) = 0
%        val f:(x:(*$\tau$*)) = e
\end{minipage} &  & 
\begin{minipage}{0.33\columnwidth}  % was .27
\small
\begin{lstlisting}
signature (*$\sigma_T$*) = 
  sig
    type t = int
    val x:int
    val f:int -> (*$\tau_f$*)
  end    
\end{lstlisting}
%        type (*$\tau$*) = int
\end{minipage} & &
\begin{minipage}{0.33\columnwidth} % was .27 
\small
\begin{lstlisting}
signature (*$\sigma_O$*) = 
  sig
    type t
    val x:t
    val f:t -> (*$\tau_f$*)
  end    
\end{lstlisting}
%        type (*$\tau$*)
%        val x: (*$\tau$*)
%        val f: (*$\tau$*) -> (*$\tau_f$*)
\end{minipage}
\end{tabular}

%\caption{Module and Signature - Example}
%\label{fig:ml:module-signatures}
%\vspace{-10pt}
%\end{figure}

In the module calculus, $M$ is $\tuple{\atcmod{\intType},\tuple{\attmod{0},\attmod{\lambda x:\intType.e}}}$.
Using abbreviations, 
$\sigma_T$ is $\tuple{\atksign{S(\intType)},\tuple{\atcsign{\intType}, \atcsign{\intType \rightarrow \tau_f}}}$ 
and $\sigma_O$ is 
$\Sigma \alpha:\atksign{\basekind}.\tuple{\atcsign{\alpha},\atcsign{\alpha \rightarrow \tau_f}}$.\footnote{Expanding abbreviations, 
$\sigma_T$ is $\Sigma \alpha:\atksign{S(\intType)}.\Sigma \beta:\atcsign{\intType}.\atcsign{\intType \rightarrow \tau_f}$ and 
$\sigma_O$ is $\Sigma \alpha:\atksign{\basekind}.\Sigma \beta:\atcsign{\alpha}.\atcsign{\alpha \rightarrow \tau_f}$.}
%% In the module calculus, $M = \tuple{\atcmod{\intType},\tuple{\attmod{0},\attmod{\lambda x:\intType.e}}}$ , $\sigma_T = \Sigma \alpha:\atksign{S(\intType)}.\Sigma \beta:\atcsign{\intType}.\atcsign{\intType \rightarrow \tau_f}$, and 
%% $\sigma_O = \Sigma \alpha:\atksign{\basekind}.\Sigma \beta:\atcsign{\alpha}.\atcsign{\alpha \rightarrow \tau_f}$.
%% Abbreviated, the signatures are 
%% $\tuple{\atksign{S(\intType)},\tuple{\atcsign{\intType}, \atcsign{\intType \rightarrow \tau_f}}}$
%% and 
%% $\Sigma \alpha:\atksign{\basekind}.\tuple{\atcsign{\alpha},\atcsign{\alpha \rightarrow \tau_f}}$.
The signature $\sigma_T$ is a transparent signature of $M$ since $\sigma_T$ exposes the information of the static part of $M$,
as $\atksign{S(\intType)}$.
The signature $\sigma_O$ is an opaque signature of $M$ since $\sigma_O$ hides the information of the static part of $M$,
as $\atksign{\basekind}$.
\end{example}

\ifshow
\begin{example}
\label{ex:ml}
We consider the following module with a type component $\tau$ and three values.

\begin{lstlisting}
struct
    type (*$\tau$*) = int
    val secret:(*$\tau$*) = 10
    fun f(x:(*$\tau$*)) = ...
    fun a(x:(*$\tau$*)) = ...
end    
\end{lstlisting}

In the core calculus, the module is represented as $M = \tuple{\intType, \tuple{10, \tuple{\text{definition of $f$},\text{definition of $a$}}}}$.
For simplicity, we write $M$ as \tuple{\intType, 10, \text{definition of $f$}, \text{definition of $a$}} and we write $M.\tau$, $M.secret$, $M.f$, and $M.a$ to project components from $M$.

A signature of $M$ is \tuple{\atksign{S(\intType)},\atcsign{\intType},\atcsign{\intType \rightarrow \intType}, \atcsign{\intType \rightarrow\intType}}.
\redtext{Another signature of $M$ where we do not know detailed information about $\tau$ is \tuple{\atksign{\basekind},\atcsign{\intType},\atcsign{\intType \rightarrow \intType}, \atcsign{\intType \rightarrow\intType}}}.
\end{example}
\fi

\begin{figure*}[!t]
%\vspace{-20pt}
\centering
\small
\begin{tabular}{lllll}
\typeEval{}{\Gamma\ \ok} & $\quad\quad$ & {Well-formed context} & $\quad\quad$ & (Fig~\ref{fig:apd:well-formed-context}) \\
\typeEval{\Gamma}{k:\kind} & & {Well-formed kind} && (Fig~\ref{fig:apd:well-formed-kind})\\
\typeEval{\Gamma}{k_1 \equiv k_2: \kind} & & {Kind equivalence} && (Fig~\ref{fig:apd:kind-equivalence})\\
\typeEval{\Gamma}{k_1 \leq k_2:\kind} & & {Subkinding} && (Fig~\ref{fig:apd:subkinding})\\
\typeEval{\Gamma}{c:k} & & Well-formed constructor && (Fig~\ref{fig:apd:well-formed-constructor})\\
\typeEval{\Gamma}{c_1 \equiv c_2:\kind} & & Constructor equivalence && (Fig~\ref{fig:apd:constructor-equivalence})\\
%\typeEval{\Gamma}{\tau} & & Well-formed type \\
\typeEval{\Gamma}{e:\tau} & & \purpletext{Well-typed term} && (Fig~\ref{fig:apd:well-typed-term})\\
\typeEval{\Gamma}{\sigma:\sign} & & Well-formed signature && (Fig~\ref{fig:apd:well-formed-signature})\\
\typeEval{\Gamma}{\sigma_1 \equiv \sigma_2: \sign} & & Equivalence signature && (Fig~\ref{fig:apd:signature-equivalence})\\
\typeEval{\Gamma}{\sigma_1 \leq \sigma_2:\sign} & & {Subsignature} && (Fig~\ref{fig:apd:subsignature})\\
\typeEvalP{\Gamma}{M:\sigma} & & \purpletext{Pure well-formed module} && (Fig~\ref{fig:apd:well-formed-module})\\
\typeEvalI{\Gamma}{M:\sigma} & & \purpletext{Impure well-formed module} && (Fig~\ref{fig:apd:well-formed-module})
\end{tabular}
\vspace*{-0.7ex}
\caption{Judgment forms in the static semantics}
\label{fig:ml:judgment-forms}
\vspace{-15pt}
\end{figure*}

\paragraph{Static semantics.}
The judgment forms in the static semantics are described in Figure~\ref{fig:ml:judgment-forms}. 
%The full description can be found in the appendix.
W.r.t. the static semantics, for the signatures described in Example~\ref{ex:ml:module-signature}, it follows that the transparent signature $\sigma_T$ is a sub-signature of the opaque signature $\sigma_O$.

%\subsection{Static semantics}
\begin{figure}[!h]
\begin{mathpar}
\footnotesize
\LabelRule{\WfeE}{~}
{\typeEval{}{.\ \ok}}\and
%%%%
\LabelRule{\WfeC}{\typeEval{}{\Gamma\ \ok} \\ 
\typeEval{\Gamma}{k:\kind}
}
{\typeEval{}{\Gamma,\alpha:k\ \ok}}\and
%%%%%%%
\LabelRule{\WfeT}{
    \typeEval{}{\Gamma\ \ok}\\
    \typeEval{\Gamma}{\tau:\basekind}
}
{\typeEval{}{\Gamma,x:\tau\ \ok}}\and
%%%%%%%
\LabelRule{\WfeM}{
    \typeEval{}{\Gamma\ \ok}\\
    \typeEval{\Gamma}{\sigma:\sign}
}
{\typeEval{}{\Gamma,\alpha/m:\sigma\ \ok}}\\
\end{mathpar}
\caption{\text{Well-formed context}\ \typeEval{}{\Gamma\ \ok}}
\label{fig:apd:well-formed-context}
\end{figure}

\begin{figure}[!h]
\begin{mathpar}
\footnotesize
\LabelRule{\WfkB}{~}
{\typeEval{\Gamma}{\basekind:\kind}}\and
%%%%
\LabelRule{\WfkS}{\typeEval{\Gamma}{c:\basekind}
}
{\typeEval{\Gamma}{S(c):\kind}}\and
%%%%%%%
\LabelRule{\WfkU}{~
}
{\typeEval{\Gamma}{\unitkind:\kind}}\and
%%%%%%%
%%%%%%%
\LabelRule{\WfkF}{
    \typeEval{\Gamma}{k_1:\kind}\\
    \typeEval{\Gamma,\alpha:k_1}{k_2:\kind}
}
{\typeEval{\Gamma}{\Pi \alpha:k_1.k_2:\kind}}\and
%%%%%%%
\LabelRule{\WfkP}{
    \typeEval{\Gamma}{k_1:\kind}\\
    \typeEval{\Gamma,\alpha:k_1}{k_2:\kind}
}
{\typeEval{\Gamma}{\Sigma \alpha:k_1.k_2:\kind}}
\end{mathpar}

\caption{\text{Well-formed kind}\ \typeEval{\Gamma}{k:\kind}}
\label{fig:apd:well-formed-kind}
\end{figure}

%%%%%%%%%%%%%%%%%%%%%%%%%%%%%%%%%%%%%%%%%%555
%%%%%%%%%%%%%%%%%%%%%%%%%%%%%%%%%%%%%%%%%%%%%
\begin{figure}[!h]
\begin{mathpar}
\footnotesize
\LabelRule{\EqkR}{
    \typeEval{\Gamma}{k:\kind}   
}
{\typeEval{\Gamma}{k \equiv k:\kind}}\and
%%%%%%%%%%%%%%%
\LabelRule{\EqkS}{
    \typeEval{\Gamma}{k_1 \equiv k_2:\kind}   
}
{\typeEval{\Gamma}{k_2 \equiv k_1:\kind}}\and
%%%%%%%%%%%%%%%
\LabelRule{\EqkT}{
    \typeEval{\Gamma}{k_1 \equiv k_2:\kind} \\
    \typeEval{\Gamma}{k_2 \equiv k_3:\kind}   
}
{\typeEval{\Gamma}{k_1 \equiv k_3:\kind}}\and
%%%%%%%%%%%%%%%
\LabelRule{\EqkSg}{
    \typeEval{\Gamma}{c_1 \equiv c_2:\basekind}   
}
{\typeEval{\Gamma}{S(c_1) \equiv S(c_2):\kind}}\and
%%%%%%%%%%%%%%%
\LabelRule{\EqkF}{
    \typeEval{\Gamma}{k_1 \equiv k_2:\kind} \\
    \typeEval{\Gamma,\alpha:k_1}{k_3 \equiv k_4:\kind}     
}
{\typeEval{\Gamma}{\Pi \alpha:k_1.k_3 \equiv \Pi \alpha:k_2.k_4:\kind}}\and
%%%%%%%%%%%%%%%
\LabelRule{\EqkP}{
    \typeEval{\Gamma}{k_1 \equiv k_2:\kind} \\
    \typeEval{\Gamma,\alpha:k_1}{k_3 \equiv k_4:\kind}     
}
{\typeEval{\Gamma}{\Sigma \alpha:k_1.k_3 \equiv \Sigma \alpha:k_2.k_4:\kind}}
\end{mathpar}
\caption{\text{Kind equivalence}\ \typeEval{\Gamma}{k_1 \equiv k_2:\kind}}
\label{fig:apd:kind-equivalence}
\end{figure}

\begin{figure}[!h]
\begin{mathpar}
\footnotesize
%\boxed{\text{Subkinding}\ \typeEval{\Gamma}{k_1 \leq k_2:\kind}}\\
\LabelRule{\SubkE}{
    \typeEval{\Gamma}{k_1 \equiv k_2:\kind}   
}
{\typeEval{\Gamma}{k_1 \leq k_2:\kind}}\and
%%%%%%%%%%%%%%%
\LabelRule{\SubkT}{
    \typeEval{\Gamma}{k_1 \leq k_2:\kind} \\
    \typeEval{\Gamma}{k_2 \leq k_3:\kind}   
}
{\typeEval{\Gamma}{k_1 \leq k_3:\kind}}\and
%%%%%%%%%%%%%%%
\LabelRule{\SubkSg}{
    \typeEval{\Gamma}{c:\basekind}   
}
{\typeEval{\Gamma}{S(c) \leq \basekind:\kind}}\and
%%%%%%%%%%%%%%%
\LabelRule{{\SubkDFn}}{
    \typeEval{\Gamma}{k_1' \leq k_1:\kind} \\
    \typeEval{\Gamma,\alpha: k_1'}{k_2 \leq k_2':\kind}\\
    \typeEval{\Gamma,\alpha:k_1}{k_2:\kind}     
}
{\typeEval{\Gamma}{\Pi \alpha:k_1.k_2 \leq \Pi \alpha:k_1'.k_2':\kind}}\and
%%%%%%%%%%%%%%%
\LabelRule{{\SubkDPr}}{
    \typeEval{\Gamma}{k_1 \leq k_1':\kind} \\
    \typeEval{\Gamma,\alpha:k_1}{k_2 \leq k_2':\kind}\\
    \typeEval{\Gamma,\alpha:k_1'}{k_2':\kind}
}
{\typeEval{\Gamma}{\Sigma \alpha:k_1.k_2 \leq \Sigma \alpha:k_1'.k_2':\kind}}
\end{mathpar}
\caption{\text{Subkinding}\ \typeEval{\Gamma}{k_1 \leq k_2:\kind}}
\label{fig:apd:subkinding}
\end{figure}

%%%%%%%%%%%%%%%%%%%%%%%%%%%%%%%%%%%%%%%%%%%%%%%%%%5
%%%%%%%%%%%%%%%%%%%%%%%%%%%%%%%%%%%%%%%%%%%%%%%%%%

\begin{figure}[!h]
\begin{mathpar}
\footnotesize
%\boxed{\text{Well-formed constructor}\ \typeEval{\Gamma}{c:k}}\\
\LabelRule{\WfcV}{
    \Gamma(\alpha) = k
}
{\typeEval{\Gamma}{\alpha:k}}\and
%%%%%%%%%%%%%%%%%%%%    
\LabelRule{\WfcDFun}{
    \typeEval{\Gamma}{k_1:\kind}\\
    \typeEval{\Gamma,\alpha:k_1}{c:k_2}
}
{\typeEval{\Gamma}{\lambda\alpha:k_1.c: \Pi \alpha: k_1.k_2 }}\and
%%%%%%%%%%%%%%%%%%%%    
\LabelRule{\WfcDApp}{
    \typeEval{\Gamma}{c_1: \Pi \alpha:k_1.k_2}\\
    \typeEval{\Gamma}{c_2:k_1}
}
{\typeEval{\Gamma}{c_1c_2: k_2[\alpha\mapsto c_2]} }\and
%%%%%%%%%%%%%%%%%%%%    
\LabelRule{{\WfcPr}}{
    \typeEval{\Gamma}{c_1: k_1}\\
    \typeEval{\Gamma}{c_2:k_2[\alpha\mapsto c_1]}\\
    {\typeEval{\Gamma,\alpha:k_1}{k_2:\kind}}
}
{\typeEval{\Gamma}{\tuple{c_1,c_2}: \Sigma \alpha:k_1.k_2 }}\and
%%%%%%%%%%%%%%%%%
\LabelRule{{\WfcPrj1}}{
    \typeEval{\Gamma}{c: \Sigma \alpha:k_1.k_2}
}
{\typeEval{\Gamma}{\prj{1}{c}: k_1 }}\and
%%%%%%%%%%%%%%%%%
\LabelRule{{\WfcPrj2}}{
    \typeEval{\Gamma}{c: \Sigma \alpha:k_1.k_2}
}
{\typeEval{\Gamma}{\prj{2}{c}: k_2[\alpha\mapsto \prj{1}{c}] }}\and
%%%%%%%%%%%%%%%%%
\LabelRule{\WfcUc}{~}
{\typeEval{\Gamma}{\unitcon:\unitkind}}\and
%%%%%%%%%%%%%%%%%
\LabelRule{\WfcUt}{~}
{\typeEval{\Gamma}{\unittype:\basekind}}\and
%%%%%%%%%%%%%%%%%
\LabelRule{\WfcInt}{~}
{\typeEval{\Gamma}{\intType:\basekind}}\and
%%%%%%%%%%%%%%%%%
\LabelRule{\WfcFn}{
    \typeEval{\Gamma}{\tau_1:\basekind}\\
    \typeEval{\Gamma}{\tau_2:\basekind}\\
}
{\typeEval{\Gamma}{\tau_1 \rightarrow \tau_2: \basekind}}\and
%%%%%%%%%%%%%%%%%%%%
\LabelRule{\WfcPrd}{
    \typeEval{\Gamma}{\tau_1:\basekind}\\
    \typeEval{\Gamma}{\tau_2:\basekind}\\
}
{\typeEval{\Gamma}{\tau_1 \times\tau_2: \basekind}}\and
%%%%%%%%%%%%%%%%%%%%
\LabelRule{\WfcAbs}{
    \typeEval{\Gamma}{k:\kind}\\
    \typeEval{\Gamma,\alpha:k}{\tau:\basekind}\\
}
{\typeEval{\Gamma}{\forall \alpha:k.\tau: \basekind}}\and
%%%%%%%%%%%%%%%%%%%%
\LabelRule{\WfcEx}{
    \typeEval{\Gamma}{k:\kind}\\
    \typeEval{\Gamma,\alpha:k}{\tau:\basekind}\\
}
{\typeEval{\Gamma}{\exists \alpha:k.\tau: \basekind}}\and
%%%%%%%%%%%%%%%%%%%%
\LabelRule{{\WfcSg}}{
    \typeEval{\Gamma}{c:\basekind}
}
{\typeEval{\Gamma}{c: S(c)}}\and
%%%%%%%%%%%%%%%%%%%
\LabelRule{{\WfcROne}}{
\typeEval{\Gamma}{c:\Pi \alpha:k_1.k_2}\\
{\typeEval{\Gamma,\alpha:k_1}{c\alpha: k_2'}}
}
{\typeEval{\Gamma}{c:\Pi \alpha:k_1.k_2'}}\and
%%%%%%%%%%%%%%%%%%%
\LabelRule{{{\WfcRTwo}}}{
    \typeEval{\Gamma}{\prj{1}{c}:k_1}\\
    \typeEval{\Gamma}{\prj{2}{c}:k_2[\alpha\mapsto\prj{1}{c}]}\\
%    \typeEval{\Gamma,\alpha:k_1}{c\alpha:k_2}
    \typeEval{\Gamma,\alpha:k_1}{k_2:\kind}
}
{\typeEval{\Gamma}{c:\Sigma \alpha:k_1.k_2}}\and
%%%%%%%%%%%%%%%%%%%%%%
\LabelRule{\WfcSub}{
    \typeEval{\Gamma}{c:k}\\
    \typeEval{\Gamma}{k \leq k':\kind}
}
{\typeEval{\Gamma}{c:k'}}
\end{mathpar}
\caption{\text{Well-formed constructor}\ \typeEval{\Gamma}{c:k}}
\label{fig:apd:well-formed-constructor}
\end{figure}

%%%%%%%%%%%%%%%%%%%%%%%%%%%%%%%%%%%%%%%%%%%%%%%%%%%%%%%%%
%%%%%%%%%%%%%%%%%%%%%%%%%%%%%%%%%%%%%%%%%%%%%%%%%%%%%%%%%
\begin{figure*}[!h]
\begin{mathpar}
\footnotesize
%\boxed{\text{Constructor equivalence}\ \typeEval{\Gamma}{c_1 \equiv c_2:k}}\\
%%%%%%%%%%%
\mbox{\LabelRule{\EqcR}{
    \typeEval{\Gamma}{c:k}
}
{\typeEval{\Gamma}{c\equiv c: k}}\and
%%%%%%%%%%%
\LabelRule{\EqcS}{
    \typeEval{\Gamma}{c_1 \equiv c_2:k}
}
{\typeEval{\Gamma}{c_2\equiv c_1: k}}\and
%%%%%%%%%%%
\LabelRule{\EqcT}{
    \typeEval{\Gamma}{c_1 \equiv c_2:k}\\
    \typeEval{\Gamma}{c_2 \equiv c_3:k}
}
{\typeEval{\Gamma}{c_1\equiv c_3: k}}}\and
%%%%%%%%%%%
\mbox{\LabelRule{{\EqcDFn}}{
    \typeEval{\Gamma}{k_1 \equiv k_1':\kind}\\
    \typeEval{\Gamma,\alpha:k_1}{c \equiv c':k_2}
}
{\typeEval{\Gamma}{\lambda \alpha:k_1.c \equiv \lambda \alpha:k_1'.c': \Pi \alpha:k_1.k_2}}\and
%%%%%%%%%%%%%
\LabelRule{{\EqcDApp}}{
    \typeEval{\Gamma}{c_1 \equiv c_1': \Pi \alpha:k_1.k_2} \\
    \typeEval{\Gamma}{c_2 \equiv c_2':k_1}
}
{\typeEval{\Gamma}{c_1 c_2 \equiv c_1'c_2': k_2[\alpha\mapsto c_2]}}} \and
%%%%%%%%%%%%%
\LabelRule{{\EqcDPr}}{
    \typeEval{\Gamma}{c_1 \equiv c_1':k_1}\\
    \typeEval{\Gamma}{c_2 \equiv c_2':k_2[\alpha\mapsto c_1]}\\
    \typeEval{\Gamma,\alpha:k_1}{k_2:\kind}
}
{\typeEval{\Gamma}{\tuple{c_1,c_2} \equiv \tuple{c_1',c_2'}:\Sigma \alpha:k_1.k_2}}\and
%%%%%%%%%%%%%
\LabelRule{{\EqcPrj1}}{
    \typeEval{\Gamma}{c \equiv c':\Sigma \alpha:k_1.k_2}
}
{\typeEval{\Gamma}{\prj{1}{c} \equiv \prj{1}{c'}:k_1}}\and
%%%%%%%%%%%%%
\LabelRule{{\EqcPrj2}}{
    \typeEval{\Gamma}{c \equiv c':\Sigma \alpha:k_1.k_2}
}
{\typeEval{\Gamma}{\prj{2}{c} \equiv \prj{2}{c'}:k_2[\alpha\mapsto \prj{1}{c}]}}\and
%%%%%%%%%%%%%
\mbox{\LabelRule{{\EqcFn}}{
    \typeEval{\Gamma}{\tau_1 \equiv \tau_1':\basekind}\\
    \typeEval{\Gamma}{\tau_2 \equiv \tau_2':\basekind}
}
{\typeEval{\Gamma}{\tau_1 \rightarrow \tau_2 \equiv \tau_1' \rightarrow \tau_2':\basekind}}\and
%%%%%%%%%%%%%
\LabelRule{{\EqcPrd}}{
    \typeEval{\Gamma}{\tau_1 \equiv \tau_1':\basekind}\\
    \typeEval{\Gamma}{\tau_2 \equiv \tau_2':\basekind}
}
{\typeEval{\Gamma}{\tau_1 \times \tau_2 \equiv \tau_1' \times \tau_2':\basekind}}}\and
%%%%%%%%%%%%%
\LabelRule{{\EqcAbs}}{
    \typeEval{\Gamma}{k \equiv k':\kind}\\
    \typeEval{\Gamma,\alpha:k}{\tau \equiv \tau':\basekind}
}
{\typeEval{\Gamma}{\forall \alpha:k.\tau \equiv \forall \alpha:k'.\tau':\basekind}}\and
%%%%%%%%%%%%%
\LabelRule{{\EqcEx}}{
    \typeEval{\Gamma}{k \equiv k':\kind}\\
    \typeEval{\Gamma,\alpha:k}{\tau \equiv \tau':\basekind}
}
{\typeEval{\Gamma}{\exists \alpha:k.\tau \equiv \exists \alpha:k'.\tau':\basekind}}\and
%%%%%%%%%%%%%
\LabelRule{{\EqcSg}}{
    \typeEval{\Gamma}{c \equiv c':\basekind}
}
{\typeEval{\Gamma}{c \equiv c':S(c)}}\and
%%%%%%%%%%%%%
\LabelRule{{\EqcSgTwo}}{
    \typeEval{\Gamma}{c:S(c')}
}
{\typeEval{\Gamma}{c \equiv c':\basekind}}\and
%%%%%%%%%%%%%
\LabelRule{{\EqcROne}}{
    \typeEval{\Gamma}{c:\Pi \alpha:k_1.k_2'}\\
    \typeEval{\Gamma}{c':\Pi \alpha:k_1.k_2''}\\    
    \typeEval{\Gamma,\alpha:k_1}{c\alpha \equiv c'\alpha:k_2}
}
{\typeEval{\Gamma}{c \equiv c':\Pi\alpha:k_1.k_2}}\and
%%%%%%%%%%%%%
\LabelRule{{\EqcRTwo}}{
    \typeEval{\Gamma}{c \equiv c':\Pi \alpha:k_1.k_2'}\\
    {\typeEval{\Gamma,\alpha:k_1}{c\alpha \equiv c'\alpha:k_2}}
}
{\typeEval{\Gamma}{c\equiv c': \Pi\alpha:k_1.k_2}}\and
%%%%%%%%%%%%%
\LabelRule{{\EqcRThree}}{
    \typeEval{\Gamma}{\prj{1}{c} \equiv \prj{1}{c'}:k_1}\\
    \typeEval{\Gamma}{\prj{2}{c} \equiv \prj{2}{c'}:k_2[\alpha\mapsto\prj{1}{c}]} \\
    \typeEval{\Gamma,\alpha:k_1}{k_2:\kind}
}
{\typeEval{\Gamma}{c \equiv c': \Sigma \alpha:k_1.k_2}}\and
%%%%%%%%%%%%%%
\mbox{\LabelRule{{\EqcRFour}}{
    \typeEval{\Gamma}{c:\unitkind}\\
    \typeEval{\Gamma}{c':\unitkind}
}
{\typeEval{\Gamma}{c \equiv c':\unitkind}}\and
%%%%%%%%%%%%%
\LabelRule{{\EqcSub}}{
    \typeEval{\Gamma}{c\equiv c':k}\\
    \typeEval{\Gamma}{k \leq k':\kind}
}
{\typeEval{\Gamma}{c \equiv c':k'}}} \and
%%%%%%%%%%%%%
\LabelRule{{\EqcBeta}}{
    \typeEval{\Gamma,\alpha:k_1}{c_2:k_2} \\
    \typeEval{\Gamma}{c_1:k_1}
}
{\typeEval{\Gamma}{(\lambda \alpha:k_1.c_2)c_1 \equiv c_2[\alpha\mapsto c_1]:k_2[\alpha\mapsto c_2]}}\and
%%%%%%%%%%%%%
\LabelRule{{\EqcRFive}}{
    \typeEval{\Gamma}{c_1:k_1}\\
    \typeEval{\Gamma}{c_2:k_2}       
}
{\typeEval{\Gamma}{\prj{1}{\tuple{c_1,c_2}}\equiv c_1:k_1}}\and
%%%%%%%%%%%%%
\LabelRule{{\EqcRSix}}{
    \typeEval{\Gamma}{c_1:k_1}\\
    \typeEval{\Gamma}{c_2:k_2}       
}
{\typeEval{\Gamma}{\prj{2}{\tuple{c_1,c_2}}\equiv c_2:k_2}}
\end{mathpar}
%\todoinline{MN}{I think that the {\EqcDFn}\ should be 
%$$\LabelRule{\redtext{\EqcDFn}}{
%    \typeEval{\Gamma}{k_1 \equiv k_1':\kind}\\
%    \typeEval{\Gamma,\alpha:k_1}{c \equiv c':k_2}
%}
%{\typeEval{\Gamma}{\lambda \alpha:k_1.c \equiv \lambda \alpha:k_1'.c': \Pi \alpha:k_1.k_2}}$$
%since $c$ and $c'$ may not be in an arrow form and we may not be able to have $c\alpha$ and $c'\alpha$.
%I checked the Coq development and its implementation is 
%eqc\_lam \{G k1 k1' k2 c c'\}
%    : eqk G k1 k1'
%      -> eqc (G; cl\_cn k1) c c' k2
%      -> eqc G (cn\_lam k1 c) (cn\_lam k1' c') (kd\_pi k1 k2).      
%}
\caption{\text{Constructor equivalence}\ \typeEval{\Gamma}{c_1 \equiv c_2:k}}
\label{fig:apd:constructor-equivalence}
\end{figure*}

%%%%%%%%%%%%%%%%%%%%%%%%%%%%%%%%%%%%%%%%%%%%%%%%%%%%%%%%%%%%%%%%%

%%%%%%%%%%%%%%%%%%%%%%%%%%%%%%%%%%%%%%%%%%%%%%%%%%%%%%%%%%%%%%%%%
%%%%%%%%%%%%%%%%%%%%%%%%%%%%%%%%%%%%%%%%%%%%%%%%%%%%%%%%%%%%%%%%%
\begin{figure}[!h]
\begin{mathpar}
\footnotesize
%\boxed{\text{Well-typed term}\ \typeEval{\Gamma}{e:\tau}}\\
\LabelRule{{\WttV}}{
    \Gamma(x) = \tau
}
{\typeEval{\Gamma}{x:\tau}}\and
%%%%%%%%%%%%%%%%%%%%%%%%
\LabelRule{{\WttUnit}}{
    ~
}
{\typeEval{\Gamma}{\unitval:\unittype}}\and
%%%%%%%%%%%%%%%%%%%%%%%%
\LabelRule{{\WttInt}}{
    ~
}
{\typeEval{\Gamma}{n:\intType}}\and
%%%%%%%%%%%%%%%%%%%%%%%%
\LabelRule{{\WttAbs}}{
    \typeEval{\Gamma}{\tau_1:\basekind}\\
    \typeEval{\Gamma,x:\tau_1}{e:\tau_2}
}
{\typeEval{\Gamma}{\lambda x:\tau_1.e:\tau_1 \rightarrow \tau_2}}\and
%%%%%%%%%%%%%%%%%%%%%%%%
\LabelRule{{\WttApp}}{
    \typeEval{\Gamma}{e_1:\tau_1 \rightarrow \tau_2}\\
    \typeEval{\Gamma}{e_2:\tau_1}
}
{\typeEval{\Gamma}{e_1e_2:\tau_2}}\and
%%%%%%%%%%%%%%%%%%%%%%%%
\LabelRule{{\WttPr}}{
    \typeEval{\Gamma}{e_1:\tau_1}\\
    \typeEval{\Gamma}{e_2:\tau_2}
}
{\typeEval{\Gamma}{\tuple{e_1,e_2}:\tau_1\times \tau_2}}\and
%%%%%%%%%%%%%%%%%%%%%%%%
\LabelRule{{\WttPrj1}}{
    \typeEval{\Gamma}{e:\tau_1 \times \tau_2}
}
{\typeEval{\Gamma}{\prj{1}{e}:\tau_1}}\and
%%%%%%%%%%%%%%%%%%%%%%%%
\LabelRule{{\WttPrj2}}{
    \typeEval{\Gamma}{e:\tau_1\times\tau_2}
}
{\typeEval{\Gamma}{\prj{2}{e}:\tau_2}}\and
%%%%%%%%%%%%%%%%%%%%%%%%
\LabelRule{{\WttUnv}}{
    \typeEval{\Gamma}{k:\kind}\\
    \typeEval{\Gamma,\alpha:k}{e:\tau}
}
{\typeEval{\Gamma}{\Lambda \alpha:k.e: \forall \alpha:k.\tau}}\and
%%%%%%%%%%%%%%%%%%%%%%%%
\LabelRule{{\WttInst}}{
    \typeEval{\Gamma}{e:\forall\alpha:k.\tau}\\
    \typeEval{\Gamma}{c:k}
}
{\typeEval{\Gamma}{e[c]:\tau[\alpha\mapsto c]}}\and
%%%%%%%%%%%%%%%%%%%%%%%%
\LabelRule{{\WttPk}}{
    \typeEval{\Gamma}{c:k}\\
    \typeEval{\Gamma}{e:\tau[\alpha\mapsto c]}\\
    \typeEval{\Gamma,\alpha:k}{\tau:\basekind}
}
{\typeEval{\Gamma}{\pack{c,e}{\exists \alpha:k.\tau}:\exists \alpha:k.\tau}}\and
%%%%%%%%%%%%%%%%%%%%%%%%
\LabelRule{{\WttUp}}{
    \typeEval{\Gamma}{e_1:\exists\alpha:k.\tau}\\
    \typeEval{\Gamma,\alpha:k,x:\tau}{e_2:\tau'}\\
    \typeEval{\Gamma}{\tau':\basekind}
}
{\typeEval{\Gamma}{\unpack{\alpha,x}{e_1}{e_2}:\tau'}}\and
%%%%%%%%%%%%%%%%%%%%%%%%
\LabelRule{{\WttRc}}{
    \typeEval{\Gamma}{e:(\unittype \rightarrow \tau) \rightarrow \tau}
}
{\typeEval{\Gamma}{\fix{\tau}{e}:\tau}}\and
%%%%%%%%%%%%%%%%%%%%%%%%
\LabelRule{{\WttLtOne}}{
    \typeEval{\Gamma}{e_1:\tau_1}\\
    \typeEval{\Gamma,x:\tau_1}{e_2:\tau_2}
}
{\typeEval{\Gamma}{\letexp{x=e_1}{e_2}:\tau_2}}\and
%%%%%%%%%%%%%%%%%%%%%%%%
\LabelRule{{\WttLtTwo}}{
    \typeEvalI{\Gamma}{M:\sigma}\\
    \typeEval{\Gamma,\alpha/m:\sigma}{e:\tau}\\
    \typeEval{\Gamma}{\tau:\basekind}
}
{\typeEval{\Gamma}{\letexp{\alpha/m=M}{e}:\tau}}\and
%%%%%%%%%%%%%%%%%%%%%%%
\LabelRule{{\WttExt}}{
    \typeEvalI{\Gamma}{M:\atcmod{\tau}}
}
{\typeEval{\Gamma}{\extract{M}:\tau}}\and
%%%%%%%%%%%%%%%%%%%%%%%
\LabelRule{{\WttEq}}{
    \typeEval{\Gamma}{e:\tau}\\
    \typeEval{\Gamma}{\tau \equiv \tau':\basekind}
}
{\typeEval{\Gamma}{e:\tau'}}
\end{mathpar}
\caption{\text{Well-typed term}\ \typeEval{\Gamma}{e:\tau}}
\label{fig:apd:well-typed-term}
\end{figure}

%%%%%%%%%%%%%%%%%%%%%%%%%%%%%%%%%%%%%%%%%%%%%%%%%%%%%%%%%%%%%%%%%%%%%%%%%%5
%%%%%%%%%%%%%%%%%%%%%%%%%%%%%%%%%%%%%%%%%%%%%%%%%%%%%%%%%%%%%%%%%%%%%%%%%%5

\begin{figure}[!h]

\begin{mathpar}
\footnotesize
%\boxed{\text{Well-formed signature}\ \typeEval{\Gamma}{\sigma:\sign}}\\
\LabelRule{\WfsOne}{
    ~
}
{\typeEval{\Gamma}{\unitsign:\sign}}\and
%%%%%%%%%%%%%%%%%%%%%%%
\LabelRule{\WfsTwo}{
    \typeEval{\Gamma}{k:\kind}
}
{\typeEval{\Gamma}{\atksign{k}:\sign}}\and
%%%%%%%%%%%%%%%%%%%%%%%
\LabelRule{\WfsThree}{
    \typeEval{\Gamma}{\tau:\basekind}
}
{\typeEval{\Gamma}{\atcsign{\tau}:\sign}}\and
%%%%%%%%%%%%%%%%%%%%%%%
\LabelRule{\WfsFour}{
    \typeEval{\Gamma}{\sigma_1:\sign}\\
    \typeEval{\Gamma,\alpha:\fstsign{\sigma_1}}{\sigma_2:\sign}\\
}
{\typeEval{\Gamma}{\Piap \alpha:\alpha:\sigma_1.\sigma_2:\sign}}\and
%%%%%%%%%%%%%%%%%%%%%%%
\LabelRule{\WfsFive}{
    \typeEval{\Gamma}{\sigma_1:\sign}\\
    \typeEval{\Gamma,\alpha:\fstsign{\sigma_1}}{\sigma_2:\sign}\\
}
{\typeEval{\Gamma}{\Pign \alpha:\alpha:\sigma_1.\sigma_2:\sign}}\and
%%%%%%%%%%%%%%%%%%%%%%%
\LabelRule{\WfsSix}{
    \typeEval{\Gamma}{\sigma_1:\sign}\\
    \typeEval{\Gamma,\alpha:\fstsign{\sigma_1}}{\sigma_2:\sign}\\
}
{\typeEval{\Gamma}{\Sigma \alpha:\sigma_1.\sigma_2:\sign}}
\end{mathpar}
\caption{\text{Well-formed signature}\ \typeEval{\Gamma}{\sigma:\sign}}
\label{fig:apd:well-formed-signature}
\end{figure}
%%%%%%%%%%%%%%%%%%%%%%%
%%%%%%%%%%%%%%%%%%%%%%%%%%%%%%%%%%%%%%%%%%%%%%%%%%%%%%%%%%%%%%%%%%%%%%%%%%%%%%%%
%%%%%%%%%%%%%%%%%%%%%%%%%%%%%%%%%%%%%%%%%%%%%%%%%%%%%%%%%%%%%%%%%%%%%%%%%%5

\begin{figure}[!h]
\begin{mathpar}
\footnotesize
%\boxed{\text{Signature equivalence}\ \typeEval{\Gamma}{\sigma\equiv \sigma':\sign}}\\
\LabelRule{\SeqR}{
    \typeEval{\Gamma}{\sigma:\sign}
}
{\typeEval{\Gamma}{\sigma \equiv \sigma:\sign}}\and
%%%%%%%%%%%%%%%%%%%
\LabelRule{\SeqS}{
    \typeEval{\Gamma}{\sigma \equiv \sigma':\sign}
}
{\typeEval{\Gamma}{\sigma' \equiv \sigma:\sign}}\and
%%%%%%%%%%%%%%%%%%%
\LabelRule{\SeqT}{
    \typeEval{\Gamma}{\sigma \equiv \sigma'':\sign}\\
    \typeEval{\Gamma}{\sigma'' \equiv \sigma':\sign}\\
}
{\typeEval{\Gamma}{\sigma \equiv \sigma':\sign}}\and
%%%%%%%%%%%%%%%%%%%
\LabelRule{\SeqAtk}{
    \typeEval{\Gamma}{k \equiv k':\kind}
}
{\typeEval{\Gamma}{\atksign{k} \equiv \atksign{k'}:\sign}}\and
%%%%%%%%%%%%%%%%%%%
\LabelRule{\SeqAtc}{
    \typeEval{\Gamma}{\tau \equiv \tau':\basekind}
}
{\typeEval{\Gamma}{\atcsign{\tau} \equiv \atcsign{\tau'}:\sign}}\and
%%%%%%%%%%%%%%%%%%%
\LabelRule{\SeqDFnGn}{
    \typeEval{\Gamma}{\sigma_1 \equiv \sigma_1':\sign}\\
    \typeEval{\Gamma,\alpha:\fstsign{\sigma_1}}{\sigma_2 \equiv \sigma_2':\sign}
}
{\typeEval{\Gamma}{\Pign \alpha:\sigma_1.\sigma_2 \equiv \Pign \alpha:\sigma_1':\sigma_2':\sign}}\and
%%%%%%%%%%%%%%%%%%%
\LabelRule{\SeqDFnAp}{
    \typeEval{\Gamma}{\sigma_1 \equiv \sigma_1':\sign}\\
    \typeEval{\Gamma,\alpha:\fstsign{\sigma_1}}{\sigma_2 \equiv \sigma_2':\sign}
}
{\typeEval{\Gamma}{\Piap \alpha:\sigma_1.\sigma_2 \equiv \Piap \alpha:\sigma_1'.\sigma_2':\sign}}\and
%%%%%%%%%%%%%%%%%%%
\LabelRule{\SeqDPr}{
    \typeEval{\Gamma}{\sigma_1 \equiv \sigma_1':\sign}\\
    \typeEval{\Gamma,\alpha:\fstsign{\sigma_1}}{\sigma_2 \equiv \sigma_2':\sign}
}
{\typeEval{\Gamma}{\Sigma \alpha:\sigma_1.\sigma_2 \equiv \Sigma \alpha:\sigma_1'.\sigma_2':\sign}}
\end{mathpar}
\caption{\text{Signature equivalence}\ \typeEval{\Gamma}{\sigma\equiv \sigma':\sign}}
\label{fig:apd:signature-equivalence}
\end{figure}

%%%%%%%%%%%%%%%%%%%%%%%%%%%%%%%%%%%%%%%%%%%%%%%%%%%%%%%%%%%%%%%%%%%%%%%%%%%%%%%
%%%%%%%%%%%%%%%%%%%%%%%%%%%%%%%%%%%%%%%%%%%%%%%%%%%%%%%%%%%%%%%%%%%%%%%%%%%%%%%
%%%%%%%%%%%%%%%%%%%%%%%%%%%%%%%%%%%%%%%%5
\begin{figure}[!h]
\begin{mathpar}
\footnotesize
%\boxed{\text{Subsignature}\ \typeEval{\Gamma}{\sigma \leq \sigma':\sign}}\\
\LabelRule{\SubsEq}{
    \typeEval{\Gamma}{\sigma \equiv \sigma':\sign}
}
{\typeEval{\Gamma}{\sigma\leq \sigma':\sign}}\and
%%%%%%%%%%%%%%
\LabelRule{\SubsT}{
    \typeEval{\Gamma}{\sigma \leq \sigma'':\sign}\\
    \typeEval{\Gamma}{\sigma'' \leq \sigma':\sign}
}
{\typeEval{\Gamma}{\sigma \leq \sigma':\sign}}\and
%%%%%%%%%%%%%%
\LabelRule{\SubsK}{
    \typeEval{\Gamma}{k \leq k':\kind}
}
{\typeEval{\Gamma}{\atksign{k} \leq \atksign{k'}:\sign}}\and
%%%%%%%%%%%%%%
\LabelRule{\SubsDFnGn}{
    \typeEval{\Gamma}{\sigma_1' \leq \sigma_1:\sign}\\
    \typeEval{\Gamma,\alpha:\fstsign{\sigma_1'}}{\sigma_2 \leq \sigma_2':\sign}\\
    \typeEval{\Gamma,\alpha:\fstsign{\sigma_1}}{\sigma_2:\sign}
}
{\typeEval{\Gamma}{\Pign \alpha:\sigma_1.\sigma_2 \leq \Pign \alpha:\sigma_1'.\sigma_2':\sign}}\and
%%%%%%%%%%%%%%
\LabelRule{\SubsDFnAp}{
    \typeEval{\Gamma}{\sigma_1' \leq \sigma_1:\sign}\\
    \typeEval{\Gamma,\alpha:\fstsign{\sigma_1'}}{\sigma_2 \leq \sigma_2':\sign}\\
    \typeEval{\Gamma,\alpha:\fstsign{\sigma_1}}{\sigma_2:\sign}
}
{\typeEval{\Gamma}{\Piap \alpha:\sigma_1.\sigma_2 \leq \Piap \alpha:\sigma_1'.\sigma_2':\sign}}\and
%%%%%%%%%%%%%%
\LabelRule{\SubsDPr}{
    \typeEval{\Gamma}{\sigma_1 \leq \sigma_1':\sign}\\
    \typeEval{\Gamma,\alpha:\fstsign{\sigma_1}}{\sigma_2 \leq \sigma_2':\sign}\\
    \typeEval{\Gamma,\alpha:\fstsign{\sigma_1'}}{\sigma_2':\sign}
}
{\typeEval{\Gamma}{\Sigma \alpha:\sigma_1.\sigma_2 \leq \Sigma \alpha:\sigma_1'.\sigma_2':\sign}}
\end{mathpar}
\caption{\text{Subsignature}\ \typeEval{\Gamma}{\sigma \leq \sigma':\sign}}
\label{fig:apd:subsignature}
\end{figure}
%%%%%%%%%%%%%%
%%%%%%%%%%%%%%%%%%%%%%%%%%%%%%%%%%%%%%%%%%%%%%%%%%%%%%%%%%%
%%%%%%%%%%%%%%%%%%%%%%%%%%%%%%%%%%%%%%%%%%%%%%%%%%%%%%%%%%%

\begin{figure}[!h]
\begin{mathpar}
\footnotesize
%\boxed{\text{Well-formed module}\ \wfmod{\Gamma}{\kappa}{M:\sigma}}\\
\LabelRule{\WfmV}{
    \Gamma(m) = \sigma
}
{\wfmod{\Gamma}{P}{m:\sigma}}\and
%%%%%%%%%%%%%%%%%
\LabelRule{\WfmU}{~
}
{\wfmod{\Gamma}{P}{\unitmod:\unitsign}}\and
%%%%%%%%%%%%%%%%%
\LabelRule{\WfmAtc}{
    \typeEval{\Gamma}{c:k}
}
{\wfmod{\Gamma}{P}{\atcmod{c}:\atksign{k}}}\and
%%%%%%%%%%%%%%%%%
\LabelRule{\WfmAtt}{
    \typeEval{\Gamma}{e:\tau}
}
{\wfmod{\Gamma}{P}{\attmod{e}:\atksign{\tau}}}\and
%%%%%%%%%%%%%%%%%
\LabelRule{\WfmAbsI}{
    \typeEval{\Gamma}{\sigma:\sign}\\
    \typeEvalI{\Gamma,\alpha/m:\sigma}{M:\sigma'}
}
{{\wfmod{\Gamma}{P}{\lambdagn \alpha/m:\sigma.M: \Pign \alpha:\sigma.\sigma'}}}\and
%%%%%%%%%%%%%%%%%
\LabelRule{\WfmAppI}{
    \typeEvalI{\Gamma}{M_1:\Pign \alpha:\sigma.\sigma'}\\
    \typeEvalP{\Gamma}{M_2:\sigma}\\
    \typeEval{\Gamma}{\fstop{M_2}{c_2}}
}
{\wfmod{\Gamma}{I}{M_1M_2:\sigma'[\alpha\mapsto c_2]}}\and
%%%%%%%%%%%%%%%%%
\LabelRule{\WfmAbsP}{
    \typeEval{\Gamma}{\sigma:\sign}\\
    \typeEvalP{\Gamma,\alpha/m:\sigma}{M:\sigma'}
}
{\wfmod{\Gamma}{P}{\lambdaap \alpha/m:\sigma.M: \Piap \alpha:\sigma.\sigma'}}\and
%%%%%%%%%%%%%%%%%
\LabelRule{\WfmAppK}{
    \typeEvalK{\Gamma}{M_1:\Piap \alpha:\sigma.\sigma'}\\
    \typeEvalP{\Gamma}{M_2:\sigma}\\
    \typeEval{\Gamma}{\fstop{M_2}{c_2}}
}
{\wfmod{\Gamma}{\kappa}{M_1\appapp M_2: \sigma'[\alpha\mapsto c_2]}}\and
%%%%%%%%%%%%%%%%%
\LabelRule{\WfmPr}{
    \typeEvalK{\Gamma}{M_1:\sigma_1}\\
    \typeEvalK{\Gamma}{M_2:\sigma_2}\\
    {\alpha\not\in FV(\sigma_2)}
}
{\wfmod{\Gamma}{\kappa}{\tuple{M_1,M_2}:\Sigma \alpha:\sigma_1.\sigma_2}}\and
%%%%%%%%%%%%%%%%%
\LabelRule{\WfmPrj1}{
    {\typeEvalP{\Gamma}{M:\Sigma\alpha:\sigma_1.\sigma_2}}
}
{{\wfmod{\Gamma}{P}{\prj{1}{M}:\sigma_1}}}\and
%%%%%%%%%%%%%
\LabelRule{\WfmPrj2}{
    \typeEvalP{\Gamma}{M:\Sigma \alpha:\sigma_1.\sigma_2}\\
    \typeEval{\Gamma}{\fstop{M}{c}}
}
{\wfmod{\Gamma}{P}{\prj{2}{M}:\sigma_2[\alpha\mapsto \purpletext{\prj{1}{c}}]}}\and
%%%%%%%%%%%%%%%%%
\LabelRule{\WfmUnp}{
    \typeEval{\Gamma}{e:\exists\alpha:k.\tau}\\
    \typeEvalI{\Gamma,\alpha:k,x:\tau}{M:\sigma}\\
    \typeEval{\Gamma}{\sigma:\sign}
}
{\wfmod{\Gamma}{I}{\unpack{\alpha,x}{e}{M:\sigma}:\sigma}}\and
%%%%%%%%%%%%%%%%%
\LabelRule{\WfmLetOne}{
    \typeEval{\Gamma}{e:\tau}\\
    \typeEvalK{\Gamma,x:\tau}{M:\sigma}
}
{\wfmod{\Gamma}{K}{\letexp{x=e}{M}:\sigma}}\and
%%%%%%%%%%%%%%%%
\LabelRule{\WfmLetTwo}{
    \typeEvalI{\Gamma}{M_1:\sigma}\\
    \typeEvalI{\Gamma,\alpha/m:\sigma}{M_2:\sigma'}\\
    \typeEval{\Gamma}{\sigma':\sign}
}
{\wfmod{\Gamma}{I}{\letexp{\alpha/m = M_1}{(M_2:\sigma')}:\sigma'}}\and
%%%%%%%%%%%%%%%%%
\LabelRule{\WfmSeal}{
    \typeEvalI{\Gamma}{M:\sigma}\\
}
{\wfmod{\Gamma}{I}{(M\seal \sigma):\sigma)}}\and
%%%%%%%%%%%%%%%%%
\LabelRule{{\WfmROne}}{
    \typeEvalP{\Gamma}{M:\atksign{k'}}\\
    \typeEval{\Gamma}{\fstop{M}{c}}\\
    \typeEval{\Gamma}{c:k}
}
{\wfmod{\Gamma}{P}{M:\atksign{k}}}\and
%%%%%%%%%%%%%%%%%
\LabelRule{\WfmRTwo}{
    \typeEvalP{\Gamma}{M:\Piap \alpha:\sigma_1.\sigma_2'}\\
    \typeEvalP{\Gamma,\alpha/m:\sigma_1}{M\appapp m:\sigma_2}
}
{\wfmod{\Gamma}{P}{M: \Piap \alpha:\sigma_1.\sigma_2}}\and
%%%%%%%%%%%%%%%%%
\LabelRule{\WfmRThree}{
    \typeEvalP{\Gamma}{\prj{1}{M}:\sigma_1}\\
    \typeEvalP{\Gamma}{\prj{2}{M}:\sigma_2}\\
    \alpha\not\in FV(\sigma_2)
}
{\wfmod{\Gamma}{P}{M:\Sigma \alpha:\sigma_1.\sigma_2}}\and
%%%%%%%%%%%%%%%%%
\LabelRule{\WfmRFour}{
    \typeEvalP{\Gamma}{M:\sigma}
}
{\wfmod{\Gamma}{I}{M:\sigma}}\and
%%%%%%%%%%%%%%%%
\LabelRule{\WfmSub}{
    \typeEvalK{\Gamma}{M:\sigma}\\
    \typeEval{\Gamma}{\sigma\leq \sigma':\sign}
}
{\wfmod{\Gamma}{K}{M:\sigma'}}
%%%%%%%%%%%%%%%%
\end{mathpar}
\caption{\text{Well-formed module}\ \wfmod{\Gamma}{\kappa}{M:\sigma}}
\label{fig:apd:well-formed-module}
\end{figure}

%\subsection{Dynamic semantics}
%The dynamic semantics is call-by-value semantics. Closed (term) values and module values are as below.
%\begin{align*}
%v & := x \sep \unitval \sep \lambda x:\tau.e \sep \tuple{v,v} \sep \Lambda \alpha:k.e \sep \pack{c,v}{\exists \alpha:k.\tau}\\
%V & := \unitmod \sep \atcmod{c} \sep \attmod{v} \sep \tuple{V,V} \sep \lambdagn \alpha/m:\sigma.M \sep \lambdaap \alpha/m:\sigma.M
%\end{align*}
%
%Notice that the evaluation of a module may diverge since a dynamic value in a module may diverge.

\clearpage  % output all figures before proceeding; to make proofs easier to read

\begin{figure*}[!t]
$$\small
\LabelRule{\SubsDPr}{
\LabelRuleProof{\SubsK}{
\LabelRuleProof{\SubkSg}{
\LabelRuleProof{\WfcInt}{~}
{\typeEval{}{\intType:\basekind}}}
{\typeEval{}{S(\intType)\leq \basekind:\kind}}}
{\typeEval{}{\atksign{S(\intType)} \leq \atksign{\basekind}:\sign}}\hspace{40pt}
%%%%%
\LabelRuleProof{\SubsEq}{
\LabelRuleProof{\SeqAtc}{
\LabelRuleProof{\EqcS}{
\LabelRuleProof{\EqcSgTwo}{
\LabelRuleProof{\WfcV}{~}
{\typeEval{\alpha:S(\intType)}{\alpha:S(\intType)}}}
{\typeEval{\alpha:S(\intType)}{\alpha \equiv \intType:\basekind}}}
{\typeEval{\alpha:S(\intType)}{\intType \equiv \alpha:\basekind}}   }
{\typeEval{\alpha:S(\intType)}{\atcsign{\intType} \equiv \atcsign{\alpha}:\sign}}}
{\typeEval{\alpha:S(\intType)}{\atcsign{\intType} \leq \atcsign{\alpha}:\sign }}\hspace{40pt}
%%%%%
\LabelRuleProof{\WfsThree}{
\LabelRule{\WfcV}{~}
{\typeEval{\alpha:\basekind}{\alpha:\basekind}}}
{\typeEval{\alpha:\basekind}{\atcsign{\alpha}:\sign}}
}
{\typeEval{}{\Sigma \alpha:\atksign{S(\intType)}.\atcsign{\intType} \leq \Sigma \alpha:\atksign{\basekind}.\atcsign{\alpha}}}
$$
\caption{Derivation of \typeEval{}{\Sigma \alpha:\atksign{S(\intType)}.\atcsign{\intType} \leq \Sigma \alpha:\atksign{\basekind}.\atcsign{\alpha}}}
\label{fig:derivation:example:sub_signatures}
\end{figure*}

\begin{example}[Opaque signature]
\label{ex:ml:opaque-sign}
We consider a module $M = \tuple{\atcmod{\intType},\attmod{0}}$ and show that it has the opaque signature $\sigma_O = \Sigma \alpha:\atksign{\basekind}.\atcsign{\alpha}$ 
\footnote{Notice that in Example~\ref{ex:ml:module-signature}, we have $M = \tuple{\atcmod{\intType},\tuple{\attmod{0},\attmod{\lambda x:\intType.e}}}$ and $\sigma_O = \Sigma \alpha:\atksign{\basekind}.\Sigma \beta:\atcsign{\alpha}.\atcsign{\alpha \rightarrow \tau_f}$. From the static semantics, we can also derive that \typeEvalP{}{M:\sigma_O}. Here, to simplify the presentation, we have $M = \tuple{\atcmod{\intType},\attmod{0}}$ and $\sigma_O = \Sigma \alpha:\atksign{\basekind}.\atcsign{\alpha}$}.
We then have that $M$ is a pure module of the signature $\sigma_O$ (i.e. \typeEvalP{}{M:\sigma_O}).

First, we have that $M$ is a module of a transparent signature, i.e. \typeEvalP{}{M:\Sigma\alpha:\atksign{S(\intType)}.\atcsign{\intType}} by instantiating the \WfmPr\ rule.
$$\small \LabelRule{\WfmPr}{
    \typeEvalP{}{\atcmod{\intType}:\atksign{S(\intType)}}\\
    \typeEvalP{}{\attmod{0}:\atcsign{\intType}}\\
    \alpha \not\in FV(\atcsign{\intType})
}
{\typeEvalP{}{M: \Sigma \alpha:\atksign{S(\intType)}.\atcsign{\intType}}}$$

Next, we have that \typeEval{}{\Sigma \alpha:\atksign{S(\intType)}.\atcsign{\intType} \leq \Sigma \alpha:\atksign{\basekind}.\atcsign{\alpha}}, by the derivation described in Fig.~\ref{fig:derivation:example:sub_signatures}:

Finally, it follows that \typeEvalP{}{M:\sigma_O}.
$$\small
\LabelRule{\WfmSub}{
    \typeEvalP{}{\tuple{\atcmod{\intType},\attmod{0}}:\Sigma \alpha:\atksign{S(\intType)}.\atcsign{\intType} } \\
     \typeEval{}{\Sigma \alpha:\atksign{S(\intType)}.\atcsign{\intType} \leq \Sigma \alpha:\atksign{\basekind}.\atcsign{\alpha}}}
{\typeEvalP{}{M:\Sigma \alpha:\atksign{\basekind}.\atcsign{\alpha} }}
$$
\end{example}

From \cite{Crary-POPL-17}, we have the following lemma about the correctness of \fstop{M}{c} operation.
\begin{lemma}
If \typeEvalP{\Gamma}{M:\sigma} then \typeEval{\Gamma}{\fstop{M}{c}} and \typeEval{\Gamma}{c:\fstsign{\sigma}}.
\end{lemma}

To facilitate the proofs about TRNI for ML, from the static semantics, we have the following lemma.

%\todoinline{DN}{Please move proof to appendix.
%\newline
%[MN] I moved it.
%}

\begin{lemma}[Weakening]
\label{lem:ml:weakening}
{Suppose that \typeEval{}{\Gamma,\alpha:k\ \ok}.}
It follows that:
\begin{itemize}
\item if \typeEval{\Gamma}{\sigma:\sign}, then \typeEval{\Gamma,\alpha:k}{\sigma:\sign},
\item if \typeEval{\Gamma}{c:k'}, then \typeEval{\Gamma,\alpha:k}{c:k'},
\item if \typeEval{\Gamma}{k':\kind}, then \typeEval{\Gamma,\alpha:k}{k':\kind},
\item if \typeEval{\Gamma}{k_1 \equiv k_2:\kind}, then \typeEval{\Gamma,\alpha:k}{k_1 \equiv k_2:\kind},
\item if \typeEval{\Gamma}{k_1 \leq k_2:\kind}, then \typeEval{\Gamma,\alpha:k}{k_1 \leq k_2:\kind},
\item if \typeEval{\Gamma}{c_1 \equiv c_2:k'}, then \typeEval{\Gamma,\alpha:k}{c_1 \equiv c_2:k'}.
\end{itemize}
\end{lemma}
\begin{proof}
We prove this lemma by induction on the derivation of \typeEval{\Gamma}{\sigma:\sign}, \typeEval{\Gamma}{c:k'}, \typeEval{\Gamma}{k':\kind}, \typeEval{\Gamma}{k_1 \equiv k_2:\kind}, \typeEval{\Gamma}{k_1 \leq k_2:\kind}, and \typeEval{\Gamma}{c_1 \equiv c_2:k'}.
\end{proof}

\paragraph{Dynamic semantics.}
The dynamic semantics is given by call-by value semantics. 
We have dynamic semantics for terms \typeEval{\Gamma}{e \transitml e'} and for modules \typeEval{\Gamma}{M \transitml M'} (see Fig.~\ref{fig:ml:semantics:dynamic:term} and Fig.~\ref{fig:ml:semantics:dynamic:module}), where 
{the context $\Gamma$ is only used to extract the static part of module values}.
Open term values and module values are as below.
\begin{align*}
v & := x \sep \unitval \sep n \sep \lambda x:\tau.e \sep \tuple{v,v} \sep \Lambda \alpha:k.e & \text{Term values} \\ 
	& \quad\sep \pack{c,v}{\exists \alpha:k.\tau} 
\\
V & := m \sep \unitmod \sep \atcmod{c} \sep \attmod{v} \sep \tuple{V,V} & \text{Module values}\\
	&\quad \sep \lambdagn \alpha/m:\sigma.M \sep \lambdaap \alpha/m:\sigma.M
\end{align*}

%\redtext{For dynamics semantics, we have \typeEval{\Gamma}{e \transitml e'} and \typeEval{\Gamma}{M \transitml M'} described in respectively Fig.~\ref{fig:ml:semantics:dynamic:term} and Fig.~\ref{fig:ml:semantics:dynamic:module}.
In the tstep\_fix rule, $\lambda \_:\unittype.\fix{\tau}{e}$ means that the term variable bound by $\lambda$ is a fresh variable.
To be precise, the variable must not be in \dom{\Gamma}, and in addition it should be canonically chosen, to maintain strict determinacy of evaluation. In Crary's deBruin representation this is automatic.

\purpletext{We use $V,W$ as metavariables for module values.
We write \terminating{e} when the evaluation of $e$ terminates.
Similarly, we have \terminating{M}}.

%\todoinline{DN}{
%I think we should add the following: ``To be precise, the variable must not be in \dom{\Gamma}, and in addition it should be canonically chosen, to maintain strict determinacy of evaluation. In Crary's deBruin representation this is automatic.''
%\\
%By the way, to be sure I understand fix, I wrote this simple example, assuming if/else available.
%Let $ty := (\intType \to \intType)$
%and $\Gamma := mul: \intType\to\intType\to\intType, pred: \intType\to\intType$.
%Let $body := \lambda rec: (\unittype\to ty). \lambda n:\intType. 
%        ~ if ~ n==0 ~ then ~ 1 ~ else ~ mul~n~(rec~\unitval~(pred~n))$
%and let $fact := \fix{ty}{body}$.
%Have $\typeEval{\Gamma}{body: (\unittype\to ty)\to ty}$
%hence 
%$\typeEval{\Gamma}{ fact : ty }$.
%\newline
%[MN] I agree.
%}

\begin{figure}
\begin{mathpar}
\footnotesize
\LabelRule{tstep\_app1}{
\typeEval{\Gamma}{e_1 \transitml e_1'}}
{\typeEval{\Gamma}{e_1\ e_2 \transitml e_1'\ e_2}} \and
%%%%%%%%%%%%%
\LabelRule{tstep\_app2}{
    \typeEval{\Gamma}{e_2 \transitml e_2'}
}
{\typeEval{\Gamma}{v_1\ e_2 \transitml v_1\ e_2'}} \and
%%%%%%%%%%%%%%
\LabelRule{tstep\_app3}{
    ~
}
{\typeEval{\Gamma}{(\lambda x:\tau.e_1)\ v_2 \transitml e_1[x\mapsto v_2]}} \and
%%%%%%%%%%%%%%
\LabelRule{tstep\_pair1}{
    \typeEval{\Gamma}{e_1 \transitml e_1'}
}
{\typeEval{\Gamma}{\tuple{e_1,e_2} \transitml \tuple{e_1',e_2}}} \and
%%%%%%%%%%%%%%
\LabelRule{tstep\_pair2}{
    \typeEval{\Gamma}{e_2 \transitml e_2'}
}
{\typeEval{\Gamma}{\tuple{v_1,e_2} \transitml \tuple{v_1,e_2'}}} \and
%%%%%%%%%%%%%%
\LabelRule{tstep\_pi11}{
    \typeEval{\Gamma}{e \transitml e'}
}
{\typeEval{\Gamma}{\prj{1}{e} \transitml \prj{1}{e'}}} \and
%%%%%%%%%%%%%%
\LabelRule{tstep\_pi12}{
    ~
}
{\typeEval{\Gamma}{\prj{1}{\tuple{v_1,v_2}} \transitml v_1}} \and
%%%%%%%%%%%%%%
\LabelRule{tstep\_pi21}{
    \typeEval{\Gamma}{e \transitml e'}
}
{\typeEval{\Gamma}{\prj{2}{e} \transitml \prj{2}{e'} }} \and
%%%%%%%%%%%%%%
\LabelRule{tstep\_pi22}{
    ~
}
{\typeEval{\Gamma}{\prj{2}{\tuple{v_1,v_2}} \transitml v_2}} \and
%%%%%%%%%%%%%%
\LabelRule{tstep\_papp1}{
    \typeEval{\Gamma}{e \transitml e'}
}
{\typeEval{\Gamma}{e[c] \transitml e'[c]}} \and
%%%%%%%%%%%%%%
\LabelRule{tstep\_papp2}{
    ~
}
{\typeEval{\Gamma}{(\Lambda \alpha:k.e)[c] \transitml e[\alpha\mapsto c]}} \and
%%%%%%%%%%%%%%
\LabelRule{tstep\_pack}{
    \typeEval{\Gamma}{e \transitml e'}
}
{\typeEval{\Gamma}{\pack{c,e}{\exists\alpha:k.\tau} \transitml \pack{c,e'}{\exists\alpha:k.\tau}}} \and
%%%%%%%%%%%%%%
\LabelRule{tstep\_unpack1}{
    \typeEval{\Gamma}{e_1 \transitml e_1'}
}
{\typeEval{\Gamma}{\unpack{\alpha,x}{e_1}{e_2} \transitml \unpack{\alpha,x}{e_1'}{e_2}}} \and
%%%%%%%%%%%%%%
\LabelRule{tstep\_unpack2}{
    ~
}
{\typeEval{\Gamma}{\unpack{\alpha,x}{(\pack{c,v}{\exists \alpha:k.\tau})}{e_2} \transitml} \\\\ {e[\alpha\mapsto c,x\mapsto v]}} \and
%%%%%%%%%%%%%%
\LabelRule{tstep\_fix}{
    ~
}
{\typeEval{\Gamma}{\fix{\tau}{e} \transitml e\ (\lambda \_:\unittype.\fix{\tau}{e})}} \and
%%%%%%%%%%%%%%
\LabelRule{tstep\_lett1}{
    \typeEval{\Gamma}{e_1 \transitml e_1'}
}
{\typeEval{\Gamma}{\letexp{x=e_1}{e_2} \transitml \letexp{x=e_1'}{e_2}}} \and
%%%%%%%%%%%%%%
\LabelRule{tstep\_lett2}{
    ~
}
{\typeEval{\Gamma}{\letexp{x=v_1}{e_2} \transitml e_2[x\mapsto v_1]}} \and
%%%%%%%%%%%%%%
\LabelRule{tstep\_letm1}{
    \typeEval{\Gamma}{M \transitml M'}
}
{\typeEval{\Gamma}{\letexp{\alpha/m = M}{e} \transitml \letexp{\alpha/m = M'}{e}}} \and
%%%%%%%%%%%%%%
\LabelRule{tstep\_letm2}{
    \typeEval{\Gamma}{\fstop{V}{c}}
}
{\typeEval{\Gamma}{\letexp{\alpha/m = V}{e} \transitml e[\alpha\mapsto c,m\mapsto V]}} \and
%%%%%%%%%%%%%%
\LabelRule{tstep\_ext1}{
    \typeEval{\Gamma}{M \transitml M'}
}
{\typeEval{\Gamma}{\extract{M} \transitml \extract{M'}}} \and
%%%%%%%%%%%%%%
\LabelRule{tstep\_ext2}{
    ~
}
{\typeEval{\Gamma}{\extract{\attmod{v}} \transitml v}} 
%%%%%%%%%%%%%
\end{mathpar}
\caption{Dynamic semantics - Terms \typeEval{\Gamma}{e \transitml e'}}
\label{fig:ml:semantics:dynamic:term}
\end{figure}

\begin{figure}
\begin{mathpar}
\footnotesize
\LabelRule{mstep\_dyn}{
    \typeEval{\Gamma}{e \transitml e'}
}
{\typeEval{\Gamma}{\attmod{e} \transitml \attmod{e'}}} \and
%%%%%%%%%%%%%%%%%%%%%%%%%%%%
\LabelRule{mstep\_appgn1}{
    \typeEval{\Gamma}{M_1 \transitml M_1'}
}
{\typeEval{\Gamma}{M_1\ M_2 \transitml M_1'\ M_2}} \and
%%%%%%%%%%%%%%%%%%%%%%%%%%%%
\LabelRule{mstep\_appgn2}{
    \typeEval{\Gamma}{M_2 \transitml M_2'}
}
{\typeEval{\Gamma}{V_1\ M_2 \transitml V_1\ M_2'}} \and
%%%%%%%%%%%%%%%%%%%%%%%%%%%%
\LabelRule{mstep\_appgn3}{
    \typeEval{\Gamma}{\fstop{V_2}{c_2}}
}
{\typeEval{\Gamma}{(\lambdagn \alpha/m:\sigma.M_1)\ V_2 \transitml M_1[\alpha\mapsto c_2,m\mapsto V_2]}} \and
%%%%%%%%%%%%%%%%%%%%%%%%%%%%
\LabelRule{mstep\_appap1}{
    \typeEval{\Gamma}{M_1 \transitml M_1'}
}
{\typeEval{\Gamma}{M_1\appapp M_2 \transitml M_1'\appapp M_2}} \and
%%%%%%%%%%%%%%%%%%%%%%%%%%%%
\LabelRule{mstep\_appap2}{
    \typeEval{\Gamma}{M_2 \transitml M_2'}
}
{\typeEval{\Gamma}{V_1\appapp M_2 \transitml V_1\appapp M_2'}} \and
%%%%%%%%%%%%%%%%%%%%%%%%%%%%
\LabelRule{mstep\_appap3}{
    \typeEval{\Gamma}{\fstop{V_2}{c_2}}
}
{\typeEval{\Gamma}{(\lambdaap \alpha/m:\sigma.M_1)\ V_2 \transitml M_2[\alpha\mapsto c_2,m\mapsto V_2]}} \and
%%%%%%%%%%%%%%%%%%%%%%%%%%%%
\LabelRule{mstep\_pair1}{
    \typeEval{\Gamma}{M_1 \transitml M_1'}
}
{\typeEval{\Gamma}{\tuple{M_1,M_2} \transitml \tuple{M_1',M_2}}} \and
%%%%%%%%%%%%%%%%%%%%%%%%%%%%
\LabelRule{mstep\_pair2}{
    \typeEval{\Gamma}{M_2 \transitml M_2'}
}
{\typeEval{\Gamma}{\tuple{V_1,M_2} \transitml \tuple{V_1,M_2'}}} \and
%%%%%%%%%%%%%%%%%%%%%%%%%%%%
\LabelRule{mstep\_pi11}{
    \typeEval{\Gamma}{M \transitml M'}
}
{\typeEval{\Gamma}{\prj{1}{M} \transitml \prj{1}{M'} }} \and
%%%%%%%%%%%%%%%%%%%%%%%%%%%%
\LabelRule{mstep\_pi12}{
    ~
}
{\typeEval{\Gamma}{\prj{1}{\tuple{V_1,V_2}} \transitml V_1}} \and
%%%%%%%%%%%%%%%%%%%%%%%%%%%%
\LabelRule{mstep\_pi21}{
    \typeEval{\Gamma}{M \transitml M'}
}
{\typeEval{\Gamma}{\prj{2}{M} \transitml \prj{2}{M'}}} \and
%%%%%%%%%%%%%%%%%%%%%%%%%%%%
\LabelRule{mstep\_pi22}{
    ~
}
{\typeEval{\Gamma}{\prj{2}{\tuple{V_1,V_2}} \transitml V_2}} \and
%%%%%%%%%%%%%%%%%%%%%%%%%%%%
\LabelRule{mstep\_unpack1}{
    \typeEval{\Gamma}{e \transitml e'}
}
{\typeEval{\Gamma}{\unpack{\alpha,x}{e}{(M:\sigma)} \transitml } \\\\  {\unpack{\alpha,x}{e'}{(M:\sigma)}}} \and
%%%%%%%%%%%%%%%%%%%%%%%%%%%%
\LabelRule{mstep\_unpack2}{
    ~
}
{\typeEval{\Gamma}{\unpack{\alpha,x}{(\pack{c,v}{\exists \alpha:k.\tau})}{(M:\sigma)} \transitml}\\\\ {M[\alpha\mapsto c, x\mapsto v]}} \and
%%%%%%%%%%%%%%%%%%%%%%%%%%%%
\LabelRule{mstep\_lett1}{
    \typeEval{\Gamma}{e_1 \transitml e_1'}
}
{\typeEval{\Gamma}{\letexp{x=e_1}{M_2} \transitml \letexp{x=e_1'}{M_2}}} \and
%%%%%%%%%%%%%%%%%%%%%%%%%%%%
\LabelRule{mstep\_lett2}{
    ~
}
{\typeEval{\Gamma}{\letexp{x=v_1}{M_2} \transitml M_2[x\mapsto v_1]}} \and
%%%%%%%%%%%%%%%%%%%%%%%%%%%%
\LabelRule{mstep\_letm1}{
    \typeEval{\Gamma}{M_1 \transitml M_1'}
}
{\typeEval{\Gamma}{\letexp{\alpha/m=M_1}{(M_2:\sigma)} \transitml \letexp{\alpha/m=M_1'}{(M_2:\sigma)}}} \and
%%%%%%%%%%%%%%%%%%%%%%%%%%%%
\LabelRule{mstep\_letm2}{
    \typeEval{\Gamma}{\fstop{V}{c}}
}
{\typeEval{\Gamma}{\letexp{\alpha/m=V}{(M:\sigma)} \transitml M[\alpha\mapsto c, m\mapsto V]}} \and
%%%%%%%%%%%%%%%%%%%%%%%%%%%%
\LabelRule{mstep\_seal}{
    ~
}
{\typeEval{\Gamma}{(M\seal \sigma) \transitml M}}
%%%%%%%%%%%%%%%%%%%%%%%%%%%%
\end{mathpar}
\caption{Dynamic semantics - Modules \typeEval{\Gamma}{M \transitml M'}}
\label{fig:ml:semantics:dynamic:module}
\end{figure}

\clearpage
\subsection{Logical relation}
%\subsection{Logical relation}
%\label{app:ml:logical-relation}
In order to define logical relation, we define some auxiliary notions as in \cite{Crary-POPL-17}
Given a relation $R \in \relDel{\tau,\tau'}$, \rarb{s}{R} contains continuations that agree on values related by $R$.
Conversely, given a relation $S$ on continuations, related continuations by $S$ agree on terms in \rarb{t}{S}.
\purpletext{From \rarb{s}{\_} and \rarb{t}{\_}, we define Pitts closed relations}.
%\newtext{We say that a relation $R \in \relDel{\tau,\tau'}$ is {\em Pitts closed} if $R = \rarb{st}{R}$}.

\begin{definition}[Closure]
\label{def:closure}
For $R \in \relDel{\tau,\tau'}$, define
$$\rarb{s}{R} \triangleq \{\tuple{v:\tau \rightarrow\unittype, v':\tau\rightarrow\unittype} \sep 
\forall \tuple{w,w'} \in R: \terminating{vw} \Leftrightarrow \terminating{v'w'}\}.
$$
For $S \in \relDel{\tau \rightarrow \unittype, \tau'\rightarrow\unittype}$, define
\( \rarb{t}{S} \triangleq \{ \tuple{w:\tau,w':\tau}\sep \forall \tuple{v,v'} \in S, \terminating{v w} \Leftrightarrow \terminating{v' w'}\} \).

For $R \in \relDel{\tau,\tau'}$, define 
$$\rarb{ev}{R}  \triangleq \{\tuple{e:\tau,e':\tau'} \sep \terminating{e} \Leftrightarrow \terminating{e'}, 
	 \forall v,v'. e \reduce v \implies e' \reduce v' \implies \tuple{v,v'} \in R \} 
$$
Relation $R \in \relDel{\tau,\tau'}$ is {\em Pitts closed} if $R = \rarb{st}{R}$.
\end{definition}

\purpletext{We have similar definitions for relations defined on closed signatures $\sigma_1$ and $\sigma_2$, where the signatures for continuations are $\Pign \alpha:\sigma_1.\unitsign$ and $\Pign \alpha:\sigma_2.\unitsign$,
which can be abbreviated as $\sigma_1\to\unitsign$ and $\sigma_2\to\unitsign$.
}

Apropos the \textsf{stev} closure, we may be able to infer indirectly that two terms are related by \rarb{stev}{R} when these two terms depend on terms related by \rarb{stev}{Q}.
%\redtext{The \textsf{stev} closure allows us to infer indirectly some properties.
%That is we may be able to infer that a pair of terms related by \rarb{stev}{R} depends on terms related by a $\rarb{stev}{Q}$ only by looking at terms related by $Q$}.

\begin{definition}
{Suppose \typeEval{x:\varrho}{e:\tau}.
We say that $x$ is {\em active} in $e$ if for all closed $e'$ s.t. $\typeEval{}{e':\varrho}$, \terminating{e[x\mapsto e']} implies \terminating{e'}.}
\end{definition}

\begin{lemma}
\label{lem:ml:monadic}{Suppose $Q \in \relDel{\varrho_1,\varrho_2}$, $R \in \relDel{\tau_1,\tau_2}$, \typeEval{x:\varrho_i}{e_i:\tau_i}, and $x$ is active in $e_1$ and $e_2$.
If for all $\tuple{v_1,v_2} \in Q$, $\tuple{e_1[x\mapsto v_1],e_2[x\mapsto v_2]} \in \rarb{stev}{R}$, then 
$$\forall \tuple{e_1',e_2'} \in \rarb{stev}{Q}. \tuple{e_1[x\mapsto e_1'],e_2[x\mapsto e_2']} \in \rarb{stev}{R}.$$}
\end{lemma}

From \cite{Crary-POPL-17} and Lemma wf\_mr in \cite{crary-17-impl}, we have the following lemma about properties of logical interpretations of constructors, kinds and signatures.

\begin{lemma}
\label{lem:ml:crary:interpretation}
Suppose that \typeEval{}{\Gamma\ \ok} and $\tuple{\rho,\rho'} \in \interenv{\Gamma}$. Then
\begin{itemize}
\item If \typeEval{\Gamma}{c:k}, then $\tuple{\rho_L(c),\rho_R(c), \inter{c}{\rho},\inter{c}{\rho'}} \in \inter{k}{\rho}$.
\item If \typeEval{\Gamma}{k_1 \equiv k_2}, then $\inter{k_1}{\rho} = \inter{k_2}{\rho}$
\item If \typeEval{\Gamma}{\sigma:\sign}, then $\inter{\sigma}{\rho} = \inter{\sigma}{\rho'}$.
\item If \typeEval{\Gamma}{\sigma \equiv \sigma':\sign}, then $\inter{\sigma}{\rho} = \inter{\sigma'}{\rho}$.
\end{itemize}
\end{lemma}

\paragraph{Precandiate.}
As in \cite{Crary-POPL-17}, we next present simple kinds which are used in the definitions of logical interpretations of dependent kinds. 
A {\em simple kind} is a kind that does not have singleton kinds.
Given a kind $k$, \simple{k} returns a simple kind by replacing singleton kinds in $k$ with \basekind.
A {\em candidate} of a kind is a pre-candidate which is in an interpretation of a kind.
%\newtext{A {\em simple kind} is a kind that does not have singleton kinds.
%Given a kind $k$, \simple{k} returns a simple kind by replacing singleton kinds in $k$ with \basekind. 
%We then define {\em pre-candidates} indexed by simple kinds}.
\purpletext{We use $Q$ as a meta-variable for pre-candidates in general, $\Phi$ for pre-candidates over function kinds, $P$ for pre-candidates over pair kinds, and $R$ for pre-candidates over \basekind.}

\begin{align*}
\small
\valSet & \triangleq \{v \sep \exists \tau. \typeEval{}{v:\tau}\}  \\
\conSet & \triangleq \{c \sep \exists k. \typeEval{}{c:k}\} \\
\precan{\basekind} & \triangleq \power{\valSet \times \valSet} \\
{\precan{\unitkind}} & \triangleq \set{\tuple{}}\\
\precan{k_1 \rightarrow k_2} & \triangleq \conSet \times \conSet \times \precan{k_1} \to \precan{k_2}\\
\precan{k_1 \times k_2} & \triangleq \precan{k_1} \times \precan{k_2}\\
\end{align*}

\paragraph{Logical interpretations for kinds and constructors.}
\airforcetext{Following \cite{Crary-POPL-17}, we generalize $\rho$ presented in \S\ref{sec:abstraction} to a mapping that maps constructor variables to tuples of the form \tuple{c,c',Q}, term variables to tuples of the form \tuple{v,v'}, and module variables to tuples of the form \tuple{V,V'}}.
Notice that in the simple language, for any $\alpha \in \dom{\rho}$, $\rho(\alpha) = R \in \relDel{\tau_1,\tau_2}$ for some $\tau_1$ and $\tau_2$.
If we use the notation in this section, then we have that $\rho(\alpha) = \tuple{\tau_1,\tau_2,R}$.
\purpletext{We write $\rho_L$ and $\rho_R$ for the substitutions that map every variable in the domain of $\rho$ to respectively the first element and the second element of the tuple that $\rho$ maps that variable to}.
If we do not have module variables, then, $\rho_L(\_)$ is similar to $\delta_1\gamma_1(\_)$ and $\rho_R(\_)$ is similar to $\delta_2\gamma_2(\_)$ in \S\ref{sec:abstraction}, where $\delta_1$ and $\delta_2$ are type substitutions in $\rho$, and $\gamma_1$ and $\gamma_2$ are term substitutions in $\rho$.

\purpletext{Logical interpretations for kinds and constructors are presented in Fig.~\ref{fig:ml:logrel:kind-con} 
\footnote{The presentation of logical interpretations here is similar to the one in \cite{Crary-POPL-17}. Notice that we can write definitions of logical interpretations in the form of inference rules as in \S\ref{sec:abstraction}.}.}
The interpretation of a kind $k$ is a tuple \tuple{c,c',Q,Q'} where $c$ and $c'$ are closed constructors of the kind $k$, and $Q$ and $Q'$ are candidates relating $c$ and $c'$.
Notice that since pre-candidates are defined only for simple kinds, in the logical interpretations of $\Pi \alpha:k_1.k_2$ and $\Sigma \alpha:k_1.k_2$, we use respectively \precan{\simple{\Pi \alpha:k_1.k_2}} and \precan{\simple{\Sigma \alpha:k_1.k_2}}.
In addition, in the definition for $\Pi \alpha:k_1.k_2$, since $\Phi$ and $\Phi'$ are pre-candidates of $\simple{\Pi \alpha:k_1.k_2}$, we require that $Q$ and $Q'$ are pre-candidates of \simple{k_1}
\footnote{Notice that from definition of \precan{\Pi \alpha:k_1.k_2}, we have that $\Phi \tuple{d,d',Q}$ and $\Phi \tuple{d'',d''',Q'}$ are in \precan{\simple{k_2}}.}.
The definitions for types of \basekind\ (e.g. $\tau_1 \rightarrow \tau_2$) are similar to the ones of the simple language.

\begin{figure*}
%\vspace{-10pt}
\small
\begin{align*}
\inter{\basekind}{\rho} & \triangleq \{ 
        \tuple{\tau,\tau',R,R} \sep \typeEval{}{\tau,\tau':\basekind},  
        R \in \relDel{\tau,\tau'}, 
        \text{$R$ Pitts closed}\}\\
\inter{S(c)}{\rho} & \triangleq \{\tuple{\tau,\tau,\inter{c}{\rho},\inter{c}{\rho}} \sep  
    \typeEval{}{\tau \equiv \rho_L(c):\basekind}, 
    \typeEval{}{\tau' \equiv \rho_R(c):\basekind},
    \inter{c}{\rho} \in \relDel{\tau,\tau'},
    \text{\inter{c}{\rho}\ Pitts closed} 
    \}\\
\inter{\Pi \alpha:k_1.k_2}{\rho} & \triangleq \{\tuple{c,c',\Phi,\Phi'} \sep
    \typeEval{}{c: \rho_L (\Pi \alpha: k_1. k_2)},
    \typeEval{}{c':\rho_R (\Pi \alpha: k_1. k)2)}, \\
    & \hspace{70pt}\Phi, \Phi' \in \precan{\simple{\Pi \alpha:k_1.k_2}}, \\
    & \hspace{70pt} \forall d,d',Q,Q', d'', d'''. {Q, Q' \in \precan{\simple{k_1}}} \implies \\
    & \hspace{90pt} \tuple{d,d',Q,Q'} \in \inter{k_1}{\rho} \implies \\
    & \hspace{110pt}\typeEval{}{d \equiv d'':\rho_L(k_1)} \implies \typeEval{}{d' \equiv   d''' \in \rho_R(k_1)}  \implies \\
    & \hspace{140pt}\tuple{cd, c'd', \Phi\tupletwo{d,d',Q},\Phi'\tupletwo{d'', d''', Q'}} \in \inter{k_2}{\rho,\alpha\mapsto \tuple{d,d',Q}}
     \}    \\
\inter{\Sigma \alpha:k_1.k_2}{\rho} & \triangleq \{\tuple{c,c',P,P'} \sep 
    \typeEval{}{c:\rho_L(\Sigma \alpha:k_1.k_2)},
    \typeEval{}{c':\rho_R(\Sigma \alpha:k_1.k_2)},\\
    & \hspace{70pt} P,P' \in \precan{\simple{\Sigma \alpha:k_1.k_2}} \\
    & \hspace{70pt} \tuple{\prj{1}{c},\prj{1}{c'},\prj{1}{P},\prj{1}{P'}} \in \inter{k_1}{\rho}, \\
    & \hspace{70pt} \tuple{\prj{2}{c},\prj{2}{c'},\prj{2}{P},\prj{2}{P'}} \in \inter{k_2}{\rho,\alpha\mapsto \tuple{\prj{1}{c},\prj{1}{c'},\prj{1}{P}}}
    \}     \\
\inter{\unitkind}{\rho} & \triangleq  \{\tuple{c,c',\tuple{},\tuple{}}\sep 
    \typeEval{}{c:\unitkind}, \typeEval{}{c':\unitkind}
    \}\\ 
\\
\interset{k}{\rho} & \triangleq \set{\tuple{c,c',Q} \sep \tuple{c,c',Q,Q} \in \inter{k}{\rho}}    \\
 \\
 %%%%%
\inter{\alpha}{\rho} & \triangleq Q,\ \text{where}\ \rho(\alpha) = \tuple{c,c',Q}
\hspace{30pt}
\inter{\lambda \alpha:k.c}{\rho} \triangleq \lambda\tuple{d,d',Q}\in \interset{k}{\rho}.\inter{c}{\rho,\alpha\mapsto\tuple{d,d',Q}}\\
\inter{c_1c_2}{\rho} & \triangleq \inter{c_1}{\rho}\tuple{\rho_L(c_2),\rho_R(c_2),\inter{c_2}{\rho}}
\hspace{25pt}
\inter{\tuple{c_1,c_2}}{\rho} \triangleq \tuple{\inter{c_1}{\rho},\inter{c_2}{\rho}}\\
\inter{\prj{i}{c}}{\rho} & \triangleq \prj{i}{\inter{c}{\rho}}
\hspace{124pt}
\inter{\unitcon}{\rho} \triangleq \tuple{}\\
%%%%%%%
\inter{\unittype}{\rho} & \triangleq \set{\tuple{\unitval,\unitval}}\\
\inter{\intType}{\rho} & \triangleq {\set{\tuple{v,v} \sep \typeEval{}{v:\intType}}}\\
\inter{\tau_1 \rightarrow\tau_2}{\rho} & \triangleq \{
    \tuple{v_1,v_2} \sep \typeEval{}{v_1:\rho_L(\tau_1 \rightarrow\tau_2}), 
        \typeEval{}{v_2:\rho_R(\tau_1 \rightarrow\tau_2)}, \forall \tuple{v_1',v_2'} \in \inter{\tau_1}{\rho}. \tuple{v_1v_1',v_2v_2'} \in \rev{\inter{\tau_2}{\rho}}
    \} \\
\inter{\tau_1 \times \tau_2}{\rho}  & \triangleq \{
    \tuple{\tuple{v_1,v_2},\tuple{v_1',v_2'}} \sep 
        \typeEval{}{\tuple{v_1,v_2}:\rho_L(\tau_1\times\tau_2)}, \:
        \typeEval{}{\tuple{v_1',v_2'}:\rho_R(\tau_1\times\tau_2)}, \\
        & \hspace{70pt} \tuple{v_1,v_1'} \in \inter{\tau_1}{\rho}, 
        \tuple{v_2,v_2'} \in \inter{\tau_2}{\rho}        
    \}\\
\inter{\forall\alpha:k.\tau}{\rho} & \triangleq \{\tuple{v,v'}\sep
    \typeEval{}{v:\rho_L(\forall\alpha:k.\tau)}, \:
    \typeEval{}{v':\rho_R(\forall\alpha:k.\tau)},\\
    & \hspace{70pt}\forall \tuple{c,c',Q} \in \interset{k}{\rho}. \tuple{v[c],v'[c']} \in \inter{\tau}{\rho[\alpha\mapsto\tuple{c,c',Q}]}
    \}\\
\inter{\exists\alpha:k.\tau}{\rho} & \triangleq \{\tuple{v,v'}\sep
    \typeEval{}{v:\rho_L(\exists\alpha:k.\tau)},
    \typeEval{}{v':\rho_R(\exists\alpha:k.\tau)},\\
    & \hspace{70pt} \exists\tuple{c,c',Q} \in \interset{k}{\rho},
        \exists v_0,v_0',k',k'',\tau',\tau''.\\
    & \hspace{100pt} v = \pack{c,v_0}{\exists \alpha:k'.\tau'}, \:
            v' = \pack{c',v_0'}{\exists \alpha:k''.\tau''},\\   
    & \hspace{140pt} \tuple{v_0,v_0'} \in \inter{\tau}{\rho,\alpha\mapsto\tuple{c,c',Q}}               
    \}^{\mathsf{st}}
\end{align*}
\vspace{-10pt}
\caption{Logical interpretation (kinds and constructors)}
\label{fig:ml:logrel:kind-con}
\vspace{-10pt}
\end{figure*}

\paragraph{Logical interpretations for signatures.}
Logical interpretations for signatures are presented in Fig.~\ref{fig:ml:logrel:sign}.
The logical interpretation of a signature $\sigma$ of modules is a set of \tuple{V_1,V_2,Q} where the dynamic values in $V_1$ and $V_2$ are related at their types, and the static part (constructors) are related by $Q$.
\purpletext{For example, we consider the signature $\Sigma \alpha:\atksign{k}.\atcsign{\tau}$.
The logical interpretation of this signature with a $\rho$ is a set of \tuple{\tuple{V_1,V_1'},\tuple{V_2,V_2'},\tuple{P,P'}} where (1) $V_1 = \atcmod{c}$, $V_2 = \atcmod{c'}$ for some $c$ and $c'$ s.t. $\tuple{c,c',P,P} \in \inter{k}{\eta}$; and (2) $V_2 = \attmod{v}$, $V_2' = \attmod{v'}$ for some $v$ and $v'$ s.t. $\tuple{v,v'} \in \inter{\tau}{\rho}$ (from the definition of \inter{\atcsign{\tau}}{\rho} we also know that $P' = \tuple{}$)}.

\begin{figure*}
\vspace{-10pt}
\small
\begin{align*}
\inter{\unitsign}{\rho} & \triangleq \{\tuple{\unitmod,\unitmod,\tuple}\}\\
%%%%%%%%%%%%%%%%%%%%%%%%%%%%%%%%%%%%%%%%%%%%%%%%%%%%%%%%%%%%
\inter{\atksign{k}}{\rho} & \triangleq \{\tuple{\atcmod{c},\atcmod{c'},Q} \sep
    \tuple{c,c',Q} \in \interset{k}{\rho}
    \}\\
%%%%%%%%%%%%%%%%%%%%%%%%%%%%%%%%%%%%%%%%%%%%%%%%%%%%%%%%%%%%    
\inter{\atcsign{\tau}}{\rho} & \triangleq   \{\tuple{\attmod{v},\attmod{v'}, \tuple{}} \sep 
    \typeEval{}{v:\rho_L(\tau)}, \typeEval{}{v:\rho_R(\tau)},
    \tuple{v,v'} \in \inter{\tau}{\rho}
    \}\\
%%%%%%%%%%%%%%%%%%%%%%%%%%%%%%%%%%%%%%%%%%%%%%%%%%%%%%%%%%%%
\inter{\Pign \alpha:\sigma_1.\sigma_2}{\rho} & \triangleq \{\tuple{V,V',\tuple{}} \sep
    \typeEvalI{}{V: \rho_L(\Pign \alpha:\sigma_1.\sigma_2)},
    \typeEvalI{}{V': \rho_R(\Pign \alpha:\sigma_1.\sigma_2)},\\
    & \hspace{70pt} \forall \tuple{W,W',Q} \in \inter{\sigma_1}{\rho}.
        \tuple{VW,V'W'} \in \riev{\inter{\sigma_2}{\rho,\alpha\mapsto \tuple{\fstoptwo{W},\fstoptwo{W'},Q}}}
    \}\\
%%%%%%%%%%%%%%%%%%%%%%%%%%%%%%%%%%%%%%%%%%%%%%%%%%%%%%%%%%%%%
\inter{\Piap \alpha:\sigma_1.\sigma_2}{\rho} & \triangleq \{\tuple{V,V',\Phi} \sep
    \typeEvalI{}{V: \rho_L(\Piap \alpha:\sigma_1.\sigma_2)},
    \typeEvalI{}{V': \rho_R(\Piap \alpha:\sigma_1.\sigma_2)},\\
    & \hspace{70pt} \tuple{\fstoptwo{V},\fstoptwo{V'},\Phi} \in \interset{\fstsign{\Piap \alpha:\sigma_1.\sigma_2}}{\rho},\\
    & \hspace{70pt} \forall \tuple{W,W',Q} \in \inter{\sigma_1}{\rho}.
        \tuple{V\appapp W,V'\appapp W'} \in \rpev{\inter{\sigma_2}{\rho,\alpha\mapsto \tuple{\fstoptwo{W},\fstoptwo{W'},Q}}}
    \}\\
%%%%%%%%%%%%%%%%%%%%%%%%%%%%%%%%%%%%%%%%%%%%%%%%%%%%%%%%%%%%%    
\inter{\Sigma \alpha:\sigma_1.\sigma_2}{\rho} & \triangleq \{
    \tuple{V,V',P} \sep
    \typeEvalI{}{V:\rho_L(\Sigma \alpha:\sigma_1.\sigma_2)},
    \typeEvalI{}{V':\rho_R(\Sigma \alpha:\sigma_1.\sigma_2)},\\
    & \hspace{70pt}\exists V_1, V_1', V_2, V_2'. V = \tuple{V_1,V_2}, 
    V' = \tuple{V_1',V_2'},
    \tuple{V_1,V_1',\prj{1}{P}} \in \inter{\sigma_1}{\rho},\\
    & \hspace{100pt}         \tuple{V_2,V_2',\prj{2}{P}} \in \inter{\sigma_2}{\rho,\alpha\mapsto\tuple{\fstoptwo{V_1},\fstoptwo{V_1'}, \prj{1}{P}}}
    \}    \\
\\[-1ex]
%%%%%%%%%%%%%%%%%%%%%%%%%%%%%%%%%%%%%%%%%%%%%%%%%%%%%%%%%%%%
\rpev{\inter{\sigma}{\rho}} &\triangleq \{
    \tuple{M,M',Q} \sep \typeEvalP{}{M:\rho_L(\sigma)}, \: \typeEvalP{}{M':\rho_R(\sigma)}, \: \terminating{M} \Leftrightarrow \terminating{M'}, \\
    & \hspace{70pt} \forall V,V'. M \reduce V \implies M' \reduce V' \implies \tuple{V,V',Q} \in \inter{\sigma}{\rho}
    \}\\
\\[-1ex]
%%%%%%%%%%%%%%%%%%%%%%%%%%%%%%%%%%%%%%%%%%%%%%%%%%%%%%%%%%%%
\rarb{i}{\inter{\sigma}{\rho}} & \triangleq \rst{\{\tuple{V,V'}\sep \exists Q. \tuple{V,V',Q} \in \inter{\sigma}{\rho} \}}
\end{align*}
\vspace{-10pt}
\caption{Logical interpretation (signatures)}
\label{fig:ml:logrel:sign}
\vspace{-10pt}
\end{figure*}

\begin{definition}
We say that $\tuple{\rho,\rho'} \in \interenv{\Gamma}$ \purpletext{if whenever $\Gamma(\alpha) = k$}, there exists $\rho(\alpha) = \tuple{c_1,c_2,Q}$ and $\rho'(\alpha) = \tuple{c_1',c_2',Q'}$ s.t. \typeEval{}{c_1 \equiv c_1':\rho_L(k)}, \typeEval{}{c_2\equiv c_2':\rho_R(k)}, and $\tuple{c_1,c_2,Q,Q'} \in \inter{k}{\rho}$,

We say that $\rho \in \envfull{\Gamma}$ if $\tuple{\rho,\rho} \in \interenv{\Gamma}$ and 
\begin{itemize}
\item for all $x:\tau \in \Gamma$, there exists $\rho(x) = \tuple{v_1,v_2}$ s.t. $\tuple{v_1,v_2} \in \inter{\tau}{\rho}$,
\item for all $\alpha/m:\sigma \in \Gamma$, there exists $\rho(m) = \tuple{V_1,V_2}$ and $\rho(\alpha) = \tuple{\fstoptwo{V_1},\fstoptwo{V_2},Q}$ s.t. $\tuple{V_1,V_2,Q} \in \inter{\sigma}{\rho}$.
\end{itemize}
\end{definition}

Terms $e$ and $e'$ are {\em logically equivalent} at $\tau$ in $\Gamma$ (written as \typeEval{\Gamma}{e \logeq e':\tau}) if \typeEval{}{\Gamma\ \ok} implies \typeEval{\Gamma}{e, e':\tau}, and for all $\rho \in \envfull{\Gamma}$, $\tuple{\rho_L(e),\rho_R(e')} \in \rev{\inter{\tau}{\rho}}$.
{Notice that equivalence holds vacuously, if $\Gamma$ is not well formed, but we are never interested in such $\Gamma$.}

%Next is the abtraction theorem for the core calculus stating that logical equivalence is reflexive. 
\begin{theorem}[Abstraction theorem]
\label{thm:ml:abstraction}
Suppose that \typeEval{}{\Gamma \ok}.
If \typeEval{\Gamma}{e:\tau}, then \typeEval{\Gamma}{e \logeq e:\tau}.
%\item If \typeEvalI{\Gamma}{M:\sigma}, then \typeEvalI{\Gamma}{M \logeq M:\sigma}.
%\item If \typeEvalP{\Gamma}{M:\sigma}, then \typeEvalP{\Gamma}{M \logeq M:\sigma}.
%\end{enumerate}
\end{theorem}

{In \cite{Crary-POPL-17}, there are similar results for pure modules and impure modules.  Later we express security in terms of sealed modules, but our security proof only relies on the abstraction theorem for expressions.}

\section{TRNI for the Module Calculus}
\label{sec:ml:trni}
%\noteinline{\noteapp{The example section will be rewritten: less focus on the formalization, more focus on SML.}}

%\vspace{-5pt}
This section recapitulates the development of \S\ref{sec:trni} but using an encoding suited to the module calculus.
The free theorem that typing implies security (Theorem~\ref{thm:ml:free-thm:opaque}) is 
formulated for an open term in context of the public view, as in Theorem~\ref{thm:trni}.
We then develop a ``wrapper'' to encapsulate the typing problem in a closed form.
That could facilitate use of an unmodified ML compiler without recourse to an API for the typechecker.

%\vspace{-5pt}
\subsection{Declassification policy encoding}
\label{sec:ml-trni:encoding}
In this section, we present the encoding for declassification policies by using the module calculus.
Here, a declassification function can be written in the module calculus with recursive functions.
However, for simplicity and for coherent policy, we assume---as in \S\ref{sec:local-policies}---that the applications of declassifiers on confidential input values always terminate.
In \S\ref{sec:trni}, a view is a typing context that declares variables for inputs and for declassifiers.
Here, those are gathered in a signature and the view is a context that declares a module of that signature.

Let $L \subseteq \varPolicy{\policy}$ be a finite list of distinct confidential input variables from \varPolicy{\policy}. 
%\todoinline{DN}{Maybe we should introduce the shorthand $L\subseteq\varPolicy{\policy}$
%to indicate $L$ is a list of distinct confidential variables, so we can succinctly remind about $L$ at later places like Def. \ref{def:policyenv}.
%\newline
%[MN] I added this notation and used it in Def.~\ref{def:policyenv} and lemmas.
%}
An empty list is \emptylist. We write $x\concat L$ to concatenate a confidential input variable to $L$.
{In \S\ref{sec:trni} we define operations \enccon{-} and \encpub{-} that 
apply to policy variables and declasifiers, yielding the encoding of policy as typing contexts.  Here we use the same notation, but apply the operations to variable lists
and encode policy as signatures.
First, we define \enccon{L} to return a transparent signature of the policy.}
It is defined inductively as described in Fig.~\ref{fig:ml:signatures}.
As in \S\ref{sec:trni}, we use fresh constructor variables with names that indicate their role in the encoding.
%When the list is empty, the signature is \unitsign.
%\todoinline{DN}{The following explanation is obscure but I didn't find a good way to improve it.  
%It may help to just describe the case for declassifier $f$, and use the abbreviated notation. 
%\newline
%[MN] I changed the text.
%}
{For a confidential input $x$, basically, the signature is a pair containing information about the kind of its type, its type, and the types of associated declassifiers.
For example, for $x$ that can be declassified via $f$, the signature contains:  (1)  the kind of the type of $x$: \atksign{S(\intType)}, (2) the type of $x$: \atcsign{\intType}, and (3) the type of $f$: \atcsign{\intType \rightarrow \tau_f}.
%When $x$ can be declassified via $f\circ a$, in addition to these, we also have the kind of the type of the result of $a$ which is \atcsign{S(\intType)}.}

\begin{figure*}[!t]
\begin{equation*}
\small
\enccon{L} \triangleq  \begin{cases}
    \unitsign & \text{if $L = \emptylist$,}\\
    \Sigma \alpha_x:\atksign{S(\intType)}.\Sigma \alpha:\atcsign{\intType}.\enccon{L'} & \text{if $L = x\concat L'$, $x \not\in \dom{\decPolicy{\policy}}$,}\\
    \Sigma \alpha_f:\atksign{S(\intType)}.\Sigma \alpha_1:\atcsign{\intType}.\Sigma \alpha_2:\atcsign{\intType \rightarrow \tau_f}.\enccon{L'} & \text{if $L=x\concat L'$, $\decPolicy{\policy}(x) = f$,}
%    \Sigma \alpha_{f\circ a}: \atksign{S(\intType)}.\Sigma \alpha_f:\atksign{S(\intType)}.\Sigma \alpha_1:\atcsign{\intType}. & \text{if $L=x\concat L'$, $\decPolicy{\policy}(x) = f\circ a$.}\\
%    \hspace{10pt}\Sigma \alpha_2:\atcsign{\intType \rightarrow \intType}.\Sigma \alpha_3:\atcsign{\intType \rightarrow \tau_f}.\enccon{L'}
\end{cases}
\end{equation*}
\begin{equation*}
\small
\encpub{L}  \triangleq  \begin{cases}
    \unitsign & \text{if $L = \emptylist$,}\\
    \Sigma \alpha_x:\atksign{\basekind}.\Sigma \alpha:\atcsign{\alpha_x}.\encpub{L'}  & \text{if $L = x\concat L'$, $x \not\in \dom{\decPolicy{\policy}}$,}\\
    \Sigma \alpha_f:\atksign{\basekind}.\Sigma \alpha_1:\atcsign{\alpha_f}.\Sigma \alpha_2:\atcsign{\alpha_f \rightarrow \tau_f}.\encpub{L'} \quad\quad & \text{if $L=x\concat L'$, $\decPolicy{\policy}(x) = f$,}
%    \Sigma \alpha_{f\circ a}: \atksign{\basekind}.\Sigma \alpha_f:\atksign{\basekind}.\Sigma \alpha_1:\atcsign{\alpha_{f\circ a}}. & \text{if $L=x\concat L'$, $\decPolicy{\policy}(x) = f\circ a$.}\\
%    \hspace{10pt}\Sigma \alpha_2:\atcsign{\alpha_{f\circ a} \rightarrow \alpha_{f}}.\Sigma \alpha_3:\atcsign{\alpha_f \rightarrow \tau_f}.\encpub{L'}
\end{cases}
\end{equation*}
\caption{Transparent and opaque signatures for a policy \policy.}
\label{fig:ml:signatures}
\end{figure*}

\begin{example}[Transparent signature]
\label{ex:ml:trni:con-view}
For \policyoe\ (Example~\ref{ex:policy:odd-even}), since $f = \lambda x:\intType.x \modop  2$ is of the type $\intType \rightarrow \intType$, by applying the third case in \enccon{-} with $\tau_f = \intType$, we get the signature 
$$\enccon{x} = \Sigma \alpha_f:\atksign{S(\intType)}.\Sigma \alpha_1:\atcsign{\intType}.\Sigma \alpha_2:\atcsign{\intType \rightarrow \intType}.\unitsign$$

that can be abbreviated as $\tuple{\atksign{S(\intType)}, \tuple{\atcsign{\intType}, \tuple{\atcsign{\intType \rightarrow \intType}, \unitsign}}}$.

%In Example~\ref{ex:ml:trni:con-view}, the signatures for \policyoe\ written in the syntax of the module calculus is as below.

In ML it looks like
$$\code{sig  type t=int  val x:int  val f:int->int  end}.$$

%\begin{lstlisting}
%  sig
%    type t = int
%    val x:int
%    val f:int -> int
%  end    
%\end{lstlisting}

%\begin{lstlisting}
%sig  type t = int  val x:int  val f: int->int  end
%\end{lstlisting}

%For \policyhash\ (Example~\ref{ex:policy:hash}), since $f = \lambda x:\intType.x \modop  2^{64}$ is of the type $\intType \rightarrow \intType$, by applying the last case in \enccon{-} with $\tau_f = \intType$, we get the signature
%%\begin{multline*}
%$$\enccon{x}  = \Sigma \alpha_{f\circ a}: \atksign{S(\intType)}.\Sigma \alpha_f:\atksign{S(\intType)}.
%	 \Sigma \alpha_1:\atcsign{\intType}.\Sigma \alpha_2:\atcsign{\intType \rightarrow \intType}.\Sigma \alpha_3:\atcsign{\intType \rightarrow \intType}.\unitsign$$
%%\end{multline*}
%that can be abbreviated to \tuple{\atksign{S(\intType)}, \tuple{\atksign{S(\intType)}, \tuple{\atcsign{\intType}, \tuple{\atcsign{\intType \rightarrow \intType}, \tuple{\atcsign{\intType \rightarrow \intType},\unitsign}}}}}.
\end{example}

\purpletext{We overload \encpub{L} to get an opaque signature of the policy.
The idea is similar to \enccon{L}, except that here we use constructor variables of the \basekind\ kind for types of confidential inputs and in types of declassification functions.
The definition is described in Fig.~\ref{fig:ml:signatures}}.

%\begin{equation*}
%\small
%\encpub{L}  \triangleq  \begin{cases}
%    \unitsign & \text{if $L = \emptylist$,}\\
%    \Sigma \alpha_x:\atksign{\basekind}.\Sigma \alpha:\atcsign{\alpha_x}.\encpub{L'} & \text{if $L = x\concat L'$, $x \not\in \dom{\decPolicy{\policy}}$,}\\
%    \Sigma \alpha_f:\atksign{\basekind}.\Sigma \alpha_1:\atcsign{\alpha_f}.\Sigma \alpha_2:\atcsign{\alpha_f \rightarrow \tau_f}.\encpub{L'} & \text{if $L=x\concat L'$, $\decPolicy{\policy}(x) = f$,}\\
%    \Sigma \alpha_{f\circ a}: \atksign{\basekind}.\Sigma \alpha_f:\atksign{\basekind}.\Sigma \alpha_1:\atcsign{\alpha_{f\circ a}}. & \text{if $L=x\concat L'$, $\decPolicy{\policy}(x) = f\circ a$.}\\
%    \hspace{10pt}\Sigma \alpha_2:\atcsign{\alpha_{f\circ a} \rightarrow \alpha_{f}}.\Sigma \alpha_3:\atcsign{\alpha_f \rightarrow \tau_f}.\encpub{L'}
%\end{cases}
%\end{equation*}

\begin{example}[Opaque signature]
\label{ex:ml:trni:pub-view}
For \policyoe\ (Example~\ref{ex:policy:odd-even}), since $f = \lambda x:\intType.x \modop  2$ is of the type $\intType \rightarrow \intType$, by applying the third case in \encpub{-} with $\tau_f = \intType$, we get the signature 
$$\encpub{x} = \Sigma \alpha_f:\atksign{\basekind}.\Sigma \alpha_1:\atcsign{\alpha_f}.\Sigma \alpha_2:\atcsign{\alpha_f \rightarrow \intType}.\unitsign
$$
that can be abbreviated as $\Sigma \alpha_f:\atksign{\basekind}.\tuple{\atcsign{\alpha_f}, \tuple{\atcsign{\alpha_f \rightarrow \intType}, \unitsign}}$.
In ML it looks like \code{sig \ type t \  val x:t \  val f: t->int \  end}.
%\begin{lstlisting}
%sig  type t  val x:t  val f: t->int  end
%\end{lstlisting}

%For \policyhash\ (Example~\ref{ex:policy:hash}), since $f = \lambda x:\intType.x \modop  2^{64}$ is of the type $\intType \rightarrow \intType$, by applying the last case of \encpub{-} with $\tau_f = \intType$, we get the signature 
%%\begin{multline*}
%$$\encpub{x} = \Sigma \alpha_{f\circ a}: \atksign{\basekind}.\Sigma \alpha_f:\atksign{\basekind}. \Sigma \alpha_1:\atcsign{\alpha_{f\circ a}}.
%\Sigma \alpha_2:\atcsign{\alpha_{f\circ a} \rightarrow \alpha_f}.\Sigma \alpha_3:\atcsign{\alpha_f \rightarrow \intType}.\unitsign$$
%%\end{multline*}
%that can be abbreviated as $\Sigma \alpha_{f\circ a}: \atksign{\basekind}.\Sigma \alpha_f:\atksign{\basekind}.\tuple{\atcsign{\alpha_{f\circ a}}, \tuple{\atcsign{\alpha_{f\circ a} \rightarrow \alpha_f}, \tuple{\atcsign{\alpha_f \rightarrow \intType}, \unitsign}}}$.

\end{example}

\purpletext{Hereafter, we abuse \varPolicy{\policy} and use it as a list and we write \sigmaPolCon{\policy} and \sigmaPol{\policy} to mean respectively \enccon{\varPolicy{\policy}} and \encpub{\varPolicy{\policy}}}.

%\purpletext{We next define a signature for a policy}.
%\begin{definition}
%\label{def:ml:sigmaPol}
%\purpletext{For any signature $\sigma$, we say that $\sigma \in \policy$ if \typeEval{}{\sigma \equiv \sigmaPol{\policy}:\sign}.}
%\end{definition}

{In order to define TRNI, we define the confidential view and the public view as in \S\ref{sec:trni}.
The confidential view is based on the constructed transparent signature \sigmaPolCon{\policy}, 
and the public view is based on the constructed opaque signature \sigmaPol{\policy}.
}
%Technically, this equivalence class is important for Section~\ref{sec:ml-trni:wrapper} when we need to infer that a term is TRNI.}
%\footnote{We can do the same thing for the confidential view but we get nothing new from it.}.
\begin{equation}\label{eq:defMLviews} 
\conView{\policy}  \triangleq \alpha_\policy/m_\policy:\sigmaPolCon{\policy}
\qquad
\pubViewTerm{\policy}  \triangleq \alpha_\policy/m_\policy:\sigmaPol{\policy}
\end{equation}
To express what in Example~\ref{ex:trni:con-view} and Example~\ref{ex:trni:pub-view} 
(for \policyoe) is written $x_f\, x$,
in the module calculus $x$ is accessed as {\extracttwo{\prjtwo{1}{\prj{2}{m_\policyoe}}}} and $x_f$ is accessed as \extracttwo{\prjtwo{1}{\prjtwo{2}{\prj{2}{m_\policyoe}}}}.

From the definitions, we have that \sigmaPolCon{\policy} and \sigmaPol{\policy} are closed and well-formed signatures,
and \sigmaPolCon{\policy} is a subsignature of \sigmaPol{\policy}.

\subsection{TRNI}
\label{sec:ml-trni:trni}
%\vspace{-5pt}
In order to define an environment $\rho$ for the policy, we define relations $R_x$ and $R_f$
similar to the relations \indval{\alpha_x} and \indval{\alpha_f} in \S\ref{sec:trni}.
\begin{align*}
R_x & = \{\tuple{v_1,v_2}| \typeEval{}{v_1:\intType},\ \typeEval{}{v_2:\intType}\}\\
R_f & = \{\tuple{v_1,v_2}| \typeEval{}{v_1:\intType}, \typeEval{}{v_2:\intType}, \purpletext{\rev{\tuple{f\ v_1, f\ v_2} \in \inter{\tau_f}{\emptyset}}} \}
%R_{f\circ a} &= \{\tuple{v_1,v_2}| \typeEval{}{v_1:\intType},\ \typeEval{}{v_2:\intType},\ \purpletext{\tuple{f(a\ v_1), f(a\ v_2)} \in \rev{\inter{\tau_f}{\emptyset}}}\}
\end{align*}

Notice that $\emptyset$ in \rev{\inter{\tau}{\emptyset}} is the empty environment.

%\todoinline{DN}{Let's not use proof env in cases like this where it is a sketch and the full proof is in appendix.
%\newline
%[MN] Yes. Should I remove the phrase ``By induction on $\dots$'' in proofs (e.g. Lemma~\ref{lem:ml:wrapper:sub-signature:list})?
%}

\purpletext{Given a list $L$ of confidential inputs from \policy\ and an environment $\rho$, we say that $\condenv{\rho}{L}{\policy}$ when }
\begin{itemize}
\item \purpletext{$\rho$ maps $m_\policy$ in \pubViewTerm{\policy} to related module values $V_1$ and $V_2$ for some $V_1$ and $V_2$ s.t. $V_1$ and $V_2$ are of the transparent signature \enccon{\varPolicy{\policy}} and functions in $V_1$ and $V_2$ are declassification functions from the policy, and} 
{the confidential values in $V_1$ and $V_2$ are related by $R_x$, $R_f$, or $R_{f\circ a}$ according to the policy, and }

\item \purpletext{$\rho$ maps $\alpha_\policy$ to a tuple \tuple{c_1,c_2,Q} where $c_1$ and $c_2$ are static parts from respectively $V_1$ and $V_2$, and $Q$ depends on the policy. 
That is if an element in $Q$ is corresponding to an $\alpha_x:\atksign{\basekind}$, then this element is $R_x$, if an element in $Q$ is corresponding to an $\alpha_f:\atksign{\basekind}$, then this element is $R_f$, and if an element in $Q$ is corresponding to an $\alpha_{f\circ a}:\atksign{\basekind}$, then this element is $R_{f\circ a}$}.
%By Lemma~\ref{lem:ml:pitts-closure}, $R_x$, $R_f$ and $R_{f\circ a}$ are Pitts closed, hence they can be used as candidates for the \basekind\ kind.
\end{itemize}

%\todoinline{DN}{Check my rephrasing.
%\newline
%[MN] I agree.}

The definition of $\condenv{\rho}{L}{\policy}$ is as below.
Hereafter, we write $\airforcetext{\rho \respectfull \policy}$ when $\condenv{\rho}{\varPolicy{\policy}}{\policy}$.

\begin{definition}[full environments for $\policy$]\label{def:policyenv}
Given {$L \subseteq \varPolicy{\policy}$}, we define the set \predenv{L}{\policy} of environments by 
$\condenv{\rho}{L}{\policy}$ iff $\dom{\rho} = \set{\alpha_\policy,m_\policy} $ and
\begin{itemize}
\item if $L = \emptylist$ then
$\rho(m_\policy)  = \tuple{\unitmod,\unitmod}$ and 
$\rho(\alpha_\policy)  = \tuple{\unitcon,\unitcon, \tuple{}}$,
%% \begin{align*}
%% \rho(m_\policy) & = \tuple{\unitmod,\unitmod},\\
%% \rho(\alpha_\policy) & = \tuple{\unitcon,\unitcon, \tuple{}},
%% \end{align*}

\item if $L = x\concat L'$ and $x\not\in \dom{\decPolicy{\policy}}$ then
there are $\tuple{v_1,v_2} \in R_x$ and \condenv{\rho'}{L'}{\policy} with 
\begin{align*}
\rho(m_\policy) & = \tuple{\tuple{\atcmod{\intType},\tuple{\attmod{v_1},V_1'}},\tuple{\atcmod{\intType},\tuple{\attmod{v_2},V_2'}}},\\
\rho(\alpha_\policy) &= \tuple{\tuple{\intType, \tuple{\unitcon,c_1'}},\tuple{\intType, \tuple{\unitcon,c_2'}}, \tuple{R_x, \tuple{\tuple{},Q'}}},
\end{align*}
where $\rho'(m_\policy) = \tuple{V_1',V_2'}$ and $\rho'(\alpha_\policy) = \tuple{c_1',c_2',Q'}$,

\item if $L = x\concat L'$ and $\decPolicy{\policy}(x) = f$ then there are $\tuple{v_1,v_2} \in R_f$ and \condenv{\rho'}{L'}{\policy} with 
\begin{align*}
\rho(m_\policy) & = \langle\tuple{\atcmod{\intType},\tuple{\attmod{v_1},\tuple{\attmod{f},V_1'}}}, \\
	& \hspace{90pt} \tuple{\atcmod{\intType},\tuple{\attmod{v_2},\tuple{\attmod{f},V_2'}}} \rangle,\\
\rho(\alpha_\policy) & = \langle \tuple{\intType,\tuple{\unitcon,\tuple{\unitcon, c_1'}}}, \tuple{\intType,\tuple{\unitcon,\tuple{\unitcon, c_2'}}}, \\
	& \hspace{90pt} \tuple{R_f, \tuple{\tuple{}, \tuple{\tuple{}, Q'}}} \rangle,
\end{align*}

where $\rho'(m_\policy) = \tuple{V_1',V_2'}$ and $\rho'(\alpha_\policy) = \tuple{c_1',c_2',Q'}$.

%\item if $L = x\concat L'$ and $\decPolicy{\policy}(x) = f\circ a$ then there are $\tuple{v_1,v_2} \in R_{f\circ a}$ and \condenv{\rho'}{L'}{\policy} with 
%\begin{align*}
%\rho(m_\policy) &= \langle \tuple{\atcmod{\intType},\tuple{\atcmod{\intType}, \tuple{\attmod{v_1},\tuple{\attmod{a},\tuple{\attmod{f},V_1'}}}}},\\
%        & \hspace{40pt}\tuple{\atcmod{\intType},\tuple{\atcmod{\intType},\tuple{\attmod{v_2},\tuple{\attmod{a},\tuple{\attmod{f},V_2'}}}}} \rangle, \\
%\rho(\alpha_\policy) &= \langle \tuple{\intType, \tuple{\intType, \tuple{\unitcon, \tuple{\unitcon, \tuple{\unitcon, c_1'}}}}}, \tuple{\intType, \tuple{\intType, \tuple{\unitcon, \tuple{\unitcon, \tuple{\unitcon, c_2'}}}}} \\
%&\hspace{105pt} \tuple{R_{f\circ a}, \tuple{R_f, \tuple{\tuple{}, \tuple{\tuple{},\tuple{\tuple{},Q'} }}}}\rangle,
%\end{align*}
%
%where $\rho'(m_\policy) = \tuple{V_1',V_2'}$ and $\rho'(\alpha_\policy) = \tuple{c_1',c_2',Q'}$.

\end{itemize}
\end{definition}

\begin{example}[\condenvalt{\rho}{\policyoe}]
\label{ex:ml:trni:env:oe}
\purpletext{In this example, we present a full environment for \policyoe.}
%Notice that by following the similar process presented in this example, we can construct a full environment for \policyhash.
We first define $R_f$, where $f = \lambda x:\intType.x \modop 2$.
$$R_f = \{\tuple{v_1,v_2}| \typeEval{}{v_1:\intType}, \typeEval{}{v_2:\intType}, (v_1 \modop 2) =_\intType (v_2\modop 2)\}.$$

Following the definition of \condenvalt{\rho}{\policyoe}, we construct $\rho$ as below. Notice that $\tuple{2,4} \in R_f$.
\begin{align*}
\rho(m_\policyoe) & = \langle \tuple{\atcmod{\intType},\tuple{\attmod{2},\tuple{\attmod{\lambda x:\intType.x \modop 2},\unitmod}}}, \\
& \hspace{50pt} \tuple{\atcmod{\intType},\tuple{\attmod{4},\tuple{\attmod{\lambda x:\intType.x \modop 2},\unitmod}}} \rangle,\\
\rho(\alpha_\policyoe) & = \langle \tuple{\intType,\tuple{\unitcon,\tuple{\unitcon,\unitcon}}}, \\
	& \hspace{60pt} \tuple{\intType,\tuple{\unitcon,\tuple{\unitcon, \unitcon}}},\tuple{R_f, \tuple{\tuple{}, \tuple{\tuple{}, \tuple{}}}} \rangle.
\end{align*}

It follows that \condenvalt{\rho}{\policyoe}.
\end{example}

{Next we prove that if $\tau$ is a type in the public view \pubViewTerm{\policy}, then all its logical interpretations are the same.}
\begin{lemma}
\label{lem:ml:interpretation-type:policy}
\tealtext{If \condenvalt{\rho_1}{\policy}, \condenvalt{\rho_2}{\policy}, and \typeEval{\pubViewTerm{\policy}}{\tau:\basekind}, then $\inter{\tau}{\rho_1} = \inter{\tau}{\rho_2}$}.
\end{lemma}

%Since the logical interpretations of $\tau$ in the public view w.r.t. the policy are the same for any $\rho_1,\rho_2 \respectfull \policy$, we define indistinguishability based on an arbitrary \condenvalt{\rho}{\policy}.

Therefore, we define indistinguishability based on an arbitrary \condenvalt{\rho}{\policy}.

\begin{definition}[Indistinguishability]
\label{def:ml:indis:signature}
Suppose $\rho$ and $\tau$ satisfy  $\condenvalt{\rho}{\policy}$ and \typeEval{\pubViewTerm{\policy}}{\tau:\basekind}.
%\vspace{-5pt}
\begin{itemize}
\item Values $v_1$ and $v_2$ are indistinguishable at $\tau$ (written as $\tuple{v_1,v_2} \in \indval{\tau}$) if $\tuple{v_1,v_2} \in \inter{\tau}{\rho}$.
\item Terms $e_1$ and $e_2$ are indistinguishable at $\tau$ (written as $\tuple{e_1,e_2} \in \indterm{\tau}$) if $\tuple{e_1,e_2} \in \rev{\inter{\tau}{\rho}}$.
\end{itemize}
\vspace{-5pt}
\end{definition}
\begin{example}[Indistinguishability]
\label{ex:ml:trni:ind:oe}
We consider \policyoe\ (Example~\ref{ex:policy:odd-even}). 
As described in Example~\ref{ex:ml:trni:pub-view}, the opaque signature of the policy is $\sigmaPol{\policyoe} = \Sigma \alpha_f:\atksign{\basekind}.\Sigma \alpha_1:\atcsign{\alpha_f}.\Sigma \alpha_2:\atcsign{\alpha_f \rightarrow\intType}.\unitsign$.
Thus, the public view \pubViewTerm{\policyoe} is $\alpha_\policyoe/m_\policyoe:\sigmaPol{\policyoe}$.
Notice that since $\alpha_\policyoe$ and $m_\policyoe$ are twinned, it follows that $\alpha_\policyoe$ is of the kind $\fstsign{\sigmaPol{\policy}} = \Sigma \alpha_f:\basekind.\Sigma \alpha_1:\unitkind.\Sigma \alpha_2:\unitkind.\unitkind$.

We consider the type \prj{1}{\alpha_\policyoe}. 
By a rule for well-formed constructors (rule \WfcPrj1), we have that \typeEval{\pubViewTerm{\policyoe}}{\prj{1}{\alpha_\policyoe}:\basekind}.
Thus, we can define indistinguishability for this type.
As presented in Example~\ref{ex:ml:trni:env:oe}, \condenvalt{\rho}{\policyoe}.
Therefore, we have that $\indval{\prj{1}{\alpha_\policyoe}} = \inter{\prj{1}{\alpha_\policyoe}}{\rho} = R_f$ (the definition of $R_f$ is in Example~\ref{ex:ml:trni:env:oe}).
\end{example}

%\begin{example}[Indistinguishability - \policyhash]
%\label{ex:ml:trni:ind:hash}
%\redtext{We now consider \prj{1}{\alpha_\policyhash}, which is well-formed in the public view.
%As presented in Example~\ref{ex:ml:trni:env:hash}, \condenvalt{\rhoPol{\policyhash}}{\policyhash}.
%We have that $\indval{\prj{1}{\alpha_{\policyhash}}} = \inter{\prj{1}{\alpha_{\policyhash}}}{\rhoPol{\policyhash}} = R_{f\circ a}$.}
%\end{example}

Next, we define TRNI for the module calculus. 
{The definition here is similar to the one in \S\ref{sec:trni}}.

%\vspace{-5pt}
\begin{definition}[TRNI for the module calculus]
\label{def:ml:trni:signature}
\purpletext{A term $e$ is \TRNI{\policy,\tau} if \typeEval{\conView{\policy}}{e}, and \typeEval{\pubViewTerm{\policy}}{\tau:\basekind}, and for all $\condenvalt{\rho}{\policy}$, it follows that $\tuple{\rho_L(e), \rho_R(e)} \in \indterm{\tau}$}.
\end{definition}

\begin{example}
\trimming{We consider the program $e = (\extracttwo{\prjtwo{1}{\prjtwo{2}{\prj{2}{m_\policyoe}}}})\ (\extracttwo{\prjtwo{1}{\prj{2}{m_\policyoe}}})$, which is corresponding to the program $x_f\ x$ in Example~\ref{ex:trni:program:trni},
as noted following Eqn.~(\ref{eq:defMLviews}).
%We have that $e$ is \TRNI{\policyoe,\intType}.
}
\trimming{We now check $e$ with the definition of TRNI.
We consider an arbitrary \condenvalt{\rho}{\policyoe}.
As described in Example~\ref{ex:ml:trni:env:oe}, $\rho$ is as below, where $\tuple{v_1,v_2} \in R_f$.}
\begin{align*}
\rho(m_\policyoe) & = \langle \tuple{\atcmod{\intType},\tuple{\attmod{v_1},\tuple{\attmod{\lambda x:\intType.x \modop 2},\unitmod}}},\\
		& \hspace{40pt}\tuple{\atcmod{\intType},\tuple{\attmod{v_2},\tuple{\attmod{\lambda x:\intType.x \modop 2},\unitmod}}} \rangle ,\\
\rho(\alpha_\policyoe) & = \tuple{\tuple{\intType,\tuple{\unitcon,\tuple{\unitcon,\unitcon}}}, \tuple{\intType,\tuple{\unitcon,\tuple{\unitcon, \unitcon}}},\tuple{R_f, \tuple{\tuple{}, \tuple{\tuple{}, \tuple{}}}}}.
\end{align*}

\trimming{We have that $\rho_L(e) = f\ v_1 = v_1 \modop 2$ and $\rho_R(e) = f\ v_2 = v_2 \modop 2$.
Since $\tuple{v_1,v_2} \in R_f$, we have that $(v_1 \modop 2) =_\intType (v_2 \modop 2)$. Thus, $\tuple{\rho_L(e),\rho_R(e)} \in \rev{\inter{\intType}{\rho}}$. 
In other words, $\tuple{\rho_L(e),\rho_R(e)} \in \indterm{\intType}$.
Therefore, $e$ is \TRNI{\policyoe,\intType}.}
\end{example}

\subsection{Free theorem: typing in the public view implies security}
To apply the abstraction theorem to get the free theorem, we need 
the following.
%that if  \condenvalt{\rho}{\policy} then $\rho \in \interenvfull{\pubViewTerm{\policy}}$.

%To this aim, we first prove that $R_x$, $R_f$ and $R_{f\circ a}$ are Pitts closed.
%\vspace{-5pt}
\begin{lemma}
\label{lem:ml:consenv:property:two}
\tealtext{Suppose that $\condenvalt{\rho}{\policy}$.
It follows that $\rho \in \interenvfull{\pubViewTerm{\policy}}$.}
\end{lemma}
%\begin{proof}
%We claim that for any $L \subseteq \varPolicy{\policy}$ and any \condenv{\rho}{L}{\policy}, it follows that \condenvalt{\rho}{\interenv{\alpha_\policy/m_\policy:\encpub{L}}}.
%Then the result follows directly from the claim.
%
%We now prove the claim.
%Suppose that $\rho(\alpha_\policy) = \tuple{c_1,c_2,Q}$ and $\rho(m_\policy) = \tuple{V_1,V_2}$.
%From Lemma~\ref{lem:ml:consenv:property:one}, we have that:
%\begin{itemize}
%\item \typeEval{}{c_1:\rho_L(k)}, \typeEval{}{c_2:\rho_R(k)} where $k = \fstsign{\encpub{L}}$, and hence, it follows that \typeEval{}{c_1\equiv c_1:\rho_L(k)} and \typeEval{}{c_2\equiv c_2:\rho_R(k)}
%\item $\tuple{c_1,c_2,Q,Q} \in \inter{k}{\rho}$.
%\end{itemize}
%
%Therefore, we have that $\tuple{\rho,\rho} \in \interenv{\alpha_\policy/m_\policy:\encpub{L}}$.
%Thus, we only need to prove two following items:
%\begin{itemize}
%\item $\rho(\alpha_\policy) = \tuple{\fstoptwo{V_1},\fstoptwo{V_2},Q}$, and
%\item $\tuple{V_1,V_2,Q} \in \interenv{\encpub{L}}$.
%\end{itemize}
%
%These two items are proven by induction on $L$.
%\end{proof}

%We next prove that typability in the public view implies typability in the confidential view.
\begin{lemma}
\label{lem:ml_trni:typability_impl}
\tealtext{If \typeEval{\pubViewTerm{\policy}}{e:\tau}, then \typeEval{\conView{\policy}}{e}.}
\end{lemma}

%Next, we have the free theorem which is similar to Theorem~\ref{thm:trni}.
%\newtext{Notice that to avoid introducing additional notations, instead of requiring that $e$ has no type variable \footnote{Notice that for a term $e$ well-typed in the public view, it may be of the form $\lambda x:\tau.e'$ where $\tau$ is equivalent to a type variable of \basekind.}, 
%we require that $e$ is well-typed in the confidential view}.
%\begin{theorem}[Free theorem]

\textbf{Theorem~\ref{thm:ml:free-thm:opaque}.}
%\begin{theorem}
%\label{thm:ml:free-thm:opaque}
{If $\typeEval{\pubViewTerm{\policy}}{e:\tau}$, then $e$ is \TRNI{\policy,\tau}.}
%\end{theorem}
\begin{proof}
Since \typeEval{\pubViewTerm{\policy}}{e:\tau}, from Theorem~\ref{thm:ml:abstraction}, we have that \typeEval{\pubViewTerm{\policy}}{e \logeq e:\tau}.
Thus, for any $\rho \in \envfull{\pubViewTerm{\policy}}$, it follows that: 
$$\tuple{\rho_L(e),\rho_R(e)} \in \rev{\inter{\tau}{\rho}}.$$

We consider an arbitrary $\rho$  s.t. $\condenvalt{\rho}{\policy}$.
\purpletext{From Lemma~\ref{lem:ml:consenv:property:two}, it follows that $\rho \in \envfull{\pubViewTerm{\policy}}$}.
As proven above, we have that $\tuple{\rho_L(e),\rho_R(e)} \in \rev{\inter{\tau}{\rho}}$.
From the definition of indistinguishability, we have that $\tuple{\rho_L(e),\rho_R(e)} \in \indterm{\tau}$.
{In addition, since \typeEval{\pubViewTerm{\policy}}{e:\tau}, from Lemma~\ref{lem:ml_trni:typability_impl}, it follows that \typeEval{\conView{\policy}}{e}.}
Therefore, $e$ is \TRNI{\policy,\tau}.
\end{proof}

\begin{example}[Typing implies TRNI]
\label{ex:ml:trni:program:trni}
We consider the policy \policyoe. As described in Example~\ref{ex:ml:trni:con-view} and Example~\ref{ex:ml:trni:pub-view}, the transparent signature and the opaque signature of the policy are as below.
\begin{align*}
\sigmaPolCon{\policyoe} &= \Sigma \alpha_f:\atksign{S(\intType)}.\Sigma \alpha_1:\atcsign{\intType}.\Sigma \alpha_2:\atcsign{\intType \rightarrow \intType}.\unitsign\\
\sigmaPol{\policyoe} &= \Sigma \alpha_f:\atksign{\basekind}.\Sigma \alpha_1:\atcsign{\alpha_f}.\Sigma \alpha_2:\atcsign{\alpha_f \rightarrow \intType}.\unitsign
\end{align*}

Thus, the confidential view is $\conView{\policyoe} = \alpha_\policyoe/m_\policyoe:\sigmaPolCon{\policyoe}$, and the public view is $\pubViewTerm{\policyoe} = \alpha_\policyoe/m_\policyoe:\sigmaPol{\policyoe}$.
We now look at the program $e = e_1\ e_2$, where $e_1 = (\extracttwo{\prjtwo{1}{\prjtwo{2}{\prj{2}{m_\policyoe}}}})$ and $e_2 = (\extracttwo{\prjtwo{1}{\prj{2}{m_\policyoe}}})$. 
This program is corresponding to the program $x_f\ x$ in Example~\ref{ex:trni:program:trni},
as noted following definition (\ref{eq:defMLviews}). 
We have that \typeEval{\conView{\policyoe}}{e_1: \intType \rightarrow\intType},  \typeEval{\conView{\policyoe}}{e_2: \intType}, \typeEval{\pubViewTerm{\policyoe}}{e_1: \prj{1}{\alpha_\policyoe} \rightarrow \intType}, and  \typeEval{\pubViewTerm{\policyoe}}{e_2: \prj{1}{\alpha_\policyoe}}.
%\begin{itemize}
%\item \typeEval{\conView{\policyoe}}{e_1: \intType \rightarrow\intType},  and \typeEval{\conView{\policyoe}}{e_2: \intType},
%
%\item \typeEval{\pubViewTerm{\policyoe}}{e_1: \prj{1}{\alpha_\policyoe} \rightarrow \intType}, and  \typeEval{\pubViewTerm{\policyoe}}{e_2: \prj{1}{\alpha_\policyoe}},
%\end{itemize}
Therefore, we have that \typeEval{\conView{\policyoe}}{e:\intType}, and  \typeEval{\pubViewTerm{\policyoe}}{e:\intType} and hence, from Theorem~\ref{thm:ml:free-thm:opaque}, the program is \TRNI{\policyoe,\intType}.

%Similarly, we have that the expression corresponding 
%%DN thinks following commented term is correct but doesn't help understanding
%%$$(\extracttwo{\prjtwo{1}{\prjtwo{2}{\prjtwo{2}{\prjtwo{2}{\prj{2}{m_\policyhash}}}}}})\ ((\extracttwo{\prjtwo{1}{\prjtwo{2}{\prjtwo{2}{\prj{2}{m_\policyhash}}}}})\ (\extracttwo{\prjtwo{1}{\prjtwo{2}{\prj{2}{m_\policyhash}}}})),$$
%to the program $x_f(x_a\ x)$ in Example~\ref{ex:ml:trni:program:trni}, is well-typed in both views of \policyhash, and in the public view, its types is \intType.
%Thus, the program is \TRNI{\policyhash,\intType}.
\end{example}

\begin{example}
\label{ex:ml:trni:program:non-trni}
The purpose of this example is similar to the one of Example~\ref{ex:trni:program:non-trni}: to illustrate that if a program is well-typed in the confidential view and is not \TRNI{\policy,\tau} for some $\tau$ well-formed in the public view, then the type of the program in the public view is not equivalent to $\tau$ or the program is not well-typed in the public view. 

We consider the policy $\policyoe$ and the program $\extracttwo{\prjtwo{1}{\prj{2}{m_\policyoe}}}$, which is corresponding to the program $x$ in Example~\ref{ex:trni:program:non-trni}.
This program is not \TRNI{\policyoe,\intType} since $\extracttwo{\prjtwo{1}{\prj{2}{m_\policyoe}}}$ itself is
confidential and cannot be directly declassified. In the public view of the policy, the type of this program is $\prj{1}{\alpha_\policyoe}$ which is not equivalent to $\intType$.

We consider another program: $e = (\extracttwo{\prjtwo{1}{\prj{2}{m_\policyoe}}}) \modop 3$, which is corresponding to the program $x \modop 3$ in Example~\ref{ex:trni:program:non-trni}.
This program is not \TRNI{\policyoe,\prj{1}{\alpha_\policyoe}} since it may map indistinguishable inputs to non-indistinguishable outputs.
\trimming{For example, we consider \condenvalt{\rho}{\policyoe} presented in Example~\ref{ex:ml:trni:ind:oe}.
We have that $\rho_L(e) = 2 \modop 3 = 2$, and  $\rho_R(e) = 4 \modop 3 = 1$.
As described in Example~\ref{ex:ml:trni:ind:oe}, $\indval{\prj{1}{\alpha_\policyoe}} = R_f = \set{\tuple{v_1,v_2} \sep (v_1 \modop 2) =_\intType (v_2 \modop 2)}$.
Therefore $\tuple{1,2} \not \in  \indval{\prj{1}{\alpha_\policyoe}}$.}

\trimming{As explained above, $e$ is not \TRNI{\policyoe, \prj{1}{\alpha_\policyoe}}.}
In the public view, it is not well-typed since the type of \extracttwo{\prjtwo{1}{\prj{2}{m_\policyoe}}} is \prj{1}{\alpha_\policyoe}, which is not equivalent to \intType, and \modop\ expects \intType\ arguments.
\end{example}

\subsection{Wrapper}
\label{sec:ml-trni:wrapper}
{In this section, we will transform an open term to a closed module.
We then prove that if the closed module is well-typed in the empty context, then the original open term is well-typed in the public view and hence, $e$ is TRNI.
Thus we can use our approach with ordinary ML implementations. % where the type checkers only accept closed terms and closed modules.
}

If the source programs are already parameterized by one module for their confidential inputs and their declassification functions, then there is no need to modify source programs at all.

%To transform an open term, we first construct a dummy module value $\dummy$ for \policy.
%Intuitively, \dummy\ is a module value of the transparent signature \sigmaPolCon{\policy} where $\intType$ is used for static components, $0$ (arbitrary choice) is used as values for confidential inputs and declassification functions and actions are as described in \policy.
%The dummy module value \dummy\ for a policy \policy\ can be constructed by \consdummy{L} described below.
%From the construction and the relation between \sigmaPolCon{\policy} and \sigmaPol{\policy}, we have that \typeEvalP{}{\dummy:\sigmaPolCon{\policy}} and hence, \typeEvalI{}{\dummy:\sigmaPol{\policy}}.

We next define a wrapper that wraps $e$ with the information from the public view.        
\begin{align*}
%\wrapcon{e} & \triangleq \lambdagn \alpha_\policy/m_\policy:\sigmaPolCon{\policy}.\attmod{e} \\
%\wrappub{e} & \triangleq \letexp{\alpha_\policy/m_\policy = (\dummy \seal \sigmaPol{\policy})}{e}.
\wrappub{e}  & \triangleq \lambdagn \alpha_\policy,\mPol:\sigmaPol{\policy}.\attmod{e}
\end{align*}

%{We next prove that if $V \in \policy$ then $V$ is of the opaque signature \sigmaPol{\policy}}.
%\begin{lemma}
%\label{lem:ml:wrapper:V-policy:transparentSign}
%\tealtext{If $V \in \policy$, then \typeEvalP{}{V:\sigmaPolCon{\policy}}.}
%\end{lemma}
%\begin{proof}
%The result follows from the definition of $\condenvalt{\rho}{\policy}$.
%\end{proof}
%
%TODO: update

%\begin{lemma}
%\label{lem:ml:wrapper:sub-signature}
%\tealtext{It follows that \typeEval{}{\sigmaPolCon{\policy} \leq \sigmaPol{\policy}:\sign}.}
%\end{lemma}
%\begin{proof}
%We claim that for any \newtext{$L\subseteq \varPolicy{\policy}$}, \typeEval{}{\enccon{L} \leq \encpub{L}:\sign}.
%The result then follows directly from the claim.
%The claim is proven by induction on $L$.
%\end{proof}

%\redtext{The dummy value \dummy\ can be sealed by the opaque signature \sigmaPol{\policy} (Lemma~\ref{lemma:ml:wrapper:V-policy:seal}).
%In addition, the signature of the sealed value ($\dummy \seal \sigmaPol{\policy}$) is $\sigma$ for some $\sigma$ s.t. \sigmaPol{\policy} is a subsignature of $\sigma$ (Lemma~\ref{lem:ml:wrapper:sealing:type-inference})}.

{From the construction, we have that if \wrappub{e} is well-typed in the empty context, then the original term is also well-typed in the public view.
In addition, we can infer the type of the original term in the public view.       
These results yield, by Theorem~\ref{thm:ml:free-thm:opaque}, 
that the original term is TRNI when the wrapper is well-typed.}
\begin{theorem}
\label{thm:ml:wrapper}
{If \typeEvalP{}{\wrappub{e}:\Pign \alphaPol:\sigmaPol{\policy}.\atcsign{\tau}}, then $e$ is \TRNI{\policy,\tau}.}
\end{theorem}
\begin{example}
In this example, we combine the ideas presented in \S\ref{sec:extension:global},  \S\ref{sec:ml-trni:encoding}, and \S\ref{sec:ml-trni:wrapper} to encode a complex policy which is inspired by two-factor authentication.
The policy \policyAut\ involves two confidential passwords and two declassifiers \code{checking1} and \code{checking2} written in SML as below, where \code{input1} and \code{input2} are respectively the first input and the second input from a user.
Notice that \code{checking2} takes a tuple of two passwords as its input.
%The policy requires that the result of the comparison between the first password and the first input can be released.
%If the first password is equal to the first input, then the result of the comparison between the second password and the second input can be released.

\begin{lstlisting}
fun checking1(password1:int) = 
  if (password1 = input1) then 1 else 0
fun checking2(passwords:int*int) = 
  if ((#1 passwords) = input1) then 
      if ((#2 passwords) = input2) then 1 else 0
  else 2
\end{lstlisting}

%Notice that in the signature, $\alpha_{x_1}$

%\begin{align*}
%\sigmaPolCon{\policy} & = \Sigma \alpha_{x_1}:\atksign{S(\intType)}.\Sigma \alpha_1:\atcsign{\intType}.\Sigma \alpha_2:\atcsign{\intType \rightarrow \tau_f}.\enccon{L'}
%\\
%\sigmaPol{\policy} & = 
%\end{align*}

Using the ideas presented in \S\ref{sec:extension:global}, we introduce a new variable which corresponding to the tuple of two passwords.
The confidential and public signatures of \policyAut\ in the module calculus are as below, where $f_1$ and $f_2$ are corresponding to \code{checking1} and \code{checking2}.

\begin{align*}
\sigmaPolCon{\policyAut} &= \Sigma \alpha_{f_1}:\atksign{S(\intType)}.\Sigma \alpha_1:\atcsign{\intType}.\Sigma \alpha_2:\atcsign{\intType \rightarrow \intType}.\\
    & \hspace{70pt} \Sigma \alpha_{x_2}:\atksign{S(\intType)}.\Sigma \alpha_3:\atcsign{\intType}.\\
    & \hspace{70pt}\Sigma \alpha_{f_2}:\atksign{S(\intType \times \intType)}.\Sigma \alpha_4:\atcsign{\intType\times \intType}. \\
    & \hspace{70pt}\Sigma \alpha_5:\atcsign{\intType \times \intType \rightarrow \intType}.\unitsign\\
%%%%%%%%%%%%%%%%%%%%%%%%%    
\sigmaPol{\policyAut} &= \Sigma \alpha_{f_1}:\atksign{\basekind}.\Sigma \alpha_1:\atcsign{\alpha_{f_1}}.\Sigma \alpha_2:\atcsign{\alpha_{f_1} \rightarrow \intType}.\\
    & \hspace{70pt} \Sigma \alpha_{x_2}:\atksign{\basekind}.\Sigma \alpha_3:\atcsign{\alpha_{x_2}}.\\
    & \hspace{70pt}\Sigma \alpha_{f_2}:\atksign{\basekind}.\Sigma \alpha_4:\atcsign{\alpha_{f_2}}.\\
    & \hspace{70pt}\Sigma \alpha_5:\atcsign{\alpha_{f_2} \rightarrow \intType}.\unitsign    
\end{align*}

The confidential and public signatures of \policyAut\ in SML are as below.

%\begin{tabular}{lll}
%\begin{minipage}{0.5 \textwidth}
%\begin{lstlisting}
%signature traSIG=sig
%  type t1 = int
%  val password1:t1 
%  val checking1:t1->int
%  type t2 = int
%  val password2:t2
%  type t3 = int * int 
%  val passwords:t3
%  val checking2:t3 ->int			
%end
%\end{lstlisting}
%\end{minipage}
% & & 
%\begin{minipage}{0.5 \textwidth}
%\begin{lstlisting}
%signature opaSIG=sig
%  type t1
%  val password1:t1 
%  val checking1:t1->int
%  type t2
%  val password2:t2
%  type t3
%  val passwords:t3
%  val checking2:t3->int			
%end
%\end{lstlisting}
%\end{minipage}
%\end{tabular}

\begin{lstlisting}
signature traSIG=sig
  type t1 = int
  val password1:t1 
  val checking1:t1->int
  type t2 = int
  val password2:t2
  type t3 = int * int 
  val passwords:t3
  val checking2:t3 ->int			
end
\end{lstlisting}

\begin{lstlisting}
signature opaSIG=sig
  type t1
  val password1:t1 
  val checking1:t1->int
  type t2
  val password2:t2
  type t3
  val passwords:t3
  val checking2:t3->int			
end
\end{lstlisting}

%\caption{Signatures of \policyAut}
%\label{fig:ex:signature:policyAut}
%\end{figure}

As in \S\ref{sec:extension:global}, we require that for $\rho \respectfull \policyAut$, $\rho_L$ and $\rho_R$ are consistent.
That is when $\rho_L(m_\policyAut) = V_L$, then $V_L.\code{passwords} = \tuple{V_L.\code{password1},V_L.\code{password2}}$ (and we have a similar requirement for $\rho_R$).
%Signatures in the module calculus and indistinguishability for \policyAut\ are in Appendix, Remark~\ref{rem:authentication}.
In order to define indistinguishability for \policyAut, we need to define 
$\rho\respectfull \policyAut$,\footnote{%
Notice that as noted in the main text, $\rho_L$ and $\rho_R$ are consistent.} 
and hence, we define $R_{f_1}$, $R_{x_2}$, $R_{f_2}$  as below.
%Notice that in the definitions, $v_1^*$ and $v_2^*$ are inputs \code{input1} and \code{input2} from a user.

\begin{align*}
R_{f_1} & = \{\tuple{v_1,v_1'} \sep \typeEval{}{v_1,v_1':\intType}, \\
        & \hspace{70pt} v_1 = v_1' = \code{input1}\ \vee \\
        & \hspace{70pt} (v_1 \neq \code{input1} \wedge v_1' \neq \code{input1})\}\\
R_{x_2} & = \{\tuple{v_1,v_1'} \sep \typeEval{}{v_2,v_2':\intType}\}\\
R_{f_2} & = \{\tuple{\tuple{v_1,v_2},\tuple{v_1',v_2'}} \sep \typeEval{}{v_1,v_1',v_2,v_2':\intType} \\ 
        & \hspace{70pt} (v_1 = v_1' = \code{input1} \wedge v_2 = v_2' = \code{input2})\ \vee \\
        & \hspace{70pt} (v_1 = v_1' = \code{input1} \wedge v_2 \neq \code{input2} \wedge \\
        & \hspace{170pt} v_2' \neq \code{input2})\ \vee \\
        & \hspace{70pt} (v_1 \neq \code{input1} \wedge v_1' \neq \code{input1})
        \} 
\end{align*}

By using the wrapper presented above, we can check that programs 
$$m_\policyAut.\code{checking2}\ \ m_\policyAut.\code{passwords},$$
and $m_\policyAut.\code{checking1}\ m_\policyAut.\code{password1}$
are \TRNI{\policyAut,\intType}, where $m_\policyAut$ is the module variable in confidential and public views of \policyAut.
\end{example}

\begin{remark}[On wrapper]
\label{rem:ml:wrapper}
We may choose an applicative functor for wrapping the original program. However, w.r.t. this choice, we need to handle more cases in proofs.
Thus, we choose a generative functor.
\end{remark}

\section{Proofs for TRNI for the Module Calculus}
\label{app:ml:trni}

\subsection{Properties of the encoding}
\begin{lemma}
\label{lem:ml:signature:property:private}
For any $L \subseteq \varPolicy{\policy}$, it follows that \typeEval{}{\enccon{L}: \sign}.
\end{lemma}
%\redtext{\textbf{Lemma~\ref{lem:ml:signature:property}.} [\S\ref{sec:ml:trni} - Part 1]}

\begin{proof}
We prove the lemma by induction on $L$.
We have four cases.

\emcase{Case 1:} $L = []$.
We have that $\enccon{\varPolicy{\policy}} = \unitsign$.
From the \WfsOne\ rule, it follows that \typeEval{}{\unitsign:\sign}.

\emcase{Case 2:} $L = x\concat L'$, $x \not\in\dom{\decPolicy{\policy}}$.
We have that $\enccon{L} = \Sigma \alpha_x:\atksign{S(\intType)}.\Sigma \alpha:\atcsign{\intType}.\enccon{L'}$.

From IH, we have that \typeEval{}{\enccon{L'}:\sign}.
Thus, we have the following derivation.
Notice that $\fstsign{\atcsign{c}} = \unitsign$ for any $c$ and $\fstsign{\atksign{S(\intType)}} = S(\intType)$

$$\footnotesize
\LabelRule{\WfsSix}{
\LabelRuleProof{\WfsTwo}{
\LabelRuleProof{\WfkS}{
\LabelRuleProof{\WfcInt}{~}
{\typeEval{}{\intType:\basekind}}}
{\typeEval{}{S(\intType):\kind}}}
{\typeEval{}{\atksign{S(\intType)}:\sign}} \hspace{50pt}
%%%%%%%%%
\LabelRuleProof{\WfsSix}{
\LabelRuleProof{\WfsThree}{
\LabelRuleProof{\WfcInt}{~}
{\typeEval{\alpha_x:S(\intType)}{\intType:\basekind}}}
{\typeEval{\alpha_x:S(\intType)}{\atcsign{\intType}:\sign}}\hspace{10pt}
%%%%%%%%%%
\EmptyRule{
\EmptyRule{}{\typeEval{}{\enccon{L'}:\sign}\ (\text{from IH)}}}
{\typeEval{\alpha_x:S(\intType),\alpha:\unitsign}{\enccon{L'}:\sign}\ (\text{Lem.~\ref{lem:ml:weakening}})}
}
{{\typeEval{\alpha_x:S(\intType)}{\Sigma \alpha:\atcsign{\intType}.\enccon{L'}:\sign}}}
%%%%%%%%%%%%%%%%%%%%%%%%%%%%%
}
%%%%
{\typeEval{}{\Sigma \alpha_x:\atksign{S(\intType)}.\Sigma \alpha:\atcsign{\intType}.\enccon{L'}: \sign}}
$$

\emcase{Case 3:} $L = x\concat L'$, $\decPolicy{\policy}(x) = f$, where \typeEval{}{f:\intType \rightarrow \tau}.
We have that $\enccon{L} = \Sigma \alpha_f:\atksign{S(\intType)}.\Sigma \alpha_1:\atcsign{\intType}.\Sigma \alpha_2:\atcsign{\intType \rightarrow \tau}.\enccon{L'}$.

We now look at $\Sigma \alpha_2:\atcsign{\intType \rightarrow \tau}.\enccon{L'}$.
Notice that since \typeEval{}{f:\intType \rightarrow \tau}, we have that \typeEval{}{\intType \rightarrow \tau:\basekind}.
%From Lemma~\ref{lem:ml:weakening}, we have that \typeEval{\alpha_f:S(\intType),\alpha_1:\unitsign}{\intType \rightarrow \tau:\basekind} and hence, 
From the \WfsThree\ rule, \typeEval{}{\atcsign{\intType \rightarrow \tau}:\sign}.
Thus, we have that:
$$\small
\EmptyRule{
\LabelRule{\WfsSix}{
\typeEval{}{\atcsign{\intType \rightarrow \tau}:\sign}\\
\EmptyRule{
\EmptyRule{~}{\typeEval{}{\enccon{L'}:\sign}\ \text{(from IH)}}}
{\typeEval{\alpha_2:\unitsign}{\enccon{L'}:\sign\ \text{(Lem.~\ref{lem:ml:weakening})}}}}
{\typeEval{}{\Sigma \alpha_2:\atcsign{\intType \rightarrow \tau}.\enccon{L'}:\sign}}}
{\typeEval{\alpha_f:S(\intType),\alpha_1:\unitsign}{\Sigma \alpha_2:\atcsign{\intType \rightarrow \tau}.\enccon{L'}:\sign}\ \text{(Lem.~\ref{lem:ml:weakening})}}
$$

In addition, by using a reasoning similar to the one in Case 2, we have:
\begin{itemize}
\item \typeEval{}{\atksign{S(\intType)}:\sign},
\item \typeEval{\alpha_f:S(\intType)}{\atcsign{\intType}:\sign}.
% and hence from Lemma~\ref{lem:ml:weakening}, \typeEval{\alpha_f:S(\intType),\alpha_1:\unitsign}{\atcsign{\intType}:\sign} (notice that \typeEval{}{\alpha_f:S(\intType)\ \ok}). 
\end{itemize}
 
Therefore, we have that:
$$\footnotesize
\LabelRule{\WfsSix}{
\EmptyRule{}{\purpletext{\typeEval{}{\atksign{S(\intType)}:\sign}}}\hspace{50pt}
%%%%%%%%%%%%5
\LabelRuleProof{\WfsSix}{
    \purpletext{\typeEval{\alpha_f:S(\intType)}{\atcsign{\intType}:\sign}}\\
    \typeEval{\alpha_f:S(\intType),\alpha_1:\unitsign}{\Sigma \alpha_2:\atcsign{\intType \rightarrow \tau}.\enccon{L'}:\sign}
}
{\typeEval{\alpha_f:S(\intType)}{\Sigma \alpha_1:\atcsign{\intType}.\Sigma \alpha_2:\atcsign{\intType\rightarrow\tau}.\enccon{L'}}}
}
{\Sigma \alpha_f:\atksign{S(\intType)}.\Sigma \alpha_1:\atcsign{\intType}.\Sigma \alpha_2:\atcsign{\intType \rightarrow \tau}.\enccon{L'}:\sign}
$$

\end{proof}

\begin{lemma}
\label{lem:ml:signature:property:public}
{For any $L\subseteq \varPolicy{\policy}$, it follows that \typeEval{}{\encpub{L}: \sign}.}
\end{lemma}
%\redtext{\textbf{Lemma~\ref{lem:ml:signature:property}.} [\S\ref{sec:ml:trni} - Part 2]}
%%\textbf{Lemma~\ref{lem:ml:sign:opaque:wf:general}}
%\redtext{For any $L\subseteq \varPolicy{\policy}$, it follows that \typeEval{}{\encpub{L}: \sign}.}
\begin{proof}
We prove the lemma by induction on $L$.
We have four cases.

\emcase{Case 1:} $L = \emptylist$.
We have that $\encpub{L} = \unitsign$. 
From the \WfsOne\ rule, we have that \typeEval{}{\unitsign:\sign}.

\emcase{Case 2:} $L = x\concat L'$, $x \not\in \dom{\decPolicy{\policy}}$.
We have that $\encpub{L} = \Sigma \alpha_x:\atksign{\basekind}.\Sigma \alpha:\atcsign{\alpha_x}.\encpub{L'}$.
We have the following derivation.
Notice that $\fstsign{\atcsign{c}} = \unitsign$ for any $c$, and $\fstsign{\atksign{\basekind)}} = \basekind$.

$$\footnotesize
\LabelRule{\WfsSix}{
\LabelRuleProof{\WfsTwo}{
{\typeEval{}{\basekind:\kind}}}
{\typeEval{}{\atksign{\basekind}:\sign}} \hspace{50pt}
%%%%%%%%%
\LabelRuleProof{\WfsSix}{
\LabelRuleProof{\WfsThree}{
\LabelRuleProof{{\WfcV}}{~}
{\typeEval{\alpha_x:\basekind}{\alpha_x:\basekind}}}
{\typeEval{\alpha_x:\basekind}{\atcsign{\alpha_x}:\sign}}\hspace{10pt}
%%%%%%%%%%
\EmptyRule{
\EmptyRule{}{\typeEval{}{\enccon{L'}:\sign}\ (\text{from IH)}}}
{\typeEval{\alpha_x:\basekind,\alpha:\unitsign}{\enccon{L'}:\sign}\ (\text{Lem.~\ref{lem:ml:weakening}})}
}
{{\typeEval{\alpha_x:\basekind}{\Sigma \alpha:\atcsign{\alpha_x}.\enccon{L'}:\sign}}}
%%%%%%%%%%%%%%%%%%%%%%%%%%%%%
}
%%%%
{\typeEval{}{\Sigma \alpha_x:\atksign{\basekind}.\Sigma \alpha:\atcsign{\alpha_x}.\enccon{L'}: \sign}}
$$

\emcase{Case 3:} $L=x\concat L'$, $\decPolicy{\policy}(x) = f$, where \typeEval{}{f:\intType \rightarrow \tau} for some $\tau$.
We have that $\encpub{L} = \Sigma \alpha_f:\atksign{\basekind}.\Sigma \alpha_1:\atcsign{\alpha_f}.\Sigma \alpha_2:\atcsign{\alpha_f \rightarrow \tau}.\encpub{L'}$.

We now look at $\Sigma \alpha_2:\atcsign{\alpha_f \rightarrow \tau}.\enccon{L'}$.
We have that \typeEval{\alpha_f:\basekind}{\alpha_f:\basekind} and \typeEval{}{\tau:\basekind} (notice that \typeEval{}{f:\intType\rightarrow\tau} and hence, \typeEval{}{\intType\rightarrow \tau:\basekind} and hence, \typeEval{}{\tau:\basekind}).
From \WfcFn, we have that \typeEval{\alpha_f:\basekind}{\alpha_f \rightarrow \tau:\basekind}.
{From Lemma~\ref{lem:ml:weakening}}, it follows that \typeEval{\alpha_f:\basekind,\alpha_1:\unitsign}{\alpha_f \rightarrow \tau:\basekind}.
From \WfsThree\, \typeEval{\alpha_f:\basekind,\alpha_1:\unitsign}{\atcsign{\alpha_f \rightarrow \tau}:\sign}.

Thus, we have that:
$$\small
\LabelRule{\WfsSix}{
\typeEval{\alpha_f:\basekind, \alpha_1:\unitsign}{\atcsign{\alpha_f \rightarrow \tau}:\sign}\\
\EmptyRule{
\EmptyRule{~}{\typeEval{}{\enccon{L'}:\sign}\ \text{(from IH)}}}
{\typeEval{\alpha_f:\basekind, \alpha_1:\unitsign,\alpha_2:\unitsign}{\enccon{L'}:\sign\ \text{(Lem.~\ref{lem:ml:weakening})}}}}
{\typeEval{\alpha_f:\basekind, \alpha_1:\unitsign}{\Sigma \alpha_2:\atcsign{\alpha_f \rightarrow \tau}.\enccon{L'}:\sign}}
$$

In addition, by using a reasoning similar to the one in Case 2, we have that:
\begin{itemize}
\item \typeEval{}{\atksign{\basekind}:\sign},
\item \typeEval{\alpha_f:\basekind}{\atcsign{\alpha_f}:\sign}.
\end{itemize}
 
Therefore, we have that:
$$\footnotesize
\LabelRule{\WfsSix}{
\EmptyRule{}{\purpletext{\typeEval{}{\atksign{\basekind}:\sign}}}\hspace{50pt}
%%%%%%%%%%%%5
\LabelRuleProof{\WfsSix}{
    \purpletext{\typeEval{\alpha_f:\basekind}{\atcsign{\alpha_f}:\sign}}\\
    \typeEval{\alpha_f:\basekind,\alpha_1:\unitsign}{\Sigma \alpha_2:\atcsign{\alpha_f \rightarrow \tau}.\enccon{L'}:\sign}
}
{\typeEval{\alpha_f:\basekind}{\Sigma \alpha_1:\atcsign{\alpha_f}.\Sigma \alpha_2:\atcsign{\alpha_f\rightarrow\tau}.\enccon{L'}}}
}
{\Sigma \alpha_f:\atksign{\basekind}.\Sigma \alpha_1:\atcsign{\alpha_f}.\Sigma \alpha_2:\atcsign{\alpha_f \rightarrow \tau}.\enccon{L'}:\sign}
$$

\end{proof}

\begin{lemma}
\label{lem:ml:signature:property:subsign}
\tealtext{For any $L \subseteq \varPolicy{\policy}$, it follows that \typeEval{}{\enccon{L} \leq \encpub{L}:\sign}.}
\end{lemma}
\begin{proof}
We prove the lemma by proving that for any $L$, \typeEval{}{\enccon{L} \leq \encpub{L}:\sign}.
We prove this by induction on $L$.
We have four cases.

\emcase{Case 1:} $L = \emptylist$.
We have that $\enccon{L} = \encpub{L} = \unitsign$.
From \SeqR, we have that \typeEval{}{\unitsign \equiv \unitsign:\sign} and hence, from the \SubsEq\ rule, it follows that \typeEval{}{\unitsign \leq \unitsign:\sign}.

\emcase{Case 2:} $L = x\concat L'$ and $x \not\in \dom{\decPolicy{\policy}}$.
We need to prove that \typeEval{}{\Sigma \alpha_x:\atksign{S(\intType)}.\Sigma \alpha:\atcsign{\intType}.\enccon{L'}\leq \Sigma \alpha_x:\atksign{\basekind}.\Sigma \alpha:\atcsign{\alpha_x}.\enccon{L'}:\sign}.

We first prove that \typeEval{\alpha_f:\atksign{S(\intType)}}{\Sigma \alpha:\atcsign{\intType}.\enccon{L'} \leq \Sigma \alpha:\atcsign{\alpha_x}.\enccon{L'}:\sign}.
To this aim, we prove that
\purpletext{\typeEval{\alpha_x:S(\intType)}{\atcsign{\intType} \leq \atcsign{\alpha_x}:\sign}}. 
$$\footnotesize
\LabelRule{\SubsEq}{
\LabelRuleProof{\SeqAtc}{
\LabelRuleProof{\EqcS}{
\LabelRuleProof{\EqcSgTwo}{
\LabelRuleProof{\WfcV}{ (\alpha_x:S(\intType))(\alpha_x) = S(\intType)}
{\typeEval{\alpha_x:S(\intType)}{\alpha_x:S(\intType)}}}
{\typeEval{\alpha_x:S(\intType)}{\alpha_x \equiv \intType:\basekind}}}
{\typeEval{\alpha_x:S(\intType)}{\intType \equiv \alpha_x:\basekind}}   }
{\typeEval{\alpha_x:S(\intType)}{\atcsign{\intType} \equiv \atcsign{\alpha_x}}}}
{\purpletext{\typeEval{\alpha_x:S(\intType)}{\atcsign{\intType} \leq \atcsign{\alpha_x}:\sign }}}
$$

From IH, we have that \typeEval{}{\enccon{L'} \leq \encpub{L'}:\sign}.
\purpletext{From Lemma~\ref{lem:ml:weakening}}, it follows that \typeEval{\alpha_x:S(\intType),\alpha:\unitkind}{\enccon{L'} \leq \encpub{L'}:\sign}.
In addition, from Lemma~\ref{lem:ml:signature:property:public}, we have that \typeEval{}{\encpub{L'}:\sign}.
From Lemma~\ref{lem:ml:weakening}, it follows that \typeEval{\alpha_x:S(\intType),\alpha:\unitkind}{\encpub{L'}:\sign}.
Since $\fstsign{\atcsign{\intType}} = \fstsign{\atcsign{\alpha_x}} = \unitkind$, we have that:

$$\footnotesize
\LabelRule{\SubsDPr}{
{\purpletext{\typeEval{\alpha_x:S(\intType)}{\atcsign{\intType} \leq \atcsign{\alpha_x}:\sign }}}\\ 
%%%%%%%%%%%%%%%
\typeEval{\alpha_x:S(\intType),\alpha:\unitkind}{\enccon{L'} \leq \encpub{L'}:\sign} \\
\typeEval{\alpha_x:S(\intType), \alpha:\unitkind}{\encpub{L'}:\sign}
}
{\typeEval{\alpha_x:{S(\intType)}}{\Sigma \alpha:\atcsign{\intType}.\enccon{L'} \leq \Sigma \alpha:\atcsign{\alpha_x}.\enccon{L'}:\sign}}
$$

We next prove that {\typeEval{}{\atksign{S(\intType)} \leq \atksign{\basekind}:\sign}}.
From \WfcInt, it follows that \typeEval{}{\intType:\basekind}.
From \SubkSg, it follows that \typeEval{}{S(\intType)\leq \basekind:\kind}.
From \SubsK, it follows that \typeEval{}{\atksign{S(\intType)}\leq \atksign{\basekind}:\sign}.

%$$\small
%\LabelRule{SubsK}{
%\LabelRule{\SubkSg}{
%\LabelRule{\WfcInt}{~}
%{\typeEval{}{\intType:\basekind}}}
%{\typeEval{}{S(\intType)\leq \basekind:\kind}}}
%{\typeEval{}{\atksign{S(\intType)} \leq \atksign{\basekind}:\sign}}$$

We now prove that $\typeEval{\alpha_x:\basekind}{\Sigma \alpha:\atcsign{\alpha_x}.\encpub{L'}:\sign}$.
Indeed, we have that \typeEval{\alpha_x:\basekind}{\atcsign{\alpha_x}:\sign}.
From Lemma~\ref{lem:ml:signature:property:public}, \typeEval{}{\encpub{L'}:\sign}.
\purpletext{From Lemma~\ref{lem:ml:weakening}}, \typeEval{\alpha_x:\basekind,\alpha:\unitkind}{\encpub{L'}:\sign}.
From \WfsSix, it follows that \typeEval{\alpha_x:\basekind}{\Sigma\alpha:\atcsign{\alpha_x}.\encpub{L'}:\sign}.

Thus, we have that:
$$\footnotesize
\LabelRule{\SubsDPr}{
{\typeEval{}{\atksign{S(\intType)} \leq \atksign{\basekind}:\sign}}\\
%%%%%
\typeEval{\alpha_x:{S(\intType)}}{\Sigma \alpha:\atcsign{\intType}.\enccon{L'} \leq \Sigma \alpha:\atcsign{\alpha_x}.\enccon{L'}:\sign}\\
%%%%%
\typeEval{\alpha_x:\basekind}{\Sigma \alpha:\atcsign{\alpha_x}.\encpub{L'}:\sign}
}
{\typeEval{}{\Sigma \alpha_x:\atksign{S(\intType)}.\Sigma \alpha:\atcsign{\intType}.\enccon{L'} \leq \Sigma \alpha_x:\atksign{\basekind}.\Sigma \alpha:\atcsign{\alpha_x}.\encpub{L'}}}
$$

\emcase{Case 3:} $L = x\concat L'$ and $\decPolicy{\policy}(x) = f$, where \typeEval{}{f:\intType \rightarrow \tau} for some $\tau$.
We need to prove that 
\begin{multline*}
\typeEval{}{\Sigma \alpha_f:\atksign{S(\intType)}.\Sigma \alpha_1:\atcsign{\intType}.\Sigma \alpha_2:\atcsign{\intType \rightarrow \tau}.\enccon{L'} \leq \\
\Sigma \alpha_f:\atksign{\basekind}.\Sigma \alpha_1:\atcsign{\alpha_f}.\Sigma \alpha_2:\atcsign{\alpha_f \rightarrow \tau}.\encpub{L'}:\sign}.
\end{multline*}

We first prove that \typeEval{\alpha_f:S(\intType)}{\atcsign{\intType \rightarrow \tau} \leq \atcsign{\alpha_f \rightarrow \tau}:\sign}.

$$\footnotesize
\LabelRule{\SubsEq}{
\LabelRuleProof{\SeqAtc}{
\LabelRuleProof{\EqcFn}{
\LabelRuleProof{\EqcS}{
\LabelRuleProof{\EqcSgTwo}{
\LabelRuleProof{\WfcV}{ (\alpha_f:S(\intType))(\alpha_f) = S(\intType)}
{\typeEval{\alpha_f:S(\intType)}{\alpha_f:S(\intType)}}}
{\typeEval{\alpha_f:S(\intType)}{\alpha_f \equiv \intType:\basekind}}}
{\typeEval{\alpha_f:S(\intType)}{\intType \equiv \alpha_f:\basekind}} \\
\EmptyRule{
\LabelRuleProof{\EqcR}{\typeEval{}{\tau:\basekind}}
{\typeEval{}{\tau \equiv \tau:\basekind}}
}{\typeEval{\alpha_f:S(\intType)}{\tau \equiv \tau:\basekind}\ \text{(Lem.~\ref{lem:ml:weakening})}}}
{\typeEval{\alpha_f:S(\intType)}{\intType \rightarrow \tau \equiv \alpha_f \rightarrow \tau:\basekind}}}
{\typeEval{\alpha:S(\intType)}{\atcsign{\intType \rightarrow \tau} \equiv \atcsign{\alpha_f \rightarrow \tau}:\sign}}}
{\typeEval{\alpha_f:S(\intType)}{\atcsign{\intType \rightarrow \tau} \leq \atcsign{\alpha_f \rightarrow \tau}:\sign}}
$$

We next prove that \typeEval{\alpha_f:S(\intType),\alpha_1:\unitkind}{\Sigma \alpha_2:\atcsign{\intType \rightarrow \tau}.\enccon{L'} \leq \Sigma \alpha_2:\atcsign{\alpha_f \rightarrow \tau}.\encpub{L'}:\sign}.
We have that 
\begin{itemize}
\item \typeEval{\alpha_f:S(\intType)}{\atcsign{\intType \rightarrow \tau} \leq \atcsign{\alpha_f \rightarrow \tau}:\sign} and hence, \purpletext{from Lemma~\ref{lem:ml:weakening}}, it follows that \typeEval{\alpha_f:S(\intType), \alpha_1:\unitkind}{\atcsign{\intType \rightarrow \tau} \leq \atcsign{\alpha_f \rightarrow \tau}:\sign},

\item \typeEval{}{\enccon{L'} \leq \encpub{L'}:\sign} (from IH) and hence, \purpletext{from Lemma~\ref{lem:ml:weakening}}, it follows that \typeEval{\alpha_f:S(\intType),\alpha_1:\unitkind,\alpha_2:\unitkind}{\enccon{L'} \leq \encpub{L'}:\sign},

\item \typeEval{}{\encpub{L'}:\sign} (from Lemma~\ref{lem:ml:signature:property:public}) and hence, from \purpletext{Lemma~\ref{lem:ml:weakening}}, it follows that \typeEval{\alpha_f:S(\intType),\alpha_1:\unitkind,\alpha_2:\unitkind}{\encpub{L'}:\sign}
\end{itemize}

Since $\fstsign{\atcsign{c}} = \unitkind$ for any $c$, we have that:
$$\footnotesize
\LabelRule{\SubsDPr}
{\typeEval{\alpha_f:S(\intType), \alpha_1:\unitkind}{\atcsign{\intType \rightarrow \tau} \leq \atcsign{\alpha_f \rightarrow \tau}:\sign} \\
\typeEval{\alpha_f:S(\intType),\alpha_1:\unitkind,\alpha_2:\unitkind}{\enccon{L'} \leq \encpub{L'}:\sign}\\
\typeEval{\alpha_f:S(\intType),\alpha_1:\unitkind,\alpha_2:\unitkind}{\encpub{L'}:\sign}
}
{\purpletext{\typeEval{\alpha_f:S(\intType),\alpha_1:\unitkind}{\Sigma \alpha_2:\atcsign{\intType \rightarrow \tau}.\enccon{L'} \leq \Sigma \alpha_2:\atcsign{\alpha_f \rightarrow \tau}.\encpub{L'}:\sign}}}
$$

We have that:
\begin{itemize}
\item \typeEval{\alpha_f:S(\intType)}{\atcsign{\intType} \leq \atcsign{\alpha_f}:\sign} (as in Case 2),

\item \purpletext{\typeEval{\alpha_f:S(\intType),\alpha_1:\unitkind}{\Sigma \alpha_2:\atcsign{\intType \rightarrow \tau}.\enccon{L'} \leq \Sigma \alpha_2:\atcsign{\alpha_f \rightarrow \tau}.\encpub{L'}:\sign}} (as proven above),

\item \typeEval{\alpha_f:S(\intType)}{{\alpha_f \rightarrow \tau}:\basekind}.
From the \WfsThree\ rule, \typeEval{\alpha_f:S(\intType)}{\atcsign{\alpha_f \rightarrow \tau}:\sign}.
From Lemma~\ref{lem:ml:weakening}, \typeEval{\alpha_f:S(\intType), \alpha_1:\unitkind}{\atcsign{\alpha_f \rightarrow \tau}:\sign}.

From Lemma~\ref{lem:ml:signature:property:public}, \typeEval{}{\encpub{L'}:\sign} and hence, from Lemma~\ref{lem:ml:weakening}, \typeEval{\alpha_f:S(\intType), \alpha_1:\unitkind,\alpha_2:\unitkind}{\encpub{L'}:\sign}.

Since \typeEval{\alpha_f:S(\intType), \alpha_1:\unitkind}{\atcsign{\alpha_f \rightarrow \tau}:\sign} and  \typeEval{\alpha_f:S(\intType), \alpha_1:\unitkind,\alpha_2:\unitkind}{\encpub{L'}:\sign}, from the \WfsSix\ rule, it follows that \typeEval{\alpha_f:S(\intType), \alpha_1:\unitkind}{\Sigma \alpha_2.\atcsign{\alpha_f \rightarrow \tau}.\encpub{L'}:\sign}.
\end{itemize}

Since $\fstsign{\atcsign{c}} = \unitkind$ for any $c$, we have that:
$$\footnotesize
\LabelRule{SubDPr}{
\typeEval{\alpha_f:S(\intType)}{\atcsign{\intType} \leq \atcsign{\alpha_f}:\sign}\\
\typeEval{\alpha_f:S(\intType),\alpha_1:\unitkind}{\Sigma \alpha_2:\atcsign{\intType \rightarrow \tau}.\enccon{L'} \leq \Sigma \alpha_2:\atcsign{\alpha_f \rightarrow \tau}.\encpub{L'}:\sign}\\
\typeEval{\alpha_f:S(\intType), \alpha_1:\unitkind}{\Sigma \alpha_2.\atcsign{\alpha_f \rightarrow \tau}.\encpub{L'}:\sign}
}
{\purpletext{\typeEval{\alpha_f:S(\intType)}{\Sigma \alpha_1: \atcsign{\intType}. \Sigma \alpha_2:\atcsign{\intType \rightarrow \tau}.\enccon{L'} \leq \Sigma \alpha_1: \atcsign{\alpha_f}.\Sigma \alpha_2:\atcsign{\alpha_f \rightarrow \tau}.\encpub{L'}:\sign}}}
$$

We have that 
\begin{itemize}
\item \typeEval{}{\atksign{S(\intType)} \leq \atksign{\basekind}:\sign} (as in Case 2), and
\item \typeEval{\alpha_f:S(\intType)}{\Sigma \alpha_1:\atcsign{\intType}.\Sigma \alpha_2:\atcsign{\intType \rightarrow \tau}.\enccon{L'} \leq \Sigma \alpha_1:\atcsign{\alpha_f}.\Sigma \alpha_2:\atcsign{\alpha_f \rightarrow \tau}.\encpub{L'}:\sign} (as proven above),

\item \typeEval{\alpha_f:\basekind}{{\alpha_f \rightarrow \tau}:\basekind}.
From the \WfsThree\ rule, \typeEval{\alpha_f:\basekind}{\atcsign{\alpha_f \rightarrow \tau}:\sign}.
From Lemma~\ref{lem:ml:weakening}, \typeEval{\alpha_f:\basekind,\alpha_1:\unitkind}{\atcsign{\alpha_f \rightarrow \tau}:\sign}.

From Lemma~\ref{lem:ml:signature:property:public}, \typeEval{}{\encpub{L'}:\sign} and hence, from Lemma~\ref{lem:ml:weakening}, \typeEval{\alpha_f:\basekind, \alpha_1:\unitkind,\alpha_2:\unitkind}{\encpub{L'}:\sign}.

Since \typeEval{\alpha_f:\basekind,\alpha_1:\unitkind}{\atcsign{\alpha_f \rightarrow \tau}:\sign} and \typeEval{\alpha_f:\basekind, \alpha_1:\unitkind,\alpha_2:\unitkind}{\encpub{L'}:\sign}, from the \WfsSix\ rule, \typeEval{\alpha_f:\basekind, \alpha_1:\unitkind}{\Sigma \alpha_2:\atcsign{\alpha_f \rightarrow \tau}.\encpub{L'}:\sign}.

Since \typeEval{\alpha_f:\basekind}{\alpha_f:\basekind}, from the \WfsThree\ rule, \typeEval{\alpha_f:\basekind}{\atcsign{\alpha_f}:\sign}.
Since \typeEval{\alpha_f:\basekind}{\atcsign{\alpha_f}:\sign} and \typeEval{\alpha_f:\basekind, \alpha_1:\unitkind}{\Sigma \alpha_2:\atcsign{\alpha_f \rightarrow \tau}.\encpub{L'}:\sign}, from the \WfsSix\ rule, \typeEval{\alpha_f:\basekind}{\Sigma \alpha_1:\atcsign{\alpha_f}.\Sigma \alpha_2:\atcsign{\alpha_f \rightarrow \tau}.\encpub{L'}:\sign}.

%From Lemma~\ref{lem:ml:weakening}, \typeEval{\alpha_f:\basekind, \alpha_1:\unitkind}{\atcsign{\alpha_f \rightarrow \tau}:\sign}.

%\typeEval{\alpha_f:\basekind}{\Sigma \alpha_1:\atcsign{\alpha_f}.\Sigma \alpha_2:\atcsign{\alpha_f \rightarrow \tau}.\encpub{L'}:\sign} \redtext{(since \typeEval{\alpha_f:\basekind, \alpha_1:\unitkind}{\Sigma \alpha_2.\atcsign{\alpha_f \rightarrow \tau}.\encpub{L'}:\sign} and \typeEval{\alpha_f:\basekind}{\atcsign{\alpha_f}:\sign})}, and 
\end{itemize}

From the \SubsDPr\ rule, we have that 
\begin{multline*}
\typeEval{}{\Sigma \alpha_f:\atksign{S(\intType)}.\Sigma \alpha_1:\atcsign{\intType}.\Sigma \alpha_2:\atcsign{\intType \rightarrow \tau}.\enccon{L'} \leq \\
\Sigma \alpha_f:\atksign{\basekind}.\Sigma \alpha_1:\atcsign{\alpha_f}.\Sigma \alpha_2:\atcsign{\alpha_f \rightarrow \tau}.\encpub{L'}:\sign}.
\end{multline*}

\end{proof}

\begin{lemma}[Pitts closure]
\label{lem:ml:pitts-closure}
\tealtext{For any $x$, $f$ and $a$ in the policy, it follows that $R_x$, $R_f$ and $R_{f\circ a}$ are Pitts closed.}
\end{lemma}
%\textbf{Lemma~\ref{lem:ml:pitts-closure} [\S\ref{sec:ml:trni}].}
%\tealtext{For any $x$, $f$ and $a$, it follows that $R_x$, $R_f$ and $R_{f\circ a}$ are Pitts closed.}
\begin{proof}
\purpletext{We consider $R_x$ first}.
The proof of this case is trivial since any $v_1$ and $v_2$ s.t. \typeEval{}{v_i:\intType}, we have that $\tuple{v_1,v_2} \in R_x$.

\purpletext{We next consider $R_f$, where \typeEval{}{f:\intType \rightarrow \tau}}.
We consider \typeEval{x:\intType}{f\ x}. 
\purpletext{We have that $x$ is active in $f\ x$}.
We now consider arbitrary $v_1$ and $v_2$ s.t. \tuple{v_1,v_2} in $R_f$.
From the definition of $R_f$, we have that $\tuple{f\ v_1,f\ v_2} \in \rev{\inter{\tau}{\emptyset}}$.
{Since \typeEval{}{\tau:\basekind}, from Lemma~\ref{lem:ml:crary:interpretation}, we have that $\tuple{\_,\_,\inter{\tau}{\emptyset},\inter{\tau}{\emptyset}} \in \inter{\basekind}{\emptyset}$ (notice that $\emptyset \in \interenvfull{.}$)}.
Therefore, \inter{\tau}{\emptyset} is Pitts closed, that is $\inter{\tau}{\emptyset} = \rst{\inter{\tau}{\emptyset}}$.
\purpletext{Since $\tuple{f\ v_1,f\ v_2} \in \rev{\inter{\tau}{\emptyset}}$, it follows that $\tuple{f\ v_1,f\ v_2} \in \rarb{stev}{\inter{\tau}{\emptyset}}$}.

We have proven that:
\begin{itemize}
\item $x$ is active in $f\ x$,
\item for all $\tuple{v_1,v_2} \in R_f$, $\tuple{(f\ x)[x\mapsto v_1], (f\ x)[x\mapsto v_2]} \in \rarb{stev}{\inter{\tau}{\emptyset}}$.
\end{itemize}

\purpletext{From Lemma~\ref{lem:ml:monadic}}, for all $\tuple{w_1,w_2} \in \rst{R_f}$, we have that $\tuple{f\ w_1, f\ w_2} \in \rarb{stev}{\inter{\tau}{\emptyset}} = \rarb{ev}{\inter{\tau}{\emptyset}}$. 
From the definition of $R_f$, we have that $\tuple{w_1,w_2} \in R_f$.

%\purpletext{We now consider $R_{f\circ a}$, where \typeEval{}{a:\intType \rightarrow\intType} and \typeEval{}{f:\intType \rightarrow\tau}}.
%We consider \typeEval{x:\intType}{f(a\ x)}. 
%\purpletext{We have that $x$ is active in $f(a\ x)$}.
%We now consider arbitrary $v_1$ and $v_2$ in $R_{f\circ a}$.
%From the definition of $R_{f\circ a}$, we have that $\tuple{f(a\ v_1),f(a\ v_2)} \in \rev{\inter{\tau}{\emptyset}}$.
%\purpletext{As proven in Case~2, we have that $\inter{\tau}{\emptyset} = \rst{\inter{\tau}{\emptyset}}$}.
%\purpletext{Since $\tuple{f(a\ v_1),f(a\ v_2)} \in \rev{\inter{\tau}{\emptyset}}$, it follows that $\tuple{f(a\ v_1),f(a\ v_2)} \in \rarb{stev}{\inter{\tau}{\emptyset}}$}.
%
%
%We have proven that:
%\begin{itemize}
%\item $x$ is active in $f(a\ x)$,
%\item for all $\tuple{v_1,v_2} \in R_{f\circ a}$, $\tuple{(f(a\ x))[x\mapsto v_1], (f(a\ x))[x\mapsto v_2]} \in \rarb{stev}{\inter{\tau}{\emptyset}}$.
%\end{itemize}
%
%
%\purpletext{From Lemma~\ref{lem:ml:monadic}}, for all $\tuple{w_1,w_2} \in \rst{R_{f\circ a}}$, we have that $\tuple{f(a\ w_1), f(a\ w_2)} \in \rarb{stev}{\inter{\tau}{\emptyset}} = \rarb{ev}{\inter{\tau}{\emptyset}}$. 
%From the definition of $R_{f\circ a}$, we have that $\tuple{w_1,w_2} \in R_{f\circ a}$.
\end{proof}

%\textbf{Lemma~\ref{lem:ml:consenv:value:property}}
\begin{lemma}
\label{lem:ml:consenv:value:property}
For any $L \subseteq \varPolicy{\policy}$ and \condenv{\rho}{L}{\policy}, it follows that
\begin{itemize}
\item \typeEvalP{}{\rho_L(m_\policy):\enccon{L}} and \typeEvalP{}{\rho_R(m_\policy):\enccon{L}}, and 
%%%
\item \typeEvalP{}{\rho_L(m_\policy):\encpub{L}} and \typeEvalP{}{\rho_R(m_\policy):\encpub{L}}.
\end{itemize}
\end{lemma}
%\tealtext{For any $L \subseteq \varPolicy{\policy}$ and \condenv{\rho}{L}{\policy}, it follows that
%\begin{itemize}
%\item \typeEvalP{}{\rho_L(m_\policy):\enccon{L}} and \typeEvalP{}{\rho_R(m_\policy):\enccon{L}}, and 
%%%%
%\item \typeEvalP{}{\rho_L(m_\policy):\encpub{L}} and \typeEvalP{}{\rho_R(m_\policy):\encpub{L}}.
%\end{itemize}}
\begin{proof}
The first part of the Lemma~\ref{lem:ml:consenv:value:property} is from the definition of \condenv{\rho}{L}{\policy}.
The second part follows from the first part, Lemma~\ref{lem:ml:signature:property:subsign}, {and the subsumption rule}. 
\end{proof}

\begin{lemma}
\label{lem:ml:consenv:property:one}
{
Suppose that $L \subseteq \varPolicy{\policy}$, $\condenv{\rho}{L}{\policy}$, $\rho(\alpha_\policy) = \tuple{c_1,c_2,Q}$, $k = \fstsign{\encpub{L}}$.
It follows that:
\begin{itemize}
\item \typeEval{}{k:\kind},
\item \typeEval{}{c_1:\rho_L(k)}, \typeEval{}{c_2:\rho_R(k)}, and
\item $\tuple{c_1,c_2,Q,Q} \in \inter{k}{\rho}$.
\end{itemize}}
\end{lemma}
%\textbf{Lemma~\ref{lem:ml:consenv:property:one}.}
%\tealtext{
%Suppose that $L \subseteq \varPolicy{\policy}$, $\condenv{\rho}{L}{\policy}$, $\rho(\alpha_\policy) = \tuple{c_1,c_2,Q}$, and $k = \fstsign{\encpub{L}}$.
%It follows that:
%\begin{itemize}
%\item \typeEval{}{{k}:\kind},
%\item \typeEval{}{c_1:\rho_L(k)}, \typeEval{}{c_2:\rho_R(k)}, and
%\item $\tuple{c_1,c_2,Q,Q} \in \inter{k}{\rho}$.
%\end{itemize}}
\begin{proof}
We prove this lemma by induction on $L$, using the definition of $\condenv{\rho}{L}{\policy}$.

\emcase{Case 1:} $L = \emptylist$.
We have that $\sigma = \unitsign$ and $k = \unitkind$, $\rho(\alpha_\policy) = \tuple{\unitcon,\unitcon,\tuple{}}$.
We can easily check that \typeEval{}{\unitkind:\kind}, \typeEval{}{\unitcon:\unitkind}, and $\tuple{\unitcon,\unitcon, \tuple{},\tuple{}} \in \inter{\unitkind}{\rho}$.

\emcase{Case 2:} $L = x\concat L'$, where $x \not \in \dom{\decPolicy{\policy}}$.
We have that 
\begin{itemize}
\item $\sigma = \Sigma \alpha_x:\atksign{\basekind}.\Sigma \alpha:\atcsign{\alpha_x}.\encpub{L'}$, and
\item $k =  \Sigma \alpha_x:\basekind.\Sigma \alpha:\unitkind.\fstsign{\encpub{L'}}$,
\item $Q = \tuple{R_x,\tuple{\tuple{}, Q'}}$,
\item $\rho(\alpha_\policy) = \tuple{\tuple{\intType, \tuple{\unitcon, c_1'}}, \tuple{\intType,\tuple{\unitcon,c_2'}}, \tuple{R_x,\tuple{\tuple{}, Q'}}}$,
\item \condenv{\rho'}{L'}{\policy}, where $\rho'(\alpha_\policy) = \tuple{c_1',c_2',Q'}$.
\end{itemize}

%where $\consenv{L'}(\alpha_\policy) = \tuple{c_1',c_2',Q'}$.

We need to prove that:
\begin{itemize}
\item \typeEval{}{\Sigma \alpha_x:\basekind.\Sigma \alpha:\unitkind.\fstsign{\encpub{L'}}:\kind},
\item \typeEval{}{\tuple{\intType, \tuple{\unitcon, c_1'}}: \rho_L(\Sigma \alpha_x:\basekind.\Sigma \alpha:\unitkind.\fstsign{\encpub{L'}})},
\item \typeEval{}{\tuple{\intType, \tuple{\unitcon, c_2'}}: \rho_L(\Sigma \alpha_x:\basekind.\Sigma \alpha:\unitkind.\fstsign{\encpub{L'}})},
\item $\tuple{\tuple{\intType,\tuple{\unitcon, c_1'}}, \tuple{\intType,\tuple{\unitcon, c_2'}}, \tuple{R_x,\tuple{\tuple{}, Q'}}} \in \inter{\Sigma \alpha_x:\basekind.\Sigma \alpha:\unitkind.\fstsign{\encpub{L'}}}{\rho}$.
\end{itemize}

We first prove that \typeEval{}{\Sigma \alpha_x:\basekind.\Sigma \alpha:\unitkind.\fstsign{\encpub{L'}}:\kind}.
From IH, we have that \typeEval{}{\fstsign{\encpub{L'}}:\kind}.
From Lemma~\ref{lem:ml:weakening}, \typeEval{\alpha:\unitkind}{\fstsign{\encpub{L'}}:\kind}.
From rule \WfkP, \typeEval{}{\Sigma \alpha:\unitkind.\fstsign{\encpub{L'}}:\kind}.
From Lemma~\ref{lem:ml:weakening}, \typeEval{\alpha_x:\basekind}{\Sigma \alpha:\unitkind.\fstsign{\encpub{L'}}:\kind}.
From rule \WfkP, \typeEval{}{\Sigma \alpha_x:\basekind.\Sigma \alpha:\unitkind.\fstsign{\encpub{L'}}:\kind}.

We next prove that \typeEval{}{\tuple{\intType, \tuple{\unitcon, c_1'}}: \rho_L(\Sigma \alpha_x:\basekind.\Sigma \alpha:\unitkind.\fstsign{\encpub{L'}})} and \typeEval{}{\tuple{\intType, \tuple{\unitcon, c_2'}}: \rho_L(\Sigma \alpha_x:\basekind.\Sigma \alpha:\unitkind.\fstsign{\encpub{L'}})}.
Since their proofs are similar, we only prove here \typeEval{}{\tuple{\intType, \tuple{\unitcon, c_1'}}: \rho_L(\Sigma \alpha_x:\basekind.\Sigma \alpha:\unitkind.\fstsign{\encpub{L'}})}. 
{Notice that since $\Sigma \alpha_x:\basekind.\Sigma \alpha:\unitkind.\fstsign{\encpub{L'}}$ is a closed kind (as proven above)}, we have that $\rho_L(\Sigma \alpha_x:\basekind.\Sigma \alpha:\unitkind.\fstsign{\encpub{L'}}) = \Sigma \alpha_x:\basekind.\Sigma \alpha:\unitkind.\fstsign{\encpub{L'}}$.

We have the following derivations.
Notice that
\begin{itemize}
\item \typeEval{}{\intType:\basekind},
\item if $k'$ is a closed kind, then $k'[\beta \mapsto c'] = k'$ for any $\beta$ and $c'$,

\item \typeEval{}{c_1':\rho'_L(\fstsign{\encpub{L'}})} (from IH) and hence, \typeEval{}{c_1':\fstsign{\encpub{L'}}} (since from IH, \fstsign{\encpub{L'}} is a closed kind).
Thus, \typeEval{}{c_1':\fstsign{\encpub{L'}}[\alpha\mapsto \unitcon]}.

\item \typeEval{}{\fstsign{\encpub{L'}}:\kind} (from IH) and hence, from Lemma~\ref{lem:ml:weakening}, \typeEval{\alpha:\unitkind}{\fstsign{\encpub{L'}}:\kind},

\item \typeEval{}{\Sigma \alpha:\unitkind.\fstsign{\encpub{L'}}:\kind} (from the \WfkP\ rule and \typeEval{\alpha:\unitkind}{\fstsign{\encpub{L'}}:\kind}) and hence \typeEval{\alpha_x:\basekind}{\Sigma \alpha:\unitkind.\fstsign{\encpub{L'}}:\kind} (from Lemma~\ref{lem:ml:weakening}).
\end{itemize}

$$\small
\LabelRule{\WfcPr}{
\typeEval{}{\unitcon:\unitkind} \\
\typeEval{}{c_1':\fstsign{\encpub{L'}}[\alpha\mapsto \unitcon]}\\
\typeEval{\alpha:\unitkind}{\fstsign{\encpub{L'}}:\kind}
}
{\typeEval{}{\tuple{\unitcon,c_1'}:\Sigma \alpha:\unitkind.\fstsign{\encpub{L'}}}}
$$

$$\small
\LabelRule{\WfcPr}{
\typeEval{}{\intType:\basekind} \\
%%%%%
{\typeEval{}{\tuple{\unitcon,c_1'}:\Sigma \alpha:\unitkind.\fstsign{\encpub{L'}}[\alpha_x\mapsto \intType]}}\\
%%%%%
\typeEval{\alpha_x:\basekind}{\Sigma \alpha:\unitkind.\fstsign{\encpub{L'}}:\kind}
}
{\typeEval{}{\tuple{\intType, \tuple{\unitcon, c_1'}}: \Sigma \alpha_x:\basekind.\Sigma \alpha:\unitkind.\fstsign{\encpub{L'}}}}
$$

We now prove that $\tuple{\tuple{\intType,\tuple{\unitcon, c_1'}}, \tuple{\intType,\tuple{\unitcon, c_2'}}, \tuple{R_x,\tuple{\tuple{}, Q'}}} \in \inter{\Sigma \alpha_x:\basekind.\Sigma \alpha:\unitkind.\fstsign{\encpub{L'}}}{\rho}$.
As proven above, we have that \typeEval{}{\tuple{\intType,\tuple{\unitcon, c_1'}}:\rho_L(\Sigma \alpha_x:\basekind.\Sigma \alpha:\unitkind.\fstsign{\encpub{L'}})} and \typeEval{}{\tuple{\intType,\tuple{\unitcon, c_2'}}:\rho_R(\Sigma \alpha_x:\basekind.\Sigma \alpha:\unitkind.\fstsign{\encpub{L'}})}.
We now need to prove that:
\begin{itemize}
\item $\tuple{R_x,\tuple{\tuple{}, Q'}} \in \precan{\simple{\encpub{L}}}$,
\item $\tuple{\intType,\intType, R_x,R_x} \in \inter{\basekind}{\rho}$,
\item \typeEval{}{\tuple{\unitcon,c_1'}:\rho_{1L}(\Sigma \alpha:\unitkind.\fstsign{\encpub{L'}})},
\item \typeEval{}{\tuple{\unitcon,c_2'}:\rho_{1R}(\Sigma \alpha:\unitkind.\fstsign{\encpub{L'}})},
\item $\tuple{\unitcon,\unitcon,\tuple{},\tuple{}} \in \inter{\unitkind}{\rho_1}$
\item $\typeEval{}{c_1':\rho_{2L}(\fstsign{\encpub{L'}})}$,
\item $\typeEval{}{c_2':\rho_{2R}(\fstsign{\encpub{L'}})}$,
\item $\tuple{c_1',c_2',Q',Q'} \in \inter{\fstsign{\encpub{L'}})}{\rho_2}$,
\end{itemize}

where $\rho_1 = \rho, \alpha_x \mapsto \tuple{\intType,\intType, R_x}$ and $\rho_2 = \rho_1, \alpha \mapsto \tuple{\unitcon,\unitcon,\tuple{}}$.

We have that:
\begin{itemize}
\item We have that $R_x \in \precan{\basekind}$, $\tuple{} \in \precan{\unitkind}$.
From IH, $\tuple{c_1',c_2',Q',Q'} \in \inter{\encpub{L'}}{\rho'}$  and hence, $Q' \in \precan{\simple{\fstsign{\encpub{L'}}}}$ (notice that $\simple{\fstsign{\encpub{L'}}} = \fstsign{\encpub{L'}}$ since there is no singleton kind in \fstsign{\encpub{L'}}).
\item From Lemma~\ref{lem:ml:pitts-closure}, $R_x = \rst{R_x}$.
Thus, $\tuple{\intType,\intType, R_x,R_x} \in \inter{\basekind}{\rho}$.

\item As shown in the first derivation in the proof above, we have that \typeEval{}{\tuple{\unitcon, c_1'}:\Sigma \alpha:\unitkind.\fstsign{\encpub{L'}}} and hence, \typeEval{}{\tuple{\unitcon, c_1'}:\rho_{1L}(\Sigma \alpha:\unitkind.\fstsign{\encpub{L'}})} (since $\Sigma \alpha:\unitkind.\fstsign{\encpub{L'}})$ is closed)

\item Similarly, it follows that \typeEval{}{\tuple{\unitcon, c_2'}:\rho_{1R}(\Sigma \alpha:\unitkind.\fstsign{\encpub{L'}})}.

\item {Since $L' \subseteq \varPolicy{\policy}$ and \condenv{\rho'}{L'}{\policy}}, from IH, $\typeEval{}{c_1':\rho'_{L}(k')}$, where $\condenv{\rho'}{L'}{\policy}$ and $k' = \fstsign{\encpub{L'}})$.
\purpletext{Also from IH, \typeEval{}{k'}}.
Thus, $\typeEval{}{c_1':k'}$ and hence, $\typeEval{}{c_1':\rho_{2L}(k')}$.
\item Similarly, $\typeEval{}{c_1':\rho_{2R}(k')}$.

\item {Since $L' \subseteq \varPolicy{\policy}$ and \condenv{\rho'}{L'}{\policy}}, from IH, $\tuple{c_1',c_2',Q',Q'} \in \inter{\fstsign{\encpub{L'}}}{\rho'}$, where $\condenv{\rho'}{L'}{\policy}$.
Also from IH, \typeEval{}{\fstsign{\encpub{L'}}:\kind}.
{Since $\tuple{\rho',\rho_2} \in \interenv{.}$, from Lemma~\ref{lem:ml:crary:interpretation}}, we have that $\inter{\fstsign{\encpub{L'}}}{\rho'} = \inter{\fstsign{\encpub{L'}}}{\rho_2}$.
Therefore, $\tuple{c_1',c_2',Q',Q'} \in \inter{\fstsign{\encpub{L'}}}{\rho_2}$.
%From Lemma~\ref{lem:ml:consenv:property:aux}, we have that $\tuple{c_1',c_2',Q',Q'} \in \inter{\fstsign{\encpub{L'}})}{\rho_2}$.
\end{itemize}

\emcase{Case 3:} $L = x \concat L'$ and $\decPolicy{\policy}(x) = f$, where \typeEval{}{f:\intType \rightarrow \tau}.
We have that $\encpub{L} = \Sigma \alpha_f:\atksign{\basekind}.\alpha_1:\atcsign{\alpha_f}.\alpha_2:\atcsign{\alpha_f \rightarrow\tau_f}.{\encpub{L'}}$, $k = \Sigma \alpha_f:\basekind.\alpha_1:\unitkind.\alpha_2:\unitkind.\fstsign{\encpub{L'}}$, and $\rho'(\alpha_\policy) = \tuple{c_1',c_2',Q'}$, and 
$$\rho(\alpha_\policy) = \tuple{\tuple{\intType,\tuple{\unitcon,\tuple{\unitcon, c_1'}}}, \tuple{\intType,\tuple{\unitcon,\tuple{\unitcon, c_2'}}},\tuple{R_f, \tuple{\tuple{}, \tuple{\tuple{}, Q'}}}},$$

As in Case 2, we have that \typeEval{}{k:\kind} and \typeEval{}{\tuple{\intType,\tuple{\unitcon,\tuple{\unitcon, c_1'}}}: k}.
%INTUITIVELY, \typeEval{}{\tuple{\intType,\tuple{\unitcon,\tuple{\unitcon, c_1'}}}: k} SINCE THERE ARE ONLY T AND 1 IN k AND THERE IS NO SINGTON KINDS IN K, AND 
%EACH COMPONENT OF c_1 IS OF THE CORRESPONDING KIND.
We now prove that $\tuple{c_1,c_2,Q,Q} \in \inter{\Sigma \alpha_f:\basekind.\alpha_1:\unitkind.\alpha_2:\unitkind.\fstsign{\encpub{L'}}}{\rho}$.
That is we need to prove that:
\begin{itemize}
\item $\tuple{R_f, \tuple{\tuple{}, \tuple{\tuple{}, Q'}}} \in \precan{\simple{\encpub{L}}}$,
\item $\tuple{\intType,\intType, R_f,R_f} \in \inter{\basekind}{\rho}$,
\item $\tuple{\unitcon,\unitcon, \tuple{},\tuple{}} \in \inter{\unitkind}{\rho_1}$, where $\rho_1 = \rho,\alpha_f \mapsto \tuple{\intType,\intType, R_f}$,
\item $\tuple{\unitcon,\unitcon, \tuple{},\tuple{}} \in \inter{\unitkind}{\rho_2}$, where $\rho_2 = \rho_1,\alpha_1 \mapsto \tuple{\unitcon,\unitcon, \tuple{}}$,
\item $\tuple{c_1',c_2', Q',Q'} \in \inter{\fstsign{\encpub{L'}}}{\rho_3}$, where $\rho_3 = \rho_2,\alpha_2 \mapsto \tuple{\unitcon,\unitcon, \tuple{}}$.
\end{itemize}

The first item can be easily verified (as in Case 2).
From Lemma~\ref{lem:ml:pitts-closure}, $R_f$ is Pitts closed and hence, $\tuple{\intType,\intType, R_f,R_f} \in \inter{\basekind}{\rho}$.
We can easily verify that $\tuple{\unitcon,\unitcon, \tuple{},\tuple{}} \in \inter{\unitkind}{\rho_1}$ and $\tuple{\unitcon,\unitcon, \tuple{},\tuple{}} \in \inter{\unitkind}{\rho_2}$.

We have that $L' \subseteq \varPolicy{\policy}$ and \condenv{\rho'}{L'}{\policy} (since \condenv{\rho}{L}{\policy}), from IH, we have that $\tuple{c_1',c_2',Q',Q'} \in \inter{\fstsign{\encpub{L'}}}{\rho'}$.
Also from IH, we have that \typeEval{}{\fstsign{\encpub{L'}}:\kind}.
Since $\tuple{\rho',\rho_3} \in \interenv{.}$, from Lemma~\ref{lem:ml:crary:interpretation}, we have that $\inter{\fstsign{\encpub{L'}}}{\rho'} = \inter{\fstsign{\encpub{L'}}}{\rho_3}$.
Therefore, we have that $\tuple{c_1',c_2',Q',Q'} \in \inter{\fstsign{\encpub{L'}}}{\rho_3}$.

\end{proof}

\textbf{Lemma~\ref{lem:ml:interpretation-type:policy} (in \S\ref{sec:ml:trni})}
\tealtext{If \condenvalt{\rho_1}{\policy}, \condenvalt{\rho_2}{\policy}, and \typeEval{\pubViewTerm{\policy}}{\tau:\basekind}, then $\inter{\tau}{\rho_1} = \inter{\tau}{\rho_2}$}.
\begin{proof}
From the definition of \condenv{\rho_i}{\varPolicy{\policy}}{\policy},
%\todoinline{DN}{Shouldn't this be \predenv{\varPolicy{\policy}}{\policy}?
%Also, there are two copies of the proof below.
%\newline
%[MN] I fixed it. 
%The first proof is for \sigmaPol{\policy}. The second proof is for $\sigma \in \policy$.
%}
we have that $\rho_1(\alpha_\policy) = \rho_2(\alpha_\policy) = \tuple{c_1,c_2,Q}$ for some $c_1$, $c_2$, and $Q$.
\purpletext{From Lemma~\ref{lem:ml:consenv:property:one} and the definition of constructor equivalence, we have that:
\begin{itemize}
\item \typeEval{}{c_1 \equiv c_1:\rho_{1L}(\fstsign{\sigmaPol{\policy}})}, \typeEval{}{c_2 \equiv c_2:\rho_{1R}(\fstsign{\sigmaPol{\policy}}))}, and
\item $\tuple{c_1,c_2,Q,Q} \in \inter{\fstsign{\sigmaPol{\policy}})}{\rho_1}$.
\end{itemize}}

In other words, $\tuple{\rho_1,\rho_2} \in \interenv{\alpha_\policy/m_\policy:\sigmaPol{\policy}}$.

%From the definition of \pubViewTerm{\policy}, $\pubViewTerm{\policy} = \alpha_\policy/m_\policy:\sigma$ for some $\sigma \in \policy$.
%From the definition of $\sigma \in \policy$, we have that \typeEval{}{\sigma \equiv \sigmaPol{\policy}:\sign}.
%From the static semantics and the definition of \fstsign{\_}, we have that \typeEval{}{\fstsign{\sigma} \equiv \fstsign{\sigmaPol{\policy}}:\kind}.
%%[THIS CAN BE PROVEN BY INDUCTION ON THE DEFINITION OF SIGNATURE EQUIVALENCE].
%As proven in Lemma~\ref{lem:ml:consenv:property:one}, \fstsign{\sigmaPol{\policy}} is a closed kind and hence, $\rho_{1L}(\fstsign{\sigmaPol{\policy}}) = \fstsign{\sigmaPol{\policy}}$ and $\rho_{1R}(\fstsign{\sigmaPol{\policy}}) = \fstsign{\sigmaPol{\policy}}$.
%Since \typeEval{}{\fstsign{\sigma} \equiv \fstsign{\sigmaPol{\policy}}:\kind}, \fstsign{\sigmaPol{\policy}} is a closed kind and hence, $\rho_{1L}(\fstsign{\sigma}) = \fstsign{\sigma}$ and $\rho_{1R}(\fstsign{\sigma}) = \fstsign{\sigma}$.
%From the statics semantics (rule \EqcSub) and Lemma~\ref{lem:ml:crary:interpretation}, we have that:
%\begin{itemize}
%\item \typeEval{}{c_1 \equiv c_1:\rho_{1L}(\fstsign{\sigma})}, \typeEval{}{c_2 \equiv c_2:\rho_{1R}(\fstsign{\sigma})}, and
%\item $\tuple{c_1,c_2,Q,Q} \in \inter{\fstsign{\sigma}}{\rho_1}$.
%\end{itemize}
%Thus, we have that $\tuple{\rho_1,\rho_2} \in \interenv{\alpha_\policy/m_\policy:\sigma}$.
%In other words, $\tuple{\rho_1,\rho_2} \in \interenv{\pubViewTerm{\policy}}$.

\purpletext{Since \typeEval{\pubViewTerm{\policy}}{\tau:\basekind}, from Lemma~\ref{lem:ml:crary:interpretation}, we have that $\tuple{\_,\_,\inter{\tau}{\rho_1},\inter{\tau}{\rho_2}} \in \inter{\basekind}{\rho_1}$.}
From the definition of \inter{\basekind}{\rho} (see Fig.~\ref{fig:ml:logrel:kind-con}), it follows that $\inter{\tau}{\rho_1} = \inter{\tau}{\rho_2}$.
\end{proof}

\textbf{Lemma~\ref{lem:ml:consenv:property:two} (in \S\ref{sec:ml:trni}).}
\tealtext{Suppose that $\condenvalt{\rho}{\policy}$.
It follows that $\rho \in \interenvfull{\pubViewTerm{\policy}}$.}
%\begin{proof}
%We prove the lemma by proving that for any $L$ and any $\condenv{\rho}{L}{\policy}$, it follows that $\rho \in \interenvfull{\alpha_\policy/m_\policy:\encpub{L}}$.
%\purpletext{Suppose that $\rho(\alpha_\policy) = \tuple{c_1,c_2,Q}$ and $\rho(m_\policy) = \tuple{V_1,V_2}$}.
%From Lemma~\ref{lem:ml:consenv:property:one}, we have that:
%\begin{itemize}
%\item \typeEval{}{c_1:\rho_L(k)}, \typeEval{}{c_2:\rho_R(k)} where $k = \fstsign{\encpub{L}}$, and
%\item $\tuple{c_1,c_2,Q,Q} \in \inter{k}{\rho}$.
%\end{itemize}
%Thus, we only need to prove that:
%\begin{itemize}
%\item $\rho(\alpha_\policy) = \tuple{\fstoptwo{V_1},\fstoptwo{V_2},Q}$,
%\item $\tuple{V_1,V_2,Q} \in \inter{\sigma}{\rho}$.
%\end{itemize}
%
%We prove this by mutual induction on the definition of $\condenv{\rho}{L}{\policy}$.
\begin{proof}
We need to prove that \condenvalt{\rho}{\interenv{\alpha_\policy/m_\policy:\encpub{\varPolicy{\policy}}}}.
We claim that for any $L\subseteq \varPolicy{\policy}$ and any \condenv{\rho}{L}{\policy}, it follows that \condenvalt{\rho}{\interenv{\alpha_\policy/m_\policy:\encpub{L}}}.
Then the proof follows directly from the claim.

We now prove the claim.
\purpletext{Suppose that $\rho(\alpha_\policy) = \tuple{c_1,c_2,Q}$ and $\rho(m_\policy) = \tuple{V_1,V_2}$}.
From Lemma~\ref{lem:ml:consenv:property:one}, we have that:
\begin{itemize}
\item \typeEval{}{c_1:\rho_L(k)}, \typeEval{}{c_2:\rho_R(k)} where $k = \fstsign{\encpub{L}}$, and hence, it follows that \typeEval{}{c_1\equiv c_1:\rho_L(k)} and \typeEval{}{c_2\equiv c_2:\rho_R(k)}

\item $\tuple{c_1,c_2,Q,Q} \in \inter{k}{\rho}$.
\end{itemize}

Therefore, $\tuple{\rho,\rho} \in \interenv{\alpha_\policy/m_\policy:\encpub{L}}$. Thus, we only need to prove two following items:
\begin{itemize}
\item $\rho(\alpha_\policy) = \tuple{\fstoptwo{V_1},\fstoptwo{V_2},Q}$,
\item $\tuple{V_1,V_2,Q} \in \inter{\encpub{L}}{\rho}$.
\end{itemize}

We prove these two items by induction on $L$, using the definition of $\condenv{\rho}{L}{\policy}$.
We have four cases.

\emcase{Case 1:} $L = \emptylist$.
We have that $\sigma = \encpub{L} = \unitsign$, \purpletext{$\rho(m_\policy) = \tuple{\unitmod,\unitmod}$ and $\rho(\alpha_\policy) = \tuple{\unitcon,\unitcon, \tuple{}}$}.
In other words, $V_1 = V_2 = \unitmod$ and $Q = \tuple{}$.
Since \typeEval{}{\fstop{\unitmod}{\unitcon}}, from the definition of \inter{\unitsign}{\rho}, we have that $\tuple{V_1,V_2,Q} \in \inter{\sigma}{\rho}$.

\emcase{Case 2:} $L = x\concat L'$, $x\not\in \dom{\decPolicy{\policy}}$.
We have that:
\begin{align*}
\purpletext{\rho(m_\policy)} & = \tuple{\tuple{\atcmod{\intType},\tuple{\attmod{v_1},V_1'}},\tuple{\atcmod{\intType},\tuple{\attmod{v_2},V_2'}}},\\
\purpletext{\rho(\alpha_\policy)} &= \tuple{\tuple{\intType, \tuple{\unitcon,c_1'}},\tuple{\intType, \tuple{\unitcon,c_2'}}, \tuple{R_x, \tuple{\tuple{},Q'}}},
\end{align*}

and $\tuple{v_1,v_2} \in R_x$, and \condenv{\rho'}{L'}{\policy}, where \purpletext{$\rho'(m_\policy) = \tuple{V_1',V_2'}$, $\rho'(\alpha_\policy) = \tuple{c_1',c_2',Q'}$}.
We also have that $\rho = \Sigma \alpha_x:\atksign{\basekind}.\Sigma \alpha:\atcsign{\alpha_x}.\encpub{L'}$.

%Since $\consenv{L'}(m_\policy) = \tuple{V_1',V_2}$ and $\consenv{L'}(\alpha_\policy) = \tuple{c_1',c_2',Q'}$, from IH, we have that:
Since \purpletext{$\rho'(m_\policy) = \tuple{V_1',V_2}$, $\rho'(\alpha_\policy) = \tuple{c_1',c_2',Q'}$, and \condenv{\rho'}{L'}{\policy}}, from IH, we have that:
\begin{itemize}
\item $c_1' = \fstoptwo{V_1'}$, $c_2' = \fstoptwo{V_2'}$, and 
\item \purpletext{$\tuple{V_1',V_2',Q'} \in \inter{\encpub{L'}}{\airforcetext{\rho'}}$}.
%\item \redtext{$\tuple{V_1',V_2',Q'} \in \inter{\encpub{L'}}{\airforcetext{}\consenv{L'}}$}.
\end{itemize}

Therefore, we have that $\fstoptwo{V_1} = \tuple{\intType, \tuple{\unitcon, \fstoptwo{V_1'}}}$, $\fstoptwo{V_2} = \tuple{\intType, \tuple{\unitcon, \fstoptwo{V_2'}}}$.
In other words, $\rho(\alpha_\policy) = \tuple{\fstoptwo{V_1},\fstoptwo{V_2},\tuple{R_x, \tuple{\tuple{},Q'}}}$.

We now need to prove that $\tuple{\tuple{\atcmod{\intType},\tuple{\attmod{v_1},V_1'}}, \tuple{\atcmod{\intType},\tuple{\attmod{v_2},V_2'}}, \tuple{R_x, \tuple{\tuple{},Q'}}} \in \inter{\Sigma \alpha_x:\atksign{\basekind}.\Sigma \alpha:\atcsign{\alpha_x}.\encpub{L'}}{\rho}$.

From the definition of \inter{\Sigma \alpha_x:\atksign{\basekind}.\Sigma \alpha:\atcsign{\alpha_x}.\encpub{L'}}{\rho}, we need to prove that:
\begin{itemize}
\item \typeEvalI{}{V_{1}:\rho_{L}(\encpub{L})},
\item \typeEvalI{}{V_{2}:\rho_{R}(\encpub{L})},
\item $\tuple{\atcmod{\intType},\atcmod{\intType}, R_x} \in \inter{\atksign{\basekind}}{\rho}$,
\item $\tuple{\tuple{\attmod{v_1},V_1'}, \tuple{\attmod{v_2},V_2'}, \tuple{\tuple{},Q'}} \in \inter{\Sigma \alpha:\atcsign{\alpha_x}.\encpub{L'}}{{\rho_1}}$, {where $\rho_1 = \rho, \alpha_x\mapsto \tuple{\intType,\intType,R_x}$}:
    \begin{itemize}
    \item {\typeEvalI{}{\tuple{\attmod{v_1},V_1'}: \Sigma \alpha:\atcsign{\intType}.\encpub{L'}}} (since $\Sigma \alpha:\atcsign{\intType}.\encpub{L'}$ is a closed signature),
%    \redtext{\typeEvalI{}{\tuple{\attmod{v_1},V_1'}: \rho_{1L}(\Sigma \alpha:\atcsign{\alpha_x}.\encpub{L'})}}, and hence, 
    \item {\typeEvalI{}{\tuple{\attmod{v_2},V_2'}: \Sigma \alpha:\atcsign{\intType}.\encpub{L'}}} (since $\Sigma \alpha:\atcsign{\intType}.\encpub{L'}$ is a closed signature),
%    \redtext{\typeEval{}{\tuple{\attmod{v_2},V_2'}: \rho_{1R}(\Sigma \alpha:\atcsign{\alpha_x}.\encpub{L'})}}, and hence, 
    \item $\tuple{\attmod{v_1}, \attmod{v_2}, \tuple{}} \in \inter{\atcsign{\alpha_x}}{{\rho_1}}$,
    \item $\tuple{V_1',V_2',Q'}\in \inter{\encpub{L'}}{{\rho_2}}$, where $\rho_2 = \rho,\alpha_x\mapsto \tuple{\intType,\intType,R_x}, \alpha\mapsto \tuple{\fstoptwo{\atcmod{v_1'}},\fstoptwo{\atcmod{v_2'}}, \tuple{}}$.    
    \end{itemize}
\end{itemize}

These items are proven as below.
\begin{itemize}
\item {From {Lemma~\ref{lem:ml:consenv:value:property}}, \typeEvalP{}{V_{1}:\encpub{L}} and hence \typeEvalI{}{V_{1}:\encpub{L}}.
From Lemma~\ref{lem:ml:signature:property:public}, \typeEval{}{\encpub{L}:\sign}.
Thus, \typeEvalI{}{V_{1}:\rho_{L}(\encpub{L})}.}

\item Similarly, we have that \typeEvalI{}{V_{2}:\rho_{R}(\encpub{L})}.

\item From Lemma~\ref{lem:ml:pitts-closure}, we have that $\tuple{{\intType},{\intType}, R_x, R_x} \in \inter{{\basekind}}{\rho}$ and hence, $\tuple{\atcmod{\intType},\atcmod{\intType}, R_x} \in \inter{\atksign{\basekind}}{\rho}$.

\item We have that \typeEvalP{}{\attmod{v_1}:\atcsign{\intType}}.
{From {Lemma~\ref{lem:ml:consenv:value:property}}, \typeEvalP{}{V_1':\enccon{L'}}}.
Since $\alpha \not \in FV(\enccon{L'})$, from the \WfmPr\ rule, \typeEvalP{}{\tuple{\attmod{v_1}, V_1}:\Sigma \alpha:\atcsign{\intType}.\enccon{L'}}.

We have that \typeEval{}{\atcsign{\intType} \leq \atcsign{\intType}:\sign},  \typeEval{\alpha:\unitkind}{\enccon{L'} \leq \encpub{L'}:\sign}, and \typeEval{\alpha:\unitkind}{\encpub{L'}:\sign}.
Thus, 

$$\LabelRule{\SubsDPr}{
    \typeEval{}{\atcsign{\intType} \leq \atcsign{\intType}:\sign}\\
    \typeEval{\alpha:\unitkind}{\enccon{L'} \leq \encpub{L'}:\sign}\\
    \typeEval{\alpha:\unitkind}{\encpub{L'}:\sign}
}
{\typeEval{}{\Sigma \alpha:\atcsign{\intType}.\enccon{L'} \leq \Sigma \alpha:\atcsign{\intType}.\encpub{L'}:\sign}}$$

Since \typeEvalP{}{\tuple{\attmod{v_1}, V_1}:\Sigma \alpha:\atcsign{\intType}.\enccon{L'}}, from \WfmSub, \typeEvalP{}{\tuple{\attmod{v_1}, V_1}:\Sigma \alpha:\atcsign{\intType}.\encpub{L'}}.
From the forgetful rule, \typeEvalI{}{\tuple{\attmod{v_1}, V_1}:\Sigma \alpha:\atcsign{\intType}.\encpub{L'}}.

\item  Similarly, we have that \purpletext{\typeEvalI{}{\tuple{\attmod{v_2},V_2'}: \Sigma \alpha:\atcsign{\intType}.\encpub{L'}}},

\item From the requirement on $v_1$ and $v_2$ in \condenv{\rho}{L}{\policy}, and the definition of \inter{\alpha_x}{\rho, \alpha_x\mapsto \tuple{\intType,\intType,R_x}}, we have that $\tuple{v_1,v_2} \in \inter{\alpha_x}{\rho, \alpha_x\mapsto \tuple{\intType,\intType,R_x}}$ and hence, $\tuple{\attmod{v_1}, \attmod{v_2}, \tuple{}} \in \inter{\atcsign{\alpha_x}}{\rho, \alpha_x\mapsto \tuple{\intType,\intType,R_x}} = \inter{\atcsign{\alpha_x}}{{\rho_1}}$.

\item From IH, we have that $\tuple{V_1',V_2',Q'} \in \inter{\encpub{L'}}{\airforcetext{\rho'}}$.
\purpletext{Since \typeEval{}{\encpub{L'}:\sign}, and $\tuple{\rho', \rho_2} \in \interenv{.}$, from Lemma~\ref{lem:ml:crary:interpretation}, we have that $\inter{\encpub{L'}}{\airforcetext{\rho'}} = \inter{\encpub{L'}}{\rho_2}$.
Therefore, we have that $\tuple{V_1',V_2',Q'} \in \inter{\encpub{L'}}{\rho_2}$.}
\end{itemize}

\emcase{Case 3:} $L = x\concat L'$, $\decPolicy{\policy}(x) = f$, where \typeEval{}{f:\intType \rightarrow\tau}.
We have that:
\begin{align*}
\purpletext{\rho(m_\policy)} & = \tuple{\tuple{\atcmod{\intType},\tuple{\attmod{v_1},\tuple{\attmod{f},V_1'}}},\tuple{\atcmod{\intType},\tuple{\attmod{v_2},\tuple{\attmod{f},V_2'}}}},\\
\purpletext{\rho(\alpha_\policy)} & = \tuple{\tuple{\intType,\tuple{\unitcon,\tuple{\unitcon, c_1'}}}, \tuple{\intType,\tuple{\unitcon,\tuple{\unitcon, c_2'}}},\tuple{R_f, \tuple{\tuple{}, \tuple{\tuple{}, Q'}}}},
\end{align*}

and $\tuple{v_1,v_2} \in R_f$, and \condenv{\rho'}{L'}{\policy}, where \purpletext{$\rho'(m_\policy) = \tuple{V_1',V_2'}$ and $\rho'(\alpha_\policy) = \tuple{c_1',c_2',Q'}$}.
We also have that $\sigma = \Sigma \alpha_f:\atksign{\basekind}.\Sigma \alpha_1:\atcsign{\alpha_f}.\Sigma \alpha_2:\atcsign{\alpha_f \rightarrow \tau}.\encpub{L'}$

{The proof that $\rho(\alpha_\policy) = \tuple{\fstoptwo{V_1},\fstoptwo{V_2},Q}$ is similar to the one in Case 2.}
We now prove that $\tuple{V_1,V_2,Q}, \in \inter{\encpub{L}}{\rho}$.
We need to prove that:
\begin{itemize}
\item \typeEvalI{}{V_{1}:\rho_{L}(\encpub{L})},
\item \typeEvalI{}{V_{2}:\rho_{R}(\encpub{L})},
\item $\tuple{\atcsign{\intType},\atcsign{\intType}, R_f} \in \inter{\atksign{\basekind}}{\rho}$
\item $\tuple{\tuple{\attmod{v_1},\tuple{\attmod{f},V_1}}, \tuple{\attmod{v_2},\tuple{\attmod{f},V_2'}},\tuple{\tuple{}, \tuple{\tuple{}, Q'}}} \in \inter{\Sigma \alpha_1:\atcsign{\alpha_f}.\Sigma \alpha_2:\atcsign{\alpha_f \rightarrow \tau}.\encpub{L'}}{\rho_1}$, where $\rho_1 = \rho, \alpha_f \mapsto \tuple{\intType,\intType, R_f}$   
\end{itemize}

We have that:
\begin{itemize}
\item From {Lemma~\ref{lem:ml:consenv:value:property}}, \typeEvalP{}{V_{1}:\encpub{L}} and hence \typeEvalI{}{V_{1}:\encpub{L}}.
From Lemma~\ref{lem:ml:signature:property:public}, \typeEval{}{\encpub{L}:\sign}.
Thus, \typeEvalI{}{V_{1}:\rho_{L}(\encpub{L})}.

\item Similarly, we have that \typeEvalI{}{V_{2}:\rho_{R}(\encpub{L})}.

\item From Lemma~\ref{lem:ml:pitts-closure}, $R_f$ is Pitts closed.
Thus, $\tuple{\atcsign{\intType},\atcsign{\intType}, R_f} \in \inter{\atksign{\basekind}}{\rho}$.
\end{itemize}

We now prove that 
$$\tuple{\tuple{\attmod{v_1},\tuple{\attmod{f},V_1'}}, \tuple{\attmod{v_2},\tuple{\attmod{f},V_2'}},\tuple{\tuple{}, \tuple{\tuple{}, Q'}}} \in \inter{\Sigma \alpha_1:\atcsign{\alpha_f}.\Sigma \alpha_2:\atcsign{\alpha_f \rightarrow \tau}.\encpub{L'}}{\rho_1}.$$

We need to prove that:
\begin{itemize}
\item \typeEvalI{}{\tuple{\attmod{v_1},\tuple{\attmod{f},V_1'}}: \rho_{1L}({\Sigma \alpha_1:\atcsign{\alpha_f}.\Sigma \alpha_2:\atcsign{\alpha_f \rightarrow \tau}.\encpub{L'}})} and hence, \typeEvalI{}{\tuple{\attmod{v_1},\tuple{\attmod{f},V_1}}: {\Sigma \alpha_1:\atcsign{\intType}.\Sigma \alpha_2:\atcsign{\intType \rightarrow \tau}.\encpub{L'}}},

\item \typeEvalI{}{\tuple{\attmod{v_2},\tuple{\attmod{f},V_2'}}: \rho_{1L}({\Sigma \alpha_1:\atcsign{\alpha_f}.\Sigma \alpha_2:\atcsign{\alpha_f \rightarrow \tau}.\encpub{L'}})} and hence, \typeEvalI{}{\tuple{\attmod{v_1},\tuple{\attmod{f},V_1}}: {\Sigma \alpha_1:\atcsign{\intType}.\Sigma \alpha_2:\atcsign{\intType \rightarrow \tau}.\encpub{L'}}},

\item $\tuple{\attmod{v_1},\attmod{v_2}, \tuple{}} \in \inter{\atcsign{\alpha_f}}{\rho_1}$,

\item $\tuple{\tuple{\attmod{f},V_1'}, \tuple{\attmod{f},V_2'}, \tuple{\tuple{}, Q'}} \in \inter{\Sigma \alpha_2:\atcsign{\alpha_f \rightarrow \tau}.\encpub{L'}}{\rho_2}$, where $\rho_2 = \rho_1,\alpha_1 \mapsto \tuple{\unitcon,\unitcon,\tuple}$.
\end{itemize}

We have that:
\begin{itemize}
\item By using similar reasoning as in Case 2, \typeEvalI{}{\tuple{\attmod{v_1},\tuple{\attmod{f},V_1'}}: {\Sigma \alpha_1:\atcsign{\intType}.\Sigma \alpha_2:\atcsign{\intType \rightarrow \tau}.\encpub{L'}}}.

\item Similarly, \typeEvalI{}{\tuple{\attmod{v_2},\tuple{\attmod{f},V_2'}}: {\Sigma \alpha_1:\atcsign{\intType}.\Sigma \alpha_2:\atcsign{\intType \rightarrow \tau}.\encpub{L'}}}.

\item From the definition of \condenv{\rho}{L}{\policy}, $\tuple{v_1,v_2} \in R_f = \inter{\alpha_f}{\rho, \alpha_f \mapsto \tuple{\intType,\intType,R_f}}$.
Thus, $\tuple{\attmod{v_1},\attmod{v_2}, \tuple{}} \in \inter{\atcsign{\alpha_f}}{\rho, \alpha_f \mapsto \tuple{\intType,\intType,R_f}} = \inter{\atcsign{\alpha_f}}{\rho_1}$.

\end{itemize}

We now prove that $\tuple{\tuple{\attmod{f},V_1'}, \tuple{\attmod{f},V_2'}, \tuple{\tuple{}, Q'}} \in \inter{\Sigma \alpha_2:\atcsign{\alpha_f \rightarrow \tau}.\encpub{L'}}{\rho_2}$.
We need to prove that:
\begin{itemize}
\item \typeEvalI{}{\tuple{\attmod{f},V_1'}: \rho_{2L}(\Sigma \alpha_2:\atcsign{\alpha_f \rightarrow \tau}.\encpub{L'})} and hence, \typeEvalI{}{\tuple{\attmod{f},V_1'}: \Sigma \alpha_2:\atcsign{\intType \rightarrow \tau}.\encpub{L'}} (since $\Sigma \alpha_2:\atcsign{\intType \rightarrow \tau}.\encpub{L'}$ is a closed signature)
\item \typeEvalI{}{\tuple{\attmod{f},V_2'}: \rho_{2R}(\Sigma \alpha_2:\atcsign{\alpha_f \rightarrow \tau}.\encpub{L'})} and hence, \typeEvalI{}{\tuple{\attmod{f},V_2'}: \Sigma \alpha_2:\atcsign{\intType \rightarrow \tau}.\encpub{L'}}

\item $\tuple{\attmod{f},\attmod{f},\tuple{}} \in \inter{\atcsign{\alpha_f\rightarrow\tau}}{\rho_2}$

\item $\tuple{V_1',V_2',Q'} \in \inter{\encpub{L'}}{\rho_3}$, where $\rho_3 = \rho_2,\alpha_2 \mapsto \tuple{\unitcon,\unitcon,\tuple{}}$.
\end{itemize}

We have that:
\begin{itemize}
\item By using similar reasoning as in Case 2, \typeEvalI{}{\tuple{\attmod{f},V_1'}: \rho_{2L}(\Sigma \alpha_2:\atcsign{\alpha_f \rightarrow \tau}.\encpub{L'})}

\item Similarly, \typeEvalI{}{\tuple{\attmod{f},V_2'}: \rho_{2R}(\Sigma \alpha_2:\atcsign{\alpha_f \rightarrow \tau}.\encpub{L'})}

\item We consider $\tuple{w_1,w_2} \in R_f$.
{From the definition of $R_f$, $\tuple{f\ w_1,f\ w_2} \in \termrelation{\tau}{\emptyset}$}.
Thus, we have that $\tuple{f,f} \in \inter{\alpha_f \rightarrow\tau}{\rho_2}$.
In other words, $\tuple{\attmod{f},\attmod{f},\tuple{}} \in \inter{\atcsign{\alpha_f\rightarrow\tau}}{\rho_2}$.

\item We now prove $\tuple{V_1',V_2',Q'} \in \inter{\encpub{L'}}{\rho_3}$.
From IH, $\tuple{V_1',V_2',Q'} \in \inter{\encpub{L'}}{\rho'}$
Since $\tuple{\rho',\rho_3} \in \interenv{.}$ and \typeEval{}{\encpub{L'}:\sign} (Lemma~\ref{lem:ml:signature:property:public}), from Lemma~\ref{lem:ml:crary:interpretation}, we have that $\inter{\encpub{L'}}{\rho'} = \inter{\encpub{L'}}{\rho_3}$.
Thus, $\tuple{V_1',V_2',Q'} \in \inter{\encpub{L'}}{\rho_3}$.

\end{itemize}

\end{proof}

\subsection{Wrapper}

\begin{lemma}
\label{lem:ml:wrapper:conf-view:P-form}
\tealtext{If \typeEvalP{\Gamma}{\attmod{e}:\sigma}, then $\sigma = \atcsign{\tau}$ for some $\tau$ and \typeEval{\Gamma}{e:\tau}.}
\end{lemma}
\begin{proof}
We prove this lemma by induction on the derivation of \typeEvalP{\Gamma}{\attmod{e}:\sigma}.
We consider the last rule applied in the derivation.
We have two cases (since other rules cannot be applied).

\emcase{Case 1:} Rule \WfmAtt.
$$\EmptyRule{
    \typeEval{\Gamma}{e:\tau}
}
{\wfmod{\Gamma}{P}{\attmod{e}:\atksign{\tau}}}$$

The proof follows from the rule.

\emcase{Case 2:} Rule \WfmSub.
$$\EmptyRule{
    \typeEvalP{\Gamma}{\attmod{e}:\sigma'}\\
    \typeEval{\Gamma}{\sigma'\leq \sigma:\sign}
}
{\wfmod{\Gamma}{P}{\attmod{e}:\sigma}}$$

From the rule, we have that \typeEvalP{\Gamma}{\attmod{e}:\sigma'}.
From IH, $\sigma' = \atcsign{\tau'}$ for some $\tau'$ and \typeEval{\Gamma}{e:\tau'}.
Since \typeEval{\Gamma}{\sigma'\leq \sigma:\sign}, from {Lemma~\ref{lem:ml:sig-sub:atc}}, \typeEval{\Gamma}{\sigma' \equiv \sigma:\sign}.
From {Lemma~\ref{lem:ml:sig-eq:atc}}, it follows that $\sigma = \atcsign{\tau}$ for some $\tau$ s.t. \typeEval{\Gamma}{\tau \equiv \tau':\basekind}.
Thus, \typeEval{\Gamma}{e:\tau}.
\end{proof}

\textbf{Lemma~\ref{lem:ml_trni:typability_impl} (in \S\ref{sec:ml:trni}).}
\tealtext{If \typeEval{\pubViewTerm{\policy}}{e:\tau}, then \typeEval{\conView{\policy}}{e}.}
\begin{proof}
First we have that \typeEvalP{}{\lambdaap \alpha_\policy,m_\policy:\sigmaPol{\policy}. \attmod{e}: \Piap \alpha_\policy:\sigmaPol{\policy}. \atksign{\tau}}.
$$\small
\LabelRule{ofm\_lamap}{
{\typeEval{}{\sigmaPol{\policy}:\sign}}\\
\LabelRule{ofm\_dyn}{\typeEval{\alpha_\policy,m_\policy:\sigmaPol{\policy}}{e:\tau}}
{\typeEvalP{\alpha_\policy,m_\policy:\sigmaPol{\policy}}{\attmod{e}:\atksign{\tau}}}
}
{\typeEvalP{}{\lambdaap \alpha_\policy,m_\policy:\sigmaPol{\policy}. \attmod{e}: \Piap \alpha_\policy:\sigmaPol{\policy}. \atksign{\tau}}}
$$

From the weakening lemma (Lemma~\ref{lem:ml:weakening}), we have that \typeEvalP{\alpha_\policy,m_\policy:\sigmaPolCon{\policy}}{\lambdaap \alpha_\policy,m_\policy:\sigmaPol{\policy}. \attmod{e}: \Piap \alpha_\policy:\sigmaPol{\policy}. \atksign{\tau}}.
In addition, we have that \typeEval{\alpha_\policy,m_\policy:\sigmaPolCon{\policy}}{\sigmaPolCon{\policy} \leq \sigmaPol{\policy}:\sign} (Lemma~\ref{lem:ml:signature:property:subsign} and Lemma~\ref{lem:ml:weakening}) and {\typeEval{\alpha_\policy,m_\policy:\sigmaPolCon{\policy}}{\fstop{m_\policy}{\alpha_\policy}}}.
%$$\LabelRule{ofm\_subsume}{
%	\typeEvalP{\alpha_\policy,m_\policy:\sigmaPolCon{\policy}}{m:\sigmaPolCon{\policy}}\\\\
%	\typeEval{\alpha_\policy,m_\policy:\sigmaPolCon{\policy}}{\sigmaPolCon{\policy} \leq \sigmaPol{\policy}:\sign}
%}
%{\typeEvalP{\alpha_\policy,m_\policy:\sigmaPolCon{\policy}}{m_\policy:\sigmaPol{\policy}}}$$
Therefore, we have that: 
$$
\small
\LabelRule{ofm\_appap}{
\typeEvalP{\alpha_\policy,m_\policy:\sigmaPolCon{\policy}}{\lambdaap \alpha_\policy,m_\policy:\sigmaPol{\policy}. \attmod{e}: \Piap \alpha_\policy:\sigmaPol{\policy}. \atksign{\tau}} \\
\typeEvalP{\alpha_\policy,m_\policy:\sigmaPolCon{\policy}}{m:\sigmaPol{\policy}}\\
\typeEval{\alpha_\policy,m_\policy:\sigmaPolCon{\policy}}{\fstop{m_\policy}{\alpha_\policy}}
}
{\typeEvalP{\alpha_\policy,m_\policy:\sigmaPolCon{\policy}}{(\lambdaap \alpha_\policy,m_\policy:\sigmaPol{\policy}. \attmod{e})\ m_\policy:  \atksign{\tau}[\alpha_\policy \mapsto \alpha_\policy]}}
$$

and hence, \typeEvalP{\alpha_\policy,m_\policy:\sigmaPolCon{\policy}}{(\lambdaap \alpha_\policy,m_\policy:\sigmaPol{\policy}. \attmod{e})\ m_\policy:  \atksign{\tau}}.

Since \typeEval{\alpha_\policy,m_\policy:\sigmaPolCon{\policy}}{\fstop{m_\policy}{\alpha_\policy}}, from mstep\_app3, 
$$\typeEval{\alpha_\policy,m_\policy:\sigmaPolCon{\policy}}{(\lambdaap \alpha_\policy,m_\policy:\sigmaPol{\policy}. \attmod{e})\ m_\policy} \transit \typeEval{\alpha_\policy,m_\policy:\sigmaPolCon{\policy}}{e[m_\policy \mapsto m_\policy,\alpha_\policy \mapsto \alpha_\policy]}.$$

And thus,
$$\typeEval{\alpha_\policy,m_\policy:\sigmaPolCon{\policy}}{(\lambdaap \alpha_\policy,m_\policy:\sigmaPol{\policy}. \attmod{e})\ m_\policy:  \atksign{\tau}} \transit \typeEval{\alpha_\policy,m_\policy:\sigmaPolCon{\policy}}{e}.$$

From the type preservation theorem (\cite[Theorem 2.2]{Crary18}), we have that $\typeEvalP{\alpha_\policy,m_\policy:\sigmaPolCon{\policy}}{\attmod{e}:  \atksign{\tau}}$.
From Lemma~\ref{lem:ml:wrapper:conf-view:P-form}, it follows that $\typeEvalP{\alpha_\policy,m_\policy:\sigmaPolCon{\policy}}{{e}}$.
\end{proof}

\begin{lemma}
\label{lem:ml:wrapper:dummy:signature}
\tealtext{It follows that \typeEvalP{}{\dummy:\sigmaPolCon{\policy}}.}
\end{lemma}
%\textbf{Lemma~\ref{lem:ml:wrapper:dummy:signature}.}
%\tealtext{It follows that \typeEvalP{}{\dummy:\sigmaPolCon{\policy}}.}
\begin{proof}
We prove the lemma by induction on $L$.
\end{proof}

\begin{lemma}
\label{lemma:ml:wrapper:V-policy:seal}
It follows that \typeEvalI{}{(\dummy \seal \sigmaPol{\policy}):\sigmaPol{\policy}}.
\end{lemma}
%\textbf{Lemma~\ref{lemma:ml:wrapper:V-policy:seal}.}
%\tealtext{It follows that \typeEvalI{}{(\dummy \seal \sigmaPol{\policy}):\sigmaPol{\policy}}.}
\begin{proof}
From Lemma~\ref{lem:ml:wrapper:dummy:signature}, we have that \typeEvalP{}{\dummy:\sigmaPolCon{\policy}}.
Since \typeEval{}{\sigmaPolCon{\policy} \leq \sigmaPol{\policy}:\sign} (by Lemma~\ref{lem:ml:signature:property:subsign}), from the \WfmSub\ rule, it follows that \typeEvalP{}{\dummy:\sigmaPol{\policy}}.
From the \WfmRFour\ rule, we have that \typeEvalI{}{\dummy:\sigmaPol{\policy}}.
From the \WfmSeal\ rule, we have that \typeEvalI{}{(\dummy\seal \sigmaPol{\policy}):\sigmaPol{\policy}}.
%The proof of \typeEvalI{}{(V\seal \enccon{\varPolicy{\policy}}):\enccon{\varPolicy{\policy}}} is similar.
\end{proof}

\begin{lemma}
\label{lem:ml:wrapper:sealing:type-inference}
For any $V$, $\sigma$ and $\sigma'$, if \typeEvalI{}{(V \seal \sigma):\sigma'} then \typeEval{}{\sigma \leq \sigma':\sign}.
\end{lemma}
%\begin{proof}
%By induction on the derivation of \typeEvalI{}{V \seal \sigma:\sigma'}.
%\end{proof}
%\textbf{Lemma~\ref{lem:ml:wrapper:sealing:type-inference}.}
%%\begin{lemma}
%%\label{lem:ml:wrapper:sealing:type-inference}
%\tealtext{If \typeEvalI{}{(V \seal \sigma):\sigma'} then \typeEval{}{\sigma \leq \sigma':\sign}.}
%%\end{lemma}
\begin{proof}
We prove the lemma by induction on the derivation of \typeEvalI{}{V \seal \sigma:\sigma'}.
We consider the last rule applied in the derivation.
We have two cases.

\emcase{Case 1:} Rule \WfmSeal. From the rule, we have that $\sigma = \sigma'$ and hence, \typeEval{}{\sigma \leq \sigma':\sign}.

\emcase{Case 2:} Rule \WfmSub. From the rule, we have \typeEvalI{}{(V\seal \sigma):\sigma''} and \typeEval{}{\sigma'' \leq \sigma':\sign}.
Since \typeEvalI{}{(V\seal \sigma):\sigma''}, from IH, \typeEval{}{\sigma \leq \sigma'':\sign}.
Since \typeEval{}{\sigma \leq \sigma'':\sign} and \typeEval{}{\sigma'' \leq \sigma':\sign}, from the \SubsT\ rule, it follows that \typeEval{}{\sigma \leq \sigma':\sign}.
\end{proof}

\begin{lemma}
\label{lem:ml:sig-eq:atc}
\tealtext{Suppose that \typeEval{\Gamma}{\tau:\basekind}.
It follows that:
\begin{itemize}
\item if \typeEval{\Gamma}{\sigma \equiv \atcsign{\tau}:\sign}, then $\sigma = \atcsign{\tau'}$ for some $\tau'$ s.t. \typeEval{\Gamma}{\tau \equiv \tau':\basekind},
\item if \typeEval{\Gamma}{\atcsign{\tau} \equiv \sigma:\sign}, then $\sigma = \atcsign{\tau'}$ for some $\tau'$ s.t. \typeEval{\Gamma}{\tau \equiv \tau':\basekind}.
\end{itemize}}
\end{lemma}
%\textbf{Lemma~\ref{lem:ml:sig-eq:atc}.}
%\tealtext{Suppose that \typeEval{\Gamma}{\tau:\basekind}.
%It follows that:
%\begin{itemize}
%\item if \typeEval{\Gamma}{\sigma \equiv \atcsign{\tau}:\sign}, then $\sigma = \atcsign{\tau'}$ for some $\tau'$ s.t. \typeEval{\Gamma}{\tau \equiv \tau':\basekind},
%\item if \typeEval{\Gamma}{\atcsign{\tau} \equiv \sigma:\sign}, then $\sigma = \atcsign{\tau'}$ for some $\tau'$ s.t. \typeEval{\Gamma}{\tau \equiv \tau':\basekind}.
%\end{itemize}}
\begin{proof}
We prove the lemma by induction on derivation of \typeEval{\Gamma}{\atcsign{\tau} \equiv \sigma:\sign} and \typeEval{\Gamma}{\atcsign{\tau} \equiv \sigma:\sign}.
We consider the last rule applied.
We have four cases (other rules cannot be the last rule of the derivation).

\emcase{Case 1:} \SeqR.
From the rule, we have that $\sigma = \atcsign{\tau}$.
Since \typeEval{\Gamma}{\tau:\basekind}, from the \EqcR rule, we have that \typeEval{\Gamma}{\tau \equiv \tau:\basekind}.

\emcase{Case 2:} \SeqS.
We consider \typeEval{\Gamma}{\atcsign{\tau} \equiv \sigma:\sign} first.
From the rule, we have that \typeEval{\Gamma}{\sigma \equiv \atcsign{\tau}}.
From IH, $\sigma = \atcsign{\tau'}$ for some $\tau'$ s.t. \typeEval{\Gamma}{\tau \equiv \tau':\basekind}.

We now consider \typeEval{\Gamma}{\sigma \equiv \atcsign{\tau}:\sign}.
From the rule, we have that \typeEval{\Gamma}{\atcsign{\tau} \equiv \sigma:\sign}.
From IH, $\sigma = \atcsign{\tau'}$ for some $\tau'$ s.t. \typeEval{\Gamma}{\tau \equiv \tau':\basekind}.

\emcase{Case 3:} \SeqT
We consider \typeEval{\Gamma}{\atcsign{\tau} \equiv \sigma:\sign} first. From the rule, we have that \typeEval{\Gamma}{\atcsign{\tau} \equiv \sigma':\sign} and \typeEval{\Gamma}{\sigma' \equiv \sigma:\sign}.
From IH, it follows that $\sigma' = \atcsign{\tau''}$ for some $\tau''$ s.t. \typeEval{\Gamma}{\tau \equiv \tau'':\basekind}.
From the static semantics, it follows that \typeEval{\Gamma}{\tau'':\basekind}.
Thus, from IH, $\sigma = \atcsign{\tau'}$ for some $\tau'$ s.t. \typeEval{\Gamma}{\tau'' \equiv \tau':\basekind}.
From the \EqcT\ rule, it follows that \typeEval{\Gamma}{\tau \equiv \tau':\basekind}.

We now consider \typeEval{\Gamma}{\sigma \equiv \atcsign{\tau}:\sign}. From the rule, we have that \typeEval{\Gamma}{\sigma \equiv \sigma':\sign} and \typeEval{\Gamma}{\sigma' \equiv \atcsign{\tau}:\sign}.
From IH, it follows that $\sigma' = \atcsign{\tau''}$ for some $\tau''$ s.t. \typeEval{\Gamma}{\tau \equiv \tau'':\basekind}.
From the static semantics, it follows that \typeEval{\Gamma}{\tau'':\basekind}.
Thus, from IH, $\sigma = \atcsign{\tau'}$ for some $\tau'$ s.t. \typeEval{\Gamma}{\tau'' \equiv \tau':\basekind}.
From the \EqcT\ rule, it follows that \typeEval{\Gamma}{\tau \equiv \tau':\basekind}.

\emcase{Case 4:} \SeqAtc.
We consider \typeEval{\Gamma}{\atcsign{\tau} \equiv \sigma:\sign} first.
From the rule, $\sigma = \atcsign{\tau'}$ for some $\tau'$ and \typeEval{\Gamma}{\tau \equiv \tau':\basekind}.

We now consider the case \typeEval{\Gamma}{\sigma \equiv \atcsign{\tau}:\sign}.
From the rule, $\sigma = \atcsign{\tau'}$ for some $\tau'$ and \typeEval{\Gamma}{\tau' \equiv \tau:\basekind}.
From the \EqcS\ rule, it follows that \typeEval{\Gamma}{\tau \equiv \tau':\basekind}.
\end{proof}

%\noteinline{The above proof uses the fact that if \typeEval{\Gamma}{c \equiv c':k} then \typeEval{\Gamma}{c:k} and \typeEval{\Gamma}{c':k}.}

\begin{lemma}
\label{lem:ml:sig-sub:atc}
\tealtext{Suppose that \typeEval{\Gamma}{\atcsign{\tau} \leq \sigma:\sign}.
It follows that $\typeEval{\Gamma}{\atcsign{\tau} \equiv \sigma:\sign}$.}
\end{lemma}
%\begin{proof}
%By induction on the derivation of \typeEval{\Gamma}{\atcsign{\tau} \leq \sigma:\sign}.
%\end{proof}
%\textbf{Lemma~\ref{lem:ml:sig-sub:atc}.}
%\tealtext{Suppose that \typeEval{\Gamma}{\atcsign{\tau} \leq \sigma:\sign}.
%It follows that $\typeEval{\Gamma}{\atcsign{\tau} \equiv \sigma:\sign}$.}
\begin{proof}
We prove this lemma by induction on the derivation of \typeEval{\Gamma}{\atcsign{\tau} \leq \sigma:\sign}.
We consider the last rule applied.
We have two cases (other rules cannot be the last rule in the derivation).

\emcase{Case 1:} \SubsEq.
From the rule, we have that $\typeEval{\Gamma}{\atcsign{\tau} \equiv \sigma:\sign}$.

\emcase{Case 2:} \SubsT.
From the rule, we have that \typeEval{\Gamma}{\atcsign{\tau} \leq \sigma':\sign} and \typeEval{\Gamma}{\sigma' \leq \sigma:\sign}.
Since \typeEval{\Gamma}{\atcsign{\tau} \leq \sigma':\sign}, from IH, we have that \typeEval{\Gamma}{\atcsign{\tau} \equiv \sigma':\sign}.
From Lemma~\ref{lem:ml:sig-eq:atc}, $\sigma' = \atcsign{\tau'}$ for some $\tau'$ s.t. \typeEval{\Gamma}{\tau \equiv \tau':\basekind}.
Since \typeEval{\Gamma}{\sigma' \leq \sigma:\sign}, we have that \typeEval{\Gamma}{\atcsign{\tau'} \leq \sigma:\sign}.
From IH, we have that \typeEval{\Gamma}{\atcsign{\tau'} \equiv \sigma:\sign}.
Since $\sigma' = \atcsign{\tau'}$, it follows that \typeEval{\Gamma}{\sigma' \equiv \sigma:\sign}.

Since \typeEval{\Gamma}{\atcsign{\tau} \equiv \sigma':\sign}, and \typeEval{\Gamma}{\sigma' \equiv \sigma:\sign}, we have that \typeEval{\Gamma}{\atcsign{\tau} \equiv \sigma:\sign}.
%
%\begin{itemize}
%\item \typeEval{\Gamma}{\atcsign{\tau} \equiv \sigma':\sign},
%\item \typeEval{\Gamma}{\sigma' \equiv \sigma:\sign}.
%\end{itemize}
%
%From \EqcT\, we have that \typeEval{\Gamma}{\atcsign{\tau} \equiv \sigma:\sign}.
\end{proof}

\begin{lemma}
\label{lem:ml:wrapper:kindeq:property}
\tealtext{It follows that:
\begin{itemize}
\item if \typeEval{\Gamma}{\basekind \equiv k:\kind} then $k$ is $\basekind$,
\item if \typeEval{\Gamma}{k \equiv \basekind:\kind} then $k$ is $\basekind$,
\item if \typeEval{\Gamma}{\unitkind \equiv k:\kind} then $k$ is $\unitkind$,
\item if \typeEval{\Gamma}{k \equiv \unitkind:\kind} then $k$ is $\unitkind$.
\end{itemize}}
\end{lemma}
\begin{proof}
We prove this lemma by induction on the derivation of \typeEval{\Gamma}{\basekind \equiv k:\kind}, \typeEval{\Gamma}{k \equiv \basekind:\kind}, \typeEval{\Gamma}{\unitkind \equiv k:\kind}, and \typeEval{\Gamma}{k \equiv \unitkind:\kind}.
We have three cases.

\emcase{Case 1:} Rule \EqkR.
We consider \unitkind\ first.
From the rule, we have that $k$ is \unitkind.
The proof for \basekind\ is similar.

\emcase{Case 2:} Rule \EqkS.
\begin{itemize}
\item Case \typeEval{\Gamma}{\basekind \equiv k:\kind}. 
From the rule, \typeEval{\Gamma}{k \equiv \basekind:\kind}.
From IH, $k$ is \basekind.

\item Case \typeEval{\Gamma}{k \equiv \basekind:\kind}.
From the rule, \typeEval{\Gamma}{\basekind \equiv k:\kind}. 
From IH, $k$ is \basekind.

\item if \typeEval{\Gamma}{\unitkind \equiv k:\kind}.
From the rule, \typeEval{\Gamma}{k \equiv \unitkind:\kind}.
From IH, $k$ is \unitkind.

\item Case \typeEval{\Gamma}{k \equiv \unitkind:\kind}.
From the rule, \typeEval{\Gamma}{\unitkind \equiv k:\kind}. 
From IH, $k$ is \unitkind.
\end{itemize}

\emcase{Case 3:} Rule \EqkT.
The proof is similar to the proof of Case 2. %\tbupdated
\end{proof}

\begin{lemma}
\label{lem:ml:wrapper:signeq:property}
\tealtext{It follows that:
\begin{itemize}
\item if \typeEval{\Gamma}{\unitkind \equiv \sigma:\sign} then $\sigma$ is $\unitkind$,
\item if \typeEval{\Gamma}{\sigma \equiv \unitkind:\sign} then $\sigma$ is $\unitkind$,
\item if \typeEval{\Gamma}{\atksign{\basekind} \equiv \sigma:\sign} then $\sigma$ is $\atksign{\basekind}$,
\item if \typeEval{\Gamma}{\sigma \equiv \atksign{\basekind}:\sign} then $\sigma$ is $\atksign{\basekind}$.
\end{itemize}}
\end{lemma}
\begin{proof}
We prove the two first items of the lemma by induction on the derivation of \typeEval{\Gamma}{\unitsign \equiv \sigma:\sign} and \typeEval{\Gamma}{\sigma \equiv \unitsign:\sign}.

\emcase{Case 1a:} Rule \SeqR.
From the rule, $\sigma$ is \unitsign.

\emcase{Case 2a:} Rule \SeqS.
\begin{itemize}
\item Case \typeEval{\Gamma}{\unitsign \equiv \sigma:\sign}. From the rule, \typeEval{\Gamma}{\sigma \equiv \unitsign:\sign}.
From IH, $\sigma$ is \unitsign.

\item Case \typeEval{\Gamma}{\sigma \equiv \unitsign:\sign}. From the rule, \typeEval{\Gamma}{\unitsign \equiv \sigma:\sign}.
From IH, $\sigma$ is \unitsign.
\end{itemize}

\emcase{Case 3a:} Rule \SeqT. The proof is similar to the proof of Case 2a.

We prove the last two items of the lemma by induction on the derivation of  \typeEval{\Gamma}{\atksign{\basekind} \equiv \sigma:\sign}, and \typeEval{\Gamma}{\sigma \equiv \atksign{\basekind}:\sign}.

\emcase{Case 1b:} Rule \SeqR.
From the rule, $\sigma$ is \atksign{\basekind}.

\emcase{Case 2b:} Rule \SeqS.
\begin{itemize}
\item Case \typeEval{\Gamma}{\atksign{\basekind} \equiv \sigma:\sign}.
From the rule, \typeEval{\Gamma}{\sigma \equiv \atksign{\basekind}:\sign}.
From IH, $\sigma$ is \atksign{\basekind}.

\item Case \typeEval{\Gamma}{\sigma \equiv \atksign{\basekind}:\sign}.
From the rule, \typeEval{\Gamma}{\atksign{\basekind} \equiv \sigma:\sign}.
From IH, $\sigma$ is \atksign{\basekind}.
\end{itemize}

\emcase{Case 3b:} Rule \SeqT. 
The proof is similar to the proof of Case 2b.

\emcase{Case 4b:} Rule \SeqAtk.
From the rule, $\sigma' = \atksign{k'}$ for some $k'$ s.t. \typeEval{\Gamma}{\basekind \equiv k':\kind}.
From Lemma~\ref{lem:ml:wrapper:kindeq:property}, $k'$ is $\basekind$.
Thus, $\sigma'$ is $\atksign{\basekind}$.
\end{proof}

\begin{lemma}
\label{lem:ml:wrapper:kindsub:property}
\tealtext{It follows that:
\begin{itemize}
\item if \typeEval{\Gamma}{\unitkind \leq k:\kind} then $k$ is $\unitkind$,
\item if \typeEval{\Gamma}{\basekind \leq k:\kind} then $k$ is $\basekind$.
\end{itemize}}
\end{lemma}
%\textbf{Lemma~\ref{lem:ml:wrapper:kindsub:property}.}
%\tealtext{It follows that:
%\begin{itemize}
%\item if \typeEval{\Gamma}{\unitkind \leq k:\kind} then $k$ is $\unitkind$,
%\item if \typeEval{\Gamma}{\basekind \leq k:\kind} then $k$ is $\basekind$.
%\end{itemize}}
\begin{proof}
We prove this lemma by induction on the derivation of \typeEval{\Gamma}{\_ \leq k:\kind}.
Here, we only prove the first part of the lemma.
The proof of the second part is similar.

\emcase{Case 1:}  Rule \SubkE. 
From the rule, \typeEval{\Gamma}{\unitkind \equiv k:\kind}.
From Lemma~\ref{lem:ml:wrapper:kindeq:property}, $k$ is $\unitkind$.

\emcase{Case 2:} Rule \SubkT.
From the rule, \typeEval{\Gamma}{\unitkind \leq k':\kind} and \typeEval{\Gamma}{k' \leq k:\kind}.

Since \typeEval{\Gamma}{\unitkind \leq k':\kind}, from IH, $k'$ is $\unitkind$.
Since \typeEval{\Gamma}{k' \leq k:\kind} and $k'$ is $\unitkind$, from IH, we have that k is $\unitkind$.
\end{proof}

\begin{lemma}
\label{lem:ml:wrapper:signsub:property}
\tealtext{It follows that:
\begin{itemize}
\item if \typeEval{\Gamma}{\unitsign \leq \sigma:\sign} then $\sigma$ is $\unitsign$,
\item if \typeEval{\Gamma}{\atksign{\basekind} \leq \sigma:\sign} then $\sigma$ is $\atksign{\basekind}$.
\end{itemize}}
\end{lemma}
%\begin{proof}
%By induction on the derivation of \typeEval{\Gamma}{\_ \leq \sigma:\sign}.
%\end{proof}
%\textbf{Lemma~\ref{lem:ml:wrapper:signsub:property}.}
%\tealtext{It follows that:
%\begin{itemize}
%\item if \typeEval{\Gamma}{\unitsign \leq \sigma:\sign} then $\sigma$ is $\unitsign$,
%\item if \typeEval{\Gamma}{\atksign{\basekind} \leq \sigma:\sign} then $\sigma$ is $\atksign{\basekind}$.
%\end{itemize}}
\begin{proof}
We prove the first part of the lemma by induction on the derivation of \typeEval{\Gamma}{\unitkind \leq \sigma:\sign}.

\emcase{Case 1a:}  Rule \SubsEq. 
From the rule, \typeEval{\Gamma}{\unitsign \equiv \sigma:\sign}.
From Lemma~\ref{lem:ml:wrapper:signeq:property}, $\sigma$ is $\unitsign$.

\emcase{Case 2a:} Rule \SubsT.
From the rule, \typeEval{\Gamma}{\unitsign \leq \sigma':\sign} and \typeEval{\Gamma}{\sigma' \leq \sigma:\sign}.

Since \typeEval{\Gamma}{\unitsign \leq \sigma':\sign}, from IH, $\sigma'$ is $\unitsign$.
Since \typeEval{\Gamma}{\sigma' \leq \sigma:\sign} and $\sigma'$ is $\unitsign$, from IH, we have that $\sigma$ is $\unitsign$.

We now prove the second part of the lemma by induction on the derivation of \typeEval{\Gamma}{\atksign{\basekind} \leq \sigma:\sign}.

\emcase{Case 1b:}  Rule \SubsEq. 
From the rule, \typeEval{\Gamma}{\atksign{\basekind} \equiv \sigma:\sign}.
From Lemma~\ref{lem:ml:wrapper:signeq:property}, $\sigma$ is $\atksign{\basekind}$.

\emcase{Case 2b:} Rule \SubsT.
From the rule, \typeEval{\Gamma}{\atksign{\basekind} \leq \sigma':\sign} and \typeEval{\Gamma}{\sigma' \leq \sigma:\sign}.

Since \typeEval{\Gamma}{\atksign{\basekind} \leq \sigma':\sign}, from IH, $\sigma'$ is $\atksign{\basekind}$.
Since \typeEval{\Gamma}{\sigma' \leq \sigma:\sign} and $\sigma'$ is $\atksign{\basekind}$, from IH, we have that $\sigma$ is $\atksign{\basekind}$.

\emcase{Case 3b:} Rule \SubsK.
From the rule, $\sigma$ is $k'$ s.t. \typeEval{\Gamma}{\basekind \leq k':\kind}.
From Lemma~\ref{lem:ml:wrapper:kindsub:property}, $k' = \basekind$.
\end{proof}

\begin{lemma}
\label{lem:ml:wrapper:signature:policy:property}
For any {$L\subseteq \varPolicy{\policy}$}, if \typeEval{\Gamma}{\encpub{L} \leq \sigma:\sign} for some $\sigma$, then $\typeEval{\Gamma}{\encpub{L} \equiv \sigma:\sign}$.
\end{lemma}
%\textbf{Lemma~\ref{lem:ml:wrapper:signature:policy:property}.}
%\tealtext{
%%Let $\Gamma$ be a context s.t. variables in \dom{\Gamma} are not free variables in \encpub{L}.
%For any $L \subseteq \varPolicy{\policy}$, if \typeEval{\Gamma}{\encpub{L} \leq \sigma:\sign} for some $\sigma$, then $\typeEval{\Gamma}{\encpub{L} \equiv \sigma:\sign}$.}
\begin{proof}
We prove the lemma by induction on $L$.

\emcase{Case 1:} $L = \emptylist$.
We have that $\encpub{L} = \unitsign$.
Therefore, we have that \typeEval{}{\unitsign \leq \sigma:\sign}.
From Lemma~\ref{lem:ml:wrapper:signsub:property}, we have that $\sigma$ is $\unitsign$.
Thus, \typeEval{\Gamma}{\unitsign \equiv \sigma:\sign}.

\emcase{Case 2:} $L = x \concat L'$ where $x \not \in \dom{\decPolicy{\policy}}$.
We have that $\encpub{L} = \Sigma \alpha_x:\atksign{\basekind}.\Sigma \alpha:\atcsign{\alpha_x}.\encpub{L'}$.
From the definition of subsignature, $\sigma = \Sigma \alpha_x:\sigma_1.\Sigma \alpha:\sigma_2.\sigma_3$.
Without loss of generality, we suppose that $\alpha_x$ and $\alpha$ are not in $\dom{\Gamma}$ (we can change the constructor variables if necessary).
Therefore, we have that $\Gamma, \alpha_x:\basekind\ \ok$ and $\Gamma, \alpha_x:\basekind, \alpha:\unitkind\ \ok$.

%Since \typeEval{\Gamma}{\encpub{L} \leq \sigma:\sign}, we have the following derivation.

Since \typeEval{\Gamma}{\encpub{L} \leq \sigma:\sign}, from the \SubsDPr\ rule, it follows that:
    \begin{itemize}
    \item \typeEval{\Gamma}{\atksign{\basekind} \leq \sigma_1:\sign}
    \item {\typeEval{\Gamma,\alpha_x:\basekind}{\Sigma \alpha:\atcsign{\alpha_x}.\encpub{L'} \leq \Sigma \alpha:\sigma_2.\sigma_3}}, and hence,
        \begin{itemize}
        \item \typeEval{\Gamma,\alpha_x:\basekind}{\atcsign{\alpha_x} \leq \sigma_2:\sign}
        \item \typeEval{\Gamma,\alpha_x:\basekind,\alpha:\unitkind}{\encpub{L'} \leq \sigma_3:\sign}
        \end{itemize}
    \end{itemize}

%$$\footnotesize
%\LabelRule{\SubsDPr}{
%\typeEval{\Gamma}{\atksign{\basekind} \leq \sigma_1:\sign}\hspace{40pt}
%\LabelRule{\SubsDPr}{
%\typeEval{\Gamma,\alpha_x:\basekind}{\atcsign{\alpha_x} \leq \sigma_2:\sign} \\
%\typeEval{\Gamma,\alpha_x:\basekind,\alpha:\unitkind}{\encpub{L'} \leq \sigma_3:\sign}\\
%\dots
%}
%{\typeEval{\Gamma,\alpha_x:\basekind}{\Sigma \alpha:\atcsign{\alpha_x}.\encpub{L'} \leq \Sigma \alpha:\sigma_2.\sigma_3}} \\
%\dots
%}
%{\typeEval{\Gamma}{\Sigma \alpha_x:\atksign{\basekind}.\Sigma \alpha:\atcsign{\alpha_x}.\encpub{L'} \leq \Sigma \alpha_x:\sigma_1.\Sigma \alpha:\sigma_2.\sigma_3}}
%$$

We have that:
\begin{itemize}
\item \typeEval{\Gamma}{\atksign{\basekind} \leq \sigma_1:\sign}.
From {Lemma~\ref{lem:ml:wrapper:signsub:property}}, $\sigma_1$ is $\atksign{\basekind}$ and hence, \typeEval{\Gamma}{\atksign{\basekind} \equiv \sigma_1:\sign}.

\item \typeEval{\Gamma,\alpha_x:\basekind}{\atcsign{\alpha_x} \leq \sigma_2:\sign}.
From {Lemma~\ref{lem:ml:sig-sub:atc}}, $\typeEval{\Gamma,\alpha_x:\basekind}{\atcsign{\alpha_x} \equiv \sigma_2:\sign}$.

\item \typeEval{\Gamma,\alpha_x:\basekind,\alpha:\unitkind}{\encpub{L'} \leq \sigma_3:\sign}.
From IH, we have that \typeEval{\Gamma,\alpha_x:\basekind,\alpha:\unitkind}{\encpub{L'}\equiv \sigma_3:\sign}.
\end{itemize}

Therefore, we have the following derivation:

$$\footnotesize
\LabelRule{\SeqDPr}{
\typeEval{\Gamma}{\atksign{\basekind} \equiv \sigma_1:\sign}\hspace{50pt}
\LabelRuleProof{\SeqDPr}{
\typeEval{\Gamma,\alpha_x:\basekind}{\atcsign{\alpha_x} \equiv \sigma_2:\sign} \\
\typeEval{\Gamma,\alpha_x:\basekind,\alpha:\unitkind}{\encpub{L'} \equiv \sigma_3:\sign}\\
}
{\typeEval{\Gamma,\alpha_x:\basekind}{\Sigma \alpha:\atcsign{\alpha_x}.\encpub{L'} \equiv \Sigma \alpha:\sigma_2.\sigma_3}}
}
{\typeEval{\Gamma}{\Sigma \alpha_x:\atksign{\basekind}.\Sigma \alpha:\atcsign{\alpha_x}.\encpub{L'} \equiv \Sigma \alpha_x:\sigma_1.\Sigma \alpha:\sigma_2.\sigma_3}}
$$

\emcase{Case 3:} $ = x \concat L'$ and $\decPolicy{\policy}(x) = f$, where \typeEval{}{f:\intType \rightarrow\tau}.
We have that $\encpub{L} = \Sigma \alpha_f:\atksign{\basekind}.\Sigma \alpha_1:\atcsign{\alpha_f}.\Sigma \alpha_2:\atcsign{\alpha_f \rightarrow \tau}.\encpub{L'}$.
From the definition of subsignature, $\sigma =  \Sigma \alpha_f:\sigma_1.\Sigma \alpha_1:\sigma_2.\Sigma \alpha_2:\sigma_3.\encpub{L'}$.
Without loss of generality, we suppose that $\alpha_f$, $\alpha_1$, and $\alpha_2$ are not in $\dom{\Gamma}$ (we can change the constructor variables if it is necessary).
Therefore, we have that $\Gamma, \alpha_f:\basekind\ \ok$, $\Gamma, \alpha_f:\basekind, \alpha_1:\unitkind\ \ok$, and $\Gamma, \alpha_f:\basekind, \alpha_1:\unitkind, \alpha_2:\unitkind\ \ok$.

Since \typeEval{\Gamma}{\encpub{L} \leq \sigma:\sign}, from the \SubsDPr\ rule, it follows that:
\begin{itemize}
\item {\typeEval{\Gamma}{\atksign{\basekind} \leq \sigma_1:\sign}}
\item \typeEval{\Gamma,\alpha_f:\basekind}{\Sigma \alpha_1:\atcsign{\alpha_f}.\Sigma \alpha_2:\atcsign{\alpha_f \rightarrow\tau}.\encpub{L'} \leq \Sigma \alpha_1:\sigma_2.\Sigma \alpha_2:\sigma_3.\sigma_4}, and hence:
    \begin{itemize}
    \item {\typeEval{\Gamma,\alpha_f:\basekind}{\atcsign{\alpha_f} \leq \sigma_2:\sign}}
    \item \typeEval{\Gamma,\alpha_f:\basekind,\alpha_1:\unitkind}{\Sigma \alpha_2:\atcsign{\alpha_f \rightarrow\tau}.\encpub{L'} \leq \Sigma \alpha_2:\sigma_3.\sigma_4:\sign}, and hence,
    \begin{itemize}
    \item {\typeEval{\Gamma,\alpha_f:\basekind,\alpha_1:\unitkind}{\atcsign{\alpha_f \rightarrow\tau} \leq \sigma_3}}
    \item {\typeEval{\Gamma,\alpha_f:\basekind,\alpha_1:\unitkind,\alpha_2:\unitkind}{\encpub{L'} \leq \sigma_4:\sign}}
    \end{itemize}

    \end{itemize}
\end{itemize}

%$$\small
%\LabelRule{}{
%    \redtext{\typeEval{\Gamma}{\atksign{\basekind} \leq \sigma_1:\sign}}\\
%    \typeEval{\Gamma,\alpha_f:\basekind}{\Sigma \alpha_1:\atcsign{\alpha_f}.\Sigma \alpha_2:\atcsign{\alpha_f \rightarrow\tau}.\encpub{L'} \leq \Sigma \alpha_1:\sigma_2.\Sigma \alpha_2:\sigma_3.\sigma_4}
%}
%{\typeEval{\Gamma}{\Sigma \alpha_f:\atksign{\basekind}.\Sigma \alpha_1:\atcsign{\alpha_f}.\Sigma \alpha_2:\atcsign{\alpha_f \rightarrow\tau}.\encpub{L'} \leq \Sigma \alpha_f:\sigma_1.\Sigma \alpha_1:\sigma_2.\Sigma \alpha_2:\sigma_3.\sigma_4}}
%$$
%
%$$
%\footnotesize
%\LabelRule{}{
%\redtext{\typeEval{\Gamma,\alpha_f:\basekind}{\atcsign{\alpha_f} \leq \sigma_2:\sign}} \\
%\LabelRule{}{~
%    \redtext{\typeEval{\Gamma,\alpha_f:\basekind,\alpha_1:\unitkind}{\atcsign{\alpha_f \rightarrow\tau} \leq \sigma_3}} \\
%    \redtext{\typeEval{\Gamma,\alpha_f:\basekind,\alpha_1:\unitkind,\alpha_2:\unitkind}{\encpub{L'} \leq \sigma_4:\sign}}
%}
%{\typeEval{\Gamma,\alpha_f:\basekind,\alpha_1:\unitkind}{\Sigma \alpha_2:\atcsign{\alpha_f \rightarrow\tau}.\encpub{L'} \leq \Sigma \alpha_2:\sigma_3.\sigma_4:\sign}}
%}
%{\typeEval{\Gamma,\alpha_f:\basekind}{\Sigma \alpha_1:\atcsign{\alpha_f}.\Sigma \alpha_2:\atcsign{\alpha_f \rightarrow\tau}.\encpub{L'} \leq \Sigma \alpha_1:\sigma_2.\Sigma \alpha_2:\sigma_3.\sigma_4}}
%$$

As in Case 2, we have that:
\begin{itemize}

\item \typeEval{\Gamma,\alpha_f:\basekind}{\atcsign{\alpha_f} \equiv \sigma_2:\sign},
\item \typeEval{\Gamma,\alpha_f:\basekind,\alpha_1:\unitkind}{\atcsign{\alpha_f \rightarrow\tau} \equiv \sigma_3},
\item \typeEval{\Gamma,\alpha_f:\basekind,\alpha_1:\unitkind,\alpha_2:\unitkind}{\encpub{L'} \equiv \sigma_4:\sign},
\item \typeEval{\Gamma}{\atksign{\basekind} \equiv \sigma_1:\sign},
\end{itemize}

Therefore, we have the following derivations:
$$
\footnotesize
\LabelRule{}{
{\typeEval{\Gamma,\alpha_f:\basekind}{\atcsign{\alpha_f} \equiv \sigma_2:\sign}} \\
\LabelRule{}{~
    {\typeEval{\Gamma,\alpha_f:\basekind,\alpha_1:\unitkind}{\atcsign{\alpha_f \rightarrow\tau} \equiv \sigma_3}} \\
    {\typeEval{\Gamma,\alpha_f:\basekind,\alpha_1:\unitkind,\alpha_2:\unitkind}{\encpub{L'} \equiv \sigma_4:\sign}}
}
{\typeEval{\Gamma,\alpha_f:\basekind,\alpha_1:\unitkind}{\Sigma \alpha_2:\atcsign{\alpha_f \rightarrow\tau}.\encpub{L'} \equiv \Sigma \alpha_2:\sigma_3.\sigma_4:\sign}}
}
{\typeEval{\Gamma,\alpha_f:\basekind}{\Sigma \alpha_1:\atcsign{\alpha_f}.\Sigma \alpha_2:\atcsign{\alpha_f \rightarrow\tau}.\encpub{L'} \equiv \Sigma \alpha_1:\sigma_2.\Sigma \alpha_2:\sigma_3.\sigma_4}}
$$

$$\footnotesize
\LabelRule{}{
    {\typeEval{\Gamma}{\atksign{\basekind} \equiv \sigma_1:\sign}}\\
    \typeEval{\Gamma,\alpha_f:\basekind}{\Sigma \alpha_1:\atcsign{\alpha_f}.\Sigma \alpha_2:\atcsign{\alpha_f \rightarrow\tau}.\encpub{L'} \equiv \Sigma \alpha_1:\sigma_2.\Sigma \alpha_2:\sigma_3.\sigma_4}
}
{\typeEval{\Gamma}{\Sigma \alpha_f:\atksign{\basekind}.\Sigma \alpha_1:\atcsign{\alpha_f}.\Sigma \alpha_2:\atcsign{\alpha_f \rightarrow\tau}.\encpub{L'} \equiv \Sigma \alpha_f:\sigma_1.\Sigma \alpha_1:\sigma_2.\Sigma \alpha_2:\sigma_3.\sigma_4}}
$$

Thus, we have that $\typeEval{\Gamma}{\Sigma \alpha_f:\atksign{\basekind}.\Sigma \alpha_1:\atcsign{\alpha_f}.\Sigma \alpha_2:\atcsign{\alpha_f \rightarrow\tau}.\encpub{L'} \equiv \Sigma \alpha_f:\sigma_1.\Sigma \alpha_1:\sigma_2.\Sigma \alpha_2:\sigma_3.\sigma_4}$.

%\emcase{Case 4:} $ = x \concat L'$ and $\decPolicy{\policy}(x) = f \circ a$, where \typeEval{}{f:\intType \rightarrow\tau}.
%The proof is similar to the proof of Case 3.
%
\end{proof}

\begin{lemma}
\label{lem:ml:wrapper:kind}
\tealtext{Let $\sigma$ be a signature s.t. \typeEval{\Gamma}{\sigmaPol{\policy} \equiv \sigma:\sign}.
It follows that $\fstsign{\sigma} = \fstsign{\sigmaPol{\policy}}$.}
\end{lemma}
\begin{proof}
We claim that for any $L \subseteq \varPolicy{\policy}$ and $\sigma$ s.t. \typeEval{}{\encpub{L} \equiv \sigma:\sign}, it follows that $\fstsign{\sigma} = \fstsign{\encpub{L}}$. The proof then follows directly from the claim.
We prove the claim by induction on $L$.

\emcase{Case 1:} $L = \emptylist$.
We have that $\encpub{L} = \unitsign$.
Since \typeEval{\Gamma}{\sigma \equiv \sigmaPol{\policy}:\sign}, from {Lemma~\ref{lem:ml:wrapper:signeq:property}}, $\sigma = \unitsign$.
From the definition of \fstsign{}, we have that $\fstsign{\sigma} = \fstsign{\encpub{L}} = \unitkind$.

\emcase{Case 2:} $L =  x \concat L'$, where $x \not\in \dom{\decPolicy{\policy}}$.
We have that $\encpub{L} = \Sigma \alpha_x:\atksign{\basekind}.\Sigma \alpha:\atcsign{\alpha_x}.\encpub{L'}$.
Since \typeEval{}{\encpub{L} \equiv \sigma:\sign}, from the \SeqDPr\ rule, $\sigma$ is $\Sigma \alpha_x:\sigma_1.\Sigma \alpha:\sigma_2.\sigma_3$ s.t.
\begin{itemize}
\item \typeEval{\Gamma}{\atksign{\basekind} \equiv \sigma_1:\sign}
\item \typeEval{\Gamma,\alpha_x:\basekind}{\Sigma \alpha:\atcsign{\alpha_x}.\encpub{L'} \equiv \Sigma \alpha:\sigma_2.\sigma_3:\sign}.
    \begin{itemize}
    \item \typeEval{\Gamma,\alpha_x:\basekind}{\atcsign{\alpha_x} \equiv \sigma_2:\sign}
    \item \typeEval{\Gamma, \alpha_x:\basekind, \alpha:\unitkind}{\encpub{L'} \equiv \sigma_3:\sign} (notice that since $\alpha_x:\basekind$, it follows that $\fstsign{\atcsign{\tau}} = \unitkind$).
    \end{itemize}
\end{itemize}

We have that:
\begin{itemize}
\item \typeEval{\Gamma}{\atksign{\basekind} \equiv \sigma_1:\sign}.
From {Lemma~\ref{lem:ml:wrapper:signeq:property}}, $\sigma_1 = \atcsign{\basekind}$.
Thus, $\fstsign{\sigma_1} = \basekind = \fstsign{\atcsign{\basekind}}$.

\item \typeEval{\Gamma,\alpha_x:\basekind}{\atcsign{\alpha_x} \equiv \sigma_2:\sign}.
From {Lemma~\ref{lem:ml:sig-eq:atc}}, $\sigma_2 = \atcsign{\tau}$ for some $\tau$ s.t. \typeEval{\Gamma,\alpha_x:\basekind}{\alpha_x \equiv \tau:\basekind}.
Thus, $\fstsign{\sigma_2} = \unitkind = \fstsign{\atcsign{\alpha_x}}$

\item \typeEval{\Gamma, \alpha_x:\basekind, \alpha:\unitkind}{\encpub{L'} \equiv \sigma_3:\sign}.
From IH, we have that $\fstsign{\encpub{L'}} = \fstsign{\sigma_3}$.
\end{itemize}

Therefore, from the definition of \fstsign{}, we have that $\fstsign{\encpub{L}} = \fstsign{\sigma}$

\emcase{Case 3:} $L = x \concat L'$, where $\decPolicy{\policy}(x) = f$.
The proof is similar to the proof of Case 2.
%\emcase{Case 4:} $L = x \concat L'$, where $\decPolicy{\policy}(x) = f \circ a$.
%The proof is similar to the proof of Case 2.
\end{proof}

\begin{lemma}
\label{lem:ml_trnl:sigma_equiv}
\begin{itemize}
\item {For any $\Gamma$, if \typeEval{\Gamma}{\sigma \equiv \atcsign{\tau}:\sign}, then $\sigma = \atcsign{\tau'}$ for some $\tau'$ s.t. \typeEval{\Gamma}{\tau \equiv \tau':\basekind}.}
\item {For any $\Gamma$, if \typeEval{\Gamma}{\atcsign{\tau} \equiv \sigma :\sign}, then $\sigma = \atcsign{\tau'}$ for some $\tau'$ s.t. \typeEval{\Gamma}{\tau \equiv \tau':\basekind}.}
\end{itemize}
\end{lemma}
\begin{proof}
We prove the lemma by induction on the derivation of the judgments.
We only have the following cases, since other rules are not applicable.

\emcase{Case 1:} eqs\_refl. The proof is trivial.

\emcase{Case 2:} eqs\_symm. The proofs of two parts are directly from IH.

\emcase{Case 3:} eqs\_trans. The proofs follows from IH and the fact that $\equiv$ of types is transitivity.

\emcase{Case 4:} eqs\_dyn. The proofs follows from the rule.

\end{proof}

\begin{lemma}
\label{lem:ml_trnl:sigma_leq}

\begin{itemize}
\item {For any $\Gamma$, if \typeEval{\Gamma}{\sigma \leq \atcsign{\tau}:\sign}, then $\sigma = \atcsign{\tau'}$ for some $\tau'$ s.t. \typeEval{\Gamma}{\tau \equiv \tau':\basekind}.}
\item {For any $\Gamma$, if \typeEval{\Gamma}{\atcsign{\tau} \leq \sigma :\sign}, then $\sigma = \atcsign{\tau'}$ for some $\tau'$ s.t. \typeEval{\Gamma}{\tau \equiv \tau':\basekind}.}
\end{itemize}
\end{lemma}
\begin{proof}
We prove both parts of the lemma by induction on the derivation of \typeEval{\Gamma}{\sigma \leq \atcsign{\tau}:\sign} and \typeEval{\Gamma}{\atcsign{\tau} \leq \sigma:\sign}.
We only have the following cases (other rules are not applicable).

\emcase{Case 1:} subs\_refl. The proof follows from Lemma~\ref{lem:ml_trnl:sigma_equiv}.

\emcase{Case 2:} subs\_trans. 
\begin{itemize}
\item Part 1: from the rule, \typeEval{\Gamma}{\sigma \leq \sigma'':\sign} and \typeEval{\Gamma}{\sigma'' \leq \atcsign{\tau}:\sign}.
From IH for \typeEval{\Gamma}{\sigma'' \leq \atcsign{\tau}:\sign}, $\sigma'' = \atcsign{\tau''}$ for some $\tau''$ s.t. \typeEval{\Gamma}{\tau'' \equiv \tau':\basekind}.
Now, we can apply IH on \typeEval{\Gamma}{\sigma \leq \atcsign{\tau''}:\sign} and we have that $\sigma = \atcsign{\tau'}$ for some $\tau'$ s.t. \typeEval{\Gamma}{\tau' \equiv \tau'':\basekind}.
Since $\equiv$ for type is transitivity, we have that \typeEval{\Gamma}{\tau \equiv \tau':\basekind}.

\item Part 2: the proof is similar.
\end{itemize}
\end{proof}

\begin{lemma}
\label{lem:ml_trni:inversion_wrapper}
{If \typeEvalP{}{\wrappub{e}: \Pign \alphaPol:\sigma^*.\atcsign{\tau}}, where \typeEval{}{\sigma^*\equiv \sigmaPol{\policy}:\sign}, then \typeEvalI{\alphaPol/\mPol:\sigma^{**}}{\attmod{e}:\atcsign{\tau'}} for some $\sigma^{**}$ and $\tau'$s.t. \typeEval{\alphaPol/\mPol:\sigma^{**}}{\tau \equiv \tau'} and \typeEval{}{\sigma^{**} \equiv \sigmaPol{\policy}}.}
\end{lemma}
\begin{proof}
We prove the lemma by induction on the derivation of \typeEvalP{}{\wrappub{e}: \Pign \alphaPol:\sigmaPol{\policy}.\atcsign{\tau}}.

\emcase{Case 1:} ofm\_lamgn.
From the rule, \typeEvalI{\alphaPol/\mPol:\sigma^*}{\attmod{e}:\atcsign{\tau}}.

\emcase{Case 2:} ofm\_subsume.
From the rule, we have \typeEvalP{}{\wrappub{e}: \Pign \alphaPol:\sigma.\sigma'} s.t. \typeEval{}{ \Pign \alphaPol:\sigma.\sigma' \leq \Pign \alphaPol:\sigma^*.\atksign{\tau}:\sign}.
From subs\_pigen, we have that \typeEval{}{\sigma^* \leq \sigma:\sign} and \typeEval{\alphaPol:\fstsign{\sigma^*}}{\sigma' \leq \atcsign{\tau}:\sign}.
% and \typeEval{\alphaPol:\fstsign{\sigma}}{\sigma:\sign}.

\begin{itemize}
\item Since \typeEval{}{\sigma^* \leq \sigma:\sign} and \typeEval{}{\sigma^* \equiv \sigmaPol{\policy}:\sign}, we have that \typeEval{}{\sigmaPol{\policy} \leq \sigma:\sign}.
From Lemma~\ref{lem:ml:wrapper:signature:policy:property}, we have that \typeEval{}{\sigmaPol{\policy} \equiv \sigma:\sign}.

\item Since \typeEval{\alphaPol:\fstsign{\sigma^*}}{\sigma' \leq \atcsign{\tau}:\sign}, from Lemma~\ref{lem:ml_trnl:sigma_leq}, we have that $\sigma' = \atcsign{\tau'}$ for some $\tau'$.
\end{itemize}

Thus, we have that \typeEvalP{}{\wrappub{e}: \Pign \alphaPol:\sigma.\atcsign{\tau'}} s.t. \typeEval{}{\sigma \equiv \sigmaPol{\policy}:\sign}.
From IH, we close this case.
\end{proof}

%\begin{corollary}
%\label{cor:ml_trni:inversion}
%\redtext{If \typeEvalP{}{\wrappub{e}: \Pign \alphaPol:\sigmaPol{\policy}.\atcsign{\tau}}, then \typeEvalI{\alphaPol/\mPol:\sigmaPol{\policy}}{\attmod{e}:\atcsign{\tau}}.}
%\end{corollary}
%\begin{proof}
%From Lemma~\ref{lem:ml_trni:inversion_wrapper:strong}, we have that 
%\end{proof}

\begin{lemma}
\label{lem:ml_trni:inversion_impure}
{If \typeEvalI{\alphaPol/\mPol:\sigma}{\attmod{e}:\atcsign{\tau}} for some $\sigma$ s.t. \typeEval{}{\sigma\equiv\sigmaPol{\policy}:\sign}, then \typeEvalP{\alphaPol/\mPol:\sigma}{\attmod{e}:\atcsign{\tau'}} for some $\tau'$ s.t. \typeEval{\alphaPol/\mPol:\sigma}{\tau \equiv \tau':\basekind}.}
\end{lemma}
\begin{proof}
We prove the lemma by induction on the derivation of \typeEvalI{\alphaPol/\mPol:\sigma}{\attmod{e}:\atcsign{\tau}}.
We only have the following case:

\emcase{Case 1:} ofm\_forget. The proof is directly from the rule.

\emcase{Case 2:} ofm\_subsume. 
From the rule, we have that \typeEvalI{\alphaPol/\mPol:\sigma}{\attmod{e}:\sigma'} and \typeEval{\alphaPol/\mPol:\sigma}{\sigma' \leq \atcsign{\tau}:\sign}.
From Lemma~\ref{lem:ml_trnl:sigma_leq}, we have that $\sigma' = \atcsign{\tau''}$ s.t. \typeEval{\alphaPol/\mPol:\sigma}{\tau \equiv \tau'':\basekind}.
Thus, we can apply IH on \typeEvalI{\alphaPol/\mPol:\sigma}{\attmod{e}:\sigma'} and close this case.
\end{proof}

\begin{lemma}
\label{lem:ml_trni:inversion_pure}
{If \typeEvalP{\alphaPol/\mPol:\sigma}{\attmod{e}:\atcsign{\tau}} for some $\sigma$ s.t. \typeEval{}{\sigma\equiv\sigmaPol{\policy}:\sign}, then \typeEval{\alphaPol/\mPol:\sigma}{{e}:{\tau'}} for some $\tau'$ s.t. \typeEval{\alphaPol/\mPol:\sigma}{\tau \equiv \tau':\basekind}.}
\end{lemma}
\begin{proof}
We prove the lemma by induction on the derivation of \typeEvalP{\alphaPol/\mPol:\sigma}{\attmod{e}:\atcsign{\tau}}.
We have the following cases.

\emcase{Case 1:} ofm\_dyn. The proof is directly from the rule.

\emcase{Case 2:} ofm\_subsume. The proof is similar to the proof of Case 2 in Lemma~\ref{lem:ml_trni:inversion_impure}.
\end{proof}

\begin{lemma}
\label{lem:ml_trni:inversion}
{If \typeEval{}{\wrappub{e}: \Pign \alphaPol:\sigmaPol{\policy}.\atcsign{\tau}}, then \typeEval{\alphaPol/\mPol:\sigma}{e:\tau'} for some $\sigma$ and $\tau'$ s.t. \typeEval{}{\sigma\equiv \sigmaPol{\policy}:\sign} and \typeEval{\alphaPol/\mPol:\sigma}{\tau \equiv \tau':\basekind}.}
\end{lemma}
\begin{proof}
Since \typeEval{}{\wrappub{e}: \Pign \alphaPol:\sigmaPol{\policy}.\atcsign{\tau}}, from Lemma~\ref{lem:ml_trni:inversion_wrapper}, we have that \typeEvalI{\alphaPol/\mPol:\sigma''}{\attmod{e}:\atcsign{\tau''}} for some $\sigma''$ and $\tau''$ s.t. \typeEval{}{\sigma''\equiv \sigmaPol{\policy}:\sign} and \typeEval{\alphaPol/\mPol:\sigma''}{\tau\equiv \tau'':\basekind}. 
From Lemma~\ref{lem:ml_trni:inversion_impure} and Lemma~\ref{lem:ml_trni:inversion_pure}, we have that \typeEval{\alphaPol/\mPol:\sigma'}{{e}:{\tau'}} for some $\sigma'$ and $\tau'$ s.t. \typeEval{}{\sigma'\equiv \sigmaPol{\policy}:\sign} and \typeEval{\alphaPol/\mPol:\sigma'}{\tau\equiv \tau':\basekind}. 
\end{proof}

\textbf{Theorem~\ref{thm:ml:wrapper} (in \S\ref{sec:ml:trni})}
{If \typeEvalP{}{\wrappub{e}:\Pign \alphaPol:\sigmaPol{\policy}.\atcsign{\tau}}, then $e$ is \TRNI{\policy,\tau}.}
\begin{proof}
%Since \typeEval{}{\wrappub{e}:\tau}, from Lemma~\ref{lem:ml:wrapper:inference:public}, it follows that \typeEval{\alpha_\policy/m_\policy:\sigma}{e:\tau} for some $\sigma$ s.t. \typeEval{}{\sigma \equiv \sigmaPol{\policy}:\sign}.
%From Lemma~\ref{lem:ml:wrapper:kind}, we have that $\fstsign{\sigma} = \fstsign{\sigmaPol{\policy}}$.
%In addition, we have that \typeEval{\alpha_\policy/m_\policy:\sigma}{\tau:\basekind}.
%Since module variables are not used in the judgments \typeEval{\Gamma}{c:k}, we have that \typeEval{\alpha_\policy/m_\policy:\sigmaPol{\policy}}{\tau:\basekind}.
%
%We now consider an arbitrary \condenvalt{\rho}{\policy}.
%As proven in Lemma~\ref{lem:ml:consenv:property:two}, $\rho \in \interenvfull{\pubViewTerm{\policy}}$.
%Since \typeEval{}{\sigma \equiv \sigmaPol{\policy}:\sign}, we have that $\rho \in \interenvfull{\alpha_\policy/m_\policy:\sigma}$.
%Since \typeEval{\alpha_\policy/m_\policy:\sigma}{e:\tau}, from Theorem~\ref{thm:ml:abstraction}, it follows that $\tuple{\rho_L(e),\rho_R(e)} \in \rev{\inter{\tau}{\rho}}$.
%From the definition of indistinguishability, $\tuple{\rho_L(e),\rho_R(e)} \in \indterm{\tau}$.
%
%Thus, we have that:
%\begin{itemize}
%%\item \typeEval{\conView{\policy}}{e} (Lemma~\ref{lem:ml:wrapper:inference:confidential}),
%\item \typeEval{\alpha_\policy/m_\policy:\sigmaPol{\policy}}{\tau:\basekind} (proven above),
%\item for any \condenvalt{\rho}{\policy}, $\tuple{\rho_L(e),\rho_R(e)} \in \indterm{\tau}$ (proven above).
%\end{itemize}
%
%Therefore, $e$ is \TRNI{\policy,\tau}.
From Lemma~\ref{lem:ml_trni:inversion}, we have that \typeEval{\alphaPol/\mPol:\sigma}{e:\tau'} for some $\sigma$ and $\tau'$ s.t. \typeEval{}{\sigma\equiv \sigmaPol{\policy}:\sign} and \typeEval{\alphaPol/\mPol:\sigma}{\tau \equiv \tau':\basekind}.
From Lemma~\ref{lem:ml:wrapper:kind}, we have that $\fstsign{\sigma} = \fstsign{\sigmaPol{\policy}}$.
In addition, since module variables are not used in the judgments \typeEval{\Gamma}{c:k}, we have that \typeEval{\alpha_\policy:\fstsign{\sigmaPol{\policy}}}{\tau \equiv \tau':\basekind}.
Thus, we have that \typeEval{\alphaPol/\mPol:\sigmaPol{\policy}}{e:\tau}.
From Theorem~\ref{thm:ml:free-thm:opaque}, $e$ is \TRNI{\policy,\tau}.
%
%We consider now an arbitrary $\rho$ s.t. $\condenvalt{\rho}{\policy}$.
%As proven in Lemma~\ref{lem:ml:consenv:property:two}, $\rho \in \interenvfull{\pubViewTerm{\policy}}$.
%Since \typeEval{}{\sigma \equiv \sigmaPol{\policy}:\sign}, we have that $\rho \in \interenvfull{\alpha_\policy/m_\policy:\sigma}$.
%%Since \typeEval{\alpha_\policy/m_\policy:\sigma}{e:\tau'}, from Theorem~\ref{thm:ml:abstraction}, it follows that $\tuple{\rho_L(e),\rho_R(e)} \in \rev{\inter{\tau'}{\rho}} = \rev{\inter{\tau}{\rho}}$ (notice that \typeEval{\alphaPol/\mPol:\sigma}{\tau \equiv \tau':\basekind}).
%Since \typeEval{\alpha_\policy/m_\policy:\sigma}{e:\tau'} and \typeEval{\alphaPol/\mPol:\sigma}{\tau \equiv \tau':\basekind}), it follows that \typeEval{\alpha_\policy/m_\policy:\sigma}{e:\tau}.
%Since \typeEval{\alpha_\policy/m_\policy:\sigma}{e:\tau}, from Theorem~\ref{thm:ml:abstraction}, it follows that $\tuple{\rho_L(e),\rho_R(e)} \in \rev{\inter{\tau}{\rho}}$.
%From the definition of indistinguishability, $\tuple{\rho_L(e),\rho_R(e)} \in \indterm{\tau}$.
%
%In addition, since \typeEval{\alpha_\policy/m_\policy:\sigma}{e:\tau}, by using a proof similar to the one of Lemma~\ref{lem:ml_trni:typability_impl}, we have that \typeEval{\conView{\policy}}{e}.
%
%Thus, we have that:
%\begin{itemize}
%\item \typeEval{\conView{\policy}}{e},
%\item \typeEval{\alpha_\policy/m_\policy:\sigmaPol{\policy}}{\tau:\basekind},
%\item for any \condenvalt{\rho}{\policy}, $\tuple{\rho_L(e),\rho_R(e)} \in \indterm{\tau}$.
%\end{itemize}
%Therefore, $e$ is \TRNI{\policy,\tau}.
\end{proof}

%\textbf{Theorem~\ref{thm:ml:wrapper}.}
%\tealtext{Suppose that $V \in \policy$.
%If \typeEvalP{}{\wrapcon{e}} and \typeEval{}{\wrappub{e,V}:\tau} for some $\tau$, then $e$ is \TRNI{\policy,\tau}.}
%\begin{proof}
%Since \typeEval{}{\wrappub{e,V}:\tau}, from Lemma~\ref{lem:ml:wrapper:conf-view}, it follows that \typeEval{\alpha_\policy/m_\policy:\sigma}{e:\tau} for some $\sigma$ s.t. \typeEval{}{\sigma \equiv \sigmaPol{\policy}:\sign}.
%From Lemma~\ref{lem:ml:wrapper:kind}, we have that $\fstsign{\sigma} = \fstsign{\sigmaPol{\policy}}$.
%In addition, we have that \typeEval{\alpha_\policy/m_\policy:\sigma}{\tau:\basekind}.
%Therefore, it follows that that \typeEval{\alpha_\policy/m_\policy:\sigmaPol{\policy}}{\tau:\basekind}.
%
%We now consider an arbitrary \condenvalt{\rho}{\policy}.
%As proven in Lemma~\ref{lem:ml:consenv:property:two}, $\rho \in \interenvfull{\pubViewTerm{\policy}}$.
%Since \typeEval{}{\sigma \equiv \sigmaPol{\policy}:\sign}, we have that $\rho \in \interenvfull{\alpha_\policy/m_\policy:\sigma}$.
%Since \typeEval{\alpha_\policy/m_\policy:\sigma}{e:\tau}, from Theorem~\ref{thm:ml:abstraction}, $\tuple{\rho_L(e),\rho_R(e)} \in \inter{\tau}{\rho}$.
%Therefore, $\tuple{\rho_L(e),\rho_R(e)} \in \indval{\tau}$.
%
%From Lemma~\ref{lem:ml:wrapper:conf-view}, we also have that \typeEval{\conView{\policy}}{e}.
%Thus, $e$ is \TRNI{\policy,\tau}.
%\end{proof}

\section{Computation at different levels}
\label{sec:multi_level:encoding}
\renewcommand{\typeEval}[2][]{\ensuremath{#1 \vdash #2}}
\renewcommand{\indval}[2][\obs]{\ensuremath{\textit{I}_{\textit{NI}}^{#1}[\![#2]\!]}}
\renewcommand{\indterm}[2][\obs]{\ensuremath{\textit{I}_{\textit{NI}}^{#1}[\![#2]\!]^{\textit{ev}}}}

\newtext{To facilitate the presentation, in this section, we first present the encoding that support computation at different levels only for noninterference.
Then we extend the encoding to support declassification policies.}

\subsection{Language and abstraction theorem}
\label{sec:multi_level:language}

\paragraph{Language.}
To have an encoding that support multiple levels, we add the universally quantified types $\forall \alpha.\tau$ to the language presented in \S\ref{sec:abstraction}.
In addition, to simplify the encoding, we add the unit type \unitType.
W.r.t. these new types, we have new values: the unit value \unitVal\ of \unitType, and values $\Lambda \alpha.e$ of type $\forall \alpha.\tau$.
Therefore, in addition to the typing rules in \S\ref{sec:abstraction}, we have the following rules.

%\begin{align*}
%v := &\ \unitVal \sep n \sep \tuple{v,v} \sep \lambda x:\tau.e \sep \sep {\Lambda \alpha.e} & \text{Values}\\
%e := &\ x \sep v \sep \tuple{e,e}\sep \prj{i}{e} \sep e\ e \sep e[\tau] & \text{Terms}\\
%\tau := &\ \alpha \sep \intType \sep \unitType \sep \tau \times \tau \sep \tau \rightarrow \tau \sep \forall \alpha.\tau & \text{Types}\\
%\Delta := &\ . \sep \Delta,\alpha & \text{Typing Contexts}\\
%\Gamma := &\ . \sep \Gamma, x:\tau & \text{Term Contexts}
%\end{align*}
%{We write \typeEval[\Delta]{\tau} to mean that $\tau$ is well-formed in $\Delta$ -- meaning that all type variables in $\tau$ are in $\Delta$}.
%The typing rules \typeEval[\Delta,\Gamma]{e:\tau} are standard. 
%The rules are to be instantiated only with $\Gamma$ that is well-formed under $\Delta$, in the sense that \typeEval[\Delta]{\Gamma(x)} for all $x \in \dom{\Gamma}$. 
%When $\Delta$ and $\Gamma$ are empty, we write \typeEval{e:\tau} instead of \typeEval[\Delta,\Gamma]{e:\tau}.
%We write $\tau_1[\tau_2/\alpha]$ \redtext{to mean the avoid capture substitution of $\alpha$ in $\tau_1$ with $\tau_2$}.

\begin{mathpar}
%\InferRule{~}{\typeEval[\Delta,\Gamma]{n:\intType}}\and
\InferRule{~}{\typeEval[\Delta,\Gamma]{\unitVal:\unitType}}\and
%
%\InferRule{\Gamma(x)=\tau}{\typeEval[\Delta,\Gamma]{x:\tau}}\and
%%
%\InferRule{\typeEval[\Delta,\Gamma]{e_1:\tau_1} \\ \typeEval[\Delta,\Gamma]{e_2:\tau_2 }}
%{\typeEval[\Delta,\Gamma]{\tuple{e_1,e_2}:\tau_1\times \tau_2}} \and
%%%
%%\InferRule{\typeEval[\Delta,\Gamma]{e:\tau \times \tau'} }
%%{\typeEval[\Delta,\Gamma]{\prj{1}{e}:\tau}} \and
%%%
%\InferRule{\typeEval[\Delta,\Gamma]{e:\tau_1 \times \tau_2} }
%{\typeEval[\Delta,\Gamma]{\prj{i}{e}:\tau_i}} \and
%
%\InferRule{\typeEval[\Delta,\Gamma,x:\tau_1]{e:\tau_2}}
%{\typeEval[\Delta,\Gamma]{\lambda x:\tau_1.e:\tau_1 \rightarrow \tau_2}} \and
%%
%\InferRule{\typeEval[\Delta,\Gamma]{e_1:\tau_1 \rightarrow\tau_2} \\ \typeEval[\Delta,\Gamma]{e_2:\tau_1}}
%{\typeEval[\Delta,\Gamma]{e_1\ e_2: \tau_2}} \and
%%
\InferRule{\typeEval[\Delta,\alpha,\Gamma]{e:\tau} }
{\typeEval[\Delta,\Gamma]{\Lambda \alpha.e: \forall \alpha.\tau}}\and
\InferRule{\typeEval[\Delta,\Gamma]{\tau'}\\ \typeEval[\Delta,\Gamma]{\Lambda \alpha.e:\forall \alpha.\tau }}
{\typeEval[\Delta,\Gamma]{e[\tau']: \tau[\tau'/\alpha]}}
\end{mathpar}

%The notions of free (term) variables and free type variables, closed terms, closed values, open terms and open values are standard.
Similar to the language in \S\ref{sec:abstraction}, the semantics is a call-by-value semantic.
We write $\lambda \_:\tau.e$ instead of $\lambda x:\tau.e$ when $x$ is not a free variable in $e$.

%\noteinline{The sum type $\tau + \tau$ will be added later.}

\paragraph{Abstraction theorem.}
\label{sec:multi_level:background:parametricity}
%This paragraph is already in Section 4.2
%Given two closed type $\tau_1$ and $\tau_2$, we write $\Rel{\tau_1,\tau_2}$ as the set of all relations defined on closed values of $\tau_1$ and $\tau_2$. That is for any $R \in \Rel{\tau_1,\tau_2}$, if $\tuple{v_1,v_2} \in R$ then \typeEval[]{v_1:\tau_1} and \typeEval{v_2:\tau_2}.

As in \S\ref{sec:ml:language}, an {\em environment} $\rho$ is a mapping from type variables to tuples of the form $\tuplemt{\tau_1,\tau_2,R}$ where $R \in \Rel{\tau_1,\tau_2}$, and from term variables to tuples of the form $\tuplemt{v,v'}$.
%We also write $\rho_L$ and $\rho_R$ for the substitutions that map every (type or term) variable to the first and second constituent of the tuple that $\rho$ maps that variable to.
For any $\rho$ s.t. $\alpha\not\in \dom{\rho}$, we write $\rho, \alpha \mapsto \tuplemt{\tau_1,\tau_2,R}$ to mean the environment $\rho'$ s.t. $\dom{\rho'} = \dom{\rho} \cup \{\alpha\}$ and $\rho'(\beta) = \tuplemt{\tau_1,\tau_2,R}$ if $\beta = \alpha$, and $\rho'(\beta) = \rho(\beta)$ otherwise.

For the logical relation, in addition to the rules in Fig.~\ref{fig:logical-relation}, we have other rules described in Fig.~\ref{fig:multi-level:logical-relation}.
%The logical relation presented in Fig.~\ref{fig:multi-level:logical-relation} is a type-indexed family of relations on values and terms based on given relations for type variables.
%For any $\tau$, $\valrelation{\tau}{\rho}$  is a relation on closed values and \termrelation{\tau}{\rho} is a relation on closed terms.
As in \S\ref{sec:abstraction}, we have Lemma~\ref{lem:multi_level:logeq:related-term} about the types of related values and terms are described.

\begin{figure}
\vspace{-10pt}
\begin{mathpar}
\small
%\InferRule[FR-Int]
%{~}
%{\tuplemt{n,n} \in \valrelation{\intType}{\rho}}\and
%%%%%%%%%%%%%%%%%%%
\InferRule[FR-Unit]
{~}
{\tuplemt{\unitVal,\unitVal} \in \valrelation{\unitType}{\rho}}\and
%%%%%%%%%%%%%%%%%%%
%\InferRule[FR-Pair]{
%\tuplemt{v_1,v_1'} \in \valrelation{\tau_1}{\rho} \\
%\tuplemt{v_2,v_2'} \in \valrelation{\tau_2}{\rho}}
%{\tuplemt{\tuple{v_1,v_2},\tuple{v_1',v_2'}}\in \valrelation{\tau_1 \times \tau_2}{\rho}}\and
%%%%%%%%%%%%%%%%%%%%
%\InferRule[FR-Fun]
%{\forall \tuplemt{v_1',v_2'} \in \valrelation{\tau_1}{\rho}.  \tuplemt{v_1\ v_1', v_2\ v_2'} \in \termrelation{\tau_2}{\rho} }
%{\tuplemt{v_1\ v_1', v_2\ v_2'}\in \valrelation{\tau_1 \rightarrow \tau_2}{\rho}}\and
%%%%%%%%%%%%%%%%%%%
%\InferRule[FR-Var]
%{\rho(\alpha) = \tuplemt{\tau_1,\tau_2,R}\\ \tuplemt{v_1,v_2} \in R}
%{\tuplemt{v_1,v_2} \in \valrelation{\alpha}{\rho[\tuplemt{\tau_1,\tau_2,R}/\alpha]}}\and
%\InferRule[FR-Var]
%{\rho(\alpha) = \tuplemt{\tau_1,\tau_2,R}\\ \tuplemt{v_1,v_2} \in R}
%{\tuplemt{v_1,v_2} \in \valrelation{\alpha}{\rho}}\and
%%%%%%%%%%%%%%%%%%%
\InferRule[FR-Par]
{\forall R \in \Rel{\tau_1,\tau_2}. \tuplemt{v_1[\tau_1],v_2[\tau_2]} \in \termrelation{\tau}{\rho,\alpha\mapsto\tuplemt{\tau_1,\tau_2,R}} }
{\tuplemt{v_1,v_2} \in \valrelation{\forall\alpha.\tau}{\rho}}\and
%%%%%%%%%%%%%%%%%%%
%\InferRule[FR-Term]
%{{\typeEval{e_1:\rho_L(\tau)}} \\
%{\typeEval{e_2:\rho_R(\tau)}} \\
%e_1 \reduce v_1\\ e_2 \reduce v_2\\ \tuplemt{v_1,v_2} \in \valrelation{\tau}{\rho}}
%{\tuplemt{e_1,e_2} \in \termrelation{\tau}{\rho}}
\end{mathpar}
\vspace{-4ex}
\caption{{Logical relation}}
\label{fig:multi-level:logical-relation}
\end{figure}

%The next lemma is about the types of related values and terms.
\begin{lemma}
\label{lem:multi_level:logeq:related-term}
We have that:
\begin{itemize}
\item if $\tuple{v_1,v_2} \in \valrelation{\tau}{\rho}$, then \typeEval{v_1: \rho_L(\tau)}, \typeEval{v_2: \rho_R(\tau)} and
\item if $\tuple{e_1,e_2} \in \termrelation{\tau}{\rho}$, then \typeEval{e_1: \rho_L(\tau)} and \typeEval{e_2: \rho_R(\tau)}.
\end{itemize}
\end{lemma}

%\redtext{We write $\rho \respect \Delta$ to mean that the set of type variables in \dom{\rho} is equal to $\Delta$ (i.e. $\Delta = \{\alpha \sep \alpha \in \dom{\rho}\}$)}.
We write $\rho \respect \Delta$ to mean that the domain of $\rho$  is $\Delta$ (i.e. $\dom{\rho} = \Delta$).
Let \restrict{\rho} be the restriction of $\rho$ to type variables.
We write $\rho \respect \Delta,\Gamma$ to mean that {$\restrict{\rho} \respect \Delta$ and
\begin{itemize}
\item the set of term variables in \dom{\rho} is equal to $\dom{\Gamma}$ (i.e. $\dom{\Gamma} = \{x \sep x \in \dom{\rho}\}$,
\item for all $x\in \dom{\rho}$, $\rho(x) \in \valrelation{\Gamma(x)}{\rho}$.
\end{itemize}

We define \typeEval[\Delta,\Gamma]{e_1 \logrel e_2:\tau} as in \S\ref{sec:abstraction} and \S\ref{sec:ml:language}.
%\begin{definition}
%\label{def:logrel}
%Two terms $e_1$ and $e_2$ are logically equivalent at $\tau$ in contexts $\Delta$ and $\Gamma$ (denoted by \typeEval[\Delta,\Gamma]{e_1 \logrel e_2:\tau}) when \typeEval[\Delta,\Gamma]{e_1,e_2:\tau}, and for all $\rho \respect \Delta,\Gamma$: 
%$$\tuplemt{\rho_L(e_1),\rho_R(e_2)} \in \interev{\tau}{\rho}.$$
%\end{definition}
The following theorem says that the logical equivalence is reflexive. 

\begin{theorem}[Abstraction Theorem]
\label{thm:multi_level:parametricity}
\tealtext{If \typeEval[\Delta,\Gamma]{e:\tau}, then \typeEval[\Delta,\Gamma]{e \logrel e:\tau}.}
\end{theorem}

\subsection{Noninterference for free}
\label{sec:multi_level:ni_free}
\subsubsection{Overview}
\label{sec:ni_free:overview}
We consider noninterference defined for an arbitrary finite security lattice \tuplemt{\lattice,\smalleq}, where $\lattice$ is the set of security levels, and $\smalleq$ is the order between them.
We use \lvln\ as a function from variables to security levels.

We consider programs that read input values from input channels and generate output values to output channels. 
Before presenting assumptions on inputs, outputs and programs, we introduce some auxiliary notations.
Let $\{\wraptype{}\}_{l \in \lattice}$ be a family of type constructors indexed by $l$.
A {\em wrapped value} of type $\tau$ at level $l$ is a value of type $\wraptype{\tau$}.
A wrapped value $v$ can be unwrapped by $\unwrap\ k\ v$, where $k$ is an appropriate key and $\unwrap\ k\ v$ can be implemented as $v\ k$.
%We can think of a wrapped value $\tau$ as a value that can only be accessed with an appropriate key.
Below are assumptions on input/output channels and programs.
\begin{assumption}
\label{assmt:input_output}
\begin{enumerate}
\item Each input/output channel is associated with a security level.
%\redtext{To read a value in an input/output channel, an observer needs to provide an appropriate key}.

\item %Each security level can be associated with $0$, $1$, or more than $1$ input channels.
The type of an input value from an input channel is \intType.
The type of an output value to an output channel is \intType.\footnote{%
In this section, the type of inputs/outputs is \intType. We choose this since we are sticking with \cite{Li-Zdancewic-05-POPL}. The result presented in this paper can be generalized to accepting confidential inputs of arbitrary types (e.g. \textbf{Bool}, \textbf{String}, etc).}

%\item There is a secured output channel $\code{IO.Channel}_{l}$ for each security level $l$.
%The type of an output value is \intType.

\item Free variables in programs are considered as inputs from input channels.

\item Input values on an input channel at level $l$ cannot be directly observed. We model this by using \wraptype{\intType} as their type, not \intType.

\item Programs are executed in a context where there are several output channels, each corresponding to a security level. 
A program will computes a tuple of wrapped output values, where each element of the tuple is unwrapped by using \unwrap\ and an appropriate key, and the unwrapped value is sent to a channel.
The assumption about execution of programs in the context is illustrated in the following pseudo program, where $e$ is a program satisfying assumptions, $o$ is the computed tuple, $\code{Output.Channel}_{l}$ is an output channel at $l$, $k_{l}$ is a key to unwrap value at $l$, \prjcode{pr}{l} projects the output value for the output channel $l$.

\begin{lstlisting}[escapechar=\%]
let o = e in 
   %$\code{IO.Channel}_{l_1}$% := %\unwrap% %$k_{l_1}$% (%\prj{l_1}{o}%)
   %$\vdots$%
   %$\code{IO.Channel}_{l_n}$% := %\unwrap% %$k_{l_n}$% (%\prj{l_n}{o}%)
\end{lstlisting}

Therefore, outputs of considered programs are of the type {$\wraptype[l_1]{\intType} \times  \dots\times \wraptype[l_n]{\intType}$}.
\end{enumerate}
\end{assumption}

\subsubsection{Encoding}
\label{sec:ni_free:encoding}
For every level $l$ in the lattice, we introduce a type variable $\alpha_l$.
To protect data of type $\tau$ at level $l \in \lattice$, 
%instead of using monads as in DCC \cite{Abadi-etal-99-POPL}, 
we use a type variable $\alpha_l$ and form the type $\alpha_l \rightarrow \tau$.
In other words, \wraptype{\intType} is $\alpha_l \rightarrow \tau$.
Therefore, an input of type \intType\ at level $l$ is encoded as a value of type $\alpha_l \rightarrow \intType$.

To do computation on values of types $\alpha_l \rightarrow \tau$, programmers have to manage {key} manually and {it is error-prone}.
To support handling values of $\alpha_l \rightarrow \tau$, we propose the following {interfaces} that satisfy monad laws \cite{Wadler-89-FPCA,Petricek:18:PL} (see proof in \S\ref{sec:multi_level:ni_free:monad}).
The detailed implementation of them will be described later.
\begin{itemize}
\item \comp{l}: this interface transfers a wrapped value at $l$ to a continuation and produce an output at $l$,
\item \wrap{l}: this interface is used to wraps terms. 
\end{itemize}

In addition, we have the interface \convup{l}{l'} (for $l \smallst l'$) which is used to convert data from a level $l$ to a higher level $l'$.
Therefore, we have the typing context \tcontextni\ and the term context \contextni\ for NI as described in Fig.~\ref{fig:encoding:ni}.
Examples illustrating \comp{l} and \conv{ll'} are in Example~\ref{ex:multi_level:ni:running_example:comp} and Example~\ref{ex:ni:running_example:conv}.
\footnote{{In these two examples, for illustration purpose, programs only generate single inputs at single levels. 
A program generating multiple outputs for multiple levels is illustrated later in Example~\ref{ex:ni:running_example:multile_levels}.}}

\begin{figure}
\begin{align*}
\tcontextni =&\ \{\alpha_l \sep l \in \lattice\}\\
\contextni =&\ \{x: \alpha_l \rightarrow \intType \sep \text{$x$ is an input} \wedge \lvl{x} = l\}\ \cup \\
   &\ \{\comp{l}: \forall \beta_1,\beta_2. (\alpha_l \rightarrow \beta_1) \rightarrow \big( \beta_1 \rightarrow (\alpha_l \rightarrow \beta_2)  \big) \rightarrow \alpha_l \rightarrow\beta_2 \sep l \in \lattice \}\ \cup \\
   &\ \{\convup{l}{l'}: \forall \beta. (\alpha_l \rightarrow \beta) \rightarrow (\alpha_{l'} \rightarrow \beta) \sep l,l' \in \lattice \wedge l \smallst l'\}\ \cup\\
   &\ \{\wrap{l}: \forall \beta.\beta \rightarrow \alpha_l \rightarrow \beta \sep l \in \lattice\}
\end{align*}
\caption{Contexts for noninterference}
\label{fig:encoding:ni}
\end{figure}

\begin{example}[Computation at a level]
\label{ex:multi_level:ni:running_example:comp}
We consider the lattice \tuplemt{\latticethree,\smalleq} where $\latticethree = \{L,M,H\}$ and $L \smalleq M \smalleq H$.
Suppose that there are three inputs \hinp, \minp\ and \linp\ and their associated levels are respectively $H$, $M$, and $L$.
By replacing \lattice\ in Fig.~\ref{fig:encoding:ni} with \latticethree, we have the typing context \tcontextni\ and the term context \contextni\ for NI defined on \tuplemt{\latticethree,\smalleq}.
Note that the types of $\hinp$, $\minp$ and $\linp$ in \contextni\ are respectively $\alpha_H \rightarrow \intType$, $\alpha_M \rightarrow \intType$, and $\alpha_L \rightarrow \intType$.

We now illustrate \comp\ by having a program that has a computation at $M$: we want to have $\minp + 1$.\footnote{Input $\minp$ is not of the \intType\ type and hence, $\minp+1$ is not well-typed. 
We use it just to express the idea of the computation.
This convention is also used in other examples.}

% thwe want to apply the function $\plusOne = \lambda y:\intType.y+1$ on $\minp$ and the result is {an output} at $M$.
In order to use \comp{M}, we first wrap \plusOne\ as below. 
$$\wrapPlusOne = \lambda x:\intType.\wrap{M}(\plusOne\ x)$$ 
%$$\wrapPlusOne = \lambda x:\intType.\lambda \_:\alpha_H.f(x)$$ 

Then we have the following program $e_1$ {of the type $\alpha_M \rightarrow \intType$} encoding the requirement that the function \plusOne\ is applied on $y$ and the result is at $M$.
$$e_1 = \comp{M}[\intType][\intType]\ {\minp}\ {\wrapPlusOne}$$
\end{example}

\begin{example}[Convert to a higher level]
\label{ex:ni:running_example:conv}
%We now use the interface described in Example~\ref{ex:multi_level:ni:running_example:comp} to have a computation at $H$ that uses $x$ at $H$ and $y$ at $M$.
{In this example, we will do the computation $\hinp + \minp$ and the result is at $H$}. 
This example illustrates the usage of \conv{ll'} (for $l \smallst l'$).

Let \add\ be a function of the type $\intType \rightarrow\intType\rightarrow \intType$.
From \add, we construct $\wrapAddC$ of the type $\intType \rightarrow \alpha_H \rightarrow \intType$, 
where $c$ is a variable of the type \intType.
%$${\wrapAddC} = \lambda y:\intType.\lambda \_:\alpha_H.\add\ \ c\ \ y$$ 
$${\wrapAddC} = \lambda y:\intType.\wrap{H}(\add\ \ c\ \ y)$$ 

%\footnote{Inputs $\hinp$ and $\minp$ are not of the \intType\ type and hence, $\hinp+\minp+1$ and $\minp+1$ are not well-typed. 
%We use them just to express the ideas of the computations.}

The following program $e_2$ of the type $\alpha_H \rightarrow \intType$ computes the sum of $\hinp$ and $\minp$ and the computed sum is at $H$. 
Note that in order to transfer $\minp$ at $M$ to the continuation $\wrapAddC$ by using \comp{H}, we  need to convert $\minp$ from level $M$ to level $H$ by using \convup{M}{H}.
$$e_2 = \comp{H}[\intType][\intType]\ {\hinp} \ (\lambda c:\intType.(\comp{H}[\intType][\intType]\ (\convup{M}{H}\ {\minp})\ \wrapAddC))$$
\end{example}

\subsubsection{Indistinguishability}
\label{sec:multi_level:ni_free:ind}
In this section, we define indistinguishability for an observer $\obs \in \lattice$. \footnote{Following \cite{Bowman-Ahmed-15-ICFP}, we use $\obs$ for observers.}
As mentioned in \S\ref{sec:ni_free:encoding}, $\alpha_l$ is used to protect data at $l$. 
We can think of closed terms related at $\alpha_l$ as keys to open wrapped data at $l$.
Thus, if an observer \obs\ can observe protected data at $l$ (i.e. $l \smalleq \obs$), this observer has the keys.
Otherwise, the observer has no key.
This idea is captured in the following definition, where we use \unitType\ as the type of keys, \footnote{Note that other types can be used.} and $\full{\unitType} = \{\tuplemt{\unitVal,\unitVal}\}$.
\footnote{When another closed type $\tau$ is used as the type of keys, we require that all keys are indistinguishable and hence,  the substitution for $\alpha_l$ is \full{\tau}  when $l \smalleq \obs$.
To emphasize this, in Def.~\ref{def:ni:env}, we use \full{\unitType} instead of the identity relation on \unitType, even though for \unitType, these two relations coincide.
}

\begin{definition}
\label{def:ni:env}
An environment $\rho$ is an {\em environment for noninterference w.r.t. an observer \obs} (denoted by $\rho \respectni$) if $\rho\respect \tcontextni$ and for any $\alpha_l \in \tcontextni$:
$$\rho(\alpha_l) = \begin{cases}
 \tuplemt{\unitType,\unitType,\full{\unitType}} & \text{if $l \smalleq \obs$},\\
 \tuplemt{\unitType,\unitType,\emptyset} & \text{if $l \not\smalleq \obs$.}
\end{cases}$$
\end{definition}

%Since the interpretation of a type $\tau$ w.r.t. $\rho$ (i.e. \valrelation{\tau}{\rho}) only depends on the mapping for type variables in $\rho$, we have the following lemma, which says that the interpretations of $\tau$ are the same for all $\rho$ having the same mapping for type variables.
%\begin{lemma}
%\label{lem:ni:log_rel:equiv}
%For any $\tau$ s.t. \typeEval[\tcontextni]{\tau}, for any $\rho_1$ and $\rho_2$ s.t. $\rho_1 \respectni$ and $\rho_2 \respectni$, it follows that $\valrelation{\tau}{\rho_1} = \valrelation{\tau}{\rho_2}$.
%\end{lemma}
%\begin{proof}
%The result follows from the definition of the logical relation.
%Note that \valrelation{\tau}{\rho} depends only on $\rho$ restricted to type variables, and when $\rho_1 \respectni$ and $\rho_2 \respectni$, the restriction of $\rho_1$ to type variables is equal to the restriction of $\rho_2$ to type variables.
%\end{proof}
%We next define indistinguishability for an observer \obs\ as an instantiation of the logical relation. 
%As proven in Lemma~\ref{lem:ni:log_rel:equiv}, the relational interpretations of a type $\tau$ are the same for all $\rho \respectni$.
%Thus, indistinguishability is defined based on an arbitrary $\rho$ s.t. $\rho \respectni$.
Note that there is only a unique an environment for NI w.r.t. an observer \obs\ and hence, hereafter we will use the environment for NI w.r.t. an observer.
We next define indistinguishability for an observer \obs\ as an instantiation of the logical relation. 
\begin{definition}
\label{def:ni:ind}
Given a type $\tau$ s.t. $\typeEval[\tcontextni]{\tau}$.
The {\em indistinguishability relations on values and terms of $\tau$  for an observer $\obs$} (denoted by resp. \indval{\tau} and \indterm{\tau}) are defined as
\begin{align*}
\indval{\tau}  = \valrelation{\tau}{\rho} \quad\quad 
\indterm{\tau} = \termrelation{\tau}{\rho}
\end{align*}
where $\rho$ is the environment for NI w.r.t the observer \obs.
\end{definition}

\begin{example}[Indistinguishability]
\label{ex:ni:ind}
We consider the lattice \tuplemt{\latticethree,\smalleq} described in Example~\ref{ex:multi_level:ni:running_example:comp}. 
%The type of $y$ is $\alpha_M \rightarrow \intType$.
We will describe the indistinguishability relations of the type $\alpha_M \rightarrow \intType$ for an observer at $L$ and for an observer at $H$. Note that $\alpha_M \rightarrow \intType$ is the type of inputs  at level $M$.

The intuition is that the observer at $H$ can observe data at $M$ and hence, two wrapped values at $M$ are indistinguishable to the observer at $H$ if the wrapped values are the same.
This intuition is captured by the definition of indistinguishability as demonstrated below, where $\rho$ is an environment s.t. $\rho \respectni[H]$.
\begin{align*}
\indval[H]{\alpha_M \rightarrow \intType} = & \expl{from Def.~\ref{def:ni:ind}, Lem.~\ref{lem:logeq:related-term}, and rule FR-Fun}\\
   &  \{\tuplemt{\lambda x:\unitType.e_1, \lambda x:\unitType.e_2} \sep \typeEval{\lambda x:\unitType.e_i:\unitType \rightarrow \intType}, \\
   & \hspace{50pt}\forall (v_1,v_2) \in \valrelation{\alpha_M}{\rho}. ((\lambda x:\unitType.e_1)\ v_1, (\lambda x:\unitType.e_2)\ v_2) \in \termrelation{\intType}{\rho} \}\\
 = & \expl{from Def.~\ref{def:ni:env}, rule FR-Term, rule FR-Int}\\
   & \{\tuplemt{\lambda x:\unitType.e_1, \lambda x:\unitType.e_2} \sep \typeEval{\lambda x:\unitType.e_i:\unitType \rightarrow \intType}, \\
   & \hspace{50pt} \exists n_1,n_2.\quad e_1[\unitVal/x] \reduce n_1,\ \ e_2[\unitVal/x]) \reduce n_2\}
\end{align*}

Therefore, for any $n_1$ and $n_2$, $(\lambda \_:\unitType.n_1,\lambda \_:\unitType.n_2) \in \indval[H]{\alpha_M \rightarrow \intType}$ iff $n_1 = n_2$.

While the observer at $H$ can observe data at $M$, the observer at $L$ cannot.
Therefore, to the observer at $L$, all wrapped values at $M$ are indistinguishable.
Below is an demonstration showing that indistinguishability for the observer at $L$ is as desired, where $\rho \respectni[L]$.
\begin{align*}
\indval[L]{\alpha_M \rightarrow \intType} = & \expl{from Def.~\ref{def:ni:ind} and rule FR-Fun}\\
   &  \{\tuplemt{\lambda x:\unitType.e_1, \lambda x:\unitType.e_2} \sep \typeEval{\lambda x:\unitType.e_i:\unitType \rightarrow \intType}, \\
   & \hspace{50pt}\forall (v_1,v_2) \in \valrelation{\alpha_M}{\rho}.((\lambda x:\unitType.e_1)\ v_1, (\lambda x:\unitType.e_2)\ v_2) \in \termrelation{\intType}{\rho} \}\\
 = & \expl{since $\indval[L]{\alpha_M} = \emptyset$, the condition on $v_1$ and $v_2$ holds vacuously}\\
   & \{\tuplemt{\lambda x:\unitType.e_1, \lambda x:\unitType.e_2} \sep \typeEval{\lambda x:\unitType.e_i:\unitType \rightarrow \intType} \}
\end{align*}

Therefore, for any $n_1$ and $n_2$, $(\lambda \_:\unitType.n_1,\lambda \_:\unitType.n_2) \in \indval[H]{\alpha_M \rightarrow \intType}$.
\end{example}

\subsubsection{Typing implies non-interference}
\label{sec:ni_free:theorem_for_free}
\paragraph{Interface for computation.}
The implementations of \comp{l}, \conv{ll'}, and \wrap{l} are respectively \compcr, \convupcr, and \wrapcr\ as below.
The construction of \convupcr\ and \wrapcr\ is straightforward.
On a protected input of type $\unitType \rightarrow \beta_1$ and a continuation of type $\beta_1 \rightarrow (\unitType \rightarrow \beta_2)$, \compcr\ first unfolds the protected input by applying it to the key \unitVal\ and then applies the continuation on the result.

\begin{align*}
\compcr & = \Lambda \beta_1, \beta_2. \lambda x:\unitType \rightarrow \beta_1. \lambda f: \beta_1 \rightarrow (\unitType\rightarrow \beta_2).f(x\ \unitVal)\\
\convupcr & = \Lambda \beta.\lambda x:\unitType \rightarrow \beta.x\\
\wrapcr &= \Lambda \beta.\lambda x:\beta.\lambda \_:\unitType.x
\end{align*}

We next define a full environment for NI which covers not only type variables but also term variables.
Basically, its restriction to type variables (denoted by $\restrict{\rho}$) is the environment for NI w.r.t. an observer \obs. 
In addition, it maps inputs to indistinguishable values and interfaces for computation to the specific implementations just defined.

\begin{definition}
An environment $\rho$ is a {\em full environment for NI w.r.t. an observer \obs} (denoted by $\rho \respectfullni$) if $\restrict{\rho} \respectni$ and
\begin{itemize}
\item for all $x$ of the type $\alpha_l \rightarrow \intType$ (i.e. different from \comp{l} and \convup{l}{l'}), 
$$\rho(x) = \tuple{v_1,v_2} \in \indval{\alpha_l \rightarrow \intType},$$
\item for all $l$, $\rho(\comp{l}) = \tuplemt{\compcr,\compcr}$, 
\item for all $l$ and $l'$ s.t. $l \smallst l'$, $\rho(\convup{l}{l'}) = \tuplemt{\convupcr,\convupcr}$
\item for all $l$, $\rho(\wrap{l}) = \tuplemt{\wrapcr,\wrapcr}$. 

\end{itemize}
\end{definition}

\paragraph{NI for free.}
In order to instantiate the abstraction theorem to prove that a program is NI, we need to prove that \compcr\ is indistinguishable to itself for any observer and so are \convupcr\ and \wrapcr. 
\footnote{Note that Lemma~\ref{lem:substitution:term:comp_convup} is not a direct result of the abstraction theorem since the types of implementations (e.g. \compcr) of defined interfaces are closed types while the types of interfaces (e.g. the type of \comp{l}) are open types.}

\begin{lemma}
\label{lem:substitution:term:comp_convup}
For any \obs, it follows that:
\begin{align*}
\tuplemt{\compcr,\compcr} \in &\ \indval{\forall \beta_1,\beta_2. (\alpha_l \rightarrow \beta_1) \rightarrow \big( \beta_1 \rightarrow (\alpha_l \rightarrow \beta_2)  \big) \rightarrow \alpha_l \rightarrow\beta_2}\\
\tuplemt{\convupcr,\convupcr} \in &\ \indval{\forall \beta. (\alpha_l \rightarrow \beta) \rightarrow (\alpha_{l'} \rightarrow \beta)}\\
\tuplemt{\wrapcr,\wrapcr} \in &\ \indval{\forall \beta. \beta \rightarrow \alpha_{l} \rightarrow \beta}.
\end{align*}
\end{lemma}

From the definition of $\rho \respectfullni$ and Lemma~\ref{lem:substitution:term:comp_convup}, we have the following corollary.
\begin{corollary}
\label{cor:substitution:term:property}
For any $\rho \respectfullni$, $\rho \respect \tcontextni,\contextni$.
\end{corollary}
%\begin{proof}
%From the definition of $\rho \respectfullni$, Lemma~\ref{lem:substitution:term:comp}, and Lemma~\ref{lem:substitution:term:conv}, for any $x \in \dom{\contextni}$ (including \convup{l} and \comp{l}), $\rho(x) \in \valrelation{\Gamma(x)}{\rho}$.
%\end{proof}

We next prove that if a program is well-typed in $\tcontextni, \contextni$, then it transforms indistinguishable inputs to indistinguishable outputs.

\begin{theorem}
\label{thm:ni_free}
If \typeEval[\tcontextni,\contextni]{e:\tau}, then for any {$\obs \in \lattice$} and $\rho \respectfullni$, 
$$\tuplemt{\rho_L(e),\rho_R(e)} \in \indterm{\tau}.$$
\end{theorem}
\begin{proof}
As proven in Corollary~\ref{cor:substitution:term:property}, we have that $\rho\respect \tcontextni,\contextni$.
Since \typeEval[\tcontextni,\contextni]{e:\tau}, from Theorem~\ref{thm:multi_level:parametricity}, we have that $\tuplemt{\rho_L(e),\rho_R(e)} \in \termrelation{\tau}{\rho}$.
From the definition of indistinguishability, it follows that $\tuplemt{\rho_L(e),\rho_R(e)} \in \indterm{\tau}$.
\end{proof}

Therefore, we have the following corollary for programs of the type $(\alpha_{l_1} \rightarrow \intType) \times \dots \times (\alpha_{l_n} \rightarrow \intType)$.

\begin{corollary}[NI for free]
\label{cor:ni_free}
If \typeEval[\tcontextni,\contextni]{e:(\alpha_{l_1} \rightarrow \intType) \times \dots \times (\alpha_{l_n} \rightarrow \intType)}, then for any {$\obs \in \lattice$} and $\rho \respectfullni$, 
$$\tuplemt{\rho_L(e),\rho_R(e)} \in \indterm{(\alpha_{l_1} \rightarrow \intType) \times \dots \times (\alpha_{l_n} \rightarrow \intType)}.$$

\end{corollary}

\begin{example}
\label{ex:ni:running_example:multile_levels}
By reusing programs in Example~\ref{ex:multi_level:ni:running_example:comp} and Example~\ref{ex:ni:running_example:conv}, we here present of a program that generates output values for levels $H$, $M$, and $L$: the output of the program in this example is of the type $(\alpha_H \rightarrow \intType) \times (\alpha_M \rightarrow \intType) \times (\alpha_L \rightarrow \intType)$.

To make the paper easier to follow, we recall here $e_1$ and $e_2$ in resp. Example~\ref{ex:multi_level:ni:running_example:comp} and Example~\ref{ex:ni:running_example:conv}. 
Programs $e_1$ and $e_2$ compute resp. $\minp + 1$ and $\hinp+\minp+1$.
\begin{align*}
e_1 &= \comp{M}[\intType][\intType]\ {\minp}\ {\wrapPlusOne}\\
e_2 &= \comp{H}[\intType][\intType]\ {\hinp} \ (\lambda c:\intType.(\comp{H}[\intType][\intType]\ (\convup{M}{H}\ {\minp})\ \wrapAddC))
\end{align*}

Below is the program $e$ that computes output values for all levels.
From Corollary~\ref{cor:ni_free}, this program is NI.
\begin{align*}
e = (\lambda \minp: \alpha_M \rightarrow \intType.\tuple{e_2 ,\tuple{\minp,\linp}})\ e_1
\end{align*}
%The first value in the output is at $H$ and corresponds to $e_2$, the second value is at $M$ and corresponds to $\minp+1$, and the last value is at $L$ and is the low input \linp.
\end{example}

\subsubsection{Monadic encoding}
\label{sec:multi_level:ni_free:monad}
%\noteinline{For a fixed level, we have monad. So we actually have a family of monads.}
As mentioned earlier, our encoding is a monadic encoding.
We first recall here the monad laws in the popular infix form \cite{Wadler-89-FPCA}, where \wrapcr\ and \compcr\ are respectively the unit and the bind expressions in monad.
Note that {we write $e \lambdaeq e'$ when $\tuplemt{e,e'} \in \termrelation{\tau}{\emptyset}$} which intuitively means that both $e$ and $e'$ reduce to an equal value.
To simplify the presentation, we do not include types  in these laws (e.g. instead of writing $\wrapcr[\tau]\ e$, we just write $\wrapcr\ e$).
%The left unit law says that binding a wrapped value $e$ to a function $f$ is the same as applying $f$ to $e$.
%the right unit law: the result of binding  wrapped value $e$ to the unit expression (which is used to wrapped values) is $e$ itself.
%The associativity law: binding a wrapped value $e$ to $g \circ f$ is the same as binding $e to f$ and then binding the result to $g$.

\begin{align*}
(\wrapcr\ e)\ \compcr\ f& \lambdaeq f\ e & \text{Left unit}\\
e\ \compcr\  \wrapcr  & \lambdaeq e & \text{Right unit}\\
(e\ \compcr\ f)\ \compcr\ g  & \lambdaeq e\ \compcr\ (\lambda x.((f\ x)\ \compcr\ g)) & \text{Associativity}
\end{align*}

%\begin{align*}
%(\wrapcr\ e)\ \compcr\ f& \lambdaeq f\ e & \text{Left unit}\\
%e\ \compcr\  \wrapcr  & \lambdaeq e & \text{Right unit}\\
%(e\ \compcr\ f)\ \compcr\ g  & \lambdaeq e\ \compcr\ (\lambda x:((f\ x)\ \compcr\ g)) & \text{Associativity}
%\end{align*}

%In our encoding, \wrap{l} and \comp{l} correspond to resp. the unit and bind expressions and their implementations satisfy monad laws as proven in Lemma~\ref{lem:ni_free:monad}.
Next we prove that our encoding satisfies monad laws in Lemma~\ref{lem:ni_free:monad}.
We do not use the infix notation in the lemma. The first, the second and the third parts of the lemma correspond to resp. the left unit, right unit and associativity laws.

\begin{lemma}
\label{lem:ni_free:monad}
For any {closed types} $\tau$ and $\tau'$:
\begin{enumerate}
\item for any $e$ s.t. \typeEval{e:\tau}, for any $f: \tau \rightarrow \unitType \rightarrow \tau'$,
$$\compcr[\tau][\tau']\ (\wrapcr[\tau]\ e)\ f \lambdaeq f\ e$$

\item for any $e$ s.t. \typeEval{e: \unitType\rightarrow \tau}, 
$$\compcr[\tau][\tau']\ e\ \wrapcr[\tau] \lambdaeq e$$

\item for any $e$ s.t. \typeEval{e: \unitType\rightarrow \tau}, \typeEval{f: \tau \rightarrow \unitType \rightarrow \tau''}, \typeEval{f:\tau'' \rightarrow \unitType \rightarrow \tau}
$$\compcr[\tau''][\tau']\ (\compcr[\tau][\tau'']\ e\ f)\ g \lambdaeq \compcr[\tau][\tau']\ e \ \big(\lambda x:\tau.\compcr[\tau''][\tau']\ (f\ x)\ g\big)$$
\end{enumerate}
\end{lemma}

\subsection{Type-based relaxed noninterference for free}
\label{sec:multi_level:trni_free}
\subsubsection{Declassification policies}
%The syntax for writing declassification functions is similar to the syntax presented in \redtext{\S\ref{sec:multi_level:language}}, except that we have another form for expression $e \oplus e$, where $\oplus$ is a primitive operator.
%Following \cite{Li-Zdancewic-05-POPL}, we additionally require that declassification functions are of the type $\intType \rightarrow \tau_f$ for some closed $\tau_f$ which is not a polymorphic type.
%For primitive operators, to simplify the presentation, we suppose that the applications of operators on well-typed arguments always terminate.
%Therefore, the evaluations of declassification functions on values always terminate.

Declassification policies considered in this section are similar to the ones in \S\ref{sec:local-policies}, except that since we have multiple levels, for an input at $l$ that can be declassified via $f$, we need to specify which level it can be declassified to. 
Therefore, we define declassification policies as below, where \decSet\ is the set of all closed declassification functions of types $\intType \rightarrow \tau_f$ for some closed $\tau_f$, and $\hookrightarrow$ is used to denote partial functions.
%Note that a tuple of two elements in the meta-theory is denoted by using parentheses (e.g. \tuplemt{f,l}) which is different from the notation for tuples described in \redtext{\S\ref{sec:multi_level:language}}.

\begin{definition}[Declassification policies]
\label{def:multi_level:policy}
A declassification policy \policy\ is a tuple  $\tuplemt{\tuplemt{\lattice,\smalleq},\inpSet,\lvln,\decMap}$, where \tuplemt{\lattice,\smalleq} is a finite lattice of security levels, $\inpSet$ is a finite set of inputs, $\lvln:\inpSet \rightarrow \lattice$ is a mapping from inputs to security levels, $\decMap: \inpSet \hookrightarrow (\decSet \times \lattice)$ is a partial function from inputs to declassification functions and security levels s.t. if $\decMap(x) = \tuplemt{f,l}$ then $\lvl{x} \not\smalleq l$. 
\end{definition}

%For simplicity we require that if $f$ appears in the policy then $f$ has type $\intType \rightarrow \tau_f$ for some  closed $\tau_f$. 
%In the definition of local policies, if an input is not associated with a declassification function, then it cannot be declassified.

\renewcommand{\policyoe}{{\ensuremath{{\policy_{\textit{MOE}}}}}}
\begin{example}[Policy \policyoe]  %[Declassification via declassification function]
\label{ex:local_policy:odd-even}
Consider  policy \policyoe\ given by $\tuplemt{\tuplemt{\latticethree,\smalleq}, \inpSet, \lvln, \decMap}$, where 
\begin{itemize}
\item \tuplemt{\latticethree,\smalleq} is the lattice with three levels $L$, $M$, and $H$ as described in Example~\ref{ex:ni:running_example:comp},

\item $\inpSet = \{\hinp,\minp,\linp\}$: there are three inputs: $\hinp$, $\minp$ and  $\linp$,

\item associated levels of inputs are respectively $H$, $M$, and $L$,

\item input {$\minp$} can be declassified via $f = \lambda x:\intType. x \modop 2$ to level $L$ (i.e. $\decMap({\minp}) = (f,L)$).
\end{itemize}

%where $\varPolicy{\policyoe} = \set{x}$
% and $\decPolicy{\policyoe}(x) = f=\lambda x:\intType. x \modop 2$.
% Policy \policyoe\ states that only the parity of the confidential input $x$ can be released to a public observer. % or an odd number.
\end{example}

%\noteinline{Express policies as contexts?}

%\subsection{Type-based Relaxed Noninterference for free}
%\label{sec:trni_free}

\subsubsection{Encoding}
\label{sec:multi_level:trni_free:encoding}
%\redtext{The types of inputs that can be declassified and their declassifiers will be presented later.}
In this section, we extend the encoding and the indistinguishability relation for NI presented in \S\ref{sec:multi_level:ni_free} to counterparts for declassification policies.
Through this section, we consider a fixed policy \policy\ given by \tuplemt{\tuplemt{\lattice,\smalleq},\inpSet,\lvln,\decMap}.

{In addition to assumptions in \S\ref{sec:multi_level:ni_free}, for declassification policies, we additionally assume that declassifiers are also inputs of programs: that is for each input $x$ that can be declassified via $f$, we have an input $x_f$ which is corresponding for $f$.}

As in NI, we have type variables $\alpha_l$ for all levels $l\in \lattice$, and we also have interfaces \comp{l}, \convup{ll'}, and \wrap{l}.
The encoding for inputs that cannot be declassified (i.e. inputs that are not in the domain of \decMap) is the same as the encoding for inputs in NI: 
{the type of an input $x$ at level $l$ is $\wraptype{\intType} = \alpha_l \rightarrow \intType$}.

We consider an input $x$ at level $l$ that can be declassified via $f$ of the type $\intType \rightarrow \tau_f$ to level $l'$ where $l' \smallst l$.
As for NI, we may use $\alpha_l \rightarrow \intType$ as the type for this input and require that only observers at $l$ or higher levels have the key to unwrap values at $l$. However, as explained in Example~\ref{ex:ni:ind}, to an observer at $l'$, all wrapped values at $l$ are indistinguishable and we do not want this.
Therefore, we introduce a fresh type variable $\alpha_l^f$ for keys of this input and we require that the observer at $l'$ also has keys.
Since the observer at $l'$ has key, if we use $\alpha_l^f \rightarrow \intType$ as the type for the input, as explained in Example~\ref{ex:ni:ind}, only equal values are indistinguishable.
Therefore, we use $\alpha_l^f \rightarrow \alpha^f$ as the type of the input, where $\alpha_f$ is a fresh variable.

We now look at the input $x_f$ which corresponds to the declassifier $f$ for $x$. 
Following the idea of the monadic encoding, we can use $\alpha^f \rightarrow \alpha_{l'} \rightarrow \tau_f$ as the type for $x_f$ and we will define an interface \compdec{x}\ as below to apply $x_f$ to $x$.
$$\compdec{x}: (\alpha_l^f \rightarrow \alpha^f) \rightarrow (\alpha_f \rightarrow \alpha_{l'} \rightarrow \tau_f) \rightarrow \alpha_{l'}\rightarrow \tau_f$$

This interface can be used to apply different functions of type $\alpha_f \rightarrow \alpha_{l'} \rightarrow \tau_f$ to $x$.\footnote{%
Such functions can only handle $x$ via $x_f$ or handle $x$ parametrically.}
However, we do not need this flexibility. 
Therefore, we use $(\alpha_l^f \rightarrow \alpha^f) \rightarrow \alpha_{l'} \rightarrow \tau_f$ as the type for $x_f$.

In addition to $x_f$, in order to support computation on $x$, we have interface \conv{x}.

The encoding for policy \policy\ is described in Fig.~\ref{fig:encoding:local_policies}.
The additional (type and term) variables and conditions are highlighted in blue.

\begin{figure}
\begin{align*}
\tcontextp =&\ \{\alpha_l \sep l \in \lattice\} \cup {\{\alpha_l^f, {\alpha^f} \sep \exists x \in \inpSet. \decMap(x) = \tuplemt{f,\_} \wedge \lvl{x} = l\}}\\
\contextp =&\ \{x: \alpha_l \rightarrow \intType \sep x\in \inpSet \wedge \wedge \lvl{x} = l \wedge {x \not\in \dom{\decMap}}\}\ \cup \\
   &\ \{\comp{l}: \forall \beta_1,\beta_2. (\alpha_l \rightarrow \beta_1) \rightarrow \big( \beta_1 \rightarrow (\alpha_l \rightarrow \beta_2)  \big) \rightarrow \alpha_l \rightarrow\beta_2 \sep l \in \lattice \}\ \cup \\
   &\ \{\convup{l}{l'}: \forall \beta. (\alpha_l \rightarrow \beta) \rightarrow (\alpha_{l'} \rightarrow \beta) \sep l,l' \in \lattice \wedge l \smallst l'\}\ \cup \\
   &\ \{\wrap{l}: \forall \beta.\beta \rightarrow \alpha_l \rightarrow \beta \sep l \in \lattice\}\ \cup \\
   &\ {\{x: {\alpha_l^f\rightarrow \alpha^f},\quad  x_f: {(\alpha_l^f \rightarrow \alpha^f)} \rightarrow \alpha_{l'} \rightarrow \tau_f, \quad } \\
   & \quad\quad \conv{x}: {(\alpha_l^f \rightarrow \alpha^f)} \rightarrow \alpha_l \rightarrow \intType \sep x \in \inpSet \wedge \decMap(x) = \tuplemt{f,l'}  \wedge \lvl{x} = l\}
\end{align*}
\caption{Contexts for declassification  policies}
\label{fig:encoding:local_policies}
\end{figure}

\begin{example}[Input declassified correctly]
\label{ex:local_policy:encoding:oe:declassification}
We consider the \policyoe\ in Example~\ref{ex:local_policy:odd-even}.
By using the contexts for \policyoe, we can have the program {$\minp_f\ \minp$} where the input $\minp$ at $M$ is declassified correctly to $L$ via the interface {$\minp_f$}.
\end{example}

\begin{example}[Computation on input that can be declassified]
\label{ex:local_policy:encoding:oe:computation}
Since the type of {$\minp$} is $\alpha_M^f$, in order to do a computation on it, we must convert it as illustrated by the following program, where \wrapPlusOne\ is of the type $\intType \rightarrow \alpha_M \rightarrow \intType$ (the description of this function is in Example~\ref{ex:ni:running_example:comp}).
$${\comp{M}[\intType][\intType]\ {(\conv{{\minp}}\ \minp)}\  {\wrapPlusOne}}$$

This program uses {\conv{\minp}} to convert {$\minp$} to a wrapped value and then transfer it to the continuation \wrapPlusOne. 
In the context of the policy, this program is of the type $\alpha_M \rightarrow \intType$, that is {the output of the program is at $M$}.
\end{example}

%\begin{remark}[Alternative type for $x_f$]
%\label{rem:rni_free:declassifier}
%\redtext{For an input $x$ at $l$ that can be declassified to $l'$ via $f$, following the idea of the monadic encoding, we can use $\alpha^f \rightarrow \alpha_{l'} \rightarrow \tau_f$ as the type for $x_f$}.
%In this case, we will define the following interface to apply $x_f$ to $x$.
%$$\compdec{x}: (\alpha_l^f \rightarrow \alpha^f) \rightarrow (\alpha_f \rightarrow \alpha_{l'} \rightarrow \tau_f) \rightarrow \alpha_{l'}\rightarrow \tau_f$$
%
%This interface can be used to apply different functions of type $\alpha_f \rightarrow \alpha_{l'} \rightarrow \tau_f$ to $x$.
%\footnote{From the abstraction theorem, such functions can only handle $x$ via $x_f$ or handle $x$ parametrically.}
%However, we expect that \compdec{x} will be applied on $x_f$ which corresponds to the declassifier for $x$.
%Therefore, we choose the encoding for $x_f$ as in Fig.~\ref{fig:encoding:local_policies} since it is simpler.
%\end{remark}

\subsubsection{Indistinguishability}
\label{sec:multi_level:trni_free:ind}
As in \S\ref{sec:multi_level:ni_free:ind}, we first define environments w.r.t. an observer. 
The relational interpretation for $\alpha_l$ is as in the one for NI: $\alpha_l$ is interpreted as pairs of keys to open pairs of protected data.
{For $\alpha_l^f$ which corresponds to $x$ that can be declassified to $l'$, observers that can observe values at $l$ or $l'$ can have the key}.\footnote{%
If only observers who can observe values at $l$ have the key, then to an observer at $l'$, all values at $l$ are indistinguishable (see a similar explanation in Example~\ref{ex:ni:ind}.}
Different from the interpretation for $\alpha_l$,  the interpretation for $\alpha_l^f$ directly captures the notion of indistinguishability for values declassified via $f$.
{In other words, the interpretations of $\alpha^f$ are different for different observers even though they have keys.}

\begin{itemize}
\item if $\obs$ can observe data at $l$ (i.e. $l \smalleq \obs$), then indistinguishable values are equal values,
\item if $\obs$ can observe declassified data at $l'$ but not at $l$ (i.e. $l' \smalleq \obs$ and $l \not\smalleq \obs$), two values $v_1$ and $v_2$ are indistinguishable if they are indistinguishable via $f$,
\item otherwise, any two values are indistinguishable.
\end{itemize}

Let $R_f = \{\tuplemt{v_1,v_2} \sep \tuplemt{f\ v_1, f\ v_2} \in \valrelation{\tau_f}{\emptyset}\}$, where $\emptyset$ is the empty environment, be the relation s.t. two values are related if via $f$ they behave the same.
Let $id_\intType = \{\tuplemt{n,n}\sep n:\intType\}$ be the identity relation on \intType, and $\full{\intType} = \{\tuplemt{n,n'} \sep \typeEval{n,n':\intType}\}$ be the full relation on \intType.
The definition of environments w.r.t. an observer is as below.

\begin{definition}
\label{def:local_policy:env}
An environment $\rho$ is an {\em environment\ for \policy\ w.r.t. an observer \obs} (denoted by $\rho \respectp \policy$) if {$\rho\respect\tcontextp$} and 
\begin{itemize}
%\item if there is no input at $l$ that can be declassified: 
\item for any $\alpha_l \in \tcontextp$:
$$\rho(\alpha_l) = \begin{cases}
 \tuplemt{\unitType,\unitType,\full{\unitType}} & \text{if $l \smalleq \obs$,}\\
 \tuplemt{\unitType,\unitType,\emptyset} & \text{if $l \not\smalleq \obs$,}
\end{cases}$$

\item for any $\alpha_l^f \in \tcontextp$ and $\alpha_l^f \in \tcontextp$ which corresponds to $x$ s.t. $\decMap(x) = (f,l')$:
\begin{align*}
\rho(\alpha_l^f) &= \begin{cases}
 \tuplemt{\unitType,\unitType,\full{\unitType}} & \text{if $l \smalleq \obs$ or {$l' \smalleq \obs$}}\\
 \tuplemt{\unitType,\unitType,\emptyset} & \text{{otherwise},}
\end{cases}\\
\rho(\alpha_l^f) &= \begin{cases}
 \tuplemt{\intType,\intType,id_\intType} & \text{if $l \smalleq \obs$,}\\
 \tuplemt{\intType,\intType,R_f} & \text{if $l' \smalleq \obs$ and $l \not\smalleq \obs$,}\\
 \tuplemt{\intType,\intType,\full{\intType}} & \text{if $l' \not\smalleq \obs$ and $l \not\smalleq \obs$.}
\end{cases}
\end{align*}
%$$$$

%\item for any $\alpha_l^f \in \tcontextp$ which corresponds to $x$ s.t. $\decMap(x) = (f,l')$:
%$$$$
\end{itemize}
\end{definition}

%\begin{lemma}
%\label{lem:ni:log_rel:equiv}
%For any $\tau$ s.t. \typeEval[\tcontextni]{\tau}, for any $\rho_1$ and $\rho_2$ s.t. $\rho_1 \respectni$ and $\rho_2 \respectni$, it follows that $\valrelation{\tau}{\rho_1} = \valrelation{\tau}{\rho_2}$.
%\end{lemma}

We next define indistinguishability for \policy\ as an instantiation of the logical relation.
{The definition is based on the environment of $\policy$ w.r.t. an observer \obs.}

\begin{definition}
\label{def:local_policy:ind}
Given a type $\tau$ s.t. $\typeEval[\tcontextp]{\tau}$.
The {\em indistinguishability relations on values and terms of $\tau$ for an observer $\obs$ for \policy} (denoted by resp. \indvalp{\tau}{\policy} and \indtermp{\tau}{\policy}) are defined as
\begin{align*}
\indvalp{\tau}{\policy}  = \valrelation{\tau}{\rho} \quad\quad 
\indtermp{\tau}{\policy} = \termrelation{\tau}{\rho}
\end{align*}
where $\rho$ is the environment for \policy\ w.r.t. the observer \obs.
\end{definition}

\begin{example}[Indistinguishability for \policyoe]
\label{ex:local_policy:ind:oe}
For \policyoe\ presented in Example~\ref{ex:local_policy:odd-even}, the type of the input $\minp$ that can be declassified to $L$ via $f$ is $\alpha_M^f \rightarrow \alpha^f$.
An observer $H$ can observe values at $M$ and hence, two values of $\minp$ are indistinguishable when they are equal.
This is indeed captured in Def.~\ref{def:local_policy:ind} and Def.~\ref{def:local_policy:env}: 
the observer at $H$ has the key to open wrapped values at $M$ ($\indvalp[H]{\alpha_M^f}{\policyoe} = \full{\unitType}$), and 
two equal integer values are indistinguishable ($\indvalp[H]{\alpha^f}{\policyoe} = id_\intType$).

An observer $L$ can only observe values of $\minp$ after $f$.
Therefore, the observer at $L$ has the key to open values ($\indvalp[L]{\alpha_M^f}{\policyoe} = \full{\unitType}$) and two wrapped values $n_1$ and $n_2$ of $\minp$ are indistinguishable if $f(n_1) = f(n_2)$ as expressed by $\indvalp[L]{\alpha_M^f}{\policyoe} = R_f = \{\tuplemt{n_1,n_2} \sep f(n_1) = f(n_2)\}$.
\end{example}

\subsubsection{Typing implies security}
\label{sec:multi_level:trni_free:theorem_for_free}
The implementations for $x_f$ and \conv{x} are as below:
\begin{align*}
\dec{f} &= \lambda x:{\unitType \rightarrow \intType}.\lambda \_:\unitType.f(x\ \unitVal)\\
\convcr & = \lambda x:{\unitType \rightarrow \intType}.x
%\convcr & = \lambda x:{\unitType \rightarrow \intType}.\lambda \_:\unitType.(x\ \unitVal)
\end{align*}

\begin{definition}
\label{def:local_policy:env:full}
An environment $\rho$ is a {\em full environment for \policy w.r.t. an observer \obs} (denoted by $\rho \respectfullp \policy$) if {$\restrict{\rho} \respectp \policy$} and

\begin{itemize}
\item for all $x$ of the type $\alpha_l \rightarrow \intType$, $\rho(x) \in \indvalp{\alpha_l \rightarrow \intType}{\policy}$,
\item for all $l$, $\rho(\comp{l}) = \tuplemt{\compcr,\compcr}$, \item for all $l$ and $l'$ s.t. $l \smallst l'$, $\rho(\convup{l}{l'}) = \tuplemt{\convupcr,\convupcr}$,
\item for all $x$ of the type {$\alpha_l^f \rightarrow \alpha^f$}, {$\rho(x) \in \indvalp{\alpha_l^f\rightarrow\alpha^f}{\policy}$},
\item for all $x_f$, $\rho(x_f) = \tuplemt{\dec{f},\dec{f}}$,
\item for all $\conv{x}$, $\rho(\conv{x}) = \tuplemt{\convcr,\convcr}$.
\end{itemize}
\end{definition}

%\paragraph{TRNI for free}
In order to instantiate the abstraction theorem to prove security, we need to prove that \compcr, \convupcr, \dec{f}, \convcr\ are indistinguishable to themselves for any observer.

\begin{lemma}
\label{lem:trni:substitution:term:interface}
For any \obs, it follows that:
\begin{align*}
\tuplemt{\compcr,\compcr} \in &\ \indvalp{\forall \beta_1,\beta_2. (\alpha_l \rightarrow \beta_1) \rightarrow \big( \beta_1 \rightarrow (\alpha_l \rightarrow \beta_2)  \big) \rightarrow \alpha_l \rightarrow\beta_2}{\policy}\\
\tuplemt{\convupcr,\convupcr} \in &\ \indvalp{\forall \beta. (\alpha_l \rightarrow \beta) \rightarrow (\alpha_{l'} \rightarrow \beta)}{\policy}\\
\tuplemt{\wrapcr,\wrapcr} \in &\ \indvalp{\forall \beta. \beta \rightarrow \alpha_{l} \rightarrow \beta}{\policy}\\
\tuplemt{\convcr,\convcr} \in &\ {\indvalp{(\alpha_l^f \rightarrow \alpha^f)\rightarrow \alpha_l \rightarrow \intType}{\policy}}
\end{align*}
\end{lemma}

\begin{lemma}
\label{lem:trni:substitution:term:declassification}
For any \obs, for any $\alpha_l^f$ and $\alpha^f$ corresponding to $x$  s.t. $\lvl{x} = l$ and $\decMap(x) = \tuplemt{f,l'}$:
$${\tuplemt{\dec{f},\dec{f}} \in \indvalp{(\alpha_l^f \rightarrow \alpha_f) \rightarrow \alpha_{l'} \rightarrow \tau_f}{\policy}.}$$
\end{lemma}

We next prove that if a program is well-typed in \tcontextp, \contextp, then it is type-based relax noninterferent w.r.t. \policy, that is it transforms indistinguishable inputs to indistinguishable outputs.

\begin{theorem}
\label{thm:trni_free}
If \typeEval[\tcontextp,\contextp]{e:\tau}, then for any {$\obs \in \lattice$} and $\rho \respectfullp \policy$, 
$$\tuplemt{\rho_L(e),\rho_R(e)} \in \indtermp{\tau}{\policy}.$$
\end{theorem}
\begin{proof}
Since $\rho \respectfullp \policy$, from the definition of $\rho \respectfullp \policy$, Lemma~\ref{lem:trni:substitution:term:interface}, and Lemma~\ref{lem:trni:substitution:term:declassification}, it follows that $\rho \respectp \tcontextp, \contextp$.
%Therefore, we have that $\rho\respectp \tcontextp$.
Since \typeEval[\tcontextp,\contextp]{e:\tau}, from Theorem~\ref{thm:multi_level:parametricity}, we have that $\tuplemt{\rho_L(e),\rho_R(e)} \in \termrelation{\tau}{\rho}$.
From the definition of indistinguishability, it follows that $\tuplemt{\rho_L(e),\rho_R(e)} \in \indtermp{\tau}{\policy}$.
\end{proof}

%\begin{example}
%\tbupdated
%\end{example}

\begin{corollary}[RNI for free]
\label{cor:trni_free}
{If \typeEval[\tcontextp,\contextp]{e:(\alpha_{l_1} \rightarrow \intType) \times \dots \times (\alpha_{l_n} \rightarrow \intType)}, then for any {$\obs \in \lattice$} and $\rho \respectfullp \policy$, it is that $\tuplemt{\rho_L(e),\rho_R(e)} \in \indtermp{(\alpha_{l_1} \rightarrow \intType) \times \dots \times (\alpha_{l_n} \rightarrow \intType)}{\policy}$.}
\end{corollary}

\subsection{Extensions}
\label{sec:multi_level:extension}
%This section is for the simple extensions: two declassification functions for an input, equivalence forms of declassification (e.g. $f = g\circ a$), a declassification function involved more than two inputs. \tbupdated. 
{Similar to the encoding in \S\ref{sec:trni}, the encoding in this section can also be extended to support richer policies.
The ideas behind the extensions are similar to the ones in \S\ref{sec:extension}.
Here we only present two extensions: multiple declassification functions for an input and more inputs involved in a declassification functions (global policies)}.
The remaining extension is about declassifying via equivalent functions can be obtained by combining the idea in \S\ref{sec:extension} and \S\ref{sec:multi_level:extension:local-policy:mult_dec}.

\subsubsection{More declassification functions} 
\label{sec:multi_level:extension:local-policy:mult_dec}
In general an input can be declassified in more than one way.  
To show how this can be accommodated, we present an extension for a policy \policymul\ defined on the lattice \tuplemt{\latticedm,\smalleq}, where $\latticedm = \{H,M_1,M_2,L\}$ and $L \smalleq M_i \smalleq H$ (for $i \in \{1,2\}$).
The input set is $\inpSet = \{\hinp, \minpl\}$, where inputs \hinp\ and $\minpl$ are at resp. $H$ and $M_1$.

The policy allows an input to be declassified via multiple functions to different levels.
Specifically, the policy allows that:
\begin{itemize}
\item input $\hinp$ can be declassified via $f_1$ to $M_1$ or $f_2$ to $M_2$ for some $f_1$ and $f_2$, where \typeEval{f_1: \intType \rightarrow \tau_{f_1}} and \typeEval{f_2:\intType \rightarrow \tau_{f_2}},

\item input $\minpl$ can be declassified via $g_1$ or $g_2$ to $L$ for some $g_1$ and $g_2$, where \typeEval{g_1: \intType \rightarrow \tau_{g1}} and \typeEval{g_2:\intType \rightarrow \tau_{g2}}
\end{itemize}

The policy however, does not allow $\hinp$ to be declassified to $L$ via $g$ and $f_i$.

For this policy, we have the context described in Fig.~\ref{fig:encoding:ext:multi_dec} where the details of interfaces \comp{l}, \convup{l}{l'} and \wrap{l} are as in Fig.~\ref{fig:encoding:local_policies} and are omitted.
As in {\S\ref{sec:multi_level:trni_free}}, we have new type variables for inputs \hinp\ and \minpl\ and interfaces $\hinp_{f_1}$, $\hinp_{f_2}$, $\minpl_{g_1}$, $\minpl_{g_2}$ for declassification functions.
%Detailed description: \tbupdated

\begin{figure}
\begin{align*}
\tcontextp[\policymul] =&\ \{\alpha_l \sep l \in \latticedm\} \cup {\{\alpha_H^{f_1,f_2}, {\alpha^{f_1,f_2}}, \alpha_{M_1}^{g_1,g_2}, {\alpha^{g_1,g_2}} \}}\\
\contextp[\policymul] =&\ \{\comp{l}: \dots \}\ \cup 
    \{\convup{l}{l'}: \dots \}\ \cup 
    \{\wrap{l}: \dots \}\ \cup \\
   &\ {\{\hinp: {\alpha_H^{f_1,f_2} \rightarrow \alpha^{f_1,f_2}} \}}\ \cup \\
   &\ {\{\hinp_{f_1}: {(\alpha_H^{f_1,f_2} \rightarrow \alpha^{f_1,f_2})} \rightarrow \alpha_{M_1} \rightarrow \tau_{f_1},\quad \hinp_{f_2}: {(\alpha_H^{f_1,f_2} \rightarrow \alpha^{f_1,f_2})} \rightarrow \alpha_{{M_2}} \rightarrow \tau_{f_2} \} }\ \cup \\
   &\ {\{\conv{\hinp}:{(\alpha_H^{f_1,f_2} \rightarrow \alpha^{f_1,f_2})} \rightarrow \alpha_H \rightarrow \intType \}}\ \cup \\
   &\ {\{\minpl: \alpha_{M_1}^{g_1,g_2} \rightarrow \alpha^{g_1,g_2} \}}\ \cup \\
   &\ {\{\minpl_{g_1}: {(\alpha_{M_1}^{g_1,g_2} \rightarrow \alpha^{g_1,g_2})} \rightarrow \alpha_L \rightarrow \tau_{g_1},\quad \minpl_{g_2}: {(\alpha_{{M_1}}^{g_1,g_2} \rightarrow \alpha^{g_1,g_2})} \rightarrow \alpha_L \rightarrow \tau_{g_2} \} }\ \cup \\
   &\ {\{\conv{\minpl}:{(\alpha_{M_1}^{g_1,g_2} \rightarrow \alpha^{g_1,g^2})} \rightarrow \alpha_{M_1} \rightarrow \intType \}}
\end{align*}
\caption{Contexts for policy \policymul}
\label{fig:encoding:ext:multi_dec}
\end{figure}

We next define the full environment $\rho$\ for \policymul\ w.r.t. an observer \obs\ (denoted by $\rho \respectfullp \policymul$).
For such a $\rho$, the mapping for term variables is as in Def.~\ref{def:local_policy:env:full}.
For a type variable $\alpha_l \in \tcontextp[\policymul]$, $\rho(\alpha_l)$ is as in Def.~\ref{def:local_policy:env}.
{For $\alpha_H^{f_1,f_2}$, an observer at $M_1$, $M_2$ or $H$ has the key}.
$${\rho(\alpha_{H}^{f_1,f_2})} = \begin{cases}
   \tuplemt{\unitType,\unitType,\full{\unitType}} & \text{if $obs \in \{M_1,M_2,H\}$}\\
   \tuplemt{\unitType,\unitType,\emptyset} & \text{otherwise}
\end{cases}$$

For $\alpha_{f_1,f_2}$, the definition is straightforward. 
$$
{\rho({\alpha_{f_1,f_2}}) = \begin{cases}
 \tuplemt{\intType,\intType,id_\intType} & \text{if $\obs = H$,}\\
 \tuplemt{\intType,\intType,R_{f_i}} & \text{if $\obs = M_i$,}\\
 \tuplemt{\intType,\intType,\full{\intType}} & \text{if $\obs = L$.}
\end{cases}}
$$

For $\alpha_{M_1}{g_1,g_2}$, all observers have the key (since the input can be declassified to $L$).
$${\rho(\alpha_{M_1}^{g_1,g_2}) = \tuplemt{\unitType,\unitType,\full{\unitType}}}$$

{For {$\alpha^{g_1,g_2}$}, when \obs\ is $L$, since this observer can apply $g_1$ and $g_2$ to \minpl}, two wrapped values are indistinguishable if they cannot be distinguished by both $g_1$ and $g_2$.
In addition, since data at $L$ can flow to $M_2$, for an observer at $M_2$, two wrapped values at $M_1$ are indistinguishable if they are indistinguishable at $L$.
Therefore, $\rho(\alpha^{g_1,g_2})$ is as below:
$$
\rho({\alpha^{g_1,g_2}}) = \begin{cases}
 \tuplemt{\intType,\intType,id_\intType} & \text{if $M_1 \smalleq \obs$,}\\
% \tuplemt{\intType,\intType,\full{\intType}} & \text{if $\obs = M_2$,}\\
 \tuplemt{\intType,\intType,R_{g_1,g_2}} & \text{if $M_1 \not\smalleq \obs$,}
\end{cases}
$$

where $R_{g_1,g_2} = \{(n_1,n_2) \sep (g_1\ n_1, g_1\ n_2) \in \termrelation{\tau_{g_1}}{\emptyset} \wedge (g_2\ n_1, g_2\ n_2) \in \termrelation{\tau_{g_2}}{\emptyset}\}$.

We define indistinguishability w.r.t. an \obs\ as an instantiation of the logical relation with an arbitrary $\rho \respectfullp \policymul$.
The implementations of \comp{\_}, \conv{\_}, \convup{l}{l'}, $\hinp_{f_i}$, and $\minpl_{g_i}$ are as in \S\ref{sec:multi_level:trni_free:theorem_for_free}.
We also have that these implementations are indistinguishable to themselves for any observer.
Therefore, from the abstraction theorem, we again obtain that for any program $e$, if \typeEval[\Delta_{\policymul},\Gamma_{\policymul}]{e:\tau}, then this program maps indistinguishable inputs to indistinguishable outputs.
Proofs are in \S\ref{sec:multi_level:proof:extension}.

For example, we consider programs $e_{E1} = \hinp_{f_1}\ \hinp$ and $e_{E2} = \tuple{\minpl_{g_1}\ \minpl}$.
{These programs are well-typed in the context of \policymul, and their types are respectively $\alpha_M \rightarrow \tau_{f_1}$ and $\alpha_L \rightarrow \tau_{g_1}$ and hence, on indistinguishable inputs, their outputs are indistinguishable at resp. $M_1$ and $L$.}

\subsubsection{Global policies}
\label{sec:multi_level:extension:global}
We now consider policies where a declassifier can involve more than one input.
For simplicity, in this subsection, we consider a policy \policyglb\ defined on the lattice \tuplemt{\latticedm,\smalleq} described in {\S\ref{sec:multi_level:extension:local-policy:mult_dec}}.
There are two inputs: $\minpl$\ and $\minpr$ at respectively $M_1$ and $M_2$.
The average of these inputs can be declassified to $L$, i.e. they can be declassified to $L$ via $f = \lambda x:\intType_1 \times \intType_2.(\prj{1}{x}+\prj{2}{x})/2$.\footnote{%
We can extend the encoding presented in this section to have policies where different subsets of \inpSet\ can be declassified and to have more than one declassifier associated with a set of confidential inputs.}
Notice that here we use subscripts for the input type of $f$ to mean that the confidential inputs $\minpl$ and \minpr\ are corresponding to resp. the first and second elements of an input of $f$.

To encode the requirement that $\minpl$ and $\minpr$ can be declassified via $f$, we introduce a new variable $y$, which is corresponding to the tuple of inputs $\minpl$ and $\minpr$.\footnote{{This idea can be generalized to capture the requirement in which there are more than two inputs that can be declassified. When there are $n$ inputs that can be declassified, we just introduce a fresh variable $y$ which is correspond to the $n$-tuple of inputs.}}
Individual inputs cannot be declassified, and only $y$ can be declassified via $f$.
Thus, we have the context described in Fig.~\ref{fig:encoding:ext:global}, where the details of interfaces \comp{l}, \convup{l}{l'} and \wrap{l} are as in {Fig.~\ref{fig:encoding:local_policies}} and are omitted
Note that {we introduce new type variables $\alpha^f$ and $\alpha_{M_1,M_2}^f$, where $\alpha_{M_1,M_2}^f$ is the type of the key to open $y$}. 
Since $\minpl$ cannot be declassified directly, its type is $\alpha_{M_1} \rightarrow \intType$.
Similarly, the type of \minpr\ is $\alpha_{M_2} \rightarrow \intType$.

\begin{figure}
\begin{align*}
\tcontextp[\policyglb] =&\ \{\alpha_l \sep l \in \latticedm\} \cup {\{{\alpha^f, \alpha_{M_1,M_2}^f}\}}\\
\contextp[\policyglb] =&\ \{\comp{l}: \dots \}\ \cup 
    \{\convup{l}{l'}: \dots \}\ \cup 
    \{\wrap{l}: \dots \}\ \cup \\
   &\ {\{y: {\alpha_{M_1,M_2}^{f} \rightarrow \alpha^f}, \quad y_{f}: {(\alpha_{M_1,M_2}^f \rightarrow \alpha^f)} \rightarrow \alpha_{L} \rightarrow {\intType} \}}\ \cup \\
   &\ {\{\minpl: \alpha_{M_1} \rightarrow \intType,\quad \minpr: \alpha_{M_2} \rightarrow \intType\}}
\end{align*}
\caption{Contexts for policy \policyglb}
\label{fig:encoding:ext:global}
\end{figure}

%In order to define indistinguishability, we define a full environment $\rho$ for \policyglb\ w.r.t. an observer \obs\ (denoted by $\rho \respectfullp \policyglb$).
\paragraph{Environment for \policyglb.}
We next define an environment for \policyglb\ (denoted by $\rho \respectp \policyglb$).
The definition for $\rho \respectp \policyglb$ is similar to Def.~\ref{def:local_policy:env}, except for {$\alpha_{M_1,M_2}^f$ and $\alpha^f$}.

For $\alpha_{M_1,M_2}^f$, since $y$ can be declassified to $L$, all observers have the key.
$${\rho(\alpha_{M_1,M_2}^f) = \tuplemt{\unitType,\unitType,\full{\unitType}}}$$

Since $\alpha_f$ is the type of $y$ which is corresponding to both inputs, its concrete type is $\intType\times \intType$.
\begin{itemize}
\item When $\obs = H$ (i.e. the observer $\obs$ can observe both inputs at $M_1$ and $M_2$), the interpretation of $\alpha_f$ is just $id_{\intType \times \intType}$.
\item When $\obs = L$, since the observer can apply $f$ on $y$, two tuples of inputs are indistinguishable if the results of $f$ on them are the same.
\item When $\obs = M_i$, since declassified data at $L$ can be observed at $M_i$, two tuples of inputs are indistinguishable if they are indistinguishable at $L$.
\end{itemize}

Therefore, we define $\rho(\alpha_f)$ as below:
$$\rho(\alpha_f) = \begin{cases}
 \tuplemt{\intType\times \intType,\intType\times \intType,id_{\intType\times\intType}} & \text{if $\obs = H$,}\\
% \tuplemt{\intType\times \intType,\intType\times \intType, {\relLeft} } & \text{if $\obs = M_1$,}\\
% \tuplemt{\intType\times \intType,\intType\times \intType, {\relRight}} & \text{if $\obs = M_2$,}\\
 \tuplemt{\intType\times \intType,\intType\times \intType, {\relDec}} & \text{otherwise,}
\end{cases}$$

where $\relDec = \{\tuplemt{\tuple{v,u},\tuple{v',u'}} \sep \tuplemt{f\ \tuple{v,u}, f\ \tuple{v',u'}} \in \termrelation{\intType}{\emptyset} \}$.
%\begin{align*}
%%\relLeft & = \{\tuplemt{\tuple{v,u},\tuple{v,u'}} \sep \typeEval{v,u,u':\intType} \}\\
%%\relRight & = \{\tuplemt{\tuple{v,u},\tuple{v',u}} \sep \typeEval{v,v',u:\intType} \}\\
%\relDec &= \{\tuplemt{\tuple{v,u},\tuple{v',u'}} \sep \tuplemt{f\ \tuple{v,u}, f\ \tuple{v',u'}} \in \termrelation{\intType}{\emptyset} \}\\
%\end{align*}

\paragraph{Full environment for \policyglb.}
Different from previous sections, in order to define full environments, we need to encode the correspondence between inputs and argument of the declassifier.
We say that an environment $\rho$ of \policyglb\ is {\em consistent} w.r.t. an \obs\ if $\rho \respectp \policyglb$ and 
\begin{align*}
\prj{1}{{(\rho_L(y)\ \unitVal)}} = \rho_L(\minpl)\ \unitVal & &  \prj{2}{{(\rho_L(y)\ \unitVal)}} = \rho_L(\minpr)\ \unitVal\\
\prj{1}{{(\rho_R(y)\ \unitVal)}} = \rho_R(\minpl)\ \unitVal & &  \prj{2}{{(\rho_R(y)\ \unitVal)}} = \rho_R(\minpr)\ \unitVal
\end{align*}

This additional condition takes care of the correspondence of inputs and the arguments of the designated declassifier.

We construct $\dec{f}$ which is used as the concrete input for $y_f$.
$$\lambda x:{\unitType \rightarrow \intType\times\intType}.\lambda \_:\unitType.f{(x\ \unitVal)}\\$$

{An environment $\rho$ is full for \policyglb\ w.r.t. an \obs\ (denoted by $\rho \respectfullp \policyglb$) if $\rho$ is consistent w.r.t. \obs, it maps $y_f$ to $\tuplemt{\dec{f},\dec{f}}$, and the mapping of other term variables is as in Def.~\ref{def:local_policy:env:full}}.

We then define indistinguishability as an instantiation of the logical relation as in Def.~\ref{def:local_policy:ind}.
We also have the free theorem saying that if \typeEval[\Delta_\policyglb,\Gamma_\policyglb]{e:\tau}, then $e$ maps indistinguishable inputs to indistinguishable outputs.
The proof goes through without changes.

We illustrate this section by considering program $y_f\ y$ where the declassifier is applied correctly.
This program is well-typed in the contexts for \policyglb\ and its type is $\alpha_L \rightarrow \intType$ and hence, this program maps indistinguishable inputs to indistinguishable outputs.

%\begin{example}[Average can be declassified - cont.]
%Here we present the encoding for the policy \policyAve\ described in Example~\ref{ex:average:description}.
%The typing and term contexts for this policy are similar to those in Fig.~\ref{fig:encoding:ext:global}, except that $\tau_f$ is replaced by $\intType$.
%We can easily check that the type of the program $y_f\ y$ is $\alpha_l \rightarrow \intType$ and hence, \redtext{this program satisfies the policy and its output is at $L$}.
%\end{example}

\begin{remark}
%[\orangetext{On the introduction of $y$ in the encoding}]
We may use two type variables $\alpha_1$ and $\alpha_2$ for the two inputs (i.e. $\minp_1: \alpha_1$ and $\minp_2:\alpha_2$) and use $y_f:\alpha_1 \times \alpha_2 \rightarrow \intType$ for the declassifier.
Since we used two type variables separately, we may define the indistinguishability relations for them as the full relation on \intType.
However, when we define the indistinguishability relations separately, the declassifier $f$ may not be related to itself at $\alpha_1 \times \alpha_2 \rightarrow \intType$.\footnote{In order to use the abstraction theorem to prove security, we need that the declassifier $f$ is indistinguishable to itself at $\alpha_1 \times \alpha_2 \rightarrow \intType$.}
For example, we have $3$ and $4$ are indistinguishable at $\alpha_1$, $4$ and $5$ are indistinguishable at $\alpha_2$ but $f\ \tuple{3,5} \neq f\ \tuple{4,6}$.

To have $f$ related to itself, the indistinguishability relations for $\alpha_1$ and $\alpha_2$ should depend on each other.
For example, suppose that $v_1$ and $v_1'$ are indistinguishable values for $\minp_1$.
Then two values $v_2$ and $v_2'$ for $\minp_2$ are indistinguishable when $f\tuple{v_1,v_2} = f\tuple{v_1',v_2}$.
In other words, whether $v_2$ and $v_2'$ are indistinguishable depends on $v_1$ and $v_1'$.
However, we have difficulty to express such dependency between $\alpha_1$ and $\alpha_2$ in the language presented in \S\ref{sec:multi_level:language} and hence, we introduce a new variable $y$.

%\noteinline{We may define potential indistinguishability as an instantiation of the logical relation and require that $\gamma_1$ and $\gamma_2$ are indistinguishable when they are potentially indistinguishable and furthermore, $f\tuple{\gamma_1(\minp_1),\gamma_1(\minp_2))} = f\tuple{\gamma_1(\minp_1),\gamma_1(\minp_2))}$. 
%%W.r.t. this approach, we can still have the free theorem and we do not need to introduce the new variable $y$.
%However, we do not have $f$ related to itself.
%Can we just limit the relation when we apply the abstraction theorem? This may be adhoc.
%}
\end{remark}

\section{Proofs for multi-level encodings}
\label{sec:multi_level:proof}
\subsection{{Proofs for the encoding for NI}}
\label{sec:multi_level:ni_free:proof}
\textbf{Lemma~\ref{lem:multi_level:logeq:related-term}}.
Suppose that $\rho \respect \Delta$ for some $\Delta$. 
It follows that:
\begin{itemize}
\item if $\tuplemt{v_1,v_2} \in \valrelation{\tau}{\rho}$, then \typeEval{v_1: \rho_L(\tau)}, \typeEval{v_2: \rho_R(\tau)}and
\item if $\tuplemt{e_1,e_2} \in \termrelation{\tau}{\rho}$, then \typeEval{e_1: \rho_L(\tau)} and \typeEval{e_2: \rho_R(\tau)}.
\end{itemize}
\begin{proof}
The second part of the lemma follows directly from rule FR-Term.
We prove the first part of the lemma by induction on structure of $\tau$.

\emcase{Case 1:} $\intType$.
We consider $\tuplemt{v_1,v_2} \in \valrelation{\intType}{\rho}$.
From FR-Int, we have that \typeEval{v_i:\intType}.
Since $\rho_L(\intType) = \rho_R(\intType) = \intType$, we have that \typeEval{v_1:\rho_L(\intType)} and \typeEval{v_2:\rho_R(\intType)}.

\emcase{Case 2:} $\unitType$. The proof is similar to the one of Case 1.
We consider $\tuplemt{v_1,v_2} \in \valrelation{\intType}{\rho}$.
From FR-Int, we have that \typeEval{v_i:\intType}.
Since $\rho_L(\intType) = \rho_R(\intType) = \intType$, we have that \typeEval{v_1:\rho_L(\intType)} and \typeEval{v_2:\rho_R(\intType)}.

\emcase{Case 3:} $\alpha$.
We consider $\tuplemt{v_1,v_2} \in \valrelation{\alpha}{\rho}$.
From the FR-Var rule, $\tuplemt{v_1,v_2} \in \rho(\alpha) \in  \Rel{\tau_1,\tau_2}$.
From the definition of \Rel{\tau_1,\tau_2}, we have that \typeEval{v_1:\rho_L(\alpha)} and \typeEval{v_2:\rho_R(\alpha)}.

\emcase{Case 4:} $\tau_1 \times \tau_2$.
We consider $\tuplemt{v_1,v_2} \in \valrelation{\tau_1 \times \tau_2}{\rho}$.
We then have that $v_1 = \tuplemt{v_{11},v_{12}}$ for some $v_{11}$ and $v_{12}$ and $v_2 = \tuplemt{v_{21}, v_{22}}$ for some $v_{21}$ and $v_{22}$.
From FR-Pair, it follows that $\tuplemt{v_{11},v_{21}} \in \valrelation{\tau_1}{\rho}$ and $\tuplemt{v_{12},v_{22}} \in \valrelation{\tau_2}{\rho}$.
From IH (on $\tau_1$ and $\tau_2$), we have that \typeEval{v_{11}:\rho_L(\tau_1)}, \typeEval{v_{21}:\rho_R(\tau_1)}, \typeEval{v_{12}:\rho_L(\tau_2)}, and \typeEval{v_{22}:\rho_R(\tau_2)}.
From FT-Pair, \typeEval{\tuple{v_{11},v_{12}}:\rho_L(\tau_1)\times \rho_L(\tau_2)} and \typeEval{\tuple{v_{21},v_{22}}:\rho_R(\tau_1)\times \rho_R(\tau_2)}.
Thus, \typeEval{v_1:\rho_L(\tau_1)\times \rho_L(\tau_2)} and \typeEval{v_2:\rho_R(\tau_1)\times \rho_R(\tau_2)}.
In other words, \typeEval{v_1:\rho_L(\tau_1 \times \tau_2)} and \typeEval{v_2:\rho_R(\tau_1 \times \tau_2)}.

\emcase{Case 5:} $\tau_1 \rightarrow \tau_2$.
We consider $\tuplemt{v_1,v_2} \in \valrelation{\tau_1 \rightarrow \tau_2}{\rho}$.
We now look at arbitrary $\tuplemt{v_1',v_2'} \in \valrelation{\tau_1}{\rho}$.
From FR-Fun, it follows that $\tuplemt{v_1v_1',v_2v_2'} \in \termrelation{\tau_2}{\rho}$.
From IH on $\tau_1$ and the second part of the lemma on $\tau_2$, we have that \typeEval{v_1':\rho_L(\tau_1)}, \typeEval{v_2':\rho_R(\tau_1)}, \typeEval{v_1v_1':\rho_L(\tau_2)}, and \typeEval{v_2v_2':\rho_R(\tau_2)}.
From FT-App, we have that \typeEval{v_1:\rho_L(\tau_1 \rightarrow\tau_2)} and \typeEval{v_2:\rho_R(\tau_1 \rightarrow\tau_2)}.

\emcase{Case 6:} $\forall \alpha.\tau$.
We consider an arbitrary $R \in \Rel{\tau_1,\tau_2}$.
Since $\tuplemt{v_1,v_2} \in \valrelation{\forall\alpha.\tau}{\rho}$,
we have that $\tuplemt{v_1[\tau_1],v_2[\tau_2]} \in \termrelation{\tau}{\rho[\tuplemt{\tau_1,\tau_2,R}/\alpha]}$.
From IH, \typeEval{v_1[\tau_1]: \rho_L'(\tau)}, where $\rho' = \rho[\tuplemt{\tau_1,\tau_2,R}/\alpha]$.
From the typing rule, \typeEval{v_1: \rho_L'(\forall \alpha.\tau)}.
Since $\alpha$ is bound in $\tau$, we have that \typeEval{v_1: \rho_L(\forall \alpha.\tau)}.
Similarly, \typeEval{v_2: \rho_R(\forall \alpha.\tau)}.
\end{proof}

%\textbf{Theorem~\ref{thm:multi_level:parametricity}.}
%\tealtext{If \typeEval[\Delta,\Gamma]{e:\tau}, then \typeEval[\Delta,\Gamma]{e \logrel e:\tau}.}
%\begin{proof}
%Similar to the proof of Theorem~\ref{thm:type-abstraction}.
%\end{proof}

%\susection{Proofs for \S\ref{sec:ni_free}}

\textbf{Lemma~\ref{lem:substitution:term:comp_convup} - Part 1.}
For any \obs, it follows that:
$$\tuplemt{\compcr,\compcr} \in \indval{\forall \beta_1,\beta_2. (\alpha_l \rightarrow \beta_1) \rightarrow \big( \beta_1 \rightarrow (\alpha_l \rightarrow \beta_2)  \big) \rightarrow \alpha_l \rightarrow\beta_2}.$$

\begin{proof}
Let $\rho \respectni$.
We first prove for \comp{l}:
$$\rho(\comp{l}) \in \indval{\forall \beta_1,\beta_2. (\alpha_l \rightarrow \beta_1) \rightarrow \big( \beta_1 \rightarrow (\alpha_l \rightarrow \beta_2)  \big) \rightarrow \alpha_l \rightarrow\beta_2}.$$

From the definition of indistinguishability, we need to prove that:
$$\rho(\comp{l}) \in \valrelation{\forall \beta_1,\beta_2. (\alpha_l \rightarrow \beta_1) \rightarrow \big( \beta_1 \rightarrow (\alpha_l \rightarrow \beta_2)  \big) \rightarrow \alpha_l \rightarrow\beta_2}{\rho},$$
where $\rho \respectni$.
From the definition of the logical relation, we need to prove that for any $R \in Rel(\tau_1,\tau_1')$ and $S \in Rel(\tau_2,\tau_2')$, we have that:
$$\tuplemt{\compcr[\tau_1][\tau_2],\compcr[\tau_1'][\tau_2']} \in
 \termrelation{(\alpha_l \rightarrow \beta_1) \rightarrow \big( \beta_1 \rightarrow (\alpha_l \rightarrow \beta_2)  \big) \rightarrow \alpha_l \rightarrow\beta_2}{\rho'}$$

%\tuplemt{\compcr[\tau_1][\tau_2],\compcr[\tau_1'][\tau_2']} \in \termrelation{(\alpha_l \rightarrow \beta_1) \rightarrow \big( \beta_1 \rightarrow (\alpha_l \rightarrow \beta_2)  \big) \rightarrow \alpha_l \rightarrow\beta_2}{\rho, \beta_1\mapsto \tuplemt{\tau_1,\tau_1',R},\beta_2\mapsto \tuplemt{\tau_2,\tau_2',S}}

where $\rho' = \rho,\beta_1\mapsto\tuplemt{\tau_1,\tau_1',R}, \beta_2\mapsto\tuplemt{\tau_2,\tau_2',S}$.

We now need to prove that for all $\tuplemt{v,v'} \in \valrelation{(\alpha_l \rightarrow \beta_1)}{\rho'}$, for all $\tuplemt{f,f'} \in \valrelation{\beta_1 \rightarrow (\alpha_l \rightarrow \beta_2)}{\rho'}$, it follows that 
$$\tuplemt{f(v\ \unitVal), f'(v'\ \unitVal)} \in \termrelation{\alpha_l \rightarrow \beta_2}{\rho'}.$$

We now have two cases:
\begin{itemize}
\item $l \not\smalleq \obs$. We have that $\rho(\alpha_l) = \tuplemt{\unitType,\unitType,\emptyset}$ and hence, $\tuplemt{f(v\ \unitVal), f'(v'\ \unitVal)} \in \termrelation{\alpha_l \rightarrow \beta_2}{\rho'}$ holds vacuously. 
%\redtext{\bf [CHECK THE TYPES OF TERMS - Both of them have type $\alpha_l \rightarrow \beta_2$ where $\alpha_l \equiv \unitType$ and $\beta_2$ is $\sigma$ and $\sigma'$.]}.

\item $l \smalleq \obs$. We have that $\rho(\alpha_l) = \tuplemt{\unitType,\unitType,\full{\unitType}}$.
Since $\tuplemt{f,f'} \in \valrelation{\beta_1 \rightarrow (\alpha_l \rightarrow \beta_2)}{\rho'}$, we just need to prove that $\tuplemt{v\ \unitVal,v'\ \unitVal} \in \termrelation{\beta_1}{\rho'}$.
And this follows from the fact that $\tuple{v,v'} \in \valrelation{\alpha_l\rightarrow\beta_1}{\rho'}$.
\end{itemize}
\end{proof}

\textbf{Lemma~\ref{lem:substitution:term:comp_convup} - Part 2.}
For any \obs, it follows that:
$$\tuplemt{\convupcr,\convupcr} \in \indval{\forall \beta. (\alpha_l \rightarrow \beta) \rightarrow (\alpha_{l'} \rightarrow \beta)}.$$
\begin{proof}
We need to prove that for any $\rho\respectni$,
$$\tuplemt{\convupcr,\convupcr} \in \valrelation{\forall \beta. (\alpha_l \rightarrow \beta) \rightarrow (\alpha_{l'} \rightarrow \beta)}{\rho}.$$

That is for any closed types $\tau$ and $\tau'$, for any $R \in \Rel{\tau,\tau'}$,
$$\tuplemt{\convupcr[\sigma],\convupcr[\sigma']} \in \valrelation{(\alpha_l \rightarrow \beta) \rightarrow (\alpha_{l'} \rightarrow \beta)}{\rho'},$$

where $\rho' = \rho,\beta\mapsto \tuplemt{\tau,\tau',R}$.

From the definition of \convupcr, we need to prove that for any $(v,v') \in \valrelation{\alpha_l \rightarrow \beta}{\rho'}$, it follows that $(v,v') \in \valrelation{\alpha_{l'} \rightarrow \beta}{\rho'}$.
We have the following cases:
\begin{itemize}

\item $l' \smalleq \obs$. Since $l \smallst l'$, we have that $l \smallst l' \smalleq \obs$.
Since $\rho \respectni$, we have that $\rho(\alpha_l) = \rho(\alpha_l') = {\tuplemt{\unitType,\unitType,\full{\unitType}}}$.
Thus, for any $\tuplemt{v,v'} \in \valrelation{\alpha_l \rightarrow \beta}{\rho'}$,
We have that $\tuplemt{v,v'} \in \valrelation{\alpha_{l'} \rightarrow \beta}{\rho'}$.

\item $l' \not\smalleq \obs$. 
We have that $\rho(\alpha_l') = \tuplemt{\unitType,\unitType,\emptyset}$.
Since the interpretation of $\alpha_{l'}$ is $\emptyset$, for any $\tuplemt{v,v'} \in \valrelation{\alpha_l \rightarrow \beta}{\rho'}$, we have that $\tuplemt{v,v'} \in \valrelation{\alpha_{l'} \rightarrow \beta}{\rho'}$.
\end{itemize}
\end{proof}

\textbf{Lemma~\ref{lem:substitution:term:comp_convup} - Part 3.}
For any \obs, it follows that:
$$\tuplemt{\wrapcr,\wrapcr} \in \indval{\forall \beta. \beta \rightarrow \alpha_l \rightarrow \beta}.$$
\begin{proof}
We need to prove that for any $\rho \respectni$,
$$\tuplemt{\wrapcr,\wrapcr} \in \valrelation{\forall \beta. \beta \rightarrow \alpha_l \rightarrow \beta}{\rho}$$

That is for any $\tau_1$, $\tau_2$ and $R$ s.t. $R \in \Rel{\tau_1,\tau_2}$,
$$\tuplemt{\wrapcr[\tau_1],\wrapcr[\tau_2]} \in \valrelation{\beta \rightarrow \alpha_l \rightarrow \beta}{\rho,\beta\mapsto \tuplemt{\tau_1,\tau_2,R}}$$
 
Let $\rho' = \rho,\beta\mapsto \tuplemt{\tau_1,\tau_2,R}$, we need to prove that for any $(v_1,v_2) \in R$, we have that:
$$\tuplemt{\wrapcr[\tau_1]\ v_1,\wrapcr[\tau_2]\ v_2} \in \valrelation{\alpha_l \rightarrow \beta}{\rho'}.$$

From the definition of \wrapcr, we need to prove that:
$$\tuplemt{\lambda \_:\unitType.v_1,\lambda \_:\unitType.v_2} \in \valrelation{\alpha_l \rightarrow \beta}{\rho'}.$$

When $l \not\smalleq \obs$, $\rho'(\alpha_l) = \tuplemt{\unitType,\unitType,\emptyset}$ and hence, the statement hold vacuously.
When $l \smalleq \obs$, $\rho'(\alpha_l) = \tuplemt{\unitType,\unitType,\full{\unitType}}$ and hence, we only need to prove that $\tuplemt{v_1,v_2} \in R$.
This is trivial from the assumption that $\tuplemt{v_1,v_2} \in R$.
\end{proof}

\textbf{Lemma~\ref{lem:ni_free:monad}.}
For any closed types $\tau$ and $\tau'$:
\begin{enumerate}
\item for any $e$ s.t. \typeEval{e:\tau}, for any $f: \tau \rightarrow \unitType \rightarrow \tau'$,
$$\compcr[\tau][\tau']\ (\wrapcr[\tau]\ e)\ f \lambdaeq f\ e$$

\item for any $e$ s.t. \typeEval{e: \unitType\rightarrow \tau}, 
$$\compcr[\tau][\tau']\ e\ \wrapcr[\tau] \lambdaeq e$$

\item for any $e$ s.t. \typeEval{e: \unitType\rightarrow \tau}, \typeEval{f: \tau \rightarrow \unitType \rightarrow \tau''}, \typeEval{f:\tau'' \rightarrow \unitType \rightarrow \tau}
$$\compcr[\tau''][\tau']\ (\compcr[\tau][\tau'']\ e\ f)\ g \lambdaeq \compcr[\tau][\tau']\ e \ \big(\lambda x:\tau.\compcr[\tau''][\tau']\ (f\ x)\ g\big)$$
\end{enumerate}

\begin{proof}
{We prove (1)}.
From the definition of \wrapcr, $\wrapcr[\tau]\ e \reduce \lambda\_:\unitType.e$.
From the definition of \compcr, we have that:
$$\compcr[\tau][\tau']\ (\wrapcr[\tau]\ e)\ f  \reduce f(\lambda \_:\unitType.e\ \unitVal) \reduce f(e)$$

Thus, $\compcr[\tau][\tau']\ (\wrapcr[\tau]\ e)\ f \lambdaeq f\ e$.

{We prove (2)}.
From the implementation of \wrapcr, $\wrapcr[\tau] \reduce \lambda x:\tau.\lambda \_:\unitType.x$.
Let $v$ be the value s.t. $e \reduce v$ {(note that in our language, all closed terms can be reduced to values)}.
Let $v'$ of the type $\tau$ s.t. $(v\ \unitVal)\reduce v'$.
Therefore, we have that $(e\ \unitVal) \reduce v\ \unitVal \reduce v'$.

From the implementation of \compcr, 
$$\compcr[\tau][\tau']\ e\ \wrapcr[\tau]\reduce  (\lambda x:\tau.\lambda \_:\unitType.x)(v\ \unitVal) \reduce \lambda\_:\unitType.v'$$

Thus, $(\compcr[\tau][\tau']\ e\ \wrapcr[\tau])\ \unitVal \reduce v'$.
As proven above, $(e\ \unitVal) \reduce v\ \unitVal \reduce v'$.
Thus, $\compcr[\tau][\tau']\ e\ \wrapcr[\tau] \lambdaeq e$.

{We prove (3)}.
Let $v$, $v'''$, $v''$ and $v'$ be values s.t.
\begin{itemize}
\item $v$ is of type $\unitType \rightarrow \tau$ and $e \reduce v$,
\item $v'''$ is of type $\tau$ and $v\ \unitVal \reduce v'''$,
\item $v''$ is of the type $\unitType \rightarrow \tau''$ s.t. $f\ v''' \reduce v''$
\item $v'$ is of the type $\unitType \rightarrow \tau'$ s.t. $g(v''\ \unitVal) \reduce v'$.
\end{itemize}

We first look at $\compcr[\tau''][\tau']\ (\compcr[\tau][\tau'']\ e\ f)\ g$.
We have that
%$$\compcr[\tau][\tau'']\ e\ f \reduce \compcr[\tau][\tau'']\ v\ f \reduce f(v\ \unitVal) \reduce f(v''') \reduce v''.$$
$$\compcr[\tau][\tau'']\ e\ f \reduce f(v\ \unitVal) \reduce f(v''') \reduce v''.$$
 
Thus, $\compcr[\tau''][\tau']\ (\compcr[\tau][\tau'']\ e\ f)\ g \reduce g(v''\ \unitVal) \reduce v'$.

We now look at $\compcr[\tau][\tau']\ e \ \big(\lambda x:\tau.\compcr[\tau''][\tau']\ (f\ x)\ g\big)$.
We need to prove that $\compcr[\tau][\tau']\ e \ \big(\lambda x:\tau.\compcr[\tau''][\tau']\ (f\ x)\ g\big) \reduce v'$.
We have that:
\begin{align*}
\compcr[\tau][\tau']\ e \ \big(\lambda x:\tau.\compcr[\tau''][\tau']\ (f\ x)\ g\big) & \reduce \big(\lambda x:\tau.\compcr[\tau''][\tau']\ (f\ x)\ g\big)(v\ \unitVal) \\
   & \reduce \big(\lambda x:\tau.\compcr[\tau''][\tau']\ (f\ x)\ g\big)(v''')\\
   & \reduce \compcr[\tau''][\tau']\ (f\ v''')\ g\\
   & \reduce g(v'' \unitVal) \reduce v'
\end{align*}
\end{proof}

\subsection{{Proofs for the encoding for TRNI}}
\label{sec:multi_level:trni_free:proof}
\textbf{Lemma~\ref{lem:trni:substitution:term:interface}.}
For any \obs, it follows that:
\begin{align*}
\tuplemt{\compcr,\compcr} \in &\ \indvalp{\forall \beta_1,\beta_2. (\alpha_l \rightarrow \beta_1) \rightarrow \big( \beta_1 \rightarrow (\alpha_l \rightarrow \beta_2)  \big) \rightarrow \alpha_l \rightarrow\beta_2}{\policy}\\
\tuplemt{\convupcr,\convupcr} \in &\ \indvalp{\forall \beta. (\alpha_l \rightarrow \beta) \rightarrow (\alpha_{l'} \rightarrow \beta)}{\policy}\\
\tuplemt{\wrapcr,\wrapcr} \in &\ \indvalp{\forall \beta. \beta \rightarrow \alpha_{l} \rightarrow \beta}{\policy}\\
\tuplemt{\convcr,\convcr} \in &\ {\indvalp{(\alpha_l^f \rightarrow \alpha^f) \rightarrow \alpha_l \rightarrow \intType}{\policy}}
\end{align*}
\begin{proof}
The proofs for the first three parts are similar to the corresponding ones in Lemma~\ref{lem:substitution:term:comp_convup}.
We now prove the last part.
Let $\rho$ be an environment s.t. $\rho \respectp \policy$.
From Def.~\ref{def:local_policy:ind}, we need to prove that 
$\tuplemt{\convcr,\convcr} \in \valrelation{{(\alpha_l^f \rightarrow \alpha^f)} \rightarrow \alpha_l \rightarrow \intType}{\rho}$.
That is for any $(v,v') \in \valrelation{{\alpha_l^f \rightarrow \alpha^f}}{\rho}$,
$$\tuplemt{\convcr\ v,\convcr\ v'} \in \valrelation{\alpha_l \rightarrow \intType}{\rho}.$$

%Let $g = \lambda \_:\unitType.x$.
From the definition of \convcr, we need to prove that:
%$$\tuplemt{\lambda \_:\unitType.\redtext{(v\ \unitVal)},\lambda \_:\unitType.\redtext{(v'\ \unitVal)}} \in \valrelation{\alpha_l \rightarrow \intType}{\rho}.$$
$$\tuplemt{v,v'} \in \valrelation{\alpha_l \rightarrow \intType}{\rho}.$$

We have two cases:
\begin{itemize}
\item $l \smalleq \obs$. We have that $\rho(\alpha_l) = {\rho(\alpha_l^f)} = \tuplemt{\unitType,\unitType,\full{\unitType}}$ and $\rho(\alpha_l^f) = \tuplemt{\intType,\intType,id_\intType}$.
{We need to prove that $\tuplemt{v\ \unitVal,v'\ \unitVal} \in \termrelation{\alpha^f}{\rho}$}.

Since $(v,v') \in \valrelation{{\alpha_f^l\rightarrow \alpha^f}}{\rho}$, from $\rho(\alpha_l^f)$ and $\rho(\alpha^f)$ it follows that {$v\ \unitVal = v'\ \unitVal$}.
Therefore, we have that 
$\tuplemt{v,v'} \in \valrelation{\alpha_l \rightarrow \intType}{\rho}$.

\item $l \not\smalleq \obs$.
We have that $\rho(\alpha_l) = \tuplemt{\unitType,\unitType,\emptyset}$ and hence, $\tuplemt{v,v'} \in \valrelation{\alpha_l \rightarrow \intType}{\rho}$ vacuously. 

%From the encoding, we have that $\alpha_l^f$ and $\alpha^f$ are for an input at $l$ that can be declassified via $f$ to $l'$.
%We have two sub-cases.
%   \begin{itemize}
%   \item $l' \smalleq \obs$
%   \item $l' \not\smalleq \obs$.
%   \end{itemize}
\end{itemize} 
\end{proof}

\textbf{Lemma~\ref{lem:trni:substitution:term:declassification}.}
%For any $\alpha_l^f$ corresponding to $x$  s.t. $\lvl(x) = l$ and $\decMap(x) = \tuplemt{f,l'}$:
%$$\tuplemt{\dec{f},\dec{f}} \in \indvalp{\alpha_l^f \rightarrow \alpha_{l'} \rightarrow \tau_f}{\policy}.$$
For any \obs, for any $\alpha_l^f$ and $\alpha^f$ corresponding to $x$  s.t. $\lvl{x} = l$ and $\decMap(x) = \tuplemt{f,l'}$:
$${\tuplemt{\dec{f},\dec{f}} \in \indvalp{(\alpha_l^f \rightarrow \alpha_f) \rightarrow \alpha_{l'} \rightarrow \tau_f}{\policy}.}$$

\begin{proof}
Let $\rho$ be an environment s.t. $\rho \respectp \policy$.
From Def.~\ref{def:local_policy:ind}, we need to prove that:
$$\tuplemt{\dec{f},\dec{f}} \in \valrelation{{(\alpha_l^f \rightarrow \alpha^f)} \rightarrow \alpha_{l'} \rightarrow \tau_f}{\rho}.$$

That is for any $\tuplemt{v_1,v_2} \in \valrelation{{\alpha_l^f\rightarrow\alpha^f}}{\rho}$, {$\tuplemt{\dec{f}\ v_1,\dec{f}\ v_2} \in \termrelation{\alpha_{l'} \rightarrow \tau_f}{\rho}$}.
From the definition of \dec{f}, we need to prove that
$$\tuplemt{\lambda \_:\unitType.f{(v_1\ \unitVal)},\lambda \_:\unitType.f{(v_2\ \unitVal)}} \in \valrelation{\alpha_{l'} \rightarrow \tau_f}{\rho}$$

When $l' \not\smalleq \obs$, we have that $\rho(\alpha_{l'}) = \tuplemt{\unitType,\unitType,\emptyset}$ and hence, the statement holds vacuously.
We now consider the case $l' \smalleq \obs$.
We have that $\rho(\alpha_{l'}) = \tuplemt{\unitType,\unitType,\full{\unitType}}$ and {$\rho(\alpha_l^f) = \tuplemt{\unitType,\unitType,\full{\unitType}}$}.
Therefore we need to prove that
$$\tuplemt{f({v_1\ \unitVal}),f({v_2\ \unitVal})} \in \termrelation{\tau_f}{\rho}$$

Let $v_i' = v_i\ \unitVal$.
We need to prove that
$$\tuplemt{f({v_1'}),f({v_2'})} \in \termrelation{\tau_f}{\rho}$$

Note that since $\tuplemt{v_1,v_2} \in \valrelation{{\alpha_l^f\rightarrow\alpha^f}}{\rho}$ and {$\rho(\alpha_l^f) = \tuplemt{\unitType,\unitType,\full{\unitType}}$}, we have that 
$$\tuplemt{v_1\ \unitVal,v_2\ \unitVal} \in \termrelation{\alpha^f}{\rho}.$$

In other words, $\tuplemt{v_1',v_2'} \in \valrelation{\alpha^f}{\rho}$.

We have two sub-cases:
\begin{itemize}
\item $l \not\smalleq \obs$. We have that $\rho(\alpha_l^f) = \tuplemt{\intType,\intType,R_f}$.
Since $\tuple{v_1',v_2'} \in \valrelation{\alpha^f}{\rho} = R_f$, we have that $\tuplemt{f\ v_1',f\ v_2'} \in \termrelation{\tau_f}{\emptyset}$.
Therefore, $\tuplemt{f(v_1\ \unitVal),f(v_2\ \unitVal)} \in \termrelation{\tau_f}{\emptyset}$ and hence, $\tuplemt{f(v_1\ \unitVal),f(v_2\ \unitVal)} \in \termrelation{\tau_f}{\rho}$ (note that $\tau_f$ is a closed type).

\item $l \smalleq \obs$. We have that $\rho(\alpha_l^f) = \tuplemt{\intType,\intType,id_\intType}$.
Therefore ${\tuplemt{v_1',v_2'}} \in id_\intType$ and hence, $v_1' = v_2' = v$.
Since \typeEval{f\ v:\tau_f}, from Theorem~\ref{thm:multi_level:parametricity}, we have that $\tuplemt{f\ v, f\ v} \in \termrelation{\tau_f}{\emptyset}$.
Therefore, $\tuplemt{f(v_1\ \unitVal), f(v_2\ \unitVal)} \in \termrelation{\tau_f}{\emptyset}$ and hence, $\tuplemt{f(v_1\ \unitVal), f(v_2\ \unitVal)} \in \termrelation{\tau_f}{\rho}$ (note that $\tau_f$ is a closed type).
\end{itemize}
\end{proof}

\subsection{Proofs for extensions of the multi-level encoding}
\label{sec:multi_level:proof:extension}
\subsubsection{Proofs for the extension with multiple declassifiers}
We prove that the implementations of \comp{\_}, \conv{\_}, \convup{\_}{\_}, $\hinp_{f_i}$, and $\minpl_{g_i}$ are indistinguishable to themselves for any observer.

\begin{lemma}
\label{lem:trni:ext:one:substitution:term:interface}
For any \obs, it follows that:
\begin{align*}
\tuplemt{\compcr,\compcr} \in &\ \indvalp{\forall \beta_1,\beta_2. (\alpha_l \rightarrow \beta_1) \rightarrow \big( \beta_1 \rightarrow (\alpha_l \rightarrow \beta_2)  \big) \rightarrow \alpha_l \rightarrow\beta_2}{\policymul}\\
\tuplemt{\convupcr,\convupcr} \in &\ \indvalp{\forall \beta. (\alpha_l \rightarrow \beta) \rightarrow (\alpha_{l'} \rightarrow \beta)}{\policymul}\\
\tuplemt{\wrapcr,\wrapcr} \in &\ \indvalp{\forall \beta. \beta \rightarrow \alpha_{l} \rightarrow \beta}{\policymul}\\
\tuplemt{\convcr,\convcr} \in &\ \indvalp{{(\alpha_H^{f_1,f_2}\rightarrow \alpha^{f_1,f_2})} \rightarrow \alpha_{H} \rightarrow \intType}{\policymul}\\
\tuplemt{\convcr,\convcr} \in &\ \indvalp{{(\alpha_{M_1}^{g_1,g_2} \rightarrow \alpha^{g_1,g_2})} \rightarrow \alpha_{M_1} \rightarrow \intType}{\policymul}
\end{align*}
\end{lemma}
\begin{proof}
The proofs for \compcr, \convupcr, and \wrapcr\ are as the corresponding ones in the proof of Lemma~\ref{lem:substitution:term:comp_convup}.
We now prove the first part about \convcr. The proof for the second part about \convcr\ is similar.

Let $\rho$ be an environment s.t. $\rho \respectfullp \policymul$.
From the definition of indistinguishability, we need to prove that
$\tuplemt{\convcr,\convcr} \in \valrelation{{(\alpha_H^{f_1,f_2} \rightarrow \alpha^{f_1,f_2})} \rightarrow \alpha_{H} \rightarrow \intType}{\rho}$.
That is for any $(v,v') \in {\valrelation{\alpha_H^{f_1,f_2}\rightarrow \alpha^{f_1,f_2}}{\rho}}$,
$$\tuplemt{\convcr\ v,\convcr\ v'} \in \valrelation{\alpha_{H} \rightarrow \intType}{\rho}.$$

%Let $g = \lambda \_:\unitType.x$.
From the definition of \convcr, we need to prove that:
$$\tuplemt{v,v'} \in \valrelation{\alpha_{H} \rightarrow \intType}{\rho}.$$

We have two cases:
\begin{itemize}
\item $H \smalleq \obs$. We have that $\rho(\alpha_H) = \rho(\alpha_H^{f_1,f_2}) = \tuplemt{\unitType,\unitType,\full{\unitType}}$ and $\rho(\alpha^{f_1,f_2}) = \tuplemt{\intType,\intType,id_\intType}$.
We need to prove that $\tuple{v\ \unitVal,v'\ \unitVal} \in \termrelation{\intType}{\rho}$.

Since $(v,v') \in \valrelation{\alpha_H^{f_1,f_2}\rightarrow \alpha^{f_1,f_2}}{\rho}$, from $\rho(\alpha_H)$  and $\rho(\alpha_H^{f_1,f_2})$, it follows that {$v\ \unitVal = v'\ \unitVal$}.

Therefore, we have that ${\tuplemt{v,v'}} \in \valrelation{\alpha_H \rightarrow \intType}{\rho}$.

\item $H \not\smalleq \obs$.
We have that $\rho(\alpha_H) = \tuplemt{\unitType,\unitType,\emptyset}$ and hence, $\tuplemt{v,v'} \in \valrelation{\alpha_l \rightarrow \intType}{\rho}$ vacuously. 
\end{itemize} 
\end{proof}

\begin{lemma}
\label{lem:trni:ext:one:substitution:term:declassification:one}
For any \obs,
\begin{align*}
\tuplemt{\dec{f_1},\dec{f_1}} & \in \indvalp{{(\alpha_H^{f_1,f_2} \rightarrow \alpha^{f_1,f_2})}  \rightarrow \alpha_{M_1} \rightarrow \tau_{f_1}}{\policymul}\\
\tuplemt{\dec{f_2},\dec{f_2}} & \in \indvalp{{(\alpha_H^{f_1,f_2} \rightarrow \alpha^{f_1,f_2})} \alpha_{M_2} \rightarrow \tau_{f_2}}{\policymul}
\end{align*}
\end{lemma}
\begin{proof}
We only prove the first part. The proof of the second part is similar.

Let $\rho$ be an environment s.t. $\rho \respectfullp \policymul$.
From the definition of indistinguishability, we need to prove that:
$$\tuplemt{\dec{f_1},\dec{f_1}} \in \valrelation{{(\alpha_H^{f_1,f_2} \rightarrow \alpha^{f_1,f_2})} \rightarrow \alpha_{M_1} \rightarrow \tau_{f_1}}{\rho}.$$

That is for any $\tuplemt{v_1,v_2} \in \valrelation{{\alpha_H^{f_1,f_2} \rightarrow \alpha^{f_1,f_2}}}{\rho}$, $\tuplemt{\dec{f_1}\ v_1,\dec{f_1}\ v_2} \in \termrelation{\alpha_{M_1} \rightarrow \tau_{f_1}}{\rho}$.
From the definition of \dec{f_1}, we need to prove that
$$\tuplemt{\lambda \_:\unitType.f_1(v_1\ \unitVal),\lambda \_:\unitType.f_1(v_2\ \unitVal)} \in \valrelation{\alpha_{M_1} \rightarrow \tau_{f_1}}{\rho}$$

Let $v_i' = v_i\ \unitVal$. We need to prove that:
$${\tuplemt{\lambda \_:\unitType.f_1(v_1'),\lambda \_:\unitType.f_1(v_2'} \in \valrelation{\alpha_{M_1} \rightarrow \tau_{f_1}}{\rho}}$$

When $M_1 \not\smalleq \obs$, we have that $\rho(\alpha_{M_1}) = \tuplemt{\unitType,\unitType,\emptyset}$ and hence, the statement holds vacuously.
We now consider the case $M_1 \smalleq \obs$.
We have that $\rho(\alpha_{M_1}) = \tuplemt{\unitType,\unitType,\full{\unitType}}$.
Therefore we need to prove that
$$\tuplemt{f_1(v_1'),f_1(v_2')} \in \termrelation{\tau_{f_1}}{\rho}$$

Note that since $M_1 \smalleq \obs$, we have that {$\rho(\alpha^{f_1,f_2}) = \tuplemt{\unitType,\unitType, \full{\unitType}}$}.
Since $\tuplemt{v_1,v_2} \in \valrelation{{\alpha_H^{f_1,f_2} \rightarrow \alpha^{f_1,f_2}}}{\rho}$, we have that $\tuplemt{v_1\ \unitVal, v_2\ \unitVal} \in \termrelation{{\alpha^{f_1,f_2}}}{\rho}$.
In other words, $\tuplemt{v_1',v_2'} \in \valrelation{{ \alpha^{f_1,f_2}}}{\rho}$.

We have two sub-cases:
\begin{itemize}
\item $\obs = M_1$. We have that $\rho(\alpha_H^{f_1,f_2}) = \tuplemt{\intType,\intType,R_{f_1}}$.
Since $\tuple{v_1',v_2'} \in \valrelation{{ \alpha^{f_1,f_2}}}{\rho} = R_{f_1}$, we have that $\tuplemt{f_1\ v_1',f_1\ v_2'} \in \termrelation{\tau_{f_1}}{\emptyset}$
 and hence, $\tuplemt{f_1\ v_1',f_1\ v_2'} \in \termrelation{\tau_{f_1}}{\rho}$ (note that $\tau_f$ is a closed type).
Therefore, $\tuplemt{f_1(v_1\ \unitVal),f_1(v_2\ \unitVal)} \in \termrelation{\tau_{f_1}}{\rho}$

\item $\obs = H$. We have that $\rho(\alpha_l^{f_1,f_2}) = \tuplemt{\intType,\intType,id_\intType}$.
Therefore $\tuple{v_1',v_2'} \in \valrelation{{ \alpha^{f_1,f_2}}}{\rho} = id_\intType$ and hence, $v_1' = v_2' = v$.
Since \typeEval{f_1\ v:\tau_f}, from Theorem~\ref{thm:multi_level:parametricity}, we have that $\tuplemt{f_1\ v, f_1\ v} \in \termrelation{\tau_{f_1}}{\emptyset}$ and hence, $\tuplemt{f_1\ v, f_1\ v} \in \termrelation{\tau_{f_1}}{\rho}$ (note that $\tau_f$ is a closed type).
Thus, {$\tuplemt{f_1(v_1\ \unitVal), f_1(v_2\ \unitVal)} \in \termrelation{\tau_{f_1}}{\rho}$}.
\end{itemize}
\end{proof}

\begin{lemma}
\label{lem:trni:ext:one:substitution:term:declassification:two}
For any \obs,
\begin{align*}
\tuplemt{\dec{g_1},\dec{g_1}} & \in \indvalp{{(\alpha_{M_1}^{g_1,g_2} \rightarrow \alpha^{g_1,g_2})} \rightarrow \alpha_{L} \rightarrow \tau_{g_1}}{\policymul}\\
\tuplemt{\dec{g_2},\dec{g_2}} & \in \indvalp{{(\alpha_{M_1}^{g_1,g_2} \rightarrow \alpha^{g_1,g_2})} \rightarrow \alpha_{L} \rightarrow \tau_{g_2}}{\policymul}
\end{align*}
\end{lemma}
\begin{proof}
We only prove the first part. The proof of the second part is similar.

Let $\rho$ be an environment s.t. $\rho \respectfullp \policymul$.
From the definition of indistinguishability, we need to prove that:
$$\tuplemt{\dec{g_1},\dec{g_1}} \in \valrelation{{\alpha_{M_1}^{g_1,g_2} \rightarrow \alpha^{g_1,g_2}} \rightarrow \alpha_{L} \rightarrow \tau_{g_1}}{\rho}.$$

That is for any $\tuplemt{v_1,v_2} \in \valrelation{{\alpha_{M_1}^{g_1,g_2} \rightarrow \alpha^{g_1,g_2}}}{\rho}$, $\tuplemt{\dec{g_1}\ v_1,\dec{g_1}\ v_2} \in \termrelation{\alpha_{L} \rightarrow \tau_{g_1}}{\rho}$.
From the definition of \dec{g_1} and the logical relation, we need to prove that
$$\tuplemt{\lambda \_:\unitType.g_1(v_1\ \unitVal),\lambda \_:\unitType.g_1(v_2\ \unitVal)} \in \valrelation{\alpha_{L} \rightarrow \tau_{g_1}}{\rho}$$

Let $v_i'$ s.t. $v_i\ \unitVal \reduce v_i'$.
We need to prove that:
$$\tuplemt{\lambda \_:\unitType.g_1(v_1'),\lambda \_:\unitType.g_1(v_2')} \in \valrelation{\alpha_{L} \rightarrow \tau_{g_1}}{\rho}$$

For all $\obs$, we have that $L \smalleq \obs$ and hence, $\rho(\alpha_{L}) = \tuplemt{\unitType,\unitType,\full{\unitType}}$.
Therefore we need to prove that
$$\tuplemt{g_1(v_1'),g_1(v_2')} \in \termrelation{\tau_{g_1}}{\rho}$$

Note that $\rho(\alpha_{M_1}^{g_1,g_2}) = \tuplemt{\unitType,\unitType,\full{\unitType}}$.
Since $\tuplemt{v_1,v_2} \in \valrelation{{\alpha_{M_1}^{g_1,g_2} \rightarrow \alpha^{g_1,g_2}}}{\rho}$, we have that 
$\tuplemt{v_1',v_2'} \in \valrelation{{\alpha^{g_1,g_2}}}{\rho}$.

We have two sub-cases:
\begin{itemize}
\item $M_1 \not\smalleq \obs$ (i.e. $\obs = L$ or $\obs = M_2$). We have that $\rho(\alpha^{g_1,g_2}) = \tuplemt{\intType,\intType,R_{g_1,g_2}}$.
Since $\tuple{v_1',v_2'} \in \valrelation{\alpha^{g_1,g_2}}{\rho} = R_{g_1,g_2}$, we have that $\tuplemt{g_1\ v_1',g_1\ v_2'} \in \termrelation{\tau_{g_1}}{\emptyset}$ and hence, $\tuplemt{g_1\ v_1', g_1\ v_2'} \in \termrelation{\tau_{g_1}}{\rho}$ (note that $\tau_{g_1}$ is a closed type).

\item $M_1 \smalleq \obs$. We have that $\rho(\alpha^{g_1,g_2}) = \tuplemt{\intType,\intType,id_\intType}$.
Therefore $\tuple{v_1',v_2'} \in id_\intType$ and hence, $v_1' = v_2' = v$.
Since \typeEval{g_1\ v:\tau_f}, from Theorem~\ref{thm:multi_level:parametricity}, we have that $\tuplemt{g_1\ v, g_1\ v} \in \termrelation{\tau_{g_1}}{\emptyset}$ and hence, $\tuplemt{g_1\ v, g_1\ v} \in \termrelation{\tau_{g_1}}{\rho}$ (note that $\tau_{g_1}$ is a closed type).
\end{itemize}
\end{proof}

\begin{theorem}
If \typeEval[\tcontextp,\contextp]{e:\tau}, then for any {$\obs \in \latticedm$} and $\rho \respectfullp \policymul$, 
$$\tuplemt{\rho_L(e),\rho_R(e)} \in \indtermp{\tau}{\policymul}.$$
\end{theorem}
\begin{proof}
Since $\rho \respectfullp \policy$, from the definition of $\rho \respectfullp \policymul$, Lemma~\ref{lem:trni:ext:one:substitution:term:interface}, Lemma~\ref{lem:trni:ext:one:substitution:term:declassification:one}, and Lemma~\ref{lem:trni:ext:one:substitution:term:declassification:two}, it follows that $\rho \respectp \tcontextp[\policymul], \contextp[\policymul]$.
%Therefore, we have that $\rho\respectp \tcontextp$.
Since \typeEval[\Delta_\policymul,\Gamma_\policymul]{e:\tau}, from Theorem~\ref{thm:multi_level:parametricity}, we have that $\tuplemt{\rho_L(e),\rho_R(e)} \in \termrelation{\tau}{\rho}$.
From the definition of indistinguishability, it follows that $\tuplemt{\rho_L(e),\rho_R(e)} \in \indtermp{\tau}{\policymul}$.
\end{proof}

\subsubsection{Proofs for the extension for global policies}
{The proof of the main result for this extension (in Section~\ref{sec:multi_level:extension:global}) is similar to the one in Section~\ref{sec:multi_level:extension:local-policy:mult_dec}.
Here, we only prove that \dec{f} relates to itself}.

\begin{lemma}
\label{lem:trni:glb:substitution:term:declassification}
For any \obs,
\begin{align*}
\tuplemt{\dec{f},\dec{f}} & \in \indvalp{{(\alpha_{M_1,M_2}^f \rightarrow \alpha^{f})} \rightarrow \alpha_{L} \rightarrow \intType}{\policyglb}
\end{align*}
\end{lemma}
\begin{proof}
Let $\rho$ be an environment s.t. $\rho \respectfullp \policyglb$.
From the definition of indistinguishability, we need to prove that:
$$\tuplemt{\dec{f_1},\dec{f_1}} \in \valrelation{{\alpha_{M_1,M_2}^f \rightarrow \alpha^f)} \rightarrow \alpha_{L} \rightarrow \intType}{\rho}.$$

That is for any $\tuplemt{v_1,v_2} \in \valrelation{{\alpha_{M_1,M_2}^f \rightarrow \alpha^f}}{\rho}$, $\tuplemt{\dec{f}\ v_1,\dec{f}\ v_2} \in \termrelation{\alpha_{L} \rightarrow \intType}{\rho}$.
From the definition of \dec{f}, we need to prove that
$$\tuplemt{\lambda \_:\unitType.f(v_1\ \unitType),\lambda \_:\unitType.f(v_2\ \unitType)} \in \valrelation{\alpha_{L} \rightarrow \intType}{\rho}$$

Let $v_i'$ be s.t. $v_i\ \unitVal \reduce v_i'$.
We need to prove that:
$$\tuplemt{\lambda \_:\unitType.f(v_1'),\lambda \_:\unitType.f(v_2')} \in \valrelation{\alpha_{L} \rightarrow \intType}{\rho}$$

%Hereafter, we write $v_1$ as \tuple{v,u} and $v_2$ as \tuple{v',u'}.
Since $L$ is the smallest element in the lattice, for all \obs, $\rho(\alpha_L) = \full{\unitType}$.
Thus, we need to prove that $f\ v_1' = f\ v_2'$.

Note that $\rho(\alpha_{M_1,M_2}^f) = \tuplemt{\unitType,\unitType,\full{\unitType}}$.
Since  $\tuplemt{v_1,v_2} \in \valrelation{{\alpha_{M_1,M_2}^f \rightarrow \alpha^f}}{\rho}$, we have that $\tuplemt{v_1',v_2'} \in \valrelation{{\alpha^f}}{\rho}$.

\begin{itemize}
\item $\obs = H$.
We have that $\rho(\alpha^f) = \tuplemt{\intType\times \intType,\intType\times \intType,id_{\intType\times\intType}}$.
Therefore,  we have that $v_1' = v_2'$ and hence, $f\ v_1' = f\ v_2'$. 

\item $\obs \in \{M_1,M_2,L\}$. 
We have that $\rho(\alpha_f) = \tuplemt{\intType\times \intType,\intType\times \intType, {\relDec} }$, where 
$$\relDec = \{\tuplemt{v_1,v_2} \sep \tuplemt{f\ v_1, f\ v_2} \in \termrelation{\intType}{\emptyset} \}$$

Since $\tuplemt{v_1',v_2'} \in \valrelation{{\alpha^f}}{\rho} = \relDec$, from the definition of \termrelation{\intType}{\emptyset} we have that $f\ v_1' = f\ v_2'$.
\end{itemize}
\end{proof}

\iffalse
\section{Battleship example}
Here we present the indistinguishability for the battleship example.
Here we use $f$ and $g$ for declassifiers, where $f$ is for the input from $M_2$ and $g$ is for both inputs.
In the main text, we do not construct $\dec{f}$ and \dec{g}
Then we have the following
\begin{align*}
\rho(\alpha_{f}) &= \begin{cases}
   id_\intType & \text{if $M_2 \leq \obs$}\\
   R_f & \text{otherwise}
\end{cases}\\
\\
\rho(\alpha_{M_1}) &= \text{as general cases}\\
\\
\rho(\alpha_g) &= \begin{cases}
   id_{\intType\times\intType} & \text{if $\obs = H$}\\
   R_g & \text{otherwise}
\end{cases}
\end{align*}
where $R_g = \{\tuplemt{\tuple{v,u},\tuple{v',u'}} \sep g \tuple{v,u} = g\tuple{v',u'}\}$.

For full environment, we require that they respect indistinguishability at each type and being consistent. 
\fi

\end{document}